\documentclass[manuscript]{acmart}

\renewcommand\footnotetextcopyrightpermission[1]{} 
\renewcommand\footnotetextcopyrightpermission[1]{}

\usepackage{xpatch}
\makeatletter
\xpatchcmd{\ps@firstpagestyle}{Manuscript submitted to ACM}{}{\typeout{First patch succeeded}}{\typeout{first patch failed}}
\xpatchcmd{\ps@standardpagestyle}{Manuscript submitted to ACM}{}{\typeout{Second patch succeeded}}{\typeout{Second patch failed}}    \@ACM@manuscriptfalse
\makeatother

\setcopyright{none}

\usepackage{booktabs}   
\usepackage{subcaption} 

\newcommand{\SUBSTCLO}[1]{}
\newcommand{\LONGVERSION}[1]{#1}
\newcommand{\LONGVERSIONCHECKED}[1]{#1}
\newcommand{\SHORTVERSION}[1]{}

\ifdefined\LONGVERSION
  \relax
\else
 \newcommand{\SUBSTCLO}[1]{}
 \newcommand{\LONGVERSION}[1]{}
 \newcommand{\LONGVERSIONCHECKED}[1]{}
 \newcommand{\SHORTVERSION}[1]{#1}
\fi

\usepackage{proof}
\usepackage{amssymb}
\usepackage{mathtools}
\usepackage{longtable}
\usepackage{enumitem}
\usepackage{xspace}
\usepackage{natbib}
\usepackage{cdsty}
\usepackage{hyperref}

\usepackage[final]{listings}
\usepackage{lstextract}
\usepackage{srcltx}

\usepackage{bbding}
\usepackage{pifont}
\usepackage{wasysym}
\usepackage{amssymb}

\usepackage{todonotes}

\newcommand{\highlight}[1]{{\color{green}{#1}}}

\newcommand{\beluga}{\textsc{Beluga}\xspace}
\newcommand{\cocon}{\textsc{Cocon}\xspace}

\newcommand{\hatctx}[1]{\hat{#1}}
\newcommand\floor[1]{\lfloor#1\rfloor}
\newcommand\ceil[1]{\lceil#1\rceil}
\newcommand{\cbox}[1]{\ceil {#1}}
\newcommand{\unbox}[2]{\floor {#1}_{#2}}
\newcommand{\unboxc}[1]{ #1}

\newcommand{\csubclo}[2]{\{#1 / #2\}_{clo}}
\newcommand{\lfss}[2]{[#1 / #2]}
\newcommand{\lfs}[2]{[#1 / \hatctx #2]}

\definecolor{DimGrey}{rgb}{0.8,0.8,0.8}

\lstdefinelanguage{ContextualML}
{
  morekeywords={and, block, case, of, mlam, fn, impossible, let, in, schema,
    some, rec, type, ctype, prop, stratified, inductive, coinductive, LF, if, then,
    else, total, with},
  keepspaces=true,
  sensitive,
  morecomment=[l]{\%},
  morecomment=[n]{\%\{}{\}\%},
  morestring=[b]"
}[keywords,comments,strings]

\newcommand{\pcase}[1]{\emph{Case}. #1 \\[0.75em]}

\newcommand{\prf}[1]{#1\\[0.15em]}

\newcommand{\LF}{\mathsf{LF}}

\newcommand{\evaln}[3]{#1 \Downarrow_#3 #2}

\newcommand{\der}{\vdash}
\newcommand{\sem}{\Vdash}
\newcommand{\semlf}{\Vdash_{\mathsf{LF}}}

\newcommand{\D}{{\mathcal{D}}}
\newcommand{\Ca}{{\mathcal{C}}}

\newcommand{\M}{{\mathcal{M}}}

\newcommand{\IH}{{\mathcal{I}}}
\newcommand{\J}{{\mathcal{J}}}

\newcommand{\JLF}{\J_{\mathsf{LF}}}
\newcommand{\Jcomp}{\J_{\mathsf{comp}}}

\newcommand{\Ax}{\mathfrak{A}}
\newcommand{\Ru}{\mathfrak{R}}

\newcommand{\rappto}{~}
\newcommand{\cappto}{\ll}
\newcommand{\typeof}{\mathsf{typeof}}

\newcommand{\ann}[1]{ \breve{#1}}

\newcommand{\pos}{\mathrm{lookup}~}
\newcommand{\trunc}{\mathrm{trunc}}

\newcommand{\clam}{\mathtt{lam}~\xspace}
\newcommand{\capp}{\mathtt{app}~\xspace}

\newcommand{\bnfas}{\;\mathrel{::=}\;}
\newcommand{\bnfalt}{\, \mid \,}
\newcommand{\wvec}[1]{\overrightarrow{#1}}

\newcommand{\Pityp}[3]{\Pi #1{:}#2.#3}
\newcommand{\const}[1]{\textsf{#1}}
\newcommand{\wwk}[1]{\textbf{wk}}
\newcommand{\wk}[1]{\textsf{wk}_{#1}}
\newcommand{\sclo}[2]{\wk{#1} \circ #2}
\newcommand{\id}{\textsf{id}}
\renewcommand{\arrow}{\Rightarrow}
\newcommand{\univ}[1]{\textsf{U}_{#1}}

\newcommand{\lftype}[0]{\mathsf{type}}
\newcommand{\lfkind}[0]{\mathsf{kind}}
\newcommand{\ctx}[0]{\mathsf{ctx}}

\newcommand{\Nat}{\textsf{Nat}}

\newcommand{\FV}{\mathsf{FV}}

\newcommand{\tmctx}{\mathsf{tm\_ctx}}
\newcommand{\tm}{\mathsf{tm}}
\newcommand{\tp}{\mathsf{tp}}

\newcommand{\R}{\mathcal{B}}
\newcommand{\trec}[3]{\textsf{rec}^{#3}~\ensuremath{{#1}}}
\newcommand{\titer}[3]{\textsf{rec}^{#3}~\ensuremath{{#1}}}

\newcommand{\mto}{\Rightarrow}

\newcommand{\tmfn}[2]{\textsf{fn } #1 \Rightarrow #2}

\newcommand{\tmrec}[4]{\textsf{rec}^{#1} (#2 \mid #3 \mid #4)~}
\newcommand{\tmrecctx}[3]{\textsf{rec}^{#1} (#2 \mid #3)~}

\newcommand{\whnf}{\searrow}
\newcommand{\lfwhnf}{\searrow_{\mathsf{LF}}}

\newcommand{\tightoverset}[2]{%
  \mathop{#2}\limits^{\vbox to -.5ex{\kern-0.95ex\hbox{$#1$}\vss}}}

\newcommand{\neut}{\textsf{wne}~}

\newcommand{\norm}{\textsf{whnf}~}

\newcommand{\defiff}{: \Leftrightarrow}

\begin{document}

%
\title[\cocon: Computation in Contextual Type Theory]{\cocon:
  Computation in Contextual Type Theory}         


\author{Brigitte Pientka}
\orcid{nnnn-nnnn-nnnn-nnnn}             
\affiliation{
  \position{Associate Professor}
  \department{School of Computer Science}             
  \institution{McGill University}                     
  \country{Canada}
}
\email{bpientka@cs.mcgill.ca}         

\author{Andreas Abel}
\orcid{0000-0003-0420-4492}             
\affiliation{
  \position{Senior lecturer}
  \department{Department of Computer Science and Engineering}              
  \institution{Gothenburg University}  
   \country{Sweden}
}
\email{andreas.abel@gu.se}          

\author{Francisco Ferreira}
\orcid{0000-0003-0420-4492}             
\affiliation{
  \position{Research Associate}
  \department{School of Computer Science}             
  \institution{Imperial College London}           
  \country{United Kingdom}
}
\email{f.ferreira-ruiz@imperial.ac.uk}          

\author{David Thibodeau}
\orcid{0000-0003-0420-4492}             
\affiliation{
  \position{PhD Student}
  \department{School of Computer Science}             
  \institution{McGill University}           
  \country{Canada}
}
\email{david.thibodeau@mail.mcgill.ca}          

\author{Rebecca Zucchini}
\orcid{0000-0003-0420-4492}             
\affiliation{
  \position{Master Student}
  \department{}             
  \institution{ENS Paris Saclay}           
  \streetaddress{sss}
  \city{xxx}
  \state{xx}
  \postcode{xxxxx}
  \country{France}
}
\email{XXX}          


\begin{abstract}
We describe a Martin-L{\"o}f-style dependent type theory, called \cocon, that allows us to
mix the intensional function space that is used to represent
higher-order abstract syntax (HOAS) trees with the extensional
function space that describes (recursive) computations. We mediate
between HOAS representations and computations using contextual modal
types.  Our type theory also supports an infinite hierarchy of
universes and hence supports type-level computation---thereby providing
metaprogramming and (small-scale) reflection.
Our main contribution is the development of a Kripke-style model for
\cocon that allows us to prove normalization. From the normalization
proof, we derive subject reduction and consistency. Our work lays the
foundation to incorporate the methodology of logical frameworks into
systems such as Agda and bridges the longstanding gap between
these two worlds.
\end{abstract}






\maketitle

\section{Introduction}

Higher-order abstract syntax (HOAS)  is an elegant and deceptively simple idea of encoding syntax and more generally formal systems given via axioms and inference rules. The basic idea
is to map uniformly binding structures in our object language to the function space in a meta-language thereby inheriting $\alpha$-renaming and capture-avoiding substitution. In the logical framework LF \cite{Harper93jacm}, for example, we encode a simple object language consisting of functions, function application, and let-expressions using a type \lstinline!tm! together with the constants as follows:

\begin{lstlisting}
lam : (tm -> tm) -> tm.
app : tm -> tm -> tm.
letv: tm -> (tm -> tm) -> tm.
\end{lstlisting}

The object language term $(\mathsf{lam}~x.\mathsf{lam}~y.\mathsf{let}~w = x~y~\mathsf{in}~w~y)$ is then encoded as
\lstinline[basicstyle=\ttfamily\small]!(lam \x.lam \y.letv (app x y) \w.app w y)!
using the LF abstractions to model binding.
Object level substitution is modelled through LF application; for instance, the fact that
$((\mathsf{lam}~x.M)~N)$ reduces to $[N/x]M$ in our object language is expressed as
\lstinline[basicstyle=\ttfamily\small]!(app (lam M) N)!
reducing to
\lstinline[basicstyle=\ttfamily\small]!(M N)!.
%
%
This approach can offer substantial benefits: programmers do not need to build up the
basic mathematical infrastructure, they can work at a higher-level of abstraction, encodings are more compact, and hence it is easier to mechanize formal systems together with their meta-theory. 

However, this approach relies on the fact that we use an
\emph{intensional} function space that lacks recursion, case analysis,
inductive types, and universes to adequately represent syntax.  In the
logical framework LF \citep{Harper93jacm} for example we use the
dependently-typed lambda calculus as a meta-language to represent
formal systems. Two LF objects are equal if they have the same
$\beta\eta$-normal form.
Under this view, intensional functions represent syntactic binding structures. However, we cannot write recursive programs about such syntactic structures \emph{within} LF, as we lack the power of recursion. We only have a way to represent data.
In contrast, to describe computation we rely on the \emph{extensional}
function space. Under this view, two functions are (extensionally)
equal if they \emph{behave} in the same way, i.e. when
they produce equal results when applied to equal inputs. Under
this view, functions are opaque.

 \subsection{Intensional and Extensional Functions -- A World of a Difference }
To understand the fundamental difference between defining HOAS trees in LF vs.
defining HOAS-style trees using inductive types, let us consider an inductive type
\lstinline!D! with one constructor
\lstinline!lam: (D -> D) -> D!. What is the problem with such a
definition in type theory? -- In functional ML-like languages, this is, of course,
possible, and types like \lstinline|D| can be explained using domain theory
\cite{scott:dataTypesAsLattices}.
However, the function argument to the constructor \lstinline!lam! is
opaque and we would not be able to pattern match deeper on the
argument to inspect the shape and structure of the syntax
tree that is described by it. We can only observe it by applying it to
some argument. The resulting encoding also would not be adequate,
i.e.~there are terms of type \lstinline!D! that are in normal form but
do not uniquely correspond to a term in the object language we try to
model. As a consequence, we may need to rule out such ``exotic''
representations \cite{Despeyroux:TLCA95}. But there is a more fundamental problem. In proof
assistants based on type theory such as Coq or Agda, we cannot afford
to work within an inconsistent system and we demand that all programs we write are
terminating. The definition of a constructor \lstinline!lam! as given
previously would be forbidden as it violates what is known as the
positivity restriction. Were we to allow it, we can easily write
non-terminating programs by pattern matching -- even without making a recursive call.

\noindent
\begin{minipage}{4cm}
\begin{lstlisting}
apply : D -> (D -> D)
apply  (lam f) = f
\end{lstlisting}
\end{minipage}
\begin{minipage}{5cm}
\begin{lstlisting}
omega : D
omega = lam (\x -> apply x x)
\end{lstlisting}
\end{minipage}
\begin{minipage}{6cm}
\begin{lstlisting}
Omega : D
Omega = apply omega omega
\end{lstlisting}
\end{minipage}

Here we simply write two functions: the function \lstinline!apply! unpacks an
object of type \lstinline!D! using pattern matching and the function
\lstinline!omega! creates an object of type \lstinline!D!. Using
\lstinline!apply! and \lstinline!omega! we can now write a
non-terminating program that will continue to reproduce itself.

This example begs two questions: How can we reason inductively about
LF definitions if they are seemingly not inductive? Do we have to
simply give up on HOAS definitions to model syntactic structures
within type theory to remain consistent?

\subsection{Towards Bridging the Gap between Intensional and Extensional Functions}

Over the past two decades, we have made substantial progress in
bringing the intensional and extensional views closer
together. \citet{Despeyroux97} made the key observation that we can
mediate between the weak LF and the strong computation-level
function space using a box-modality. The authors describe a simply
typed lambda-calculus with iteration and case constructs which
preserves the adequacy of higher-order abstract syntax encodings. The
well-known paradoxes are avoided through the use
of a modal box operator which obeys the laws of S4. In addition to
being simply typed, all computation had to be on closed HOAS trees. \citet{Despeyroux99}
sketch an extension of this line of work to dependent type theory --
however it lacks a normalization proof. %

\beluga \cite{Pientka:IJCAR10,Pientka:CADE15} took another important step towards writing inductive proofs
about HOAS trees by generalizing the box-modality to a contextual
modal type \cite{Nanevski:ICML05,Pientka:POPL08}. This allows us to
characterize HOAS trees that depend on a context of assumptions. More
importantly, \beluga allows programmers to analyze these contextual
HOAS trees using case distinction and recursion. Exploiting the
Curry-Howard isomorphism, inductive proofs about HOAS trees can be
described using recursive functions.
However, the gap between full dependent type theories with recursion
and universes such as Martin-L{\"of} type theory, and weak dependent
type theories such as LF remains. In particular, \beluga cleanly
separates representing syntax from reasoning about syntax. The
resulting language is an indexed type system in the tradition of
\citet{Zenger:TCS97} and \citet{Xi99popl} where the index language is
completely different from the computation language which is used to
write recursive programs. In \beluga, contextual LF is taken as the
index domain. This has the key advantage that meta-theoretic proofs
are modular and hinge on the fact that equality in the index domain is
decidable.  However, this approach also gives up a lot of
expressivity; in particular we can only express properties of HOAS
trees, but we lack the power to express properties of the functions we
write about them. This prevents us from fully exploiting the power of
metaprogramming and reflection.


\subsection{The Best of Both Worlds}
In this paper, we present the Martin-L{\"o}f style dependent type
theory \cocon where we mediate between intensional syntactic
structures and extensional computations using contextual types
\cite{Nanevski:ICML05,Pientka:POPL08}. Following \beluga, we pair a LF
object together with its surrounding LF context and embed it as a
contextual object into computations using the box-modality. For
example, $\cbox{x,y \vdash \capp x~y}$ 
describes a contextual LF object that has the contextual LF type
$\cbox{x{:}\tm,y{:}\tm \vdash \tm}$. In contrast to \beluga, we also
allow computations to be embedded within LF objects. For example, if a
program $t$ promises to compute a value
$\cbox{x{:}\tm,y{:}\tm \vdash \tm}$, then we can embed $t$ directly into
an LF object writing
$\clam \lambda x. \clam \lambda y. \capp \unbox{t}{}~x$. In general,
we can use a computation that produces a value of type
$\cbox{\Psi \vdash A}$ when constructing a LF object in a LF context
$\Phi$ by unboxing it together with a LF substitution that moves the
value from the LF context $\Psi$ to the current LF context
$\Phi$. This is written as $\unbox{t}\sigma$. In the example, we
omitted 
the substitution as the computation already
promised to produce a value in the appropriate LF
context. 

Being able to embed functions into LF objects is key to express properties about function we write about them. For example, we might implement a function that evaluates a $\tm$-object and another function \lstinline!trans! that eliminates let-expressions from our $\tm$ language. Then we would like to know whether both the original term and the translated term evaluate to the same value.

Allowing computation within LF objects, might seem like a small
change, but it has far reaching consequences. To establish consistency
of the type theory, we cannot consider normalization of LF separately
from normalization of computations anymore, as it is done in
\citet{Pientka:TLCA15} and \citet{JacobRao:stratified2018}. As
Martin-L{\"o}f type theory \citeyearpar{Martin-Loef73a}, \cocon is a
predicative type theory and supports an infinite hierarchy of
universes. This allows us to write type-level computation, i.e. we can
compute types whose shape depends on a given value. Such recursively
defined types are sometimes called large eliminations
\cite{Werner:1992}. Due to the presence of type-level computations
dependencies cannot be erased from the model. As a consequence, the
simpler proof technique of \citet{Harper03tocl} which considers
approximate shape of types and has been used to prove completeness of
equivalence algorithm for LF's type theory cannot be used in our
setting. Instead, we follow recent work by \citet{Abel:LMCS12} and
\citet{Abel:POPL18} on proving normalization of our fully dependent
type theory using a Kripke logical relation. Our semantic model
highlights the intensional character of the LF function space and the
extensional character of computations.
Our main contribution is the design of the Kripke-style model for the
dependent type theory \cocon that allows us to establish
normalization. From the normalization proof, we derive type
uniqueness, subject reduction, and consistency. 

We believe \cocon lays the foundation to incorporate the methodology of logical frameworks into systems such as Agda \cite{Norell:phd07} or Coq \cite{bertot/casteran:2004}.
 This finally allows us to combine the world of type theory and logical frameworks inheriting the best of both worlds.

\subsection{Outline of the Technical Development}
Before delving into the technical details, we sketch here the main structure of the technical development.
\cocon consists of two mutually defined layers: LF to define HOAS and
computation to write recursive programs. The syntax and typing rules
of \cocon together with definitional equality are described in
Sec.~\ref{sec:cocon}.
We distinguish between two different kinds of variables,
LF variables and computation variables.
In particular, we define two different substitution
operations. We then proceed to prove some elementary properties about
LF (Sec. \ref{sec:proplf}) and computation (Sec. \ref{sec:propcomp}),
in particular well-formedness of LF contexts, LF Weakening and LF
Substitution properties. For LF we also establish functionality of LF
typing from which injectivity of LF function types follows.

Similarly, we establish some elementary properties about computation-level contexts and compu\-ta\-tion-level substitutions. We then proceed to define weak head reductions for LF and computations (Sec. \ref{sec:whred}) and show that they are closed under weakening (renaming).

Using weak head reduction, we define semantic equality using a Kripke model (Sec. \ref{sec:logrel}). Our model is Kripke-style in the sense that it is closed under weakening. It contains all well-typed terms in weak head normal form (whnf) and is built on top of definitional equality. We do not define semantic typing, but say a term is semantically well typed, if it is semantically equal to itself.
%
Since we embed computations inside LF terms, our typing rules for LF and computations are mutually defined, and one might wonder how we can break this cycle to arrive at a well-founded definition of semantic equality. We consider two LF terms $M$ and $N$ that weak head reduce to a $\unbox t \sigma$ and $\unbox {t'}{\sigma'}$ resp. semantically equal, if the computations $t$ and $t'$ are definitionally equal and the corresponding LF substitutions are semantically equal. 
This allows us to first define semantic equality for LF objects and subsequently semantic equality for computations breaking the cycle.

%
%
As we allow type-level computation, semantic equality for computations cannot be inductively defined on the structure of computation-level types. Instead, we use semantic kinding for types as a measure to define the semantic equality for computations.

Our semantic equality definitions are stable under renaming (weakening). We also prove symmetry, transitivity and type conversion for semantic equality that are the cornerstone of the development. This allows us to show that our semantic definition for terms is backwards closed and that neutral terms are semantically equal (see Sec. \ref{sec:semprop}).
%
Using the Kripke-model, we then show normalization and subject
reduction (see Sec. \ref{sec:validity}).
Logical consistency follows.
The full development including the proofs can be found in the accompanying long version.

\paragraph{Summary of Contributions}

\begin{itemize}
\item We describe \cocon, a Martin-L{\"o}f style type theory with an
an infinite hierarchy of universes and two intertwined layers: on
the LF layer, we can define HOAS trees referring to values produced by
computations and on the computation layer we can write recursive
functions on HOAS trees and exploit the power of large
eliminations. We mediate between these layers using contextual modal
types. This allows us to bridge the gap between the intensional LF
function space and the extensional function space used for writing
recursive computations.

\item We give a Kripke-style model to describe semantic equality for well-typed LF objects and well-typed computations highlighting the difference between intensional and extensional functions. Using this model we prove normalization.

\end{itemize}


\section{\cocon: Computation in Contextual Type Theory}\label{sec:cocon}
\cocon combines the logical framework LF \cite{Harper93jacm} with a full dependent type theory that supports recursion over HOAS objects and universes. For clarity, we split \cocon's grammar into different syntactic
categories (see Fig.~\ref{fig:grammar}). The LF layer describes LF objects, LF types, LF contexts; the computation layer consists of terms and types that describe recursive computation and universes. We mediate the interaction between LF objects and computations via a (contextual) box modality following \citet{Pientka:POPL08}: we embed contextual LF objects into  computations, by pairing an LF object with its LF contexts and we embed computations within LF objects by unboxing the result
of a computation. This allows us to not only write functions about LF
objects, but also establish proofs about such functions and opens the way for metaprogramming and writing programs using reflection.

\begin{figure}[htb]
  \centering
\[
\begin{array}{p{3cm}@{~}l@{~}r@{~}l}
LF Kinds           & K           & \bnfas & \lftype \bnfalt \Pityp x A K  \\
LF Types           & A, B        & \bnfas & \const{a}~M_1 \ldots M_n \bnfalt \Pityp x A B\\
LF Terms           & M, N       & \bnfas & \lambda x.M \bnfalt x \bnfalt \const{c} \bnfalt M~N \bnfalt \unbox t \sigma \\
LF Constants       & \const{a}, \const{c} & \bnfas & \tm \bnfalt \clam\bnfalt\capp\ \\
LF Contexts        & \Psi, \Phi & \bnfas & \psi \bnfalt \cdot \bnfalt \Psi, x{:}A\\
LF Context (Erased) & \hatctx{\Psi},\hatctx{\Phi} & \bnfas &  \psi \bnfalt \cdot \bnfalt \hatctx{\Psi},~x\\
LF Substitutions   & \sigma   & \bnfas & \cdot \bnfalt \wk{\hatctx\Psi} \bnfalt \sigma, M \SUBSTCLO{\bnfalt \sclo {\hatctx\Psi} {\unbox t \sigma}} \\
LF Signature       & \Sigma   & \bnfas & \tm{:}\lftype,~ \clam{:} \Pityp y {(\Pityp x \tm \tm)} \tm,~\mathsf{app}{:} \Pityp x \tm {\Pityp y \tm \tm}
\\[0.25em]
\hline
\\[-0.75em]
Contextual Types & T & \bnfas &
                     \Psi \vdash A \bnfalt      
                     \Psi \vdash_\# A           
   \SUBSTCLO{\bnfalt \Psi \vdash \Phi \mid \Psi \vdash_\# \Phi}
\\
Contextual Objects & C & \bnfas &
                      \hatctx{\Psi} \vdash M    
    \SUBSTCLO{\bnfalt \hatctx{\Psi} \vdash \sigma }
\\[0.25em]
\hline
\\[-0.75em]

Sorts     & u          & \bnfas & \univ k \\
Domain of Discourse & \ann\tau & \bnfas & \tau \bnfalt \tmctx \\
Types and & \tau, \IH, & \bnfas &  u \bnfalt \cbox T
                             \bnfalt (y :\ann{\tau}_1) \arrow \tau_2 
\\
Terms & t, s &  \bnfalt &  y \bnfalt  \cbox C \bnfalt t_1~t_2 \bnfalt \tmfn y t 
             \bnfalt \titer{\R}{}{\IH} \rappto\Psi~t
\\
Branches  & {\R} & \bnfas &
  ({\psi, p \mto t_v} \mid{\psi, m, n, f_m, f_n \mto t_{\mathsf{app}}}\mid {\psi, m, f_m \mto t_{\clam}})
\LONGVERSION{\\[0.25em]
         & & \bnfalt &
   ({\psi \mto t_x} \mid {\psi, y, f_y \mto t_y}) }
\SUBSTCLO{\\[0.25em]
          & & \bnfalt &
  ({ \cdot \mto t_e} \mid {\psi, \sigma, m, f_{\sigma} \mto t_c})}
\\
Contexts & \Gamma & \bnfas & \cdot \bnfalt \Gamma, y:\ann\tau 
\end{array}
\]
  \caption{Syntax of \cocon}
  \label{fig:grammar}
\end{figure}

\subsection{Syntax}
\paragraph{Logical Framework LF with Embedded Computations}
As in LF we allow dependent kinds and types; LF terms can be defined by
LF variables, constants, LF applications, and LF lambda-abstractions. In addition, we allow a computation $t$ to be embedded into LF terms using a closure $\unbox t {\sigma}$. Here the computation $t$ eventually computes to a \emph{contextual} object that depends on assumptions $\Psi$ following \citet{Pientka:POPL08}. Once computation of $t$ produces a contextual object $\hatctx\Psi\vdash M$ we can embed the result by applying the substitution $\sigma$ to $M$ moving $M$ from the LF context $\Psi$ to the current context $\Phi$.

We distinguish between computations that characterize a general LF term $M$ of type $A$ in a context $\Psi$ using the contextual type $\Psi \vdash A$ and computations that are guaranteed to return a variable in a context $\Psi$ of type $A$ using the contextual type $\Psi \vdash_\# A$. This is essential when describing recursors over contextual LF terms, but also generally important when mechanizing formal systems and it is smoothly integrated in our type theory. \SHORTVERSION{For simplicity, we focus on $\Psi \vdash A$ in the subsequent development and fix the LF signature to include the type $\tm$ and the LF constants $\clam$ and $\capp$.}

\LONGVERSION{For simplicity, we fix here the LF signature to include the type $\tm$ and the LF constants $\clam$ and $\capp$. This allows us to for example define recursors on $\tm$-objects directly.}


\paragraph{LF contexts}
LF contexts are either empty or are built by extending a context with a declaration $x{:}A$.
We may also use a (context) variable $\psi$ that stands for a context prefix and must be declared on the computation-level. In particular, we can write functions where we abstract over (context) variables. Consequently, we can pass LF contexts as arguments to functions. We classify LF contexts via schemas -- for this paper, we pre-define the schema $\tmctx$ which classifies a LF context which consists of $\tm$ declarations. Such context schemas are similar to adding base types to computation-level types.
We often do not need to carry the full LF context with the type annotations, but it suffices to simply consider the erased LF context. Erased LF contexts are simply a list of variables possibly with a context variable at the head.

\LONGVERSION{At the moment, we do not support computation on context at the moment -- this simplifies the design. Recall that the head of a context denotes a possibly empty sequence of declarations. This prefix should be abstract and opaque to any LF term or LF type that is considered within this context. In other words, a LF term $M$ (or LF type $A$) should be meaningful without requiring any specific knowledge about the prefix of declarations. Second, it would be difficult to enforce well-scoping and $\alpha$-renaming. To illustrate, consider the following LF term $\capp~x~y$ in the LF context $x{:}\tm, y{:}\tm$. If we were to allow type checking to exploit equivalence relations on LF contexts that take into account computations on LF contexts, we can argue that since $x{:}\tm, y{:}\tm$ is equivalent to $\unboxc{\texttt{copy}~\cbox{x{:}\tm, y{:}\tm}}$, $\capp~x~y$ should also be meaningful in the latter LF context. However, now the LF variables $x$ and $y$ are free in $\capp~x~y$.}




\paragraph{LF Substitutions}
LF substitutions allow us to move between LF contexts. The empty LF substitution provides a mapping from an empty LF context to a LF context $\Psi$ and hence has weakening built-in. The weakening substitution
written as $\wk{\hatctx\Psi}$ where $\Psi$ describes the prefix of the range that corresponds to the domain; in other words it describes the weakening of the domain $\Psi$ to the range $\Psi, \wvec{x{:}A}$. In general, we may weaken any given LF context with the declarations $\wvec{x{:}A}$. The generality of weakening substitutions is necessary to, for example, express that we can weaken a LF context $\unboxc{\psi}$.  We may write simply $\id$, if $|\wvec{x{:}A}| = 0$.

Weakening substitutions do not subsume the empty substitutions -- only the empty substitution that maps the empty context to a concrete context $x_n{:}A_n, \ldots, x_1{:}A_1$ can be expressed as $\wk{\cdot}$ where we annotate the weakening substitution with the empty LF context. For example, we would not be able to represent a substitution with the empty context as the domain and a context variable $\psi$ as the range using a weakening substitution. Our built-in weakening substitutions are also sometimes called \emph{renamings} as they only allow contexts to be extended to the right but they do not support arbitrary weakening of a LF context where we would insert a declaration in the middle (i.e. given a context $x{:}A_1, y{:}A_3$ we can weaken it to $x{:}A_1, w{:}A_2, y{:}A_3$).

From a de Bruijn perspective, the weakening substitution $\wk{\cdot}$ which maps the empty context to ${x_n{:}A_n, \ldots, x_1{:}A_1}$ can be viewed as a shift $n$. Further, as in the de Bruijn world, $\wk {x_n{:}A_n, \ldots, x_1{:}A_1}$ can be expanded and is equivalent to $\cdot, x_n, \ldots, x_1$.
While our theory lends itself to an implementation with de Bruijn indices, we formulate our type theory using a named representation of variables. This not only simplifies our subsequent definitions of substitutions, but also leaves open how variables are realized in an implementation.

 LF substitutions can also be built by extending a LF substitution $\sigma$ with a LF term
$M$. Following \citet{Nanevski:ICML05}, we do not store the domain of a substitution, but simply write them as a list of terms. We resurrect the domain of the substitution before applying it by erasing types from a context. To apply a substitution $\sigma$ to a term $M$ in an erased context $\hatctx \Psi$, we write $[\sigma / \hatctx\Psi]M$.

\SUBSTCLO{
Finally, we allow computations to be embedded into substitutions; such computations are useful, if want to state abstractly properties about substitutions, for example if we want to relate a context $\psi:{\tmctx}$ to another context $\phi:{\tmctx}$. Such properties naturally arise meta-theoretic proofs.
 We embed computations into substitutions using LF closures, written as $\unbox t \sigma$ following \cite{Cave:LFMTP13}. Intuitively, these closures describe when applying a LF substitution gets stuck. When we compose LF substitutions, we may get stuck when we compose $\sigma$ with a computation $t$ that promised to produce eventually a LF substitution. Hence we form a LF closure $\unbox{t}\sigma$. Further, we incorporate weakening. If we compose a weakening substitution with a computation $t$, we may be again stuck until we can proceed and know the result of $t$. These two observations lead to defining suspended substitutions as $\sclo {\hatctx\Psi} {\unbox t \sigma}$. The development follows closely ideas by \citet{Cave:LFMTP13}.

In addition to general contextual types $\Psi \vdash \Phi$ that describe a (contextual) LF substitution, we can also support the restricted form of contextual substitution types $\Psi \vdash_\# \Phi$ that are only inhabited by weakening substitutions. First-class support for weakening substitutions has proven useful in encoding formal systems using contextual LF (see for example \cite{Cave:LFMTP15}). Such renamings are crucial in many proofs such as for example normalization proofs using logical predicates \cite{Cave:LFMTP13,Cave:LFMTP15,Thibodeau:Howe16,POPLMarkReloaded}. For the subsequent development, we will concentrate on LF substitutions of contextual types $\Psi \vdash \Phi$.
}

\paragraph{Contextual Objects and Types}
We mediate between the LF and computation level using contextual types. We consider here general contextual LF terms that have type $\Psi \vdash A$, and contextual variable objects that have type $\Psi \vdash_\# A$. \SUBSTCLO{Similarly, we can also add  general contextual LF substitutions that have type $\Psi \vdash \Phi$ and contextual weakening substitutions that have type $\Psi \vdash_\# \Phi$. }


\paragraph{Computations and their Types}
Computations are formed by computation-level functions, written as $\tmfn y t$, that are extensional, i.e. we can only observe their behaviour, applications, written as $t_1~t_2$,
boxed contextual objects, written as $\cbox C$, and  recursor, written as $\titer{\R}{}{\IH} \rappto\Psi~t$. We annotate the recursor with the typing invariant $\IH$ and recurse over the values computed by the term $t$. The LF context $\Psi$ describes the local LF-world in which the value computed by $t$ makes sense.
Finally, $\R$ describes the different branches we can take depending on the value computed by $t$. These branches can be generated generically following \citet{Pientka:TLCA15}. We focus on in the rest of the paper on the iterator over contextual objects of type $\cbox{\Psi \vdash \tm}$. In this case, we consider three different branches: 1) In the LF variable case,
$(\psi, p \mto t_v)$, the variable $p$ stands for a LF variable $\tm$ in the LF context $\psi$ and has type $\cbox{\psi \vdash_\# \tm}$ and $t_v$ is the body of the branch.
2) In the $\capp$-case, written as $(\psi, m, n, f_m, f_n \mto t_{\mathsf{app}})$, we pattern match on a LF term $\capp~m~n$ in the LF context $\psi$. The recursive calls are denoted by $f_m$ and $f_n$ respectively and $t_{\mathsf{app}}$ describes the body of the branch.
3) In the $\clam$-case, written as $({\psi, m, f_m \mto t_{\clam}})$, we pattern match on a LF term $\clam ~\lambda x. m$ where $m$ denotes a LF term of LF type $\tm$ in the LF context $\psi, x{:}\tm$. The recursive call is described by $f_m$ and the body of the branch is denoted by $t_{\clam}$.

\LONGVERSION{For illustration, we also include other branches to construct also recursors over contextual objects of type $\cbox{\Psi \vdash_\# \tm}$, i.e. variables of type $\tm$ in the LF context $\Psi$. In this case, we only consider two branches where $\Psi = \Psi', x{:}\tm$: in the first branch $({\psi \mto t_x})$ we pattern match against $x$ and $\psi$ will be instantiated with $\Psi'$; in the second branch, $({\psi, y, f_y \mto t_y})$, the LF variable we are looking for is not $x$, but is somewhere in $\Psi'$. In this case, $y$ denotes intuitively a LF variable that is not $x$ and has type $\cbox{\psi \vdash_\# \tm}$ and we will instantiate $\psi$ with $\Psi'$; $f_y$ is the recursive call on the smaller LF context $\Psi'$ and $t_y$ is the body of the branch. }
\LONGVERSION{We also include branches for recursing over LF substitution which have either the contextual type $\cbox{\Psi \vdash \Phi}$ or $\cbox{\Psi \vdash_\# \Phi}$. Here we consider two cases: either the LF substitution is empty then we choose the first branch, or it is of the shape $\sigma, m$ and $f_\sigma$ denotes the recursive call on the smaller LF substitution $\sigma$. }

Computation-level types consist of boxed contextual types, written as $\cbox T$, and dependent types, written as $ (y:\ann{\tau}_1) \arrow \tau_2$. We overload the dependent function space and allow as domain of discourse both computation-level types and the schema $\tmctx$ of LF context. To form both functions we use $\tmfn y t$. We also overload function application $t~s$ to eliminate dependent types $(y : \tau_1) \arrow \tau_2$ and and $(y : \tmctx) \arrow \tau_2$, although in the latter case $s$ stands for a LF context.

\cocon is a pure type system (PTS) with infinite hierarchy of predicative universes, written as $\univ k$ where $k \in \Nat$.
The universes are not cumulative. We use sorts $u,k \in \mathcal S$, axioms $\Ax = \{(\univ i,~\univ {i+1} \mid i \in \Nat \}$, and rules $\Ru = \{ (\univ i,~ \univ j,~\univ {\mathsf{max}(i,j)}) \mid i,j \in \Nat \}$. Universes add additional power.


\begin{example}
To illustrate the syntax of \cocon, we write a program that counts the
number of constructors in a given $\tm$. The type of the function is
$\IH = (\psi : {\tmctx}) \arrow (m:\cbox{\unboxc{\psi} \vdash \tm}) \arrow \mathsf{nat}$.

\[
  \begin{array}{r@{~}lcl}
\mathsf{fn}~\psi \Rightarrow \mathsf{fn}~m \Rightarrow \mathsf{rec}^\IH (& \psi, p & \mto & 0 \\
\mid & \psi,  m, n, f_n, f_m & \mto & f_n + f_m + 1 \\
\mid & \psi, m, f_m & \mto & f_m ~+~ 1~) \rappto~\psi~m
  \end{array}
%
\]

The first branch describes the variable case where $p$ describes a variable from the LF context $\unboxc{\psi}$ which has type $\cbox{\unboxc \psi \vdash_\# \tm}$. The second branch describes the application case; here $f_n$ and $f_m$ respectively denote the recursive calls and have type $\mathsf{nat}$. The third branch describes the lambda case where $f_m$ is the recursive call made on the body of the lambda-term.

\end{example}

\begin{example}
Next we implement $\mathsf{copy}$ of type
$\IH = (\psi:{\tmctx}) \arrow (m:\cbox{\unboxc{\psi} \vdash \tm}) \arrow \cbox{\unboxc{\psi} \vdash \tm}$.
We abbreviate the identity substitution $\wk{\unboxc{\psi}}$ by simply writing $\id$.

\[
  \begin{array}{r@{~}lcl}
\mathsf{fn}~\psi \Rightarrow \mathsf{fn}~m \Rightarrow \mathsf{rec}^\IH (& \psi, p & \mto & \cbox{\unboxc{\psi} \vdash \unbox p {\id}~} \\
\mid & \psi,  m, n, f_n, f_m & \mto & \cbox{\unboxc{\psi} \vdash \capp \unbox{f_n}{\id}~\unbox{f_m}{\id}~} \\
\mid & \psi, m, f_m & \mto & \cbox{\unboxc{\psi} \vdash \clam~ \lambda x.\unbox{f_m}{\id}~}~)~\rappto~\psi~m
  \end{array}
\]

In this example the input and output type depends on $\psi$; in particular the type of the recursive call $f_m$ in the lambda case will be $\cbox{\unboxc \psi, x{:}\tm \vdash \tm}$.
\end{example}

\LONGVERSION{
\begin{example}
We return the position of a LF variable in a LF context by writing a function $\mathsf{pos}$ that has type $\IH = (\psi : \tmctx) \arrow (x : \cbox{\psi \vdash_\# \tm}) \arrow \mathsf{nat}$.

\[
  \begin{array}{r@{~}lcl}
\mathsf{fn}~\psi \Rightarrow \mathsf{fn}~x \Rightarrow \mathsf{rec}^\IH (& \psi & \mto & 0 \\
\mid & \psi, y, f_y & \mto & 1 + f_y)~\psi~x
  \end{array}
\]

\end{example}

}

\subsection{LF Substitution Operation}
Our type theory distinguishes between LF-variables and computation-level variables. We have substitution operation for both. 
\LONGVERSION{Let's consider first a few examples to get a better intuition. Let's look at a few examples to get a better intuition.
%

\paragraph{Examples 1} Consider the LF term $\capp \unbox{\cbox{x,y \vdash \capp~x~y}}{\wk{x, y}}~w$. This LF term is obviously well-typed in the (normal) LF context $x{:}\tm, y{:}\tm, w{:}\tm$ and applying the substitution $\wk{x,y}$ to $\capp~x~y$ is meaningful as $\wk{x,y}$ expands to $\cdot, x, y$. When we apply $\cdot, x, y$ to unbox $\cbox{x,y \vdash \capp~x~y}$, we resurrect the domain and apply $[\cdot, x, y ~/~ \cdot, x, y](\capp~x~y)$.

\paragraph{Examples 2} What about considering the $\alpha$-equivalent term $\unbox{\cbox{x',y' \vdash \capp~x'~y'}}{\wk{x,y}}$? --
Again we observe that $\wk{x,y}$ expands to $\cdot, x, y$; when we apply $\cdot, x, y$ to unbox $\cbox{x',y' \vdash \capp~x'~y'}$, we resurrect the domain and apply $[\cdot,x,y ~/~ \cdot, x', y'](\capp~x'~y')$ effectively renaming $x'$ and $y'$ to $x$ and $y$ respectively.

}

\begin{figure}[htb]
  \[
  \begin{array}{l@{~}c@{~}l@{~\mathrm{where~}}l}
\lfs \sigma {\Psi} ( \lambda x. M  ) & = &   \lambda x. M' & \lfs {\sigma,x}{\Psi, x}(M) = M'\quad
{\mathrm{provided\ that\ }x \notin \FV(\sigma) ~\mathrm{and}~ x \not\in \hatctx{\Psi}}
    \\[0.25em]
 \lfs \sigma \Psi (M~N) & = & M'~N' & {\lfs \sigma \Psi (M) = M'~\mathrm{and}~\lfs \sigma \Psi (N)= N'}
    \\[0.25em]
\lfs \sigma \Psi (\unbox{ t} {\sigma'})& = & \unbox t {\sigma''} &
\lfs{\sigma}{\Psi}(\sigma') = \sigma''
    \\[0.25em]
\lfs \sigma \Psi  (x)  & = & M & \pos x~\lfs \sigma \Psi = M
    \\[0.25em]
 \lfs \sigma \Psi c     & = & \multicolumn{2}{l}{c}  \\[1em]
\lfs \sigma \Psi (\cdot)  & = & \multicolumn{2}{l}{\cdot}  \\[0.25em]
\lfs \sigma \Psi (\wk{\hatctx\Phi})  & = & \sigma' & \trunc_\Phi ~(\sigma / \hatctx\Psi) = \sigma'\\[0.25em]
\SUBSTCLO{\lfs \sigma \Psi (\sclo {\hatctx\Phi} {\unbox{t}{\sigma'}}) & = & \sclo {\hatctx\Phi} {\unbox{t}{\sigma''}} & \lfs \sigma \Psi (\sigma') = \sigma''  \\[0.25em]}
\lfs \sigma \Psi (\sigma', M) & = & \sigma'', M' &
\lfs \sigma \Psi (\sigma') = \sigma''~\mathrm{and}~\lfs \sigma \Psi (M)=M'
 \end{array}
  \]

\caption{Simultaneous LF Substitution for LF Objects }
\label{fig:lfsubst}
\end{figure}

We define LF substitutions uniformly using simultaneous substitution operation written as $\lfs \sigma \Psi M$ \LONGVERSION{(and similarly $\lfs   \sigma \Psi A$ and $\lfs \sigma \Psi K$) }(see Fig.~\ref{fig:lfsubst}). As LF substitutions are simply a list of terms, we need to resurrect the domain to lookup the instantiation for a LF variable $x$ in $\sigma$. This is always possible. When pushing the substitution through an application $M~N$, we simply apply it to $M$ and $N$ respectively. When pushing the LF substitution through a $\lambda$-abstraction, we extend it.
When applying $\sigma$ to a LF variable $x$, we retrieve the corresponding instantiation from $\sigma$ using the auxiliary function $\pos$ which works mostly as expected.
When applying the LF substitution $\sigma$ to the LF closure $\unbox{t}{\sigma'}$ we leave $t$ untouched, since $t$ cannot contain any free LF variables and compose $\sigma$ and $\sigma'$.

 \LONGVERSION{
\[
\begin{array}{l@{~}c@{~}l@{~}l}
\pos x~\lfs {\sigma, M}{\Psi, x} & = & M & \\[0.25em]
\pos x~\lfs {\sigma, M}{\Psi, y} & = & \pos x~ \lfs \sigma \Psi & \\[0.25em]
\pos x~\lfs {\wk{\hatctx\Psi}}{\Psi} & = & x &\mathrm{where}~x\in\hatctx\Psi \\
\SUBSTCLO{
  \pos x~\lfs {\sclo{\hatctx\Psi} {\unbox t\sigma}} \Psi & = &
    \lfs{\sigma}{\Phi}(M) &
    \mathrm{where}~t = \cbox{\hatctx\Phi \vdash \sigma'}
 \mathrm{and}~
    \trunc_\Psi ~ (\sigma' / \hatctx\Psi) = \sigma''
\\
 & & & ~\mathrm{and}~
    \pos x~\lfs{\sigma''}{\Psi} = M
\\[0.25em]
\pos x~\lfs {\sclo{\hatctx\Psi} {\unbox t \sigma}} \Psi & = & \unbox{\cbox{\hatctx\Psi',x \vdash x}}{\sclo {\hatctx\Psi',x} {\unbox t \sigma}} &
    \mathrm{where}~t \neq {\cbox{\hatctx\Phi \vdash \sigma'}}
  ~\mathrm{and}~\hatctx\Psi = \hatctx\Psi',x, \vec y
\\
}
\pos x~\lfs \sigma \Psi & = & \mbox{fails otherwise}
 \end{array}
  \]
}

Composition of LF substitution is straightforward.
%
%
%
\SUBSTCLO{
However, we need to be careful in handling substitution closures $\sclo{\hatctx\Psi} \unbox{t}\sigma$, as we need to guarantee that we make progress, if we can. For example, if $t =  \cbox{\hatctx\Phi \vdash \sigma'}$, then we want continue to look up the variable $x$ in $\sigma'$. To ensure we work with the same domain $\hatctx\Psi$, we first truncate $\sigma'$ dropping instantiations that do not play a role. Finally, we apply $\sigma$ to the term that the variable $x$ is mapped to. If $t$ cannot be further unfolded, then we create a closure. If $r \neq  \unbox {\cbox{\hatctx\Phi \vdash \sigma'}} {\sigma}$, we create a LF closure on the LF term level. Intuitively, we proceed, once the LF closure is known. \SHORTVERSION{The full definition can be found in the accompanying long version and is omitted due to lack of space.}
}
%
%
%
When we apply $\sigma$ to $\wk{\hatctx\Psi}$, we truncate $\sigma$ and only keep those entries corresponding to the LF context $\Psi$. Recall that $\wk{\hatctx\Psi}$ provides a weakening substitution from a context $\Psi$ to another context $\Psi, \wvec{x{:}A}$ where $|\wvec{x{:}A}| = n$. The simultaneous substitution $\sigma$ provides mappings for all the variables in $\hatctx{\Psi}, \vec{x}$. The result of $[\sigma / \hatctx \Psi, \vec{x}]\wk{\hatctx\Psi}$ then should only provide mappings for all the variables in $\Psi$. We use the operation $\trunc$ to remove irrelevant instantiations. \SHORTVERSION{Its definition is straightforward and is given in the long version.}
\LONGVERSION{
The definition of truncation is straightforward.
\[
\begin{array}{l@{~}c@{~}l@{~}l}
\trunc_\Psi ~ (\sigma / \hatctx \Psi) & = & \sigma\\
\trunc_\Psi ~ (\sigma, M / \hatctx \Phi, x) & = & \trunc_\Psi ~(\sigma / \hatctx \Phi)\\
\trunc_\Psi ~ (\wk{(\hatctx\Psi, \vec{x})} / \hatctx \Psi, \vec{x}) & = & \wk{\hatctx\Psi} \\
\SUBSTCLO{\trunc_\Psi ~(\sclo {\hatctx\Psi, \vec{x}} {\unbox t \sigma}  / \hatctx \Psi, \vec{x}) & = & \sclo {\hatctx\Psi} {\unbox t \sigma}\\}
\trunc_\Psi ~(\dot / \cdot) & = & \mbox{fails}~\Psi\neq \cdot
\end{array}
\]
}

\subsection{Computation-level Substitution Operation}

The computation-level substitution operation $\{t/x\}t'$ traverses the computation $t'$ and replaces any free occurrence of the computation-level variable $x$ in $t'$ with $t$\LONGVERSION{\/(see Fig.~\ref{fig:csub})}. The interesting case is $\{t/x\}\cbox{C}$. Here we push the substitution into $C$ and we will further apply it to objects in the LF layer. When we encounter a closure such as $\unbox{t''}{\sigma}$, we continue to push it inside $\sigma$ and also into $t''$.
%
When substituting a LF context $\Psi$ for the variable $\psi$ in a context $\Phi$, we rename the declarations present in $\Phi$. This is a convention. It would equally work to rename the variable declarations in $\Psi$. For example,
in $\{(x{:}\tm, y{:}\tm) / \psi\}(\hatctx \psi, x \vdash \clam \lambda y. \capp x~y~)$, we rename the variable $x$ in $\hatctx\psi, x$ and replace $\psi$ with $(x{:}\tm, y{:}\tm)$ in $(\hatctx\psi, w \vdash \clam \lambda y. \capp w~y)$. This results in $x, y, w \vdash \clam \lambda y. \capp w~y$. When type checking this term we will eventually also $\alpha$-rename the $\lambda$-bound LF variable $y$.

\LONGVERSION{
\begin{figure}
  \centering
\[
  \begin{array}{l@{~=~}l@{~\qquad}l}
\multicolumn{3}{l}{\mbox{Computation-level Substitution for Terms}}\\[0.5em]
   \{t/y\}(\tmfn x t') & \tmfn x \{t/y\}t'    & \mathrm{provided~} x \not\in \FV(t) \\[0.1em]
   \{t/y\}(t_1~t_2)    & \{t/y\}t_1 ~ \{t/y\}t_2 & \\[0.1em]
   \{t/y\}(\titer{\R}{}{\IH} \rappto \Psi~t) & \titer {\{t/y\}b_1 \mid \{t/y\}b_k}{}{\IH} \rappto \{t/y\}\Psi~\{t/y\}t & \mathrm{where~} \R = b_1 \mid \ldots \mid b_n \\[0.1em]
   \{t/y\}\cbox C & \cbox {\{t/y\} (C)} & \\[0.1em]
   \{t/y\} (y) & t & \\[0.1em]
   \{t/y\} (x) & x & \mathrm{where~} x \neq y\\[0.1em]
   \{t/y\} (\cdot) & \cdot & \\[0.1em]
   \{t/y\} (\Psi, x:A) & \{t/y\}\Psi, x:\{t/y\}A & \mathrm{provided~} (\hatctx{\Psi},x) \not\in \FV(t)\\[0.5em]
\multicolumn{3}{l}{\mbox{Computation-level Substitution for Branches}}\\[0.5em]
   \{t/y\} (\vec x \mto t') & (\vec x \mto \{t/y\}t') &
   \mathrm{provided~} \vec x \not \in \FV(t)
\\[0.5em]
\multicolumn{3}{l}{\mbox{Computation-level Substitution for Contextual Objects}}\\[0.5em]
   \{t/y\}(\hatctx{\Psi} \vdash M) & \{t/y\}\hatctx\Psi \vdash \{t/y\}M & \mathrm{provided~} \hatctx\Psi\not\in \FV(t) \\[0.1em]
   \{t/y\}(\hatctx{\Psi} \vdash \sigma) & \{t/y\}\hatctx\Psi \vdash \{t/y\}\sigma & \mathrm{provided~} \hatctx\Psi\not\in \FV(t) \\[0.1em]
\multicolumn{3}{l}{\mbox{Computation-level Substitution for LF Objects}}\\[0.5em]
   \{t/y\}(\lambda x.M) & \lambda x. \{t/y\}M & \\[0.1em]
   \{t/y\}(M~N) & \{t/y\}M ~ \{t/y\}N & \\[0.1em]
   \{t/y\}(\unbox {t'}{\sigma})   & \unbox {\{t/y\}t'}{\{t/y\}\sigma} \\[0.1em]
   \{t/y\}(\const c) & \const c & \\[0.1em]
   \{t/y\}(x)        & x & \\[0.1em]
\multicolumn{3}{l}{\mbox{Computation-level Substitution for LF Substitutions}}\\[0.5em]
   \{t/y\}(\cdot) & \cdot & \\[0.1em]
   \{t/y\}(\sigma, M) & \{t/y\}\sigma,~\{t/y\}M & \\[0.1em]
   \{t/y\}(\wk{\hatctx\Psi})   & \wk{\{t/y\}({\hatctx\Psi})} & \\[0.1em]
\SUBSTCLO{   \{t/y\}(\sclo {\hatctx\Psi} r) & \sclo {\{t/y\}({\hatctx\Psi})} {\csubclo t y {~}r}}
  \end{array}
\]
  \caption{Computation-level Substitution}\label{fig:csub}
\end{figure}
}

\LONGVERSION{
\begin{figure}[htb]
  \centering
\[
\begin{array}{c}
\multicolumn{1}{p{13.5cm}}{
\fbox{$\Gamma ; \Phi \vdash A : \lftype$}~\mbox{and}~
\fbox{$\Gamma ; \Phi \vdash K : \lfkind$}~~LF type $A$ is well-kinded
  and LF kind $K$ is well-formed}
\\[1em]
\infer{\Gamma; \Psi \vdash \const a: K}{\Gamma \vdash \Psi : \ctx & \const a{:}K \in \Sigma}
\quad
\infer{\Gamma ; \Psi \vdash P~M : [M/x]K}
{\Gamma ; \Psi \vdash P : \Pityp x A K & \Gamma; \Psi \vdash M : A}
\quad
\infer{\Gamma ; \Psi \vdash \Pityp x A B : \lftype}
{\Gamma ; \Psi \vdash A : \lftype & \Gamma ; \Psi, x{:}A \vdash B : \lftype}
\\[1em]
\infer{\Gamma ; \Psi \vdash A : K}
{\Gamma ; \Psi \vdash A : K'  & \Gamma ; \Psi \vdash K' \equiv K : \lfkind}
\quad
\infer{\Gamma ; \Psi \vdash \lftype : \lfkind}{\Gamma \vdash \Psi : \ctx}
\quad
\infer{\Gamma ; \Psi \vdash \Pityp x A K: \lfkind}
{\Gamma ; \Psi \vdash A : \lftype & \Gamma ; \Psi, x{:}A \vdash K : \lfkind}
\end{array}
\]
\caption{Kinding Rules for LF Types}\label{fig:lfkinding}
\end{figure}
}

\begin{figure}[htb]
  \centering
\[
\begin{array}{c}
\LONGVERSION{ \multicolumn{1}{p{13cm}}{
 \fbox{$\Gamma ; \Psi \vdash_\# M : A$}~
 ~~LF term $M$ of LF type $A$ in the LF context $\Psi$ and context $\Gamma$ describes a variable}
 \\[0.3em]
\infer{\Gamma ; \Psi\vdash_\# M :  A}
  {\Gamma ; \Psi \vdash M \equiv x : A & \Psi(x) = A
}
\qquad
\infer{\Gamma ; \Psi \vdash_\# M : A}
      {\Gamma ; \Psi \vdash_\# M : B &
       \Gamma ; \Psi \vdash B \equiv A : \lftype }
\\[0.5em]
\infer{\Gamma ; \Psi \vdash_\# M : A}
{\Gamma ; \Psi \vdash M \equiv \unbox{t}{\sigma} : \lfs {\sigma}\Phi (A) &
 \Gamma \vdash t : [\Phi \vdash_\# A] &
 \Gamma; \Psi \vdash_\# \sigma : \Phi }
\\[0.5em]
}
\multicolumn{1}{p{13cm}}{
\fbox{$\Gamma ; \Psi \vdash M : A$}~
~~LF term $M$ has LF type $A$ in the LF context $\Psi$ and context $\Gamma$}
\\[0.5em]
 \infer{\Gamma ; \Psi \vdash \lambda x.M : \Pityp x A B}
         {\Gamma ; \Psi, x{:}A \vdash M : B}
\quad
\infer{\Gamma ; \Psi \vdash \unbox t {\sigma} : \lfs \sigma \Phi A}
         {\Gamma \vdash t : [\Phi \vdash A]  ~\mbox{or}~\Gamma \vdash t : [\Phi \vdash_\# A] &
          \Gamma; \Psi \vdash \sigma : \Phi}
\\[0.5em]
\infer{\Gamma ; \Psi \vdash x : A}{
\Gamma \vdash \Psi : \ctx&
\Psi(x) = A }
\qquad
\infer
      {\Gamma ; \Psi \vdash M : A}
      {\Gamma ; \Psi \vdash M : B &
       \Gamma ; \Psi \vdash B \equiv A : \lftype }
\\[0.5em]
\infer{\Gamma ; \Psi \vdash M~N : [N/x]B}
      {\Gamma ; \Psi \vdash M : \Pityp x A B &
       \Gamma ; \Psi \vdash N : A}
\quad
\infer{\Gamma ; \Psi \vdash \const{c} : A}
      {\Gamma\vdash \Psi : \ctx & \const{c}:A \in \Sigma}
\\[1em]
\multicolumn{1}{p{13.5cm}}{
\fbox{$\Gamma ; \Phi \vdash \sigma : \Psi$}~LF substitution $\sigma$ provides a mapping from the LF context $\Psi$ to $\Phi$}
\\[0.3em]
\infer{\Gamma ;\Psi, \wvec{x{:}A} \vdash \wk{\hatctx\Psi} : \Psi}
{\Gamma\vdash \Psi, \wvec{x{:}A} : \ctx }
\quad
 \infer{\Gamma ; \Phi \vdash \cdot : \cdot}{\Gamma \vdash \Phi :\ctx }
\quad
\infer{\Gamma ; \Phi \vdash \sigma, M : \Psi, x{:}A}
      {\Gamma ; \Phi \vdash \sigma : \Psi &
       \Gamma ; \Phi \vdash M : \lfs \sigma \Psi A}
\SUBSTCLO{
\\[0.3em]
\infer{\Gamma ; \Phi \vdash \sclo {\hatctx\Psi} {\unbox{t}{\sigma}} : \Psi}
{\Gamma \vdash t : \cbox{\Psi' \vdash \Psi, \wvec{x{:}A}}
 & \Gamma ; \Phi \vdash \sigma : \Psi'
&  \Gamma\vdash \Psi, \wvec{x{:}A} : \ctx }
\\[0.3em]}
\LONGVERSION{\\[1em]
 \multicolumn{1}{p{13cm}}{
 \fbox{$\Gamma ; \Psi \vdash_\# \sigma : \Phi$}~
 ~~LF substitution $\sigma$ from LF context $\Phi$ to the LF context $\Psi$ is a weakening subst.}
 \\[0.5em]
\infer{\Gamma ; \Psi, \wvec{x{:}A} \vdash_\# \sigma : \Psi}
  {\Gamma ; \Psi, \wvec{x{:}A} \vdash \sigma \equiv \wk{\hatctx\Psi} : \Psi
}
\\[0.3em]
\SUBSTCLO{\infer{\Gamma ; \Psi  \vdash_\# \sigma : \Phi}
{\Gamma ; \Psi \vdash \sigma \equiv \sclo {\hatctx\Phi} {\unbox{t}{\wk{\Phi,\vec{x}}}} : \Phi &
 \Gamma \vdash t : \cbox{\Psi \vdash_\# \Phi, \wvec{x{:}A}} &
 \Gamma ; \Psi \vdash \wk{\hatctx\Psi} : \Phi,\wvec{x{:}A}}
}}
\end{array}
\]
\caption{Typing Rules for LF Terms and LF Substitutions}\label{fig:lftyping}
\end{figure}

\subsection{LF Typing}
We concentrate here on the typing rules for LF terms, LF substitutions and LF contexts (see Fig.~\ref{fig:lftyping}). The rules for LF types and kinds are straightforward (see Fig.~\ref{fig:lfkinding}). All of the typing rules have access to a LF signature $\Sigma$ which we omit to keep the presentation compact.
In typing rules for LF abstractions $\lambda x.M$ we simply extend the LF context and check the body $M$. When we encounter a LF variable, we look up its type in the LF context. The conversion rule is important and subtle. We only allow conversion of types -- conversion of the LF context is not necessary, as we do not allow computations to appear directly in the LF context and we can keep part of the LF context abstract.
However, we deviate from \citet{Cave:POPL12} in the rule that allows us to embed computations into LF terms. Given a computation $t$ that has type $\cbox {\Psi \vdash A}$ or $\cbox{\Psi \vdash_\# A}$, we can embed it into the current LF context $\Phi$ by forming the closure $\unbox t {\sigma}$ where $\sigma$ provides a mapping for the variables in $\Psi$. This formulation generalizes previous work which only allowed variables declared in $\Gamma$ to be embedded in LF terms. This enforced a strict separation between computations and LF terms.
%
%
%
The typing rules for LF substitutions are as expected.
\SUBSTCLO{The rule for suspended composition of substitution $\sclo {\hatctx\Phi}{\unbox{t}{\sigma}}$ with domain $\Psi$  requires some care.  Here we ensure that $t$ computes a LF substitution that maps LF variables from $\Psi, \wvec{x{:}A}$ to the LF context $\Psi'$. Hence, we allow $t$ to provide a ``bigger'' substitution than required; these additional instantiations will be truncated when we know what LF substitution the term $t$ computes. The ``stuck'' LF substitution $\sigma$ must map LF variables from $\Psi'$ to the final target LF context $\Phi$. As for LF terms, we distinguish between general LF substitutions and LF renamings (or weakenings) that guarantee that we only map LF variables to LF variables.
}

\SHORTVERSION{The typing rules for LF contexts simply analyze the structure of a LF context. When we reach the head we either encounter a empty LF context or an context variable $y$ which must be declared in the computation-level context $\Gamma$. They can be found in the long version. }
\LONGVERSION{Last, we consider the typing rules for LF contexts (see Fig.~\ref{fig:lfctxtyping}). They simply analyze the structure of a LF context. When we reach the head we either encounter a empty LF context or an context variable $y$ which must be declared in the computation-level context $\Gamma$.

\begin{figure}[htb]
  \centering
\[
\begin{array}{c}
\multicolumn{1}{l}{
\fbox{$\Gamma \vdash \Psi : \ctx$}~\mbox{LF context $\Psi$ is a well-formed}} \\[0.3em]
\infer{\Gamma \vdash \cdot : \ctx}{\vdash \Gamma}
\quad
\infer{\Gamma \vdash \unboxc{y} : \ctx}{\Gamma(y) = \tmctx & \vdash \Gamma}
\quad
\infer{\Gamma \vdash \Psi, x{:}A : \ctx}
{\Gamma \vdash \Psi : \ctx & \Gamma ; \Psi \vdash A : \lftype}
\end{array}
\]
\caption{Typing Rules for LF Contexts}\label{fig:lfctxtyping}
\end{figure}
}





\subsection{Definitional LF Equality}
We now consider definitional LF equality. We omit the transitive closure rules as well as congruence rules, but concentrate here on the reduction and expansion rules. For LF terms, equality is $\beta\eta$. In addition, we can reduce $\unbox{\Psi \vdash M}{\sigma}$ by simply applying $\sigma$ to $M$.

For LF substitutions, we take into account that weakening substitutions are not unique.
 For example, the substitution $\wk\cdot$ may stand for a mapping from the empty context to another LF context; so does the empty  substitution $\cdot$.
Similarly, $\wk{x_1, \ldots x_n}$ is equivalent to $\wk{\cdot}, x_1, \ldots, x_n$.


\begin{figure}[h]
  \centering
  \[
  \begin{array}{c}
    \multicolumn{1}{p{13cm}}{\fbox{$\Gamma ; \Psi \vdash M \equiv N : A$}\quad LF Term $M$ is definitionally equal to LF Term $N$ at LF type
    $A$}\\[0.75em]
    \infer{\Gamma ; \Psi \vdash M \equiv \lambda x.M~x : \Pityp x A B}{\Gamma ; \Psi \vdash M : \Pityp x A B}
    \quad
    \infer{\Gamma ; \Psi \vdash (\lambda x.M_1)~M_2 \equiv [M_2/x]M_1 : [M_2/x]B}
          {\Gamma ; \Psi, x{:}A \vdash M_1 : B & \Gamma ; \Psi \vdash M_2 : A}
    \\[0.75em]
\infer
  {\Gamma ; \Psi \vdash \unbox{\cbox{\hatctx{\Phi} \vdash N}}{\sigma} \equiv \lfs \sigma\Phi N :\lfs \sigma\Phi A}
          {\Gamma ; \Phi \vdash N : A & \Gamma ; \Psi \vdash \sigma : \Phi}
\\[1em]
\multicolumn{1}{p{13cm}}{\fbox{$\Gamma ; \Psi \vdash \sigma \equiv \sigma' : \Phi$}~\quad LF Substitution $\sigma$ is definitionally equal to LF Substitution $\sigma'$ }\\[1em]
 \infer{\Gamma ; \Psi \vdash \wk{\cdot} \equiv \cdot : \cdot}{
 \Gamma \vdash \Psi : \ctx }
 \quad
 \infer{\Gamma ; \Phi, x{:}A, \wvec{y{:}B} \vdash \wk{\hatctx\Phi,x} \equiv \wk{\Phi}, x: \Phi, x{:}A}
 {
 \Gamma \vdash \Phi, x{:}A, \wvec{y{:}B} : \ctx
 }
\SUBSTCLO{\\[1em]
\infer
{\Gamma ; \Phi \vdash \sclo {\hatctx\Psi} {\unbox{\cbox{\hatctx {\Phi'} \vdash \sigma'}}{\sigma}}  \equiv   \lfs\sigma {\Phi'}~(\lfss {\sigma'}{\hatctx{\Psi},\vec x}  ~\wk\Psi) : \Psi}
{
\Gamma ; \Phi' \vdash \sigma' : \Psi, \wvec{x{:}A} & \Gamma ; \Phi \vdash \sigma : \Phi' }
\\[1em]
\infer
{\Gamma ; \Phi \vdash \sclo {\hatctx\Psi,x} r \equiv   \sclo{\hatctx\Psi} r, \unbox{\cbox{\hatctx\Psi,x \vdash x}} {\sclo {\hatctx\Psi,x} r} : \Psi,x{:}A}
{r = \unbox{t}{\sigma} \qquad t \neq \cbox{\hatctx {\Phi'} \vdash \sigma'} \qquad \Gamma ; \Phi \vdash \sigma : \Phi' \quad \Gamma \vdash t : \cbox{\Phi' \vdash \Psi, x{:}A,\wvec{x{:}A}}
}}
\LONGVERSION{\\[1em]
\infer{\Gamma ; \Psi \vdash \sigma, M \equiv \sigma', N : \Phi, x{:}A}
{\Gamma ; \Psi \vdash \sigma \equiv \sigma' : \Phi &
 \Gamma ; \Psi \vdash M \equiv N : \lfs \sigma \Phi A }
}
\end{array}
\]
\caption{Reduction and Expansion for LF Terms and LF Substitutions}
\end{figure}

\subsection{Contextual LF Typing and Definitional Equivalence}

We describe typing and equivalence of contextual objects in Fig.~\ref{fig:ctxtyping}. This is standard.
%
We lift definitional equality on LF terms to
contextual objects.
Note that we overload notation, writing $\hatctx{\Psi}$ for a LF context $\Psi$ where we have already erased type declarations, but we sometimes abuse notation and write $\hatctx{\Psi}$ for taking a LF context $\Psi$ and erasing its type information.

\begin{figure}[htb]
  \centering
\[
\begin{array}{c}
\LONGVERSION{
\multicolumn{1}{p{13.5cm}}{\fbox{$\Gamma \vdash T$}\quad Contextual  Type $T$ is well-kinded}
\\[1.5em]
\infer{\Gamma \vdash (\Psi\vdash A)}{
  \Gamma ; \Psi \vdash A : \lftype }
\LONGVERSION{\quad
\infer{\Gamma \vdash (\Psi\vdash_\# A)}{
  \Gamma ; \Psi \vdash A : \lftype }}
\SUBSTCLO{\\[1em]
\infer{\Gamma \vdash (\Psi\vdash \Phi)}{
\Gamma \vdash \Psi : \ctx &
\Gamma \vdash  \Phi : \ctx }
\qquad
\infer{\Gamma \vdash (\Psi\vdash_\# \Phi)}{
\Gamma \vdash \Psi : \ctx &
\Gamma \vdash \Phi : \ctx}}
\\[1em]
}
\multicolumn{1}{p{13.5cm}}{\fbox{$\Gamma \vdash C : T$}\quad Contextual
  Objects $C$ has Contextual Type $T$ in context $\Gamma$}
\\[0.5em]
\infer{\Gamma \vdash (\hatctx{\Psi} \vdash M) : (\Psi \vdash A)}
{
  \Gamma ; \Psi \vdash M : A}
\LONGVERSION{\quad
\infer{\Gamma \vdash (\hatctx{\Psi} \vdash M) : (\Psi \vdash_\# A)}
{
  \Gamma ; \Psi \vdash_\# M : A}}
\SUBSTCLO{\\[1em]
\infer{\Gamma \vdash (\hatctx{\Psi} \vdash \sigma) : (\Psi \vdash \Phi)}
{\Gamma ; \Psi \vdash \sigma : \Phi}
\quad
\infer{\Gamma \vdash (\hatctx{\Psi} \vdash \sigma) : (\Psi \vdash_\# \Phi)}
{\Gamma ; \Psi \vdash_\# \sigma : \Phi}}
\\[0.75em]
\multicolumn{1}{p{13.5cm}}{\fbox{$\Gamma \vdash C \equiv C' : T$}\quad
  Definitional Equivalence between Contextual Object}
\\[0.5em]
    \infer{\Gamma \vdash (\hatctx{\Psi} \vdash M) \equiv (\hatctx{\Psi} \vdash N) : (\Psi \vdash A)}
          {
           \Gamma ; \Psi \vdash M \equiv N : A}
\SUBSTCLO{\quad
    \infer{\Gamma \vdash (\hatctx{\Psi} \vdash \sigma_1) \equiv (\hatctx{\Psi} \vdash \sigma_2) : (\Psi \vdash \Phi)}
          {
           \Gamma ; \Psi \vdash \sigma_1 \equiv \sigma_2 : \Phi}}
\LONGVERSION
{\\[1em]
\multicolumn{1}{p{13.5cm}}{\fbox{$\Gamma \vdash T \equiv T' $}\quad
  Definitional Equivalence between Contextual Types}
\\[0.5em]
\infer{\Gamma \vdash (\Psi \vdash A) \equiv (\Phi \vdash B)}
      {\Gamma \vdash \Psi \equiv \Phi : \ctx
     & \Gamma ; \Psi \vdash A \equiv B : \lftype}
\qquad
\SUBSTCLO{
\infer{\Gamma \vdash (\Psi \vdash \Psi') \equiv (\Phi \vdash \Phi')}
      {\Gamma \vdash \Psi \equiv \Phi : \ctx
     & \Gamma \vdash \Psi' \equiv \Phi' : \ctx}
}
}
\LONGVERSION{\\[1em]
\infer{\Gamma \vdash (\Psi \vdash_\# A) \equiv (\Phi \vdash_\# B)}
      {\Gamma \vdash \Psi \equiv \Phi : \ctx
     & \Gamma ; \Psi \vdash A \equiv B : \lftype}
\qquad
\SUBSTCLO{\infer{\Gamma \vdash (\Psi \vdash_\# \Psi') \equiv (\Phi \vdash_\# \Phi')}
      {\Gamma \vdash \Psi \equiv \Phi : \ctx
     & \Gamma \vdash \Psi' \equiv \Phi' : \ctx}}
}
\end{array}
\]
  \caption{Typing and Equivalence Rules for Contextual Objects}
  \label{fig:ctxtyping}
\end{figure}

\subsection{Computation Typing}
We describe well-typed computations in Fig.~\ref{fig:comptyping} using the typing judgment $\Gamma \vdash t : \tau$.
Computations only have access to computation-level variables declared in the context $\Gamma$.
To avoid duplication of typing rules, we overload the typing judgment and write $\ann\tau$ instead of $\tau$, if the same judgment is used to check that a given LF context is of schema $\tmctx$. For example, to ensure that $(y : \ann\tau_1) \arrow \tau_2$ has kind $u_3$, we check that $\ann\tau_1$ is well-kinded. For compactness, we abuse notation writing $\Gamma \vdash \tmctx : u$ although the schema $\tmctx$ is not a proper type whose elements can be computed.
In the typing rules for computation-level (extensional) functions, the input to the function which we also call domain of discourse may either be of type $\tau_1$ or $\tmctx$. To eliminate a term $t$ of type $(y : \tau_1) \arrow \tau_2$, we check that $s$ is of type $\tau_1$ and then return $\{s/y\}\tau_2$ as the type of $t~s$.  To eliminate a term of type $(y : \tmctx) \arrow \tau$, we overload application simply writing $t~s$, although $s$ stands for a LF context and check that $s$ is of schema $\tmctx$. This distinction between the domains of discourse is important, as we only allow LF contexts to be built either by a context variable or a LF type declaration.
We can embed contextual object $C$ into computations by boxing it and transitioning to the typing rules for LF. We eliminate contextual types using the recursor.

In general, the output type of the recursor may depend on the argument we are recursing over. We hence annotate the recursor itself with an invariant $\IH$. We consider only the recursor for contextual LF terms where $\IH = (\psi : \tmctx) \arrow (y:\cbox{\psi \vdash \tm}) \arrow \tau$, but other recursors follow similar ideas. To check that the recursor $\titer{\R}{}\IH~\Psi~t$ has type $\{\Psi/\psi, t/y\}\tau$, we check that each of the three branches has the specified type $\IH$. In the base case, we may assume in addition to $\psi:{\tmctx}$ that we have a variable $p:\cbox{\unboxc{\psi} \vdash_\#\tm}$ and check that the body has the appropriate type.
%
%
If we encounter a contextual LF object built with the LF constant $\capp$, then we choose the branch $b_\capp$. We assume $\psi:\cbox{\tmctx}$, $m:\cbox{\unboxc{\psi}\vdash \tm}$, $n:\cbox{\unboxc{\psi}\vdash \tm}$, as well as $f_n$ and $f_m$ which stand for the recursive calls on $m$ and $n$ respectively. We then check that the body $t_\capp$ is well-typed.
If we encounter a LF object built with the LF constant $\clam$, then we choose the branch $b_\clam$. We assume $\psi : \cbox{\tmctx}$ and $m:\cbox{\psi, x{:}\tm \vdash \tm}$ together with the recursive call $f_m$ on $m$ in the extended LF context $\psi, x{:}\tm$. We then check that the body  $t_\clam$ is well-typed.

\begin{figure}[htb]
  \centering
\[
\begin{array}{c}
\multicolumn{1}{l}{\mbox{Well-formed Context:}~\fbox{$\vdash \Gamma$}}\\[-1.5em]
\infer{\vdash \cdot}{} \qquad
\infer{\vdash \Gamma, x{:}\ann\tau}{\vdash \Gamma & \Gamma \vdash \ann\tau : u}
\\[1em]
\multicolumn{1}{l}{\mbox{Typing and Kinding Judgments for Computations}~~\fbox{$\Gamma \vdash t : \tau$} ~\mbox{and}~\fbox{$\Gamma \vdash \tau : u$}}
\\[0.75em]
\infer[(u_1, u_2) \in \Ax]{\Gamma \vdash u_1: u_2}{\vdash \Gamma}
\quad
\infer[(u_1,~u_2,~u_3) \in \Ru]{\Gamma \vdash (y:\ann\tau_1) \arrow \tau_2 : u_3}
      {\Gamma \vdash \ann\tau_1 : u_1 &
       \Gamma, y{:}\tau_1 \vdash \tau_2 : u_2}
\quad
\infer{\Gamma \vdash \cbox{T} : u}{\Gamma \vdash T}
\\[0.75em]
\infer{\Gamma \vdash y : \ann\tau}{y:\ann\tau \in \Gamma&\vdash\Gamma}
\quad
\infer{\Gamma \vdash t~s : \{s/y\}\tau_2}
{\Gamma \vdash t : (y:\ann\tau_1) \arrow \tau_2 &
 \Gamma \vdash s : \ann\tau_1}
\quad
\infer{\Gamma \vdash \tmfn y t : (y:\ann\tau_1) \arrow \tau_2}
        {\Gamma, y:\ann\tau_1 \vdash t : \tau_2 & \Gamma \vdash (y:\ann\tau_1) \arrow \tau_2 : u}
\\[0.75em]
\infer{\Gamma \vdash \cbox C : \cbox T}{\Gamma \vdash C : T}
\quad
\infer{\Gamma \vdash t : \tau}
{\Gamma \vdash t : \tau' & \Gamma \vdash \tau' \equiv \tau : u}
\\[1em]
\multicolumn{1}{l}{\mbox{Schema checking of LF Context}~~\fbox{$\Gamma \vdash \Psi : \tmctx$}~~\mbox{;Well-Formedness of Schema}~\fbox{$\Gamma \vdash \tmctx : u$}}\\[0.75em]
\infer{\Gamma\vdash \tmctx : u}{\vdash \Gamma}
\quad
\infer{\Gamma \vdash \cdot : \tmctx}{\vdash\Gamma}
\quad
\infer{\Gamma \vdash \Psi, x{:}A : \tmctx}
{\Gamma \vdash \Psi : \tmctx & \Gamma ; \Psi \vdash A : \lftype & \Gamma ; \Psi \vdash A \equiv \tm : \lftype}
\end{array}
\]
  \caption{Typing Rules for Computations (without recursor)}
  \label{fig:comptyping}
\end{figure}

\begin{figure}[htb]
  \centering
\[
\begin{array}{c}
\LONGVERSION{
\multicolumn{1}{l}{\mbox{Recursor over LF Parameters}~\IH = (\psi : \tmctx) \arrow (q:\cbox{\psi \vdash_\# \tm}) \arrow \tau }\\[0.5em]
\infer[]
{\Gamma \vdash \tmrecctx {\IH} {\psi \mto b_e}{\psi, q, f_q \mto b_c} \rappto \Psi~t: \{\Psi/\psi,~t/y\}\tau}
{
\Gamma \vdash t :  \cbox{\Psi \vdash_\# \tm} & \Gamma \vdash \IH : u &
\Gamma \vdash (\psi \mto b_e) : \IH &
\Gamma \vdash (\psi, q, f_q \mto b_c) : \IH
 }
\\[1em]
\multicolumn{1}{p{13.5cm}}{\mbox{Branches where} \quad$\IH = (\psi : \tmctx) \arrow (y:\cbox{\psi \vdash_\# \tm}) \arrow \tau$ }
\\[1em]
\infer{\Gamma \vdash (\psi \mto b_e) : \IH}
{\Gamma, \psi:\tmctx \vdash  b_e : \{{(\psi,x{:}\tm)}/\psi,~\cbox{\psi, x \vdash x}/p \}\tau}
\\[0.5em]
\infer{\Gamma \vdash (\psi, q, f_q \mto b_c) : \IH}
{ \Gamma, \psi:\tmctx, q:\cbox{\psi \vdash_\# \tm},  f_q: \{q/p\}\tau  \vdash  b_c :
  \{(\psi,x{:}\tm)/\psi,~\{\cbox{\psi, x \vdash  \unbox{q}{\wk{\psi}}}/p \}\tau }
\\[1em]
}
\multicolumn{1}{l}{\mbox{Recursor over LF Terms}~\IH = (\psi : \tmctx) \arrow (y:\cbox{\psi \vdash \tm}) \arrow \tau }\\[0.5em]
\infer
{\Gamma \vdash \tmrec {\IH} {b_v} {b_{\mathsf{app}}} {b_{\clam}} \rappto \Psi~t : \{{\Psi}/\psi,~t/y\}\tau}
{
 \Gamma \vdash t :  \cbox{\Psi \vdash \tm} & \Gamma \vdash \IH : u &
 \Gamma \vdash b_v : \IH & \Gamma \vdash b_{\mathsf{app}} : \IH & \Gamma \vdash b_{\mathsf{lam}} : \IH}
\\[1em]
\multicolumn{1}{p{13.5cm}}{\mbox{Branches where} \quad$\IH = (\psi : \tmctx) \arrow (y:\cbox{\psi \vdash \tm}) \arrow \tau$ }
\\[1em]
\infer{\Gamma \vdash ({\psi,p \mto t_v}) : \IH }
{ \Gamma, \psi:\tmctx, p:\cbox{~\unboxc{\psi} \vdash_\# \tm}  \vdash  t_v : \{p / y\}\tau}
\\[1em]
\infer{\Gamma \vdash (\psi, m, n, f_n, f_m \mto t_{\mathsf{app}}) : \IH}
{
  \begin{array}{l@{}lcl}
\Gamma, & \psi:\tmctx, m:\cbox{~\unboxc{\psi} \vdash \tm}, n:\cbox{~\unboxc{\psi} \vdash \tm}& & \\
        & f_m: \{m/y\}\tau, f_n: \{n/y\}\tau  & \vdash &  t_{\mathsf{app}} : \{\cbox{~\unboxc{\psi} \vdash \capp~\unbox{m}{\id}~\unbox{n}{\id}}/y\}\tau
  \end{array}
}
\\[1em]
\infer{\Gamma \vdash \psi, m, f_m \mto t_{\clam} : \IH}
{ \begin{array}{l@{}lcl}
\Gamma, & \phi:\tmctx,  m:\cbox{\phi, x:\tm \vdash \tm}, & & \\
        &  f_m:\{(\phi, x:\tm)/\psi, m /y \} \tau          & \vdash & t_{\clam} : \{\phi/\psi, \cbox{~\unboxc{\phi} \vdash \clam~\lambda x.\unbox{m}{\id}~} / y\}\tau
 \end{array}
 }
\end{array}
\]
  \caption{Typing Rules for Recursors}
  \label{fig:comptypingrec}
\end{figure}

\subsection{Definitional Equality for Computations}
We now consider definitional equality for computations concentrating on the reduction rules. We omit the transitive closure and congruence rules, as they are as expected.


\begin{figure}[h]
  \centering
  \[
  \begin{array}{c}
    \multicolumn{1}{l}{\mbox{Reduction and Expansions for Computations}}\\[0.75em]
          \infer{\Gamma  \vdash (\tmfn y t)~s \equiv \{s/y\}t : \{s/y\}\tau_2}
                {\Gamma \vdash \tmfn y t : ( y{:}\ann\tau_1) \arrow \tau_2 & \Gamma \vdash s : \ann\tau_1
                }
\qquad
\infer{\Gamma \vdash \cbox{\hatctx \Psi \vdash \unbox{t}{\wk{\hatctx\Psi}}} \equiv t : \cbox{\Psi \vdash A}}{
        \Gamma \vdash t : \cbox{\Psi \vdash A}}
                \\[1em]
\multicolumn{1}{l}{
\mbox{let}~{\R} =  ({\psi,p \mto t_p} \mid {\psi,m,n,f_m, f_n \mto t_{\mathsf{app}}} \mid {\psi, m, f_m \mto t_{\clam}})}\\[0.5em]
\multicolumn{1}{l}{
\mbox{and}~\IH = (\psi : \tmctx) \arrow (y : \cbox{\psi \vdash \tm}) \arrow \tau}
\\[0.75em]
\infer{\Gamma \vdash \titer{\R}{\tm}{\IH}\rappto~\Psi~\cbox{\hatctx{\Psi} \vdash \clam \lambda x.M} \equiv \{\theta\}t_{\clam}
                  :  \{\Psi/\psi, \cbox{\hatctx{\Psi} \vdash \clam \lambda x.M}/y\}\tau }
      {\Gamma \vdash \Psi : \tmctx \qquad \Gamma; \Psi, x{:}\tm \vdash M : \tm
       {\qquad\Gamma \vdash \IH : u}
      }
 \\[0.5em]\mbox{where}~
 \theta   =  \Psi/\psi,~
             \cbox{\hatctx{\Psi}, x \vdash M}/m,~
             \titer{\R}{\tm}{\IH}~\rappto {(\Psi,x{:}\tm)}~{\cbox{\hatctx{\Psi}, x \vdash M}}/f
\\[1em]
\infer{\Gamma \vdash \titer{\R}{\tm}{\IH}~\rappto {\Psi} \cbox{\hatctx{\Psi} \vdash \capp M~N} \equiv \{\theta\} t_{\capp}
:  \{\Psi/\psi, \cbox{\hatctx{\Psi} \vdash \capp M~N}/y\}\tau }
{
\Gamma  \vdash \Psi : \tmctx  \qquad
\Gamma; \Psi   \vdash  M : \tm  \qquad
\Gamma; \Psi   \vdash  N : \tm \qquad  {\Gamma \vdash \IH : u}
}
\\[0.5em]
\mbox{where}~
\theta  =  \Psi/\psi,~\cbox{\hatctx{\Psi} \vdash M}/m,~
           \cbox{\hatctx{\Psi} \vdash N}/n,~
           \titer{\R}{\tm}{\IH}~\rappto {\Psi}~\cbox{\hatctx{\Psi} \vdash M}/f_m,~
           \titer{\R}{\tm}{\IH}~\rappto {\Psi}~\cbox{\hatctx{\Psi} \vdash N}/f_n

\\[1em]
\infer{\Gamma \vdash \trec{\R}{\tm}{\IH}~\rappto {\Psi}~{\cbox{\hatctx \Psi \vdash x} }
          \equiv \{\Psi/\psi,~\cbox{\hatctx{\Psi} \vdash x}/p\} t_p :  \{\Psi/\psi, \cbox{\Psi \vdash x}/y\}\tau }
      {x{:}\tm \in \Psi  & \Gamma \vdash \Psi : \tmctx
             {\qquad\Gamma \vdash \IH : u}}
      \\[0.75em]
  \end{array}
  \]
  \caption{Definitional Equality for Computations}
  \label{fig:etype}
\end{figure}

We consider two computations to be equal, if they evaluate to the same result. We propagate values through computations and types relying on the computation-level substitution operation. When we apply a term $s$ to a computation $\tmfn y t$, we $\beta$-reduce and replace $y$ in the body $t$ with $s$. We unfold the recursor depending on the value passed. If it is $\cbox{\hatctx{\Psi} \vdash \clam~ \lambda x.M}$, then we choose the branch $t_{\clam}$. If the value is $\cbox{\hatctx{\Psi} \vdash \capp~ M~N}$, we continue with the branch $t_{\capp}$. If it is ${\cbox{\hatctx{\Psi} \vdash x} }$, i.e. the variable case, we continue with $t_v$. Note that if $\Psi$ is empty, then the case for variables is unreachable, since there is no LF variable of type $\tm$ in the empty LF context and hence the contextual type $\cbox{\cdot \vdash_\# \tm}$ is empty.

We also include the expansion of a computation $t$ at type $\cbox{\Psi \vdash A}$; it is equivalent to unboxing $t$ with the identity substitution and subsequently boxing it, i.e. $t$ is equivalent to $\cbox{\hatctx \Psi \vdash \unbox{t}{\wk{\hatctx{\Psi}}}}$ .


\section{Elementary Properties of Typing and Definitional Equality}

We now state and prove some basic properties about our type theory
before we give its semantic interpretation and show that all
well-typed terms normalize.
%
For stating the theorems succinctly, we refer to judgments that only depend on the computation context $\Gamma$ using $\Jcomp$ and judgments that refer to both the computation context $\Gamma$ and the LF context $\Psi$ with $\JLF$. \SHORTVERSION{We concentrate here on the judgments we introduced so far.}

\[
  \begin{array}{lcl}
\Jcomp & = & \{t:\ann\tau, t\equiv t' : \ann\tau
\}\\
\JLF & = & \{M:A, M\equiv N : A,
           \LONGVERSION{A : K, A \equiv B : K, K : \lfkind, K\equiv K' : \lfkind,}
            \sigma : \Psi, \sigma\equiv \sigma' : \Psi\}
  \end{array}
\]

Next we prove some elementary properties for LF and computations. As we separate the LF variables from the computation-level variables, we establish first properties such as well-formedness of context, weakening and substitution, for LF and then we prove the dual properties for computations.



\subsection{Elementary Properties of LF}\label{sec:proplf}
\begin{theorem}[Well-Formedness of LF Context]\quad
  \label{lm:lfctxwf}
  \begin{enumerate}
    \item \label{it:lfctxwf} If $\D ::\Gamma \vdash \Psi,x{:}A,\Psi' : \ctx$ then $\Ca ::~\Gamma \vdash \Psi : \ctx$ and $\Ca<\D$.
    \item \label{it:lflfctxwf} If $\D :: \Gamma; \Psi \vdash \JLF$ then $\Ca ::~\Gamma \vdash \Psi : \ctx$ and $\Ca<\D$.
  \end{enumerate}
\end{theorem}
\begin{proof}
\ref{it:lfctxwf}(1) by induction on the structure of $\Psi'$;
\ref{it:lflfctxwf}(2) by induction on $\Gamma; \Psi \vdash \JLF$.
\LONGVERSIONCHECKED{
\\[0.5em]
First statement:  If $\D ::\Gamma \vdash \Psi,x{:}A,\Psi'$ then $\Ca ::~\Gamma \vdash \Psi$ and $\Ca<\D$
\\[0.25em]
\pcase{$\Psi' = \cdot$}
\prf{$\D :: \Gamma\vdash \Psi,x{:}A : \ctx$ \hfill by assumption}
\prf{$\Ca :: \Gamma\vdash \Psi : \ctx$ and $\Ca<\D$ \hfill by inversion}
\\[0.15em]
\pcase{$\Psi' = \Psi'', x':B$}
\prf{$\D :: \Gamma\vdash \Psi,x{:}A,\Psi'',x'{:}B : \ctx$ \hfill by assumption}
\prf{$\D :: \Gamma\vdash \Psi,x{:}A,\Psi''$ and $\D'<\Ca'$ \hfill by inversion}
\prf{$\Ca :: \Gamma\vdash \Psi$ and $\Ca<\D$ \hfill by IH}
\\[0.5em]
Second statement:If $\D :: \Gamma; \Psi \vdash \JLF$ then $\Ca ::~\Gamma \vdash \Psi : \ctx$ and $\Ca<\D$.
\\[0.75em]
  \pcase{$\D=\ibnc{\Gamma \vdash \Psi : \ctx}{\const a{:}K \in \Sigma}{\Gamma; \Psi \vdash \const a: K}{}$}
  \prf{$\Ca :: \Gamma \vdash \Psi : \ctx$ and $\Ca<\D$ \hfill by assumption}
\\[-0.5em]
  \pcase{$\D = \ianc{\Gamma ; \Psi, x{:}A \vdash M : B}{\Gamma ; \Psi \vdash \lambda x.M : \Pityp x A B}{}$}
  \prf{$\Ca' :: \Gamma \vdash \Psi, x{:}A : \ctx$ and $\Ca'<\D$ \hfill by IH}
  \prf{$\Ca :: \Gamma \vdash \Psi : \ctx$ and $\Ca<\D$ \hfill by inversion on well-formedness rules for LF contexts}
}
\end{proof}

 \begin{lemma}[LF Weakening]
   \label{lm:extlfctx}~
 Let $\wvec{y{:}B} = y_1{:}B_1, \ldots, y_k{:}B_k$.
 \\\mbox{\quad}
 If~~$\Gamma; \Psi, \wvec{y{:}B} \vdash \JLF$ and
    $\Gamma \vdash (\Psi,x{:}A,  \wvec{y{:}B}) : \ctx$
    then~~$\Gamma; \Psi,x{:}A,\wvec{y{:}B} \vdash \JLF$.
 \end{lemma}
 \begin{proof}
   By induction on the first derivation.
 \LONGVERSIONCHECKED{
\\[1em]
   \pcase{\ibnc{\Gamma \vdash \Psi,\wvec{y{:}B} : \ctx}{\const a{:}K \in \Sigma}{\Gamma; \Psi,\wvec{y{:}B} \vdash \const a: K}{}}
   \prf{$\Gamma \vdash \Psi,x{:}A, \wvec{y{:}B} : \ctx$ \hfill by assumption}
   \prf{$\Gamma; \Psi,x:A, \wvec{y{:}B} \vdash \const a: K$ \hfill by rule}
\\
   \pcase{$\ibnc{\Gamma ; \Psi,\wvec{y{:}B} \vdash A' : \lftype}
                {\Gamma ; \Psi,\wvec{y{:}B}, x'{:}A' \vdash B' : \lftype}
                {\Gamma ; \Psi,\wvec{y{:}B} \vdash \Pityp {x'} {A'} B' : \lftype}{}$}
\prf{$\Gamma ; \Psi, x{:}A, \wvec{y{:}B} \vdash A' : \lftype$ \hfill by IH}
\prf{$\Gamma \vdash \Psi, x{:}A, \wvec{y{:}B}  : \ctx$ \hfill by assumption}
\prf{$\Gamma \vdash \Psi, x{:}A, \wvec{y{:}B}, x'{:}A' : \ctx$ \hfill by rules for well-formed LF contexts}
\prf{$\D::\Gamma ; \Psi,\wvec{y{:}B}, x'{:}A' \vdash B' : \lftype$     \hfill by premise}
\prf{$\Gamma;\Psi, x{:}A, \wvec{y{:}B}, x'{:}A' \vdash B' : \lftype$ \hfill by IH}
\prf{$\Gamma;\Psi, x{:}A, \wvec{y{:}B} \vdash \Pityp {x'} {A'} B' : \lftype$ \hfill by rule}
 }
 \end{proof}

\LONGVERSION{ \begin{lemma}[LF Variable Lookup]\label{lm:lflookup}~
Let $\Gamma \vdash \Psi : \ctx$ and
   $\Psi(x) = A$.\\\mbox{\quad}
If    $\Gamma ; \Phi \vdash \sigma : \Psi$
then $\Gamma; \Phi \vdash M : \lfs{\sigma}{\Psi} A $ and $\pos
x~\lfs\sigma\Psi = M$.
 \end{lemma}
 \begin{proof}
By induction $\Gamma ; \Phi \vdash \sigma : \Psi$.
\LONGVERSIONCHECKED{
\\[0.5em]
\pcase{$\ianc {\Gamma\vdash \Psi, \wvec{y{:}B} : \ctx }
              {\Gamma ;\Psi, \wvec{y{:}B} \vdash \wk{\hatctx\Psi} : \Psi}{}$}
\\
\prf{$x \in \hatctx\Psi$ \hfill by assumption $\Psi(x) = A$}
\prf{$\pos x~\lfs {\wk{\hatctx\Psi}}{\Psi}  =  x$ \hfill by definition of $\pos$}
\prf{$(\Psi, \wvec{y{:}B})(x) = A$ \hfill since $\Psi(x) = A$}
\prf{$\Gamma ; \Psi, \wvec{y{:}B} \vdash x : A$ \hfill by typing rule}
\\
\pcase{$\ibnc{\Gamma ; \Phi \vdash \sigma : \Psi' }
             {\Gamma ; \Phi \vdash N : \lfs \sigma {\Psi '} B}
             {\Gamma ; \Phi \vdash \sigma, N : \Psi', y{:}B}{}$  where $\Psi'(x) = A$ and $x \neq y$}
\prf{$\Gamma ; \Phi \vdash M : \lfs \sigma {\Psi'} A$ \hfill by IH}
\prf{$\Gamma ; \Phi \vdash M : \lfs {\sigma, N} {\Psi', y} A$ \hfill since $y \not\in \FV(A)$}

\pcase{$\ibnc{\Gamma ; \Phi \vdash \sigma : \Psi '}
             {\Gamma ; \Phi \vdash M : \lfs \sigma {\Psi '} A}
             {\Gamma ; \Phi \vdash \sigma, M : \Psi ', x{:}A}{}$  where $(\Psi ', x{:}A)(x) = A$ }
\prf{$\pos x~\lfs {\sigma, M}{\Psi ', x}  =  M$ \hfill by def. of $\pos$}
\prf{$\Gamma ; \Phi \vdash M : \lfs \sigma {\Psi '} A$ \hfill by premise}
\prf{$\Gamma ; \Phi \vdash M : \lfs {\sigma, M} {\Psi ', x} A$ \hfill since $x \not\in \FV(A)$}

\SUBSTCLO{
\pcase{$ \ianc{\Gamma \vdash t : \cbox{\Psi' \vdash \Psi, \wvec{y{:}B}}
         \quad \Gamma ; \Phi \vdash \sigma : \Psi'
         \quad \Gamma\vdash \Psi, \wvec{y{:}B} : \ctx }
              {\Gamma ; \Phi \vdash \sclo {\hatctx\Psi} {r} : \Psi}{r = \unbox{t}{\sigma}}
$
}
\prf{$\Psi(x_i) = A$ where $\Psi = \Psi_i, x_i{:}A, \wvec{y{:}A}$\hfill by assumption}
\\
\prf{\emph{Subcase.} $r \neq \unbox {\cbox{\hatctx{\Psi'} \vdash \sigma'}} {\sigma}$}
\prf{$\pos x~\lfs {\sclo {\hatctx\Psi} r} \Psi  =
       \unbox{\cbox{\hatctx{\Psi_i},x_i \vdash x_i}}{\sclo {\hatctx{\Psi_i},x_i} r} $}
\prf{$\Gamma ; \Psi_i,x_i{:}A \vdash x_i : A$ \hfill by typing rule since $\Gamma \vdash \Psi : \ctx$ }
\prf{$\Gamma \vdash \cbox{\hatctx{\Psi_i},x_i \vdash x_i} : \cbox{\Psi_i,x_i{:}A \vdash A}$ \hfill by typing rule}
\prf{$\Gamma ; \Phi \vdash \unbox{\cbox{\hatctx{\Psi_i},x_i \vdash x_i}}{\sclo {\hatctx{\Psi_i},x_i} r} : \lfs {\sclo {\hatctx{\Psi_i},x_i} r}{\Psi_i,x_i} A$ \hfill by typing rule}
\prf{$\Gamma ; \Phi \vdash \unbox{\cbox{\hatctx{\Psi_i},x_i \vdash x_i}}{\sclo {\hatctx{\Psi_i},x_i} r} : \lfs {\sclo {\hatctx{\Psi}} r}{\Psi} A$ \hfill using
$\lfs {\sclo {\hatctx{\Psi_i},x_i} r}{\Psi_i,x_i} A = \lfs {\sclo {\hatctx{\Psi}} r}{\Psi} A$}
\\
\prf{\emph{Subcase.} $r = \unbox {\cbox{\hatctx{\Psi'} \vdash \sigma'}} {\sigma}$}
\prf{$\Gamma ; \Psi' \vdash \sigma' : \Psi, \wvec{y{:}B}$
     \hfill by $\Gamma \vdash t : \cbox{\Psi' \vdash \Psi, \wvec{y{:}B}}$}
\prf{$\Gamma ; \Psi' \vdash \sigma'' : \Psi$ where $\sigma'' = \trunc_\Psi~(\sigma' / \hatctx{\Psi},\vec y)$}
\prf{$\Gamma ; \Psi' \vdash M : \lfs{\sigma''}{\Psi}A$ and $\pos x~\lfs {\sigma'}{\Psi} = M$\hfill by IH (using $\Gamma ; \Psi' \vdash \sigma'' : \Psi$)}
\prf{$\Gamma ; \Phi \vdash \lfs \sigma {\Psi'} M : \lfs \sigma {\Psi'}( \lfs{\sigma''}{\Psi}A)$ \hfill by LF subst. lemma}
\prf{$\Gamma ; \Phi \vdash \lfs \sigma {\Psi'} M : \lfs {\sclo
    {\hatctx\Psi} {r} }{\Psi} A$ \hfill since by LF subst. def.}
}
}

 \end{proof}
}

\begin{lemma}[LF Substitution]\label{lm:lfsubst}~
If $~\Gamma ; \Psi \vdash \JLF$ and $~\Gamma ; \Phi \vdash  \sigma : \Psi$
      then $\Gamma ; \Phi \vdash \lfss{\sigma}{\Psi} \JLF$.
\end{lemma}
\begin{proof}
  By induction on the derivation on the first derivation using
  well-formedness of LF contexts (Lemma \ref{lm:lfctxwf}) and LF
  weakening (Lemma \ref{lm:extlfctx}). \SHORTVERSION{For the LF variable case we
  prove by induction on $\Gamma ; \Phi \vdash  \sigma : \Psi$ that
for any $x \in \Psi$, $\Gamma; \Phi \vdash M : \lfs{\sigma}{\Psi} A $
where $\pos x~\lfs\sigma\Psi = M$.}
\LONGVERSION{In the LF variable case, we refer to Lemma \ref{lm:lflookup}.}
\LONGVERSIONCHECKED{
Most cases are straightforward; we only show a few cases, the others are similar.
\\[1em]
\pcase{$\ianc
{
\Gamma \vdash \Psi : \ctx
\quad  \Psi(x) = A
}
{\Gamma ; \Psi \vdash x : A}{}
$}
\prf{$\Gamma ; \Phi \vdash M :  \lfss {\sigma}{\Psi}A$ and $\pos x~\lfs\sigma\Psi = M$ \hfill by Lemma \ref{lm:lflookup}}
\prf{$\Gamma ; \Phi \vdash \lfss \sigma \Psi x :  \lfss \sigma \Psi A$ \hfill by subst. def.}
\\
\pcase{$\ianc{\const a:K \in \Sigma \qquad \Gamma \vdash \Psi : \ctx }
             {\Gamma; \Psi \vdash \const a: K}{}$}
\prf{$\Gamma\vdash \Phi : \ctx$ \hfill by Lemma \ref{lm:lfctxwf}}
\prf{$\lfss\sigma\Psi (\const a:K) \in \Sigma$ \hfill as $K$ is
  closed and $\lfss\sigma\Psi K = K$}
\prf{$\Gamma; \Phi \vdash \lfss\sigma\Psi \const a: \lfss\sigma\Psi K$ \hfill by rule and substitution def. }
\\[1em]
\pcase{$\ibnc{\Gamma ; \Psi \vdash A : \lftype}
             {\Gamma ; \Psi, x:A \vdash B : \lftype}
             {\Gamma ; \Psi \vdash \Pityp x A B : \lftype}{}$}
\prf{$\Gamma ; \Phi \vdash \lfss\sigma\Psi A : \lftype$ \hfill by IH }
\prf{$\Gamma \vdash \Phi : \ctx$ \hfill by the lemma \ref{lm:lfctxwf} using $~\Gamma ; \Phi \vdash  \sigma : \Psi$}
\prf{$\Gamma \vdash \Phi, x:\lfss\sigma\Psi A : \ctx$ \hfill by rule}
\prf{$\Gamma ; \Phi \vdash  \sigma : \Psi$ \hfill by assumption}
\prf{$\Gamma ; \Phi, x:\lfss \sigma\Psi A  \vdash  \sigma : \Psi$ \hfill by Lemma \ref{lm:extlfctx}}
\prf{$\Gamma ; \Phi, x:\lfss \sigma\Psi A \vdash x : (\lfs\sigma\Psi A)$ \hfill by rule}
\prf{$\Gamma ; \Phi, x:\lfss \sigma\Psi A \vdash \sigma,~ x : \Psi, x:A$ \hfill by rule}
\prf{$\Gamma ; \Phi, x:\lfss \sigma\Psi A \vdash \lfss {\sigma, ~x}{\Psi,x}B : \lftype$ \hfill by IH}
\prf{$\Gamma ; \Phi \vdash \Pityp x {\lfss\sigma\Psi A} {\lfss{\sigma,x}{\Psi,x}B} : \lftype$ \hfill by rule}
\prf{$\Gamma ; \Phi \vdash \lfss\sigma\Psi(\Pityp x A B) : \lfss\sigma\Psi \lftype$ \hfill by substitution def.}
\\[1em]
\pcase{$\ibnc{\Gamma \vdash t : [\Phi' \vdash A] ~\mbox{or}~\Gamma \vdash t : [\Phi' \vdash_\# A]}
            {\Gamma; \Psi \vdash \sigma' : \Phi'}
            {\Gamma ; \Psi \vdash \unbox t {\sigma'} : \lfss{\sigma'}{\Phi'}A}{}$}
\prf{$\Gamma; \Phi \vdash \lfss\sigma\Psi \sigma' : \Phi'$ \hfill IH}
\prf{$\Gamma; \Phi \vdash \unbox t {\lfss\sigma\Psi\sigma'} : \lfss {\lfss\sigma\Psi \sigma'}{\Phi'}A$ \hfill by rule}
\prf{$\Gamma; \Phi \vdash \lfss\sigma\Psi (\unbox t {\sigma'}) : \lfss\sigma\Psi(\lfss{\sigma'}{\Phi'}A)$ \hfill by substitution def.}\\[1em]
\pcase{
$\ibnc{\Gamma ; \Psi \vdash M : \Pityp x A B}
      {\Gamma ; \Psi \vdash N : A}
      {\Gamma ; \Psi \vdash M~N : [N/x]B}{}$}
\prf{$\Gamma ; \Phi \vdash \lfss\sigma\Psi M :\lfss\sigma\Psi ( \Pityp x A B)$ \hfill by IH}
\prf{$\Gamma ; \Phi \vdash \lfss\sigma\Psi M : \Pityp x {\lfss\sigma\Psi A} {(\lfss{\sigma,~x}{\Psi, x}B)}$ \hfill by substitution def.}
\prf{$\Gamma ; \Phi \vdash \lfss\sigma\Psi N : \lfss \sigma\Psi A$ \hfill by IH}
\prf{$\Gamma ; \Phi \vdash (\lfss\sigma\Psi M)~(\lfss\sigma\Psi N) : [\lfss\sigma\Psi N/x](\lfss{\sigma,~x}{\Psi, x}~B)$\hfill by rule}
\prf{$\Gamma ; \Phi \vdash \lfss{\sigma}{\Psi}(M~N) : \lfss{\sigma}{\Psi}([ N /x]B)$\hfill by definition and composition of substitution }\\
\pcase{$\ianc{\Gamma \vdash \Psi, \wvec{x{:}A} :\ctx \quad }
             {\Gamma ;\Psi, \wvec{x:A} \vdash \wk{\hatctx\Psi} : \Psi}{}$}\\
    \prf{$\Gamma ; \Phi \vdash \sigma  :\Psi, \wvec{x{:}A}$ \hfill by assumption}
    \prf{$\Gamma ; \Phi \vdash \sigma' :\Psi$ \hfill by inversion where $\sigma = \sigma', M_n, \ldots M_1$}
    \prf{$\Gamma ; \Phi \vdash \lfss \sigma{\Psi,\vec x}(\wk{\hatctx\Psi}) : \Psi$ \hfill using the fact that  $\sigma' = \trunc_\Psi (\sigma / \hatctx \Psi, \vec x) = \lfss\sigma{\Psi, \vec x}(\wk{\hatctx\Psi})$}\\
    \pcase{$\ibnc{\Gamma ; \Psi \vdash \sigma' : \Phi'}
                 {\Gamma ; \Psi \vdash M : \lfss{\sigma'}{\Phi'}A}
                 {\Gamma ; \Psi \vdash \sigma', M : \Phi', x:A}{}$ }\\
    \prf{$\Gamma ; \Phi \vdash \lfss\sigma\Psi \sigma' : \Phi'$ \hfill by IH}
    \prf{$\Gamma ; \Phi \vdash \lfss\sigma\Psi M : \lfss\sigma\Psi (\lfss{\sigma'}{\Phi'} A$ \hfill by IH }
    \prf{$\Gamma ; \Phi \vdash \lfss\sigma\Psi M : \lfss{\lfss\sigma\Psi\sigma'}{\Phi'} A$ \hfill by substitution def.}
    \prf{$\Gamma ; \Phi \vdash \lfss\sigma\Psi \sigma', \lfss\sigma\Psi M : \Phi', x:A$ \hfill by rule}
    \prf{$\Gamma ; \Phi \vdash \lfss\sigma\Psi (\sigma', M) : \Phi', x:A$  \hfill by substitution def.}
}
\end{proof}

\LONGVERSION{We omit here the substitution lemma for the restricted LF typing
judgments $\Gamma ; \Psi \vdash_\# M : A$ and $\Gamma ; \Phi \vdash_\#
\sigma : \Psi$. However, it is worth noting that when we apply an LF
substitution $\sigma$ where $\Gamma ; \Phi \vdash \sigma : \Psi$ to
$M$ where $\Gamma ; \Psi \vdash_\# M : A$ we are not guaranteed to
obtain a variable and hence we can only conclude $\Gamma ; \Phi \vdash
\lfs \sigma \Psi M : \lfs \sigma \Psi A$. We can only guarantee that
we remain in $\vdash_\#$ if the LF substitution is a variable
substitution.}

\begin{lemma}[LF Context Conversion]\label{lm:lfctxconv}
Assume $\Gamma \vdash \Psi, x{:}A : \ctx$ and $\Gamma ; \Psi \vdash B : \lftype$.\\\mbox{\quad}
If $\Gamma ; \Psi, x{:}A \vdash \JLF$ and $\Gamma ; \Psi \vdash A \equiv B : \lftype$
then $\Gamma ; \Psi, x{:}B \vdash \JLF$.
\end{lemma}
\begin{proof}
Proof using LF Substitution (Lemma \ref{lm:lfsubst}).
\LONGVERSIONCHECKED{
\\[0.5em]
\prf{$\Gamma \vdash \Psi, x{:}A : \ctx$ \hfill by assumption}
\prf{$\Gamma \vdash \Psi : \ctx$ \hfill by inversion}
\prf{$\Gamma \vdash \Psi, x{:}B : \ctx$ \hfill by rule}
\prf{$\Gamma ; \Psi, x{:}B \vdash \wk{\hatctx{\Psi}}: \Psi$ \hfill by rule}
\prf{$\Gamma ; \Psi, x{:}B \vdash x : B$ \hfill by rule}
\prf{$\Gamma ; \Psi \vdash B \equiv A : \lftype$ \hfill by symmetry}
\prf{$\Gamma ; \Psi, x{:}B \vdash x : A$ \hfill by rule}
\prf{$\Gamma ; \Psi, x{:}B \vdash x :\lfs{\wk{\hatctx\Psi}}{\Psi} A$ \hfill as $\lfs{\wk{\hatctx\Psi}}{\Psi} A = A$}
\prf{$\Gamma ; \Psi, x{:}B \vdash \wk{\hatctx{\Psi}}, x : \Psi, x{:}A$ \hfill by rule}
\prf{$\Gamma ; \Psi, x{:}B \vdash \J$ \hfill by Lemma \ref{lm:lfsubst} and $\lfs{\wk{\hatctx\Psi}, x}{\Psi, x} \J = \J$}

}
\end{proof}

\begin{lemma}[Functionality of LF Typing]\label{lm:func-lftyping}$\quad$ \\\mbox{\quad}
Let $\Gamma ; \Psi \vdash \sigma_1 : \Phi$ and $\Gamma ; \Psi \vdash \sigma_2 : \Phi$, and
 $\Gamma ; \Psi \vdash \sigma_1 \equiv \sigma_2 : \Phi$.
\begin{enumerate}
 \item If $\Phi = \Phi_i, x_i{:}A, \wvec{y{:}A}$ and $\Gamma ; \Phi \vdash x_i : A$
      then $\Gamma ; \Psi \vdash \lfs{\sigma_1}\Phi (x_i) \equiv \lfs{\sigma_2}\Phi  (x_i) :
     \lfs{\sigma_1}\Phi (A)$.
 \item If $\Gamma ; \Phi \vdash \sigma : \Phi'$
 then $\Gamma ; \Psi \vdash \lfs{\sigma_1}{\Phi}\sigma \equiv \lfs{\sigma_2}{\Phi}\sigma : \Phi'$.

\item If $\Gamma ; \Phi \vdash M : A$ then
      $\Gamma ; \Psi \vdash \lfs{\sigma_1}{\Phi}M \equiv \lfs{\sigma_2}{\Phi}M : \lfs{\sigma_1}{\Phi}A$.
\LONGVERSION{\item If $\Gamma ; \Phi \vdash A : \lftype$ then
      $\Gamma ; \Psi \vdash \lfs{\sigma_1}{\Phi}A \equiv \lfs{\sigma_2}{\Phi}A : \lftype$.}
\end{enumerate}
\end{lemma}
\begin{proof}
We prove these statements by induction on the typing derivation $\Gamma ; \Phi \vdash M : A$ (resp. $\Gamma ; \Phi \vdash \sigma : \Phi'$\LONGVERSION{\/and  $\Gamma ; \Phi \vdash A : \lftype$}) and followed by another inner induction on $\Gamma ; \Psi \vdash \sigma_1 \equiv \sigma_2 : \Phi$ to prove (1).
\LONGVERSIONCHECKED{
\\[1em]
We concentrate first on the variable case (1).

\pcase{$\ianc
 {\Gamma \vdash \Phi_0, x_0{:}A_0, \wvec{y{:}B} : \ctx}
 {\Gamma ; \Phi_0, x_0{:}A_0, \wvec{y{:}B} \vdash \wk{\hatctx\Phi,x} \equiv \wk{\Phi_0}, x_0: \Phi_0, x_0{:}A_0}{}
$}
\prf{Let $x_i \in \hatctx\Phi_0$ and $\Phi_0 = \bullet, x_n, \ldots, x_1$ where $\bullet$ stands for either the empty context or a variable.
Then $\pos {x_i} \lfs{\wk{\bullet, x_n \ldots, x_1}}{x_n \ldots, x_1} = x_i$}
\\
 \prf{\emph{Subcase.} $x_i = x_0$}
 \prf{$\pos {x_i}~\lfs {\wk{\hatctx\Phi,x_0}}{\Phi_0, x_0} = x_0$ \hfill since $x_i{:}A_i \in (\Phi_0, x_0:A_0)$}
\prf{$\pos {x_i}~\lfs { \wk{\Phi_0}, x_0}{\Phi_0, x_0} = x_0$ \hfill by $\pos{}{}$}
\\
\prf{\emph{Subcase.} $x \not= x_0$ and $x_i \in x_n, \ldots, x_1$}
\prf{$\pos {x_i}~\lfs {\wk{\hatctx\Phi,x}}{\Phi_0, x_0} = x_i$ \hfill since $x_i{:}A_i \in (\Phi_0, x_0:A_0)$}
\prf{$\pos {x_i}~\lfs { \wk{\Phi_0}, x_0}{\Phi_0, x_0} = \pos
  {x_i}~\lfs{\wk{\Phi_0}}{\Phi_0} = x_i$ \\ \indent \hfill since  $\pos {x_i}  [\wk{\bullet, x_n \ldots, x_1}/{\bullet, x_n \ldots, x_1}] = x_i$}
\prf{$\Gamma ; \Phi_0,  x_0{:}A_0, \wvec{y{:}B} \vdash x_i \equiv x_i  : A_i$ \hfill
using $A_i = [\wk{\bullet, x_n, \ldots, x_{i-1}} / \bullet, x_n, \ldots, x_{i-1}]A_i$}
\\[1em]
\SUBSTCLO{
\pcase{$\ianc
{\Gamma ; \Phi' \vdash \sigma' : \Phi, \wvec{y{:}A} \qquad \Gamma ; \Psi \vdash \sigma : \Phi' }
{\Gamma ; \Psi \vdash \sclo {\hatctx\Phi} {\unbox{\cbox{\hatctx {\Phi'} \vdash \sigma'}}{\sigma}}  \equiv   \lfs\sigma {\Phi'}~(\lfss {\sigma'}{\hatctx{\Phi},\vec x}  ~\wk\Phi) : \Phi}{}
$
}
\prf{$\pos x~(\sclo {\hatctx\Phi} {\unbox{\cbox{\hatctx {\Phi'}  \vdash \sigma'}}{\sigma}})  =  \lfs{\sigma}{\Phi'}(M) $ where $\trunc_\Psi ~ (\sigma' / \hatctx\Psi) = \sigma''~\mathrm{and}~
    \pos x~\lfs{\sigma''}{\Psi} = M$ }
\prf{$\Gamma ; \Phi' \vdash M : \lfs{\sigma''}{\Psi}(A)$ \hfill by LF Variable Lookup (Lemma \ref{lm:lflookup})}
\prf{$\Gamma ; \Psi \vdash \lfs{\sigma}{\Phi'}(M) : \lfs{\sigma}{\Phi'}(\lfs{\sigma''}{\Psi}(A))$ \hfill by LF subst. lemma}
\prf{$\Gamma ; \Psi \vdash \lfs{\sigma}{\Phi'}(M) : \lfs{\sclo {\hatctx\Phi} {\unbox{\cbox{\hatctx {\Phi'} \vdash \sigma'}}{\sigma}}  }{\Phi} A$ \hfill by LF subst. prop.}
\prf{$\pos~x~ (\lfss {\sigma'}{\hatctx{\Phi},\vec x}  ~\wk\Phi / \hatctx\Phi) = \pos~x~(\sigma''/\hatctx\Phi) = M$ where $\sigma'' =\trunc_\Psi ~ (\sigma' / \hatctx\Phi)$}
\prf{$\Gamma ; \Psi \vdash \lfs{\sclo {\hatctx\Phi} {\unbox{\cbox{\hatctx {\Phi'} \vdash \sigma'}}{\sigma}}}{\Phi}(x)  \equiv   \lfs\sigma {\Phi'}~(\lfs {\lfss {\sigma'}{\hatctx{\Phi},\vec x}  ~\wk\Phi}{\Phi}(x)) : \lfs{\sclo {\hatctx\Phi} {\unbox{\cbox{\hatctx {\Phi'} \vdash \sigma'}}{\sigma}}  }{\Phi} A$}
\\[1em]
\pcase{$\ianc
{r = \unbox{t}{\sigma} \qquad t \neq \cbox{\Phi' \vdash \sigma'}\qquad \Gamma ; \Phi \vdash \sigma : \Phi' \quad \Gamma \vdash t : \cbox{\Phi' \vdash \Psi, x_0{:}A_0,\wvec{x{:}A}}}
{\Gamma ; \Phi \vdash \sclo {\hatctx\Psi,x_0} r \equiv   \sclo{\hatctx\Psi} r, \unbox{\cbox{\hatctx\Psi,x_0 \vdash x_0}} {\sclo {\hatctx\Psi,x_0} r} : \Psi,x_0{:}A_0}{}
$}
\prf{$\pos~x_i~\lfs{\sclo {\hatctx\Psi,x_0} r }{\Psi} = \unbox{\cbox{\hatctx\Psi_i,x_i \vdash x_i}}{\sclo {\hatctx\Psi_i,x_i} r}$ where $\hatctx\Psi = \hatctx\Psi_i,x_i,\ldots,x_1$}
\prf{$\pos~x_i~ \lfs{\sclo{\hatctx\Psi} r, \unbox{\cbox{\hatctx\Psi,x_0 \vdash x_0}} {\sclo {\hatctx\Psi,x_0} r}}{\Psi,x_0}$}
\\
\prf{\emph{Subcase.}$x_i = x_0$}
\prf{$\pos~x_i~\lfs{\sclo {\hatctx\Psi,x_0} r }{\Psi} = \unbox{\cbox{\hatctx\Psi,x_0 \vdash x_0}}{\sclo {\hatctx\Psi,x_0} r}$}
\prf{$\pos~x_i~ \lfs{\sclo{\hatctx\Psi} r, \unbox{\cbox{\hatctx\Psi,x_0 \vdash x_0}} {\sclo {\hatctx\Psi,x_0} r}}{\Psi,x_0} = \unbox{\cbox{\hatctx\Psi,x_0 \vdash x_0}} {\sclo {\hatctx\Psi,x_0} r}$}
\prf{$\Gamma ; \Phi \vdash \lfs{ \sclo {\hatctx\Psi,x_0} r}{\Psi_0,x_0}(x)
     \equiv \lfs{\sclo{\hatctx\Psi} r, \unbox{\cbox{\hatctx\Psi,x_0 \vdash x_0}} {\sclo {\hatctx\Psi,x_0} r}}{\Psi_0, x_0} (x) : \lfs{ \sclo {\hatctx\Psi,x_0} r}{\Psi_0, x_0}(A)$}
\\
\prf{\emph{Subcase.}$x_i \neq x_0$}
\prf{$\pos~x_i~\lfs{\sclo {\hatctx\Psi,x_0} r }{\Psi} = \unbox{\cbox{\hatctx\Psi_i,x_i \vdash x_i}}{\sclo {\hatctx\Psi_i,x_i} r}$ where $\hatctx\Psi = \hatctx\Psi_i,x_i,\ldots,x_1$}
\prf{$\pos~x_i~ \lfs{\sclo{\hatctx\Psi} r, \unbox{\cbox{\hatctx\Psi,x_0 \vdash x_0}} {\sclo {\hatctx\Psi,x_0} r}}{\Psi,x_0} = \pos x_i~\lfs{\sclo{\hatctx\Psi} r}{\Psi} = \unbox{\cbox{\hatctx\Psi_i,x_i \vdash x_i}}{\sclo {\hatctx\Psi_i,x_i} r}
$}
}
\\[1em]
\pcase{$\ianc
{\Gamma ; \Psi \vdash \sigma = \sigma' : \Phi \quad
 \Gamma ; \Psi \vdash M \equiv N : \lfs \sigma \Phi A }
{\Gamma ; \Psi \vdash \sigma, M \equiv \sigma', N : \Phi, y{:}A}{}$
}
\prf{\emph{Subcase.} $x = y$}
\prf{$\pos x~\lfs{\sigma, M}{\Phi, x{:}A} = M$ \hfill by def. of $\pos{}{}$}
\prf{$\pos x~\lfs{\sigma', N}{\Phi, x{:}A} = N$ \hfill by def. of $\pos{}{}$}
\prf{$\Gamma ; \Psi \vdash M \equiv N : \lfs \sigma \Phi A$ \hfill by premise}
\prf{$\Gamma ; \Psi \vdash M \equiv N : \lfs {\sigma, M}{\Phi,x} A$ \hfill since $ \lfs {\sigma, M}{\Phi,x} A = \lfs \sigma \Phi A$ }
\\
\pcase{$\ianc{\Gamma ; \Psi \vdash \sigma' \equiv \sigma : \Phi}
             {\Gamma ; \Psi \vdash \sigma \equiv \sigma' : \Phi}{}$}
\prf{$\Gamma ; \Psi \vdash \lfs{\sigma'}{\Phi}(x) \equiv \lfs{\sigma}{\Phi}(x) : \lfs{\sigma'}{\Phi}A$ \hfill by IH}
\prf{$\Gamma ; \Psi \vdash \lfs{\sigma'}{\Phi}A \equiv \lfs{\sigma}{\Phi} A$ \hfill by IH}
\prf{$\Gamma ; \Psi \vdash \lfs{\sigma}\Phi (x_i) \equiv \lfs{\sigma'}\Phi  (x_i) :
     \lfs{\sigma'}\Phi (A)$ \hfill by type conversion}
\\[1em]
%
}
\end{proof}

\begin{lemma}[Equality Inversion]\label{lm:lfeqinv}
If $\Gamma ; \Psi \vdash A \equiv \Pi x{:}B_1.B_2 : \lftype$
or $\Gamma ; \Psi \vdash \Pi x{:}B_1.B_2 \equiv A : \lftype$
then $A = \Pi x{:}A_1.A_2$ for some $A_1$ and $A_2$
and $\Gamma ; \Psi \vdash A_1 \equiv B_1 : \lftype$
and $\Gamma ; \Psi, x{:}A_1 \vdash A_2 \equiv B_2 : \lftype$.
\end{lemma}
\begin{proof}
By induction on the definitional equality derivation.
\end{proof}

\begin{lemma}[Injectivity of LF Pi-Types]\label{lm:lfpi-inj}
If $\Gamma ; \Psi \vdash \Pi x{:}A.B \equiv \Pi x{:}A'.B' : \lftype$ then
$\Gamma ; \Psi \vdash A \equiv A' : \lftype$ and
$\Gamma ; \Psi, x{:}A \vdash B \equiv B' : \lftype$.
\end{lemma}
\begin{proof}
By equality inversion  (Lemma \ref{lm:lfeqinv}).
\end{proof}

\subsection{Elementary Properties of Computations}\label{sec:propcomp}
\begin{theorem}[Well-Formedness of Computation Context]~
  \label{lm:ctxwf}
  \begin{enumerate}
  \item \label{it:ctxwf} If $\D::~\vdash \Gamma, x{:}\ann\tau, \Gamma' $ then  $\Ca::~\vdash \Gamma $ and $\Ca$ is a sub-derivation of $\D$, i.e. $\Ca < \D$.
  \item \label{it:lfwf} If $\D::\Gamma;\Psi \vdash \JLF$ then $\Ca::~\vdash \Gamma $ and $\Ca$ is a sub-derivation of $\D$, i.e. $\Ca < \D$.
  \item \label{it:compwf} If $\D::\Gamma \vdash \Jcomp$ then $\Ca::~\vdash \Gamma $ and $\Ca$ is a sub-derivation of $\D$, i.e. $\Ca < \D$.
  \end{enumerate}
\end{theorem}
\begin{proof}
(\ref{it:ctxwf}) by induction on the structure of $\Gamma'$;
(\ref{it:lfwf}) and (\ref{it:compwf}) by mutual induction on $\D$.
\LONGVERSIONCHECKED{
\\[0.5em]
First statement:  If $\D::~\vdash \Gamma, x{:}\ann\tau, \Gamma' $ then  $\Ca::~\vdash \Gamma $ and $\Ca$ is a sub-derivation of $\D$, i.e. $\Ca < \D$.
\\[0.5em]
\pcase{$\Gamma' = \cdot$}
\prf{$\D ::~\vdash \Gamma,  x{:}\ann\tau$ \hfill by assumption}
\prf{$\Ca ::~\vdash \Gamma $ and $\Ca<\D$\hfill by inversion}
\\[-0.5em]
\pcase{$\Gamma' = \Gamma'',  y{:}\ann\tau'$}
\prf{$\D::~\vdash (\Gamma, x{:}\ann\tau, \Gamma'',  y{:}\ann\tau') $ \hfill by assumption}
\prf{$\D'::~\vdash \Gamma,  x{:}\ann\tau, \Gamma'' $ and $\D'<\D$ \hfill by inversion}
\prf{$\vdash \Gamma $ and $\Ca<\D$ \hfill by IH }
\\[0.5em]
For the 2nd and 3rd statement we show a few cases; most of the cases are straightforward and follow either directly by applying the induction hypothesis or by the premises of a rule. We only show one case.
\\[1em]

\pcase{$\D =\ianc{\Gamma \vdash t : \tau}{\Gamma \vdash t \equiv t : \tau}{}$}
\prf{$\Ca::~\vdash\Gamma$ and  $\Ca < \D$ \hfill by IH}

\pcase{$\D =\ianc{\Gamma, y{:}\ann\tau_1 \vdash t \equiv s : \tau_2}{\Gamma \vdash \tmfn y t \equiv \tmfn y s : (y:\ann\tau_1) \arrow \tau_2}{}$}
\prf{${\Ca'::}~\vdash\Gamma, y : \ann\tau_1$ and  $\Ca' < \D$ \hfill by IH}
\prf{${\Ca~::}~~\vdash\Gamma$ and $\Ca < \D$ \hfill by well-formed context rule}
}
\end{proof}


\begin{lemma}[Computation-level Weakening]\quad
  \label{lm:wkctx}
  \begin{enumerate}
  \item If $\Gamma_1,\Gamma_2 \vdash \Jcomp$ and
    $\vdash\Gamma_1,y:\ann\tau,\Gamma_2$  then $\Gamma_1,y:\tau, \Gamma_2 \vdash \Jcomp$
  \item If $\Gamma_1, \Gamma_2 ; \Psi \vdash \JLF$ and $\vdash\Gamma_1,y:\ann\tau,\Gamma_2$
       then $\Gamma_1,y:\ann\tau,\Gamma_2; \Psi \vdash \JLF$.
  \end{enumerate}
\end{lemma}
\begin{proof}
  Proof by  mutual induction exploiting Lemma \ref{lm:ctxwf}.
\LONGVERSIONCHECKED{
\\[1em]
  \pcase{\ianc{\vdash \Gamma_1,\Gamma_2}{\Gamma_1,\Gamma_2 \vdash \cdot : \ctx}{}}
  \prf{$\vdash\Gamma_1,y:\ann\tau,\Gamma_2$ \hfill by assumption}
  \prf{$\Gamma_1,y:\ann\tau,\Gamma_2 \vdash \cdot : \ctx$ \hfill by rule}
\\
  \pcase{$\D = $
\ibnc{\Gamma_1,\Gamma_2 , y':\ann\tau_1 \vdash t :  \tau_2}
   {\Gamma_1,\Gamma_2 \vdash (y':\ann\tau_1) \arrow \tau_2 : u}
   {\Gamma_1,\Gamma_2 \vdash \tmfn {y'} t : (y':\ann\tau_1) \arrow \tau_2}{}}
  \prf{$\vdash \Gamma_1, y:\ann\tau,\Gamma_2$ \hfill by assumption}
  \prf{$\vdash \Gamma_1,\Gamma_2, y':\ann\tau_1$ \hfill by Lemma \ref{lm:ctxwf}}
  \prf{$\Gamma_1,\Gamma_2\vdash \ann\tau_1 : u$ \hfill by inversion}
  \prf{$\Gamma_1, y :\ann\tau,\Gamma_2\vdash \ann\tau_1 : u$ \hfill by IH}
  \prf{$\vdash \Gamma_1, y :\ann\tau ,\Gamma_2, y': \ann\tau_1$ \hfill by rule}
  \prf{$\Gamma_1,y:\ann\tau,\Gamma_2 \vdash (y':\ann\tau_1) \arrow \tau_2 : u$ \hfill by IH}
  \prf{$\Gamma_1,y:\tau,\Gamma_2, y':\ann\tau_1 \vdash t : \tau_2$ \hfill by IH}
  \prf{$\Gamma_1,y:\tau,\Gamma_2 \vdash \tmfn y' t : (y':\ann\tau_1) \arrow \tau_2$ \hfill by rule}
\\
  \pcase{$\D =$ \ibnc{y' :\ann\tau' \in (\Gamma_1,\Gamma_2)}
                     {\vdash (\Gamma_1, \Gamma_2)}
                     {\Gamma_1,\Gamma_2 \vdash y' : \ann\tau'}{}}
  \prf{$y' :\ann\tau' \in (\Gamma_1,y:\ann\tau,\Gamma_2)$ \hfill by since $y' :\ann\tau' \in (\Gamma_1,\Gamma_2)$}
  \prf{$\vdash \Gamma_1,y:\ann\tau,\Gamma_2$ \hfill by assumption}
  \prf{$\Gamma_1,y:\ann\tau,\Gamma_2 \vdash y' : \ann\tau'$\hfill by rule}
}
\end{proof}

\begin{lemma}[Computation-level Substitution]~\label{lm:compsub}
  \begin{enumerate}
  \item 
    If $\vdash \Gamma, y :\ann\tau , \Gamma'$ and $\Gamma\vdash t : \ann\tau$ then $\vdash\Gamma, \{t/y\}\Gamma' $
    \item If~$~\Gamma, y:\ann\tau,\Gamma' ; \Psi \vdash \JLF$ and $~\Gamma \vdash t : \ann\tau$
      then $\Gamma, \{t/y\} \Gamma'; \Psi \vdash \{t/y\}\JLF$.
    \item If~$~\Gamma, y:\ann\tau, \Gamma' \vdash \Jcomp~$ and $~\Gamma \vdash t : \ann\tau$
      then $\Gamma, \{t/y\}\Gamma'\vdash \{t/y\}\Jcomp$.
    \end{enumerate}
\end{lemma}
\begin{proof}
  By mutual induction on the first derivation exploiting Lemma \ref{lm:wkctx}.
  \LONGVERSIONCHECKED{
\\
We show here a few cases. Most cases are straightforward and only
require us to apply the induction hypothesis.
    \\ [1em]
    Part 1.\\[0.5em]
    \pcase{$\Gamma' = \cdot$}
    \prf{$\vdash \Gamma, y:\ann\tau$ \hfill by assumption}
    \prf{$\vdash \Gamma$ \hfill by inversion}
    \\
    \pcase{$\Gamma' = \Gamma_0, x{:}\ann\tau_0$}
    \prf{$\vdash \Gamma, y{:}\ann\tau, \Gamma_0$ and
         $\Gamma, y{:}\ann\tau, \Gamma_0 \vdash \ann\tau_0 : u$ \hfill by inversion on assumption}
    \prf{$\vdash \Gamma, \{t/y\}\Gamma_0$ \hfill by IH (part 1)}
    \prf{$\Gamma, \{t/y\}\Gamma_0 \vdash \{t/y\}\ann\tau_0 : u$ \hfill by IH (part 2)}
    \prf{$\vdash \Gamma, \{t/y\}\Gamma_0, \{t/y\}\ann\tau_0$ \hfill by rule}
    \prf{$\vdash \Gamma, \{t/y\}(\Gamma_0, x{:}\ann\tau_0)$ \hfill by subst. def.}
    \\[0.5em]
    Part 2.
    \\[0.5em]
    \pcase{\ibnc{y'{:}\ann\tau' \in (\Gamma, y{:}\ann\tau,\Gamma')}
                {\vdash\Gamma, y{:}\ann\tau,\Gamma'}
                {\Gamma, y{:}\ann\tau,\Gamma' \vdash y' : \ann\tau'}{}  where $y \not=y'$}\\
    \emph{Subcase}: $y'{:}\ann\tau' \in \Gamma$.\\[0.5em]
    \prf{$\{t/y\}y' = y'$ \hfill by subst. def.}
    \prf{$\{t/y\}\ann\tau' = \ann\tau'$ \hfill by subst. def. and the fact that $y \not\in \FV(\ann\tau')$}
    \prf{$\vdash \Gamma, \{t/y\}\Gamma'$ \hfill by IH (part 1)}
    \prf{$y {:}\ann\tau' \in (\Gamma, \{t/y\}\Gamma')$ \hfill by previous lines}
    \prf{$\Gamma, \{t/y\}\Gamma' \vdash y' : \ann\tau'$ \hfill by rule}
    \\
    \emph{Subcase}: $y{:}\ann\tau' \in \Gamma'$.\\[0.5em]
    \prf{$\vdash \Gamma, \{t/y\}\Gamma'$ \hfill by IH (part 1)}
    \prf{$y'{:}\{t/y\}\ann\tau' \in \{t/y\}\Gamma'$ \hfill by previous lines}
    \prf{$y {:}\{t/y\}\ann\tau' \in (\Gamma, \{t/y\}\Gamma')$ \hfill by previous lines}
    \prf{$\Gamma, \{t/y\}\Gamma' \vdash y' : \{t/y\}\ann\tau'$ \hfill by rule}
    \\
\pcase{\ibnc{y{:}\ann\tau \in (\Gamma, y{:}\ann\tau,\Gamma')}{\vdash\Gamma, y{:}\ann\tau,\Gamma'}{\Gamma, y{:}\ann\tau,\Gamma' \vdash y : \ann\tau}{}}\\
    \prf{$\Gamma \vdash t : \ann\tau$ \hfill by assumption}
    \prf{$\vdash \Gamma, \{t/y\}\Gamma'$ \hfill by IH (part 1)}
    \prf{$\{t/y\}\ann\tau = \ann\tau$ \hfill by subst. def. and the fact that $y \not\in \FV(\tau)$}
    \prf{$\Gamma, \{t/y\}\Gamma' \vdash \{t/y\}y : \{t/y\}\ann\tau$ \hfill by  Lemma \ref{lm:wkctx}}
    \\
\pcase{\ibnc{\Gamma, y{:}\ann\tau, \Gamma' ,x{:}\ann\tau_1 \vdash t' : \tau_2 }{\Gamma,y: \ann\tau,\Gamma' \vdash (x: \ann\tau_1) \arrow \tau_2 : u}{\Gamma, y{:}\ann\tau,\Gamma' \vdash \tmfn x t' : (x{:}\ann\tau_1) \arrow \tau_2}{}}\\
    \prf{$\Gamma, \{t/y\}(\Gamma', x{:}\ann\tau_1) \vdash \{t/y\}t' : \{t/y\}\tau_2$ \hfill by IH (part 2)}
    \prf{$\Gamma, \{t/y\}\Gamma', x{:}\{t/y\}\ann\tau_1 \vdash \{t/y\}t' : \{t/y\}\tau_2$ \hfill by subst. def.}
    \prf{$\Gamma,\{t/y\}\Gamma' \vdash \{t/y\}((x: \ann\tau_1) \arrow \tau_2) : \{t/y\}u$ \hfill by IH (part 2)}
    \prf{$\Gamma,\{t/y\}\Gamma' \vdash (x:\{t/y\}\ann\tau_1) \arrow \{t/y,x/x\}\tau_2) : \{t/y\}u$ \hfill by subst. def.}
    \prf{$\Gamma, \{t/y\}\Gamma' \vdash \{t/y\}(\tmfn x t') : \{t/y\}((x{:}\ann\tau_1) \arrow \tau_2)$ \hfill by rule and subst. def.}
    \\
    \pcase{\ibnc{\Gamma,y{:}\ann\tau,\Gamma'  \vdash t' : (x: \ann\tau_1) \arrow \tau_2}{\Gamma,y{:}\ann\tau,\Gamma'\vdash s : \tau_1}{\Gamma,y{:}\tau,\Gamma' \vdash t'~s : \{s/x\}\tau_2}{}}
    \prf{$\Gamma, \{t/y\}\Gamma'\vdash \{t/y\}t' : (x:\{t/y\}\ann\tau_1) \arrow \{t/y,~x/x\}\tau_2$ \hfill by IH (part 2) and definition of substitution}
    \prf{$\Gamma,\{t/y\}\Gamma'\vdash \{t/y\}s : \{t/y\}\ann\tau_1$ \hfill by IH (part 2)}
    \prf{$\Gamma,\{t/y\}\Gamma' \vdash (\{t/y\}t')~(\{t/y\}s) : \{\{t/y\}s/x\}(\{t/y, x/x\}\tau_2)$ \hfill by rule}
    \prf{$\Gamma,\{t/y\}\Gamma' \vdash \{t/y\}(t'~s) : \{t/y\}(\{s/x\}\tau_2)$ \hfill by definition and composition rules of substitution}
  }
\end{proof}




Last, we define simultaneous computation-level substitution using the judgment \fbox{$\Gamma'\vdash  \theta : \Gamma$}. For simplicity, we overload the typing judgment simply writing $\Gamma \vdash t : \ann\tau$, although if $\ann\tau = \tmctx$, $t$ stands for a LF context.

  \[
  \begin{array}{c}
    \infer{\Gamma' \vdash \cdot : \cdot}{\vdash \Gamma'}
    \quad
    \infer{\Gamma'\vdash\theta, t / x :\Gamma, x: \ann\tau}{\Gamma'\vdash \theta : \Gamma & \Gamma' \vdash t : \{\theta \}\ann\tau }
  \end{array}
  \]

We distinguish between a substitution $\theta$ that provides
instantiations for variables declared in the
computation context $\Gamma$, and a renaming substitution $\rho$ which maps variables  in the computation context $\Gamma$ to the same variables in the context $\Gamma'$ where $\Gamma'
= \Gamma, \wvec{x{:}\ann\tau}$ and $\Gamma' \vdash \rho : \Gamma$. We
write $\Gamma' \leq_\rho \Gamma$ for the latter.  We note that the substitution properties also hold for renamings.

\begin{lemma}[Well-Formed Contexts for Substitutions]
  \label{lm:wfcsub}
  If\/$\Gamma'\vdash\theta :\Gamma$ then $\vdash \Gamma'$.
\end{lemma}

\begin{proof}
  By induction on the structure of the derivation of $\Gamma'\vdash\theta : \Gamma$.
\LONGVERSIONCHECKED{
\\[1em]
\pcase{$\ianc{\vdash \Gamma'}{\Gamma' \vdash \cdot : \cdot}{}$}
\prf{$\vdash \Gamma'$ \hfill by premise}

\pcase{$\ibnc{\Gamma'\vdash \theta :\Gamma}
             {\Gamma'\vdash t : \{\theta \}\ann\tau}
             {\Gamma'\vdash\theta, t / x :\Gamma, x: \ann\tau}{}
$}
\prf{$\vdash \Gamma'$ \hfill by IH}
}
\end{proof}

\begin{lemma}[Weakening for computation-level substitutions]\label{lm:wkwksub}
  Let $y$ be a new name s.t. $y \notin dom(\Gamma')$. If $\Gamma' \vdash \theta : \Gamma$.
 and $\Gamma' \vdash \ann\tau : u$ then $\Gamma', y : \ann\tau \vdash \theta : \Gamma$.
\end{lemma}
\begin{proof}
By induction on the first derivation using Lemma \ref{lm:wkctx}.
\LONGVERSIONCHECKED{\\[1em]
\pcase{\ianc{\vdash \Gamma'}{\Gamma' \vdash \cdot : \cdot}{}}
\prf{$\vdash \Gamma'$ \hfill by premise}
\prf{$\Gamma' \vdash \ann\tau : u$ \hfill by assumption}
\prf{$\vdash \Gamma', y : \ann\tau$ \hfill by rule}
\prf{$\Gamma', y : \ann\tau \vdash \cdot : \cdot$ \hfill by rule}

\pcase{\ibnc{\Gamma'\vdash \theta : \Gamma}
             {\Gamma'\vdash t : \{\theta\}\ann\tau' }
             {\Gamma'\vdash \theta, t/ x :\Gamma, x: \ann\tau'}{}
}
\prf{$\Gamma', y: \ann\tau \vdash \theta : \Gamma$ \hfill by IH}
\prf{$\Gamma', y: \ann\tau \vdash t : \{\theta\}\ann\tau'$ \hfill by Lemma \ref{lm:wkctx}}
\prf{$\Gamma', y: \ann\tau \vdash \theta, t/ x :\Gamma, x: \ann\tau'$ \hfill by rule}
}
\end{proof}

\LONGVERSION{
\begin{corollary}[Identity Extension of Computation-level Substitution]
  \label{lm:lemaux2}
Let $y$ be a new name s.t. $y \notin dom(\Gamma')$ and $y \notin dom(\Gamma)$.
If\/$~\Gamma' \vdash \theta : \Gamma$
and $\Gamma' \vdash \{\theta\}\ann\tau : u$ then
   $\Gamma', y : \{\theta\}\ann\tau \vdash \theta,y/y :  \Gamma, y{:}\ann\tau$.
\end{corollary}
\SHORTVERSION{\smallskip}
\LONGVERSIONCHECKED{
\begin{proof}~\\
  \prf{$\Gamma', y:\{\theta\}\ann\tau \vdash \theta : \Gamma$ \hfill by Lemma \ref{lm:wkwksub} }
  \prf{$\vdash \Gamma', y:\{\theta\}\ann\tau$ \hfill by Lemma \ref{lm:wfcsub}}
  \prf{$\Gamma', y:\{\theta\}\ann\tau \vdash y:\{\theta\}\ann\tau $ \hfill by typing rule}
\prf{$\Gamma', y:\{\theta\}\ann\tau \vdash \theta, y/ y :\Gamma, y: \ann\tau$ \hfill by rule}
\end{proof}
}
}

\begin{lemma}[Computation-level Simultaneous Substitution]\label{lm:compsubst}~
  \begin{enumerate}
  \item If $\Gamma' \vdash \theta : \Gamma$ and $\Gamma; \Psi \vdash \JLF$ then $\Gamma' ; \{\theta\}\Psi \vdash \{\theta\}\JLF$.
  \item If $\Gamma' \vdash \theta : \Gamma$ and $\Gamma \vdash \Jcomp$ then $\Gamma' \vdash \{\theta\}\Jcomp$.
  \end{enumerate}
\end{lemma}
\begin{proof}
By mutual induction on the second derivation using Lemma \ref{lm:ctxwf} and  Lemma \ref{lm:wkwksub}.
\LONGVERSIONCHECKED{
\\[0.5em]
\pcase{$\ibnc{x{:} \ann\tau \in \Gamma}{\vdash \Gamma}{\Gamma \vdash x : \ann\tau}{}$}
\prf{$\Gamma' \vdash \theta : \Gamma$ \hfill by assumption}
\prf{$\Gamma' \vdash \theta_0, t/x : \Gamma_0, x{:}\ann\tau$ and $\theta = \theta_0,~t/x,~\theta_1$ and $\Gamma = \Gamma_0, x{:}\ann\tau, \Gamma_1$ \hfill by inversion}
\prf{$\Gamma' \vdash t : \{\theta_0\}\ann\tau$ \hfill by inversion}
\prf{$\Gamma' \vdash t : \{\theta\}\ann\tau$ \hfill since $\ann\tau$ does not depend on the variable in $(x{:}\ann\tau,\Gamma_1)$}
\prf{$\Gamma' \vdash \{\theta\}x : \{\theta\}\ann\tau$ \hfill since $\{\theta\}x = t$ and $t$ does not depend on the variable in $(x{:}\ann\tau,\Gamma_1)$}

\pcase{$\ibnc{\Gamma \vdash t : (y: \ann\tau_1) \arrow \tau_2 }
              {\Gamma \vdash s : \ann\tau_1}
              {\Gamma \vdash t~s : \{s/y\}\tau_2}{}
$}
\prf{$\Gamma' \vdash \{\theta\}s : \{\theta\}\ann\tau_1$ \hfill by IH}
\prf{$\Gamma' \vdash \{\theta\}t : \{\theta\}((y: \ann\tau_1) \arrow \tau_2)$ \hfill by IH}
\prf{$\Gamma' \vdash \{\theta\}t : (y: \{\theta\}\ann\tau_1) \arrow \{\theta,y/y\}\tau_2$ \hfill by subst. definition}
\prf{$\Gamma' \vdash (\{\theta\} t)~(\{\theta\}s) : \{\{\theta\}s/y\}(\{\theta,y/y\}\tau_2)$ \hfill by rule}
\prf{$\Gamma' \vdash \{\theta\}(t~s) : \{\{\theta\}s/y\}(\{\theta,y/y\}\tau_2)$ \hfill by subst. definition}
\prf{$\Gamma' \vdash \{\theta\}(t~s) :\{\theta\}(\{s/y\}\tau_2)$ \hfill by compositionality of substitution}

\pcase{$\ianc{\Gamma, y: \ann\tau_1 \vdash t : \tau_2}
             {\Gamma \vdash \tmfn y t : (y: \ann\tau_1) \arrow \tau_2}{}$}
\prf{$\Gamma' \vdash \theta : \Gamma$ \hfill by assumption}
\prf{$\vdash \Gamma, y : \ann\tau_1$ \hfill by Lemma \ref{lm:ctxwf}}
\prf{$\Gamma \vdash \ann\tau_1 : u$ \hfill by inversion}
\prf{$\Gamma' \vdash\{\theta\}\ann\tau_1 : u$ \hfill by IH}
\prf{$\Gamma', y : \{\theta\}\ann\tau_1 \vdash \theta, y/y : \Gamma, y: \ann\tau_1$ \hfill by Lemma \ref{lm:lemaux2}}
\prf{$\Gamma', y:\{\theta\}\ann\tau_1 \vdash \{\theta,y/y\}t : \{\theta,y/y\}\tau_2$ \hfill by IH}
\prf{$\Gamma' \vdash \tmfn y {\{\theta,y/y\}t} : (y:\{\theta\}\ann\tau_1) \arrow \{\theta,y/y\}\tau_2$ \hfill by rule}
\prf{$\Gamma' \vdash \{\theta\}(\tmfn y t) : \{\theta\}((y: \ann\tau_1) \arrow \tau_2)$ \hfill by subst. definition}


\pcase{\ianc{ \Gamma, \psi:\tmctx, p:\cbox{~\unboxc{\psi} \vdash_\# \tm}  \vdash t_v : \{p / y\}\tau}{\Gamma \vdash ({\psi,p \mto t_v}) : \IH }{} \\[0.5em] where $\IH = (\psi : \tmctx) \arrow (y:\cbox{\psi \vdash \tm}) \arrow \tau$}
\prf{$\Gamma'\vdash \theta : \Gamma$ \hfill by assumption}
\prf{$\vdash \Gamma, \psi:\tmctx, p:\cbox{~\unboxc{\psi} \vdash_\# \tm}$\hfill by Lemma \ref{lm:ctxwf}}
\prf{$\vdash \Gamma, \psi:\tmctx$\hfill by inversion}
\prf{$\Gamma, \psi:\tmctx \vdash \cbox{~\unboxc{\psi} \vdash_\# \tm}$\hfill by inversion}
\prf{$\Gamma'  \vdash \tmctx : u$\hfill by typing rule}
\prf{$\Gamma' \vdash \{\theta\}\tmctx : u$ \hfill by subst. def. since $\{\theta\}\tmctx = \tmctx$}
\prf{$\Gamma',\psi:\tmctx \vdash \theta, \psi/\psi : \Gamma, \psi:\tmctx$ \hfill by Lemma \ref{lm:lemaux2}}
\prf{$\Gamma', \psi:\tmctx \vdash \cbox{~\unboxc{\psi} \vdash_\# \tm} : u$\hfill by typing rules}
\prf{let $\theta_v = \theta, \psi/\psi, p/p$ \hfill \\
$\Gamma',\psi:\tmctx, p : \cbox{~\unboxc{\psi} \vdash_\# \tm} \vdash \theta_v : \Gamma, \psi:\tmctx, p:  \cbox{~\unboxc{\psi} \vdash_\# \tm}$ \hfill by Lemma \ref{lm:lemaux2}}
\prf{$\Gamma', \psi:\tmctx, p:\cbox{~\unboxc{\psi} \vdash_\# \tm}  \vdash \{\theta_v\} t_v : \{p/ y\}\{\theta_v\}\tau$ \hfill by IH and by definition of substitution}
\prf{$\Gamma' \vdash \{\theta\} ({\psi,p \mto t_v}) : \{\theta\}\IH$ \hfill by rule and by definition of substitution}
 }
\end{proof}

\LONGVERSION{
Next, we show that we can always extend a renaming substitution.

\begin{lemma}[Weakening of Renaming Substitutions]\label{lm:wkwksub2}
Let $y$ be a new name s.t. $y \notin dom(\Gamma')$.\\\mbox{\qquad}
If\/$\Gamma' \leq_\rho  \Gamma$ and $\Gamma' \vdash \tau : u$ then
   $\Gamma', y : \tau \leq_\rho  \Gamma$.
\end{lemma}
\begin{proof}
Follows from Lemma \ref{lm:wkwksub}.
\end{proof}

\begin{corollary}[Identity Extension of Renaming Computation-level Substitution]
  \label{lm:weaklemaux2}
Let $y$ be a new name s.t. $y \notin dom(\Gamma')$ and $y \notin dom(\Gamma)$.\\
If\/$\Gamma' \leq_{\rho}  \Gamma$ and $\Gamma' \vdash \{\rho\}\ann\tau : u$ then
   $\Gamma', y {:} \{\rho\}\ann\tau, \Gamma' \leq_{\rho,y/y}  \Gamma, y {:}\ann\tau $.
\end{corollary}
\LONGVERSIONCHECKED{
\begin{proof}~\\
\prf{$\Gamma', y:\{\rho\}\ann\tau \leq_{\rho} : \Gamma$ \hfill by Lemma \ref{lm:wkwksub2} }
\prf{$\Gamma', y:\{\rho\}\ann\tau \leq_{\rho, y/y} :\Gamma, y: \ann\tau$ \hfill by rule}
\end{proof}
}
\medskip

\begin{lemma}[Computation-level Renaming Lemma]~\label{lem:weakcomp}
  \begin{enumerate}
  \item If\/ $\Gamma' \leq_\rho  \Gamma$ and $\Gamma; \Psi \vdash \JLF$ then $\Gamma' ; \{\rho\}\Psi \vdash \{\rho\}\JLF$.
  \item If\/ $\Gamma' \leq_\rho  \Gamma$ and $\Gamma \vdash \Jcomp$ then $\Gamma' \vdash \{\rho\}\Jcomp$.
  \end{enumerate}
\end{lemma}
\begin{proof}
By induction on the second derivation using Lemma \ref{lm:ctxwf}
\LONGVERSIONCHECKED{We show a few cases.
\\[1em]
\pcase{$\D = \ibnc{x{:}\ann\tau \in \Gamma}{\vdash \Gamma}{\Gamma \vdash x : \ann\tau}{}$}
\prf{$\Gamma' \leq_\rho  \Gamma$ \hfill by assumption}
\prf{$\Gamma' \leq_{\rho_0, x/x}  \Gamma_0, x{:}\ann\tau$ and $\rho = \rho_0, x/x, \rho_1$ and $\Gamma = \Gamma_0, x{:}\ann\tau, \Gamma_1$ \hfill by inversion}
\prf{$\Gamma' \vdash x : \{\rho_0\}\ann\tau$ \hfill by inversion}
\prf{$\Gamma' \vdash x : \{\rho\}\ann\tau$ \hfill since $\tau$ does not depend on the variable in $(x{:}\ann\tau,\Gamma_1)$}
\prf{$\Gamma' \vdash \{\rho\}x : \{\rho\}\ann\tau$ \hfill since $\{\rho\}x = y$ }

\pcase{$\D = \ibnc{\Gamma \vdash t : (y: \ann\tau_1) \arrow \tau_2 }
              {\Gamma \vdash s : \ann\tau_1}
              {\Gamma \vdash t~s : \{s/y\}\tau_2}{}
$}
\prf{$\Gamma' \vdash \{\rho\}s : \{\rho\}\ann\tau_1$ \hfill by IH}
\prf{$\Gamma' \vdash \{\rho\}t : \{\rho\}((y: \ann\tau_1) \arrow \tau_2)$ \hfill by IH}
\prf{$\Gamma' \vdash \{\rho\}t : (y: \{\rho\}\ann\tau_1) \arrow \{\rho,y/y\}\tau_2$ \hfill by subst. definition}
\prf{$\Gamma' \vdash (\{\rho\} t)~(\{\rho\}s) : \{\{\rho\}s/y\}(\{\rho,y/y\}\tau_2)$ \hfill by rule}
\prf{$\Gamma' \vdash \{\rho\}(t~s) : \{\{\rho\}s/y\}(\{\rho,y/y\}\tau_2)$ \hfill by subst. definition}
\prf{$\Gamma' \vdash \{\rho\}(t~s) :\{\rho\}(\{s/y\}\tau_2)$ \hfill by compositionality of substitution}

\pcase{$\D = \ianc{\Gamma, y: \ann\tau_1 \vdash t : \tau_2}
             {\Gamma \vdash \tmfn y t : (y: \ann\tau_1) \arrow \tau_2}{}$}
\prf{$\Gamma' \leq_\rho  \Gamma$ \hfill by assumption}
\prf{$\Ca::\quad\vdash \Gamma, y : \ann\tau_1$ and moreover $\Ca$ is smaller than $\D$  \hfill by Lemma \ref{lm:ctxwf}}
\prf{$\Gamma \vdash \ann\tau_1 : u$ \hfill by inversion}
\prf{$\Gamma' \vdash\{\rho\}\ann\tau_1 : u$ \hfill by IH}
\prf{$\Gamma', y : \{\rho\}\ann\tau_1\leq_{\rho, y/y}  \Gamma, y: \ann\tau_1$ \hfill by Lemma \ref{lm:wkwksub2}}
\prf{$\Gamma', y:\{\rho\}\ann\tau_1 \vdash \{\rho,y/y\}t : \{\rho,y/y\}\tau_2$ \hfill by IH}
\prf{$\Gamma' \vdash \tmfn y {\{\rho,y/y\}t} : (y:\{\rho\}\ann\tau_1) \arrow \{\rho,y/y\}\tau_2$ \hfill by rule}
\prf{$\Gamma' \vdash \{\rho\}(\tmfn y t) : \{\rho\}((y: \ann\tau_1) \arrow \tau_2)$ \hfill by subst. definition}


\pcase{$\D = {\Gamma \vdash \tmrec {\IH} {\psi,p \mto t_v} {\psi, m, n, f_n, f_m \mto t_{\mathsf{app}}} {\psi, m, f_m \mto t_{\clam}} t : \{[\Psi]/\psi, t/y\}\tau}$}
\prf{where $\IH =  (\psi : \cbox{\tmctx}) \arrow (m:\cbox{\psi \vdash \tm}) \arrow \tau$.}\\[0.15em]
\prf{$\Gamma \vdash t : \cbox{\Psi \vdash \tm}$ }
\prf{$\Gamma, \psi{:}\tmctx,
       m{:}\cbox{~\psi, x{:}\tm \vdash \tm},
       f_m{:}\{\cbox{~\psi, x{:}\tm}/\psi, m /y \} \tau
      \vdash t_{\clam} : \{\psi/\psi, \cbox{~\psi \vdash \clam~\lambda x.\unbox m {\id}} / y\}\tau$}
\prf{$\Gamma, \psi{:}\tmctx,
      m{:}\cbox{~\psi \vdash \tm},
      n{:}\cbox{~\psi \vdash \tm},
    f_m{:} \{m/y\}\tau, f_n: \{n/y\}\tau  \vdash
    t_{\mathsf{app}} {:} \{\cbox{~\psi \vdash\capp~\unbox m{\id}~\unbox n{\id}}/y\}\tau$
}
\prf{$\Gamma, \psi{:}\tmctx, p{:}\cbox{~\psi \vdash_\# \tm} \vdash t_v : \{p / y\}\tau$
     \hfill by premise}
\\
\prf{ $\Gamma' \vdash \{\rho\} t : \{\rho\}\cbox{\Psi \vdash \tm}$ \hfill by IH}
\prf{$\Gamma' \vdash \{\rho\} t : \cbox{\{\rho\}\Psi \vdash \tm}$ \hfill by subst. definition}
\\
\prf{$\Gamma' \vdash \tmctx$ \hfill by rule }
\prf{$\Gamma' , \psi:\tmctx \vdash \psi : \tmctx$ \hfill by rule}
\prf{$\Gamma' , \psi:\tmctx \vdash \psi : \ctx$ \hfill by rule}
\prf{$\Gamma' , \psi:\tmctx ; \psi \vdash\tm : \lftype$ \hfill by rule}
\prf{$\Gamma' , \psi:\tmctx \vdash \psi \vdash_\# \tm : \lftype$ \hfill by rule}
\prf{$\Gamma', \psi:\tmctx \vdash \cbox{~\psi \vdash_\# \tm} : u$ \hfill by rule}
\prf{$\Gamma', \psi:\tmctx,p :\cbox{~\psi \vdash_\# \tm} \leq_{(\rho,\psi/\psi,p/p)}   \Gamma, \psi:\tmctx, m:\cbox{~\psi, x{:}\tm \vdash \tm}$
     \hfill by Lemma \ref{lm:wkwksub2}}
\prf{$\Gamma', \psi:\tmctx, p:\cbox{~\psi \vdash_\# \tm}  \vdash  \{\psi/\psi,p/p\}t_v : \{\rho,\psi/\psi, p / y\}\tau$ \hfill by IH (since $\tau$ does not depend on $\Gamma'$)}
%
\prf{$\Gamma' \vdash \tmctx$ \hfill by rule }
\prf{$\Gamma' , \psi:\tmctx \vdash \psi : \tmctx$ \hfill by rule}
\prf{$\Gamma' , \psi:\tmctx \vdash \psi : \ctx$ \hfill by rule}
\prf{$\Gamma' , \psi:\tmctx \vdash \psi, x{:}\tm :\ctx$ \hfill by rule}
\prf{$\Gamma' , \psi:\tmctx ; \psi, x{:}\tm \vdash\tm : \lftype$ \hfill by rule}
\prf{$\Gamma' , \psi:\tmctx \vdash \cbox{~\psi, x{:}\tm \vdash \tm} : u$ \hfill by rule}
\prf{$\Gamma', \psi:\tmctx,
       m:\cbox{~\psi, x:\tm \vdash \tm}
       \vdash \{\cbox{\psi, x{:}\tm}/\psi, m /y \} \tau : u
       $ \hfill by assumption and Lemma \ref{lm:ctxwf}}
\prf{$\Gamma', \psi{:}\tmctx,
               m{:}\cbox{~\psi, x:\tm \vdash \tm}
       \leq_{(\rho, \psi/\psi, m/m)}
       \Gamma, \psi{:}\tmctx, m{:}\cbox{\psi, x{:}\tm \vdash \tm}$
       \hfill by Lemma \ref{lm:wkwksub2}
     }
\prf{$\Gamma', \psi:\tmctx,
       m:\cbox{~\psi, x:\tm \vdash \tm}
       \vdash
       \{\psi/\psi, m/m\}
        (\{\cbox{~\psi, x{:}\tm}/\psi, m /y \} \tau) : u$
\hfill \\
\mbox{\quad}\hfill by IH (since $\tau$ does not depend of $\Gamma'$)
 }
\prf{$\Gamma', \psi{:}\tmctx,
       m{:}\cbox{~\psi, x:\tm \vdash \tm},
       f_m{:} \{\cbox{~\psi, x{:}\tm}/\psi, m /y \} \tau\\\mbox{\qquad}
       \leq_{(\rho, \psi/\psi, m/m, f_m / f_m)}     \Gamma, \psi{:}\tmctx,
       m{:}\cbox{~\psi, x:\tm \vdash \tm},
       f_m{:}\{\cbox{~\psi, x{:}\tm}/\psi, m /y \} \tau
        $ \hfill by Lemma \ref{lm:wkwksub2}}
\prf{$\Gamma', \psi:\tmctx,
       m:\cbox{~\psi, x:\tm \vdash \tm},
       f_m: \{\cbox{~\psi, x{:}\tm}/\psi, m /y \} \tau$}
\prf{\mbox{\hspace{2cm}}~~~~~$\hfill
       \vdash
        \{\rho,  \psi/\psi, m/m, f_m/f_m\} t_{\clam} :
       \{\psi/\psi, \cbox{~\psi \vdash \clam~\lambda x.\unbox m {\id}} / y\}\tau$}
\prf{\mbox{\hspace{1cm}}\hfill by IH (since $\tau$ does not depend on $\Gamma'$)
}
\\
        With a similar argument, we have :\\
 $\Gamma', \psi:\tmctx,
           m:\cbox{~\psi \vdash \tm},
           n:\cbox{~\unbox{\psi} \vdash \tm},
           f_m: \{m/y\}\tau, f_n: \{n/y\}\tau$\\
\mbox{\hspace{1cm}}\hfill $
  \leq_{\rho, \psi/\psi, m/m, n/n}
  \Gamma, \psi:\tmctx, m:\cbox{~\psi \vdash \tm}, n:\cbox{~\psi \vdash \tm}, f_m: \{m/y\}\tau, f_n: \{n/y\}\tau$
\\[1em]
So we have by IH: \\
$\Gamma', \psi{:}\tmctx, m{:}\cbox{\psi \vdash \tm}, n{:}\cbox{\psi \vdash \tm}, f_m{:} \{m/y\}\tau, f_n{:} \{n/y\}\tau
\vdash \{\rho\} t_{\mathsf{app}} : \{\cbox{~\psi \vdash\capp~\unbox m {\id}~\unbox n{\id}}/y\}\tau$
\\[1em]
      So we have :\\
$\Gamma' \vdash \{\rho\} \tmrec {\IH} {\psi,p \mto t_v} {\psi, m, n, f_n, f_m \mto t_{\mathsf{app}}} {\psi, m, f_m \mto t_{\clam}} t : \{\rho\}(\{\cbox{\Psi}/\psi, t/y\}\tau)$

}
\end{proof}
}


\newcommand{\smallerderiv}[1]{{\color{blue}{#1}}}
\newcommand{\redundant}[1]{{\color{orange}{#1}}}
\newcommand{\unsure}[1]{{\color{red}{#1}}}

\section{Weak head reduction}\label{sec:whred}
Before we define the operational semantics of \cocon using weak head reduction,
we characterize weak head normal forms for both, (contextual) LF and
computations (\LONGVERSION{Fig.~\ref{fig:lfwhnf} and} Fig.~\ref{fig:whnf}). They are mutually defined.
Any computation-level term $t$ that is unboxed within an LF object (i.e.\/$\unbox t \sigma$) must be neutral and not further trigger any reductions. We leave the substitution $\sigma$ that is associated with it untouched. LF substitutions are in whnf.
%


\begin{figure}[htb]
  \centering
  \[
  \begin{array}{c}
\multicolumn{1}{l}{\fbox{$\norm M$}:\mbox{LF term $M$ is in weak head normal;~}\quad~\fbox{$\neut M$}:\mbox{LF term $M$ is in weak head neutral}}\\[0.5em]
    \infer{\norm \lambda x.M}{}
    \quad
\infer{\norm \unbox t \sigma}{\neut t}
 \quad 
 \infer{\norm M}{\neut M}{} \quad
 \infer{\neut (\clam M)}{}
 \quad
 \infer{\neut (\capp M~N)}{}
 \quad
 \infer{\neut x}{}
\\[0.5em]
\multicolumn{1}{l}{\fbox{$\norm \sigma$}:~\mbox{LF substitution $\sigma$ is in weak head normal form}}\\[0.5em]
\SUBSTCLO{\infer{\norm (\sclo {\psi} {\unbox t \sigma})}{\neut t}
\quad
\infer{\norm \sigma}{\neut \sigma}
}
\qquad
\infer{\norm (\sigma, M)}{}
\quad
 \infer{\norm \cdot}{}
 \qquad
 \infer{\norm {\wk{\psi}}}{}
\LONGVERSION
{\\[1em]
\multicolumn{1}{l}{\fbox{$\norm A$}:~\mbox{LF type $A$ is in weak head normal form}}\\[1em]
\infer{\norm{\const{a}~M_1 \ldots M_n}}{}
\qquad
\infer{\norm{\Pityp x A B}}{\norm A & \norm B}
}
  \end{array}
\]
  \caption{Normal and Neutral LF Terms \LONGVERSION{and LF Types}}
\label{fig:lfwhnf}
\end{figure}

LF types are always considered to be in whnf, as computation may only produce a contextual LF term, but not a contextual LF type.
LF contexts are in weak head normal form, as we do not allow
the embedding of computations that compute a context.
Similarly, the erased context, described by $\hatctx{\Psi}$, is in weak
head normal form. Further, contextual objects $C$ and types $T$ are considered to be in weak head normal form.

Computation-level expressions are in weak head normal form, if they do not trigger any further computation-level reductions. We consider boxed objects, i.e. $\cbox{C}$, and boxed types, i.e. $\cbox{T}$, in whnf. The contextual object $C$ will be further reduced when we use them and have to unbox them.
A computation-level term $t$ is neutral (i.e.~$\neut t$) if it cannot be reduced further because it contains free variables. Variables are
neutral, applications $t~s$ are neutral, if $t$ is neutral, and $\tmrec \IH {\psi, p \mto t_v} {\psi, m, n, f_m, f_n \mto
  t_{\mathsf{app}}} {\psi, m, f_m \mto t_{\mathsf{lam}}} \rappto \Psi~t$ is neutral, if $t$ is neutral.
We note that weakening preserves weak head normal forms.

\begin{figure}[htb]
  \centering
  \[
  \begin{array}{c}
\multicolumn{1}{l}{\fbox{$\norm t$}:~\mbox{Computation (type) $t$ is in weak head normal form}}\\[0.5em]
\infer{\norm \cbox T}{}
\quad
\infer{\norm (y : \tau_1) \arrow \tau_2}{}
\quad
\infer{\norm u}{}
\quad
\infer{\norm (\tmfn y t)}{} \qquad
\infer{\norm{\cbox{C}}}{} \qquad
\infer{\norm t}{\neut t}
\\[0.75em]
\infer{\neut y}{} \qquad
\infer{\neut (t~s)}{\neut t}
\qquad
\infer{\neut (\trec{\R}{\tm}{\IH}\rappto (\Psi)~s)}{\neut s}
\quad
\infer{\neut
    (\trec{\R}{\tm}{\IH}~\rappto (\Psi)~\cbox{\hatctx{\Psi}\vdash r})}
{r = \unbox t \sigma \quad \neut t}
\quad
  \end{array}
\]
  \caption{Normal and Neutral Computations}
\label{fig:whnf}
\end{figure}

\LONGVERSION{
\begin{lemma}[Weakening preserves head normal forms]\quad\label{lem:weaknorm}\\
For LF Terms:
    \begin{enumerate}
    \item If $\norm M$ and $\Gamma' \leq _\rho \Gamma$ and
          $\Gamma ; \Psi \vdash M : A$ then $\norm \{\rho\}M$.
    \item If $\neut M$ and $\Gamma' \leq _\rho \Gamma$ and $\Gamma ; \Psi \vdash M : A$ then $\neut \{\rho\}M$.
    \end{enumerate}
For LF Substitutions
    \begin{enumerate}
    \item If $\norm \sigma$ and $\Gamma' \leq _\rho \Gamma$ and
          $\Gamma ; \Psi \vdash \sigma : \Phi$ then $\norm \{\rho\}\sigma$.
    \item If $\neut \sigma$ and $\Gamma' \leq _\rho \Gamma$ and $\Gamma ; \Psi \vdash \sigma : \Phi$ then $\neut \{\rho\}\sigma$.
    \end{enumerate}
%
For Computations:
    \begin{enumerate}
    \item  If $\norm t$ and $\Gamma' \leq _\rho \Gamma$ and
           $\Gamma  \vdash t : \tau$ then $\norm \{\rho\}t$.
    \item  If $\neut t$ and $\Gamma' \leq _\rho \Gamma$ and $\Gamma \vdash t : \tau$ then $\neut\{\rho\}t$.
    \end{enumerate}
\end{lemma}
\begin{proof}
By induction on the first derivation.

\end{proof}
}

We define weak head reductions
 for LF in Fig.~\ref{fig:lfwhnfred} and for computations in Fig.~\ref{fig:whnf}. If an LF term is not already in $\norm$ form, we have two cases: either we encounter an LF application $M~N$ and we may need to beta-reduce or  $M$ reduces to $\unbox t \sigma$. If $t$ is neutral, then we are done; otherwise $t$ reduces to a contextual object $\cbox{\hatctx{\Psi} \vdash M}$, and we continue to reduce $[\sigma / \hatctx \Psi]M$.

\begin{figure}[htb]
  \centering
  \[
    \begin{array}{c}
\multicolumn{1}{l}{\fbox{$M \lfwhnf N$}:~~\mbox{LF Term $M$ weak head
      reduces to $N$ s.t. $\norm N$}}
\\[1em]
\infer{M \lfwhnf M}
{\norm M}
\quad
\infer{M\;N \lfwhnf R}
{M \lfwhnf \lambda x.M' & [N/x]M' \lfwhnf R}
\qquad
\infer{M~N \lfwhnf R \;N}{
M \lfwhnf R & \neut R}
 \\[0.5em]
\infer{\unbox t{\sigma} \lfwhnf N}
      {t \whnf \cbox{\hatctx{\Psi} \vdash M}\quad
       \lfs {\sigma}{\Psi} M \lfwhnf N}
\quad
\infer{\unbox t{\sigma} \lfwhnf \unbox{n}{\sigma}}
      {t \whnf n \quad \neut {n}}
%
\\[0.5em]
 \multicolumn{1}{l}{\fbox{$\sigma \lfwhnf \sigma'$}:~~\mbox{LF Substitution $\sigma$ weak head reduces to $\sigma'$ s.t. $\norm \sigma'$}}
 \\[1em]
  \infer{\sigma \lfwhnf \sigma}{\norm \sigma}
\quad
\infer{\wk\cdot \lfwhnf \cdot}
      { }
\SUBSTCLO{\quad
\infer{\sclo {\cdot} {\unbox{t}{\sigma} } \lfwhnf \cdot}{ }
}
\qquad
\infer{\wk{(\hatctx{\Psi},x)} \lfwhnf \wk{\hatctx\Psi}, x}
      {}
\SUBSTCLO{\\[0.5em]
\infer{(\sclo {\psi} {\unbox{t}{\sigma}}) \lfwhnf \sigma''}{
t \whnf \cbox{\hatctx \Psi' \vdash \sigma_1} &
[\sigma / \hatctx {\Psi'}]\sigma_1 \lfwhnf \sigma' &
\trunc_\psi ~ \sigma' = \sigma''
}
\\[0.5em]
\infer{\sclo {\hatctx\Psi,x}{\unbox t \sigma} \lfwhnf \sclo{\hatctx\Psi} \unbox{t}\sigma , \unbox{\cbox{\hatctx\Psi,x \vdash x}}{\sclo {\hatctx\Psi,x} {\unbox t \sigma}}}{}
\qquad
\infer{\sclo {\psi} {\unbox{t}{\sigma} } \lfwhnf \sclo {\psi} {\unbox{n}{\sigma} }}{
t \whnf n & \neut n }
}
    \end{array}
  \]
  \caption{Weak Head Reductions for LF Terms, LF Substitutions, LF Contexts, and LF Contextual Terms.}
  \label{fig:lfwhnfred}
\end{figure}

If a computation-level term $t$ is not already in $\norm$ form, we
have either an application $t_1~t_2$ or a recursor. For an application $t_1~t_2$, we reduce $t_1$. If it reduces to a function and we continue to
beta-reduce otherwise we build a neutral application. For the recursor $\titer{\R}{\tm}{\IH}~\Psi~t$ we also consider different cases:
1) if $t$ reduces to a neutral term, then we cannot proceed; 2) if $t$ reduces to $\cbox{\hatctx{\Psi'} \vdash M}$, and then proceed to further reduce $M$. If the result is $\unbox {t'}{\sigma}$, where $t$ is neutral, then we cannot proceed; if the result is $N$ where $N$ is neutral, then we proceed and choose the appropriate branch in $\R$. We note that weak head reduction for LF and computation is deterministic.

\begin{figure}[htb]
  \centering
  \[
    \begin{array}{c}
\multicolumn{1}{l}{\fbox{$t \whnf r$}:~~\mbox{Computation-level term $t$ weak head reduces to $r$ s.t. $\norm r$}}\\[1em]
\infer{t \whnf t}{\norm t}
\qquad
 \infer{t_1\;t_2 \whnf v}
          {t_1 \whnf \tmfn y t & \{t_2/y\}t \whnf v}
\quad
 \infer{t_1\;t_2 \whnf w~t_2}
      {t_1 \whnf w & \neut w  }
 \\[0.5em]
\infer{\trec{\R}{\tm}{\IH}~\rappto \Psi~t \whnf \trec{\R}{\tm}{\IH}~\rappto \Psi~s}
{
 t \whnf s & \neut{s} }
\qquad
\infer{\trec{\R}{\tm}{\IH}~\rappto \Psi~t \whnf \titer{\R}{}{\IH}
             \rappto \Psi~(\hatctx{\Psi} \vdash \unbox {t'} \sigma)}
{
 t \whnf \cbox{\hatctx \Psi \vdash M} &
 M   \lfwhnf \unbox {t'} \sigma \quad \neut t'
}
\\[0.5em]
\infer{\titer{\R}{}{\IH}~\rappto \Psi~t \whnf v}
{
 t \whnf \cbox{\hatctx \Psi \vdash M} &
 M   \lfwhnf N \quad \neut N \quad
 {\R} \cappto (\Psi)~(\hatctx{\Psi} \vdash N) \whnf v
}
\\[0.5em]
\multicolumn{1}{l}{\mbox{let}~
\R =  ({\psi,p \mto t_v} \mid {\psi,m,n,f_m, f_n \mto t_{\mathsf{app}}}
      \mid {\psi, m, f_m \mto t_{\clam}})}
 \\[1em]
 \infer{{\R} \cappto  (\Psi)~(\hatctx{\Psi} \vdash \capp M~N)  \whnf  v}
 {
\{\Psi/\psi,~~~\cbox{\hatctx{\Psi} \vdash M}/m,~~~
  \cbox{\hatctx{\Psi} \vdash N}/n,~~
  \trec{\R}{\tm}\IH~\rappto {\Psi}~{\cbox{\hatctx{\Psi} \vdash M}}/f_m,~~~
   \trec{\R}{\tm}\IH~\rappto {\Psi}~\cbox{\hatctx{\Psi} \vdash N}/f_n\} t_{\capp} \whnf v
}
\\[1em]
 \infer{{\R} \cappto  (\Psi)~(\hatctx{\Psi} \vdash x) \whnf  v}
       { \{\Psi/\psi,~\cbox{\hatctx{\Psi} \vdash x}/p\} t_v \whnf v}
\quad
\infer[]
{{\R} \cappto (\Psi)~(\hatctx{\Psi} \vdash \clam \lambda x.M)  \whnf  v}
 {
\{ {\Psi}/\psi,~\cbox{\hatctx{\Psi}, x \vdash
      M}/m,~\trec{\R}{\tm}{\IH}\rappto {(\Psi,x{:}\tm)}~{\cbox{\hatctx{\Psi}, x \vdash M}}/f  \} t_{\clam} \whnf v
}
  \end{array}
  \]
  \caption{Weak Head Reductions for Computations}
  \label{fig:whnfred}
\end{figure}

\LONGVERSION{
\begin{lemma}[Determinacy of whnf reduction]\label{lem:detwhnf}\quad
  \begin{enumerate}
  \item If $M \lfwhnf N_1$ and $M \lfwhnf N_2$ then $N_1 = N_2$.
  \item If $\sigma \lfwhnf \sigma_1$ and $\sigma \lfwhnf \sigma_2$ then $\sigma_1 = \sigma_2$.
  \item \label{it:comp-detwhnf} If $t \whnf t_1$ and $t \whnf t_2$ then $t_1 = t_2$.
  \end{enumerate}
\end{lemma}
\begin{proof}
By inspection of the rules.
\end{proof}
}
Our semantic model for equivalence characterizes well-typed terms. To facilitate our further development  we introduce the following notational abbreviations for well-typed weak head normal forms (see Def. \ref{def:typedwhnf}) and show that whnf reductions are preserved under renamings and are stable under substitutions.

\begin{definition}[Well-Typed Whnf]\label{def:typedwhnf}
\[
  \begin{array}{r@{~}c@{~}l@{~}c@{~}l@{~}l@{\quad}c@{\quad}l}
\Gamma ; \Psi & \vdash & M & \lfwhnf & N & : A & \defiff &
\Gamma ; \Psi \vdash M : A ~\mbox{and}~
\Gamma ; \Psi \vdash N : A ~\mbox{and}~
\Gamma ; \Psi \vdash M \equiv N : A ~\mbox{and}~
M \lfwhnf N
\\
  \Gamma ; \Psi & \vdash & \sigma_1 & \lfwhnf & \sigma_2 & : \Phi & \defiff &
  \Gamma ; \Psi \vdash \sigma_1 : \Phi ~\mbox{and}~
  \Gamma ; \Psi \vdash \sigma_2 : \Phi ~\mbox{and}~
  \Gamma ; \Psi \vdash \sigma_1 \equiv \sigma_2 : \Phi~\mbox{and}~
\sigma_1 \lfwhnf \sigma_2
 \\
\Gamma & \vdash & t & \whnf & t' & : {\tau}  & \defiff &
\Gamma \vdash t : \tau~\mbox{and}~
\Gamma \vdash t' : \tau~\mbox{and}~
\Gamma \vdash t \equiv t' : \tau~\mbox{and}~
t \whnf t'
  \end{array}
\]
\end{definition}

\begin{lemma}[Weak Head Reductions preserved under Weakening]\label{lem:weakwhnf}\quad
  \begin{enumerate}
  \item If\/ $\Gamma ; \Psi \vdash M \lfwhnf N: A$ and $\Gamma' \leq_\rho \Gamma$
    then $\Gamma' ; \{\rho\}\Psi \vdash \{\rho\}M \lfwhnf \{\rho\}N : \{\rho\}A$.
  \item If\/ $\Gamma ; \Psi \vdash \sigma \lfwhnf \sigma': \Phi$ and $\Gamma' \leq_\rho \Gamma$
    then $\Gamma' ; \{\rho\}\Psi \vdash \{\rho\}\sigma \lfwhnf \{\rho\}\sigma' : \{\rho\}\Phi$.
  \item \label{it:sweakcomp} If\/ $\Gamma \vdash t \whnf t': \tau$ and $\Gamma' \leq_\rho \Gamma$
    then $\Gamma' \vdash \{\rho\}t \whnf \{\rho\}t' : \{\rho\}\tau$.

  \end{enumerate}

\end{lemma}
\begin{proof}
By mutual induction on the first derivation using the computation-level substitution lemma \ref{lm:compsubst}, as renaming $\Gamma' \leq_\rho \Gamma$ are a special case of computation-level substitutions.
\end{proof}

 \begin{lemma}[LF Weak Head Reduction is stable under LF Substitutions]\label{lm:lfwhnfsub}
 Let $\Gamma ; \Psi \vdash \sigma : \Phi$.
 \begin{enumerate}
 \item If $\Gamma ; \Phi \vdash M \lfwhnf \unbox {t_1} {\sigma_1} : A$
       then $\Gamma ; \Psi \vdash \lfs \sigma \Phi M \lfwhnf \unbox{t_1}{\lfs \sigma \Phi {\sigma_1}} : \lfs \sigma \Phi A$.
 \item If $\Gamma ; \Phi \vdash M \lfwhnf \lambda x.N : \Pi x{:}A.B$
       then $\Gamma ; \Psi \vdash \lfs \sigma \Phi M \lfwhnf \lfs \sigma \Phi (\lambda x.N) : \lfs \sigma \Phi (\Pi x{:}A.B)$.
 \item If $\Gamma ; \Phi \vdash M \lfwhnf x : A$ and $\Gamma ; \Psi \vdash \sigma(x) \lfwhnf N : \lfs\sigma\Phi A$
       then $\Gamma ; \Psi \vdash \lfs \sigma \Phi M \lfwhnf N : \lfs \sigma \Phi A$.
 \item If $\Gamma ; \Phi \vdash M \lfwhnf \capp~M_1~M_2 : \tm$
       then $\Gamma ; \Psi \vdash \lfs\sigma\Phi M \lfwhnf \lfs\sigma\Phi (\capp~M_1~M_2) : \tm$.
 \item If $\Gamma ; \Phi \vdash M \lfwhnf \clam~M_1 : \tm$
       then $\Gamma ; \Psi \vdash \lfs\sigma\Phi M \lfwhnf \lfs\sigma\Phi (\clam~M_1) : \tm$.
 \item If $\Gamma ; \Phi \vdash \sigma_1 \lfwhnf \sigma_2 : \Phi'$
       then $\Gamma ; \Psi \vdash \lfs\sigma\Phi\sigma_1 \lfwhnf \lfs\sigma\Phi\sigma_2 : \Phi'$.
 \end{enumerate}
\end{lemma}
\begin{proof}
By induction on $M \lfwhnf M'$ relation that is part of the well-typed weak head reduction using Lemma \ref{lm:lfctxwf} and \ref{lm:lfsubst}.
\LONGVERSIONCHECKED{\\[1em]
For (1): \fbox{If $\Gamma ; \Phi \vdash M \lfwhnf \unbox {t_1} {\sigma_1} : A$
       then $\Gamma ; \Psi \vdash \lfs \sigma \Phi M \lfwhnf \unbox{t_1}{\lfs \sigma \Phi {\sigma_1}} : \lfs \sigma \Phi A$.}
\\[1em]
\prf{\emph{Case} $M = \unbox{t_0}{\sigma_1}$ and $t_0 \whnf t_1$ and $\neut t_1$}
\prf{$\Gamma ; \Phi \vdash M : A$ \hfill by assumption}
\prf{$\Gamma ; \Phi \vdash \unbox{t_1}{\sigma_1} : A$ \hfill by assumption}
\prf{$\Gamma ; \Psi \vdash \lfs\sigma\Phi M : \lfs\sigma\Phi A$ \hfill by LF subst. lemma \ref{lm:lfsubst}}
\prf{$\Gamma ; \Psi \vdash  \lfs\sigma\Phi (\unbox{t_1}{\sigma_1}) :  \lfs\sigma\Phi A$ \hfill by LF subst. lemma \ref{lm:lfsubst}}
\prf{$ \lfs\sigma\Phi M \lfwhnf \unbox{t_1}{\lfs\sigma\Phi {\sigma_1}}$ \hfill since $\neut t_1$ and $t_0 \whnf t_1$ and LF subst. prop.}
\prf{$\Gamma ; \Phi \vdash \lfs\sigma\Phi M \lfwhnf \unbox{t_1}{\lfs\sigma\Phi {\sigma_1}} : \lfs\sigma\Phi A$ \hfill by well-typed whnf (Def \ref{def:typedwhnf})}
\\[0.25em]
\prf{\emph{Case} $M = \unbox{t_0}{\sigma_0}$ and $t_0 \whnf \cbox{\hatctx\Phi' \vdash M'}$, $\lfs{\sigma_0}{\Phi'}M' \lfwhnf \unbox{t_1}{\sigma_1}$}
\prf{$\Gamma ; \Phi \vdash M : A$ \hfill by assumption}
\prf{$\Gamma ; \Phi \vdash \unbox{t_1}{\sigma_1} : A$ \hfill by assumption}
\prf{$\Gamma ; \Psi \vdash \lfs\sigma\Phi M : \lfs\sigma\Phi A$ \hfill by LF subst. lemma \ref{lm:lfsubst}}
\prf{$\Gamma ; \Psi \vdash  \lfs\sigma\Phi (\unbox{t_1}{\sigma_1}) :  \lfs\sigma\Phi A$ \hfill by LF subst. lemma \ref{lm:lfsubst}}
\prf{$\lfs{\lfs\sigma\Phi{\sigma_0}}{\Phi'}M' \lfwhnf \unbox{t_1}{\lfs\sigma\Phi {\sigma_1}}$ \hfill by IH (and subst. prop)}
\prf{$\unbox{t_0}{\lfs\sigma\Phi \sigma_0} \lfwhnf \unbox{t_1}{\lfs\sigma\Phi{\sigma_1}}$ \hfill by whnf}
\prf{$\Gamma ; \Psi \vdash \lfs\sigma\Phi (\unbox{t_0}{\sigma_0}) \lfwhnf \unbox{t_1}{\lfs\sigma\Phi{\sigma_1}} : \lfs\sigma\Phi A$ \hfill by subst. prop. and  well-typed whnf (Def \ref{def:typedwhnf})}
\\[0.25em]
\prf{\emph{Case} $M = \unbox{t_1}{\sigma_1}$ and $\neut t_1$}
\prf{$\Gamma ; \Phi \vdash M : A$ \hfill by assumption}
\prf{$\Gamma ; \Phi \vdash \unbox{t_1}{\sigma_1} : A$ \hfill by assumption}
\prf{$\Gamma ; \Psi \vdash \lfs\sigma\Phi M : \lfs\sigma\Phi A$ \hfill by LF subst. lemma \ref{lm:lfsubst}}
\prf{$\Gamma ; \Psi \vdash  \lfs\sigma\Phi (\unbox{t_1}{\sigma_1}) :  \lfs\sigma\Phi A$ \hfill by LF subst. lemma \ref{lm:lfsubst}}
\prf{$\lfs\sigma\Phi(\unbox{t_1}{\sigma_1}) = \unbox{t_1}{\lfs\sigma\Phi{\sigma_1}}$ \hfill by subst. def.}
\prf{$\neut \unbox{t_1}{\lfs\sigma\Phi{\sigma_1}}$ \hfill since $\neut t_1$}
\prf{$\lfs\sigma\Phi M \lfwhnf \unbox{t_1}{\lfs\sigma\Phi {\sigma_1}}$ \hfill by whnf}
\prf{$\Gamma ; \Psi \vdash \lfs\sigma\Phi M \lfwhnf \unbox{t_1}{\lfs\sigma\Phi {\sigma_1}} : \lfs\sigma\Phi A$ \hfill by well-typed whnf (Def \ref{def:typedwhnf})}
\\[-0.75em]
\prf{\emph{Case} $M = M_1~M_2$ and $M \lfwhnf \unbox{t_1}{\sigma_1}$}
\prf{$\Gamma ; \Phi \vdash M : A$ \hfill by assumption}
\prf{$\Gamma ; \Phi \vdash \unbox{t_1}{\sigma_1} : A$ \hfill by assumption}
\prf{$\Gamma ; \Psi \vdash \lfs\sigma\Phi M : \lfs\sigma\Phi A$ \hfill by LF subst. lemma \ref{lm:lfsubst}}
\prf{$\Gamma ; \Psi \vdash  \lfs\sigma\Phi (\unbox{t_1}{\sigma_1}) :  \lfs\sigma\Phi A$ \hfill by LF subst. lemma \ref{lm:lfsubst}}
\prf{$M_1 \lfwhnf \lambda x.M'$ and $[M_2/x]M' \lfwhnf \unbox{t_1}{\sigma_1}$ \hfill by inversion}
\prf{$\lfs\sigma\Phi {[M_2/x]M'} \lfwhnf \lfs\sigma\Phi (\unbox{t_1}{\sigma_1})$ \hfill by IH}
\prf{$\lfs{\sigma, \lfs\sigma\Phi M_2}{\Phi, x}M' \lfwhnf \lfs\sigma\Phi(\unbox{t_1}{\sigma_1})$ \hfill by subst. def}
\prf{$\lfs\sigma\Phi M_1 \lfwhnf \lfs\sigma\Phi (\lambda x.M')$ \hfill by IH}
\prf{$\lfs\sigma\Phi{M_1} \lfwhnf \lambda x. \lfs{\sigma, x}{\Phi,x}M'$ \hfill by subst. def.}
\prf{$[\lfs\sigma\Phi M_2/x](\sigma, x) = \sigma, \lfs\sigma\Phi M_2$ \hfill by subst. def.}
\prf{$\lfs\sigma\Phi M_1~\lfs\sigma\Phi M_2 \lfwhnf \lfs\sigma\Phi (\unbox{t_1}{\sigma_1})$ \hfill by whnf rules}
\prf{$\lfs\sigma\Phi(M_1~M_2) \lfwhnf \unbox{t_1}{\lfs\sigma\Phi \sigma_1}$ \hfill by subst. def.}
}
\LONGVERSIONCHECKED{\\[1em]
For (3): \fbox{If $\Gamma ; \Phi \vdash M \lfwhnf x : A$ and $\Gamma ; \Psi \vdash \sigma(x) \lfwhnf N : \lfs\sigma\Phi A$
       then $\Gamma ; \Psi \vdash \lfs \sigma \Phi M \lfwhnf N : \lfs \sigma \Phi A$.}
\\[1em]
\\[0.5em]
\prf{\emph{Case} $M = x$ where $x \in \Phi$ and $\norm M$}
\prf{$\Gamma ; \Phi \vdash x : A$ and $x:A \in \Phi$ \hfill since $\Gamma ; \Phi \vdash M \lfwhnf x : A$}
\prf{$\Gamma ; \Phi \vdash M : A$ \hfill since $\Gamma ; \Phi \vdash M \lfwhnf x : A$}
\prf{$\Gamma ; \Psi \vdash \lfs\sigma\Phi M : \lfs\sigma\Phi A$ \hfill LF subst. lemma \ref{lm:lfsubst}}
\prf{$\Gamma ; \Psi \vdash N : \lfs\sigma\Phi A$ \hfill since $\Gamma ; \Psi \vdash \sigma : \Phi$ and $x:A \in \Phi$}
\prf{$\lfs\sigma\Phi M = \lfs\sigma\Phi x = \sigma(x)$ \hfill by subst. def.}
\prf{$\lfs\sigma\Phi M \lfwhnf N$ \hfill since $\sigma(x) \lfwhnf N$}
\\[-0.75em]
\prf{\emph{Case} $M = M_1~M_2$ and $M \lfwhnf x$}
\prf{$\Gamma ; \Phi \vdash x : A$ and $x:A \in \Phi$ \hfill since $\Gamma ; \Phi \vdash M \lfwhnf x : A$}
\prf{$\Gamma ; \Phi \vdash M : A$ \hfill since $\Gamma ; \Phi \vdash M \lfwhnf x : A$}
\prf{$\Gamma ; \Psi \vdash \lfs\sigma\Phi M : \lfs\sigma\Phi A$ \hfill LF subst. lemma \ref{lm:lfsubst}}
\prf{$M_1 \lfwhnf \lambda x.M'$ and $[M_2/x]M' \lfwhnf x$ \hfill by inversion}
\prf{$\lfs\sigma\Phi [M_2/x]M' \lfwhnf N$ \hfill by IH using $\sigma(x) \lfwhnf N$}
\prf{$\lfs{\sigma,\lfs{\sigma}{\Phi}M_2}{\Phi, x}M' \lfwhnf N$ \hfill by subst. def.}
\prf{$\lfs\sigma\Phi M_1 \lfwhnf \lambda x.[\sigma, x]M'$ \hfill by IH and LF subst. prop.}
\prf{$\lfs\sigma\Phi M \lfwhnf N$ \hfill by whnf rules}
\prf{$\Gamma ; \Psi \vdash \lfs\sigma\Phi M \lfwhnf N : \lfs\sigma \Phi A$ \hfill by  well-typed whnf (Def \ref{def:typedwhnf})}
\\[0.25em]
\prf{\emph{Case} $M = \unbox{t_1}{\sigma_1}$ and $M \lfwhnf x$}
\prf{$\Gamma ; \Phi \vdash x : A$ and $x:A \in \Phi$ \hfill since $\Gamma ; \Phi \vdash M \lfwhnf x : A$}
\prf{$\Gamma ; \Phi \vdash M : A$ \hfill since $\Gamma ; \Phi \vdash M \lfwhnf x : A$}
\prf{$\Gamma ; \Psi \vdash \lfs\sigma\Phi M : \lfs\sigma\Phi A$ \hfill LF subst. lemma \ref{lm:lfsubst}}
\prf{$t_1 \whnf \cbox{\hatctx\Phi' \vdash M'}$ \hfill since $\Gamma ; \Phi \vdash M \lfwhnf x : \tm$}
\prf{$\sigma_1 \lfwhnf \sigma_2$ and $\lfs{\sigma_2}{\Phi'}M' \lfwhnf x$ \hfill since $\Gamma ; \Phi \vdash M \lfwhnf x : \tm$}\\[-0.85em]
\prf{$\lfs\sigma\Phi \sigma_1 \lfwhnf \lfs\sigma\Phi \sigma_2$ \hfill by IH}\\[-0.85em]
\prf{$\lfs\sigma\Phi (\lfs{\sigma_2}{\Phi'}M' ) \lfwhnf N$ \hfill by IH}\\[-0.85em]
\prf{$\lfs{\lfs\sigma \Phi \sigma_2}{\Phi'} M'\lfwhnf N$ \hfill by subst. prop.}\\[-0.85em]
\prf{$\unbox{t_1}{\lfs\sigma \Phi \sigma_1} \lfwhnf N$ \hfill by whnf rules}\\[-0.85em]
\prf{$\Gamma ; \Psi \vdash \unbox{t_1}{\lfs\sigma \Phi \sigma_1} \lfwhnf N : \tm$ \hfill by  well-typed whnf (Def \ref{def:typedwhnf})}
}
\LONGVERSIONCHECKED{
\\[1em]
For (6): \fbox{If $\Gamma ; \Phi \vdash \sigma_1 \lfwhnf \sigma_2 : \Phi'$
       then $\Gamma ; \Psi \vdash \lfs\sigma\Phi\sigma_1 \lfwhnf \lfs\sigma\Phi\sigma_2 : \Phi'$.
}
\\[1em]
\pcase{$\norm \sigma_1$ and $\sigma_2 = \sigma_1$}
\prf{$\norm \lfs\sigma\Phi \sigma_1$ \hfill by def. of $\norm$}
\prf{$\lfs\sigma\Phi \sigma_1 \lfwhnf \lfs\sigma\Phi \sigma_1$ \hfill by whnf}
\prf{$\Gamma ; \Phi \vdash \sigma_1 : \Phi'$ and $\Gamma ; \Phi \vdash \sigma_2 : \Phi'$ \hfill by assumption}
\prf{$\Gamma ; \Psi \vdash \lfs\sigma\Phi\sigma_1 : \Phi$ and $\Gamma ; \Psi \vdash \lfs\sigma\Phi\sigma_2 : \Phi$     \hfill by LF subst. lemma}
\prf{$\Gamma ; \Psi \vdash \lfs\sigma\Phi\sigma_1 \lfwhnf \lfs\sigma\Phi\sigma_1 : \Phi'$  \hfill by well-typed whnf (Def \ref{def:typedwhnf})}
\\
\pcase{$\sigma_1 = \wk\cdot$}
\prf{$\Gamma ; \Phi \vdash \wk\cdot : \cdot$ \hfill by assumption}
\prf{$\Gamma \vdash \Phi : \ctx$ and $\Gamma \vdash \Psi : \ctx$ \hfill by well-formedness of LF context (Lemma \ref{lm:lfctxwf})}
\prf{$\Gamma ; \Phi \vdash \cdot : \cdot$ \hfill by typing rule}
\prf{$\cdot = \lfs\sigma\Phi{\wk\cdot} = \trunc_\cdot (\sigma / \hatctx\Phi)$ \hfill by subst. def.}
\prf{$\Gamma ; \Psi \vdash \cdot : \cdot$ \hfill by typing rule}
\prf{$\norm \cdot$ \hfill by whnf}
\prf{$\cdot \lfwhnf \cdot$ \hfill by whnf}
\prf{$\lfs\sigma\Phi{\wk\cdot} \lfwhnf \lfs\sigma\Phi\cdot$ \hfill since $\cdot = \lfs\sigma\Phi{\wk\cdot}$ and $\lfs\sigma\Phi\cdot = \cdot$}
\prf{$\Gamma ; \Psi \vdash \lfs\sigma\Phi \wk\cdot \lfwhnf \lfs\sigma\Phi\cdot : \cdot$ \hfill  by well-typed whnf (Def \ref{def:typedwhnf})}
\\[1em]
\pcase{$\sigma_1 = \wk{\hatctx\Phi',x}$ where $\Phi = \Phi', x{:}A, \wvec{x{:}A}$ }
\prf{$\sigma_2 = \wk{\hatctx\Phi'}, x$ \hfill by assumption}
\prf{$\Gamma ; \Phi \vdash \wk{\hatctx\Phi',x} : \Phi', x{:}A$ and $\Gamma ; \Phi \vdash \wk{\hatctx\Phi'} : \Phi'$ \hfill by assumption and typing}
\prf{$\Gamma ; \Psi \vdash \sigma : \Phi', x{:}A, \wvec{x{:}A}$ \hfill by assumption\\}
\prf{\emph{Sub-case} : $\sigma = (\sigma', M, \vec{M})$}
\prf{$\Gamma ; \Psi \vdash (\sigma', M, \vec{M}) :  \Phi', x{:}A, \wvec{x{:}A}$ \hfill where $\sigma = (\sigma', M, \vec{M})$}
\prf{$\Gamma ; \Psi \vdash \sigma' : \Phi'$ and $\Gamma ; \Psi \vdash M : \lfs{\sigma'}{\Phi'}A$ \hfill by inversion}
\prf{$ \lfs\sigma\Phi (\wk{\hatctx\Phi',x}) = \trunc_{\hatctx\Phi',x} (\sigma / \Phi) = \sigma', M$ \hfill by definition}
\prf{$\lfs\sigma\Phi (\wk{\hatctx\Phi'}) = \trunc_{\hatctx\Phi'}(\sigma / \Phi) = \sigma'$ \hfill by definition}
\prf{$\Gamma ; \Psi \vdash \sigma', M \lfwhnf \sigma', M : \Phi', x{:}A$ \hfill since $\neut (\sigma', M)$}
\prf{$\Gamma ; \Psi \vdash \lfs\sigma\Phi (\wk{\hatctx\Phi',x}) \lfwhnf \lfs\sigma\Phi (\wk{\hatctx{\Phi'}}, x) : \Phi', x{:}A$ \hfill  by well-typed whnf (Def \ref{def:typedwhnf})\\}
\prf{\emph{Sub-case} : $\sigma = \wk{\hatctx\Phi', x, \vec x}$}
\prf{$ \lfs\sigma\Phi (\wk{\hatctx\Phi',x}) = \trunc_{\hatctx\Phi',x} (\sigma / \Phi) = \trunc_{\hatctx\Phi', x}(\wk{\hatctx\Phi', x, \vec x} / \hatctx\Phi',x,\vec x) = \wk{\hatctx\Phi',x}$ \hfill by definition}
\prf{$\lfs\sigma\Phi (\wk{\hatctx\Phi'},x) = \trunc_{\hatctx\Phi'} (\sigma / \Phi) = \wk{\hatctx\Phi'}$ \hfill by definition}
\prf{$\lfs\sigma\Phi (x) = x$ \hfill since $\sigma = \wk{\hatctx\Phi', x, \vec x}$}
\prf{$\Gamma ; \Phi', x{:}A, \wvec{x:A} \vdash \wk{\hatctx\Phi',x} \lfwhnf \wk{\hatctx\Phi'}, x : \Phi', x{:}A$ \hfill by $\lfwhnf$ rule and  by well-typed whnf (Def \ref{def:typedwhnf})}
\SUBSTCLO{\\[1em]
\prf{\emph{Sub-case} : $\sigma = \sclo {\hatctx\Phi',x,\vec x} \unbox{t}{\sigma'}$  }
\prf{$\Gamma ; \Psi \vdash \sclo {\hatctx\Phi',x,\vec x} \unbox{t}{\sigma'} : \Phi', x{:}A, \wvec{x{:}A}$ \hfill by assumption}
\prf{$ \lfs\sigma\Phi \wk{\hatctx\Phi',x} = \trunc_{\hatctx\Phi',x} (\sigma / \Phi) = \sclo {\hatctx\Phi',x}{\unbox{t}{\sigma'}}$ \hfill by definition}
\prf{$\lfs\sigma\Phi \wk{\hatctx\Phi'} = \trunc_{\hatctx\Phi'} (\sigma / \Phi) = \sclo {\hatctx\Phi'}{\unbox{t}{\sigma'}}$ \hfill by definition}
\prf{$\lfs\sigma\Phi (x) = \unbox{\cbox{\hatctx\Phi \vdash x}}{\sigma}$ \hfill by def.}
\prf{$\Gamma ; \Phi', x{:}A, \wvec{x{:}A} \vdash \sclo {\hatctx\Phi',x}{\unbox{t}{\sigma'}} \lfwhnf \sclo {\hatctx\Phi'}{\unbox{t}{\sigma'}},~\unbox{\cbox{\hatctx\Phi \vdash x}}{\sigma} : \Phi', x{:}A$ \hfill by rules $\lfwhnf$}
\prf{$\Gamma ; \Phi', x{:}A, \wvec{x:A} \vdash \lfs\sigma\Phi (\wk{\hatctx\Phi',x}) \lfwhnf \lfs\sigma\Phi (\wk{\hatctx{\Phi'}}, x) : \Phi', x{:}A$ \hfill by well-typed whnf (Def \ref{def:typedwhnf}) }
}
\SUBSTCLO{
\\[1em]
\pcase{$\sigma_1 = (\sclo {\psi} {\unbox{t}{\xi}})$ and $t \whnf \cbox{\psi,\vec{w} \vdash \xi_1}$ and\\
$\lfs {\xi}{\psi, \vec{w}} \xi_1 \lfwhnf \xi'$ and
$\lfs{\xi'}{\psi,\vec{w}}(\wk\psi) = \trunc_\psi ~ \xi' = \sigma_2$}
\prf{$\Gamma ; \lfs\sigma\Phi \lfs {\xi}{\psi, \vec{w}} \xi_1 \lfwhnf \lfs\sigma\Phi\xi'$ \hfill by IH}
\prf{$\Gamma ; \Phi', x{:}A \wvec{x:A} \vdash \lfs{\lfs\sigma\Phi \xi}{\psi,\vec{w}} \xi_1 \lfwhnf \lfs\sigma\Phi\xi'$ \hfill by subst. prop.}
\prf{$\Gamma ; \Phi', x{:}A, \wvec{x:A} \vdash (\sclo {\psi} {\unbox{t}{\lfs\sigma\Phi \xi}})\lfwhnf \lfs{\lfs\sigma\Phi\xi'}{\psi,\vec{w}}(\wk\psi) $ \hfill by $\lfwhnf$ and \\
\mbox{$\quad$}\hfill well-typed whnf (Def \ref{def:typedwhnf})}
\prf{$\Gamma ; \Phi', x{:}A, \wvec{x:A} \vdash \lfs\sigma\Phi (\sclo {\psi} {\unbox{t}{\xi}}) \lfwhnf \lfs\sigma\Phi (\sigma_2) : \Phi', x{:}A$ \hfill by subst. prop.
}
\\[1em]
\pcase{$\sigma_1 = \sclo {\phi} {\unbox t \xi} $ and $t \whnf t'$ and $\neut t'$ and $\sigma_2 = \sclo{\phi}\unbox{t'}{\xi}$}
\prf{$\Gamma ; \Psi \vdash \sclo {\phi}{\unbox t {\lfs\sigma\Phi \xi}} \lfwhnf \sclo{\phi}{\unbox{t'}{\lfs\sigma\Phi \xi}}$\hfill since $\neut t'$\\
\mbox{$\quad$}\hfill and  by well-typed whnf (Def \ref{def:typedwhnf})}
\prf{$\Gamma ; \Psi \vdash \lfs \sigma\Phi ( \sclo {\phi} {\unbox{t}{\xi} }) \lfwhnf \lfs\sigma\Phi (\sclo{\phi}{\unbox{t'}{\xi}})$ \hfill by subst. prop.}
}
\\[1em]
}
\LONGVERSIONCHECKED{
\fbox{For (4): If $\Gamma ; \Phi \vdash M \lfwhnf \capp~M_1~M_2 : \tm$
       then $\Gamma ; \Psi \vdash \lfs\sigma\Phi M \lfwhnf \lfs\sigma\Phi (\capp~M_1~M_2) : \tm$.}
\\[1em]
\pcase{$\ianc
{\norm (\capp~M_1~M_2)}
{\capp~M_1~M_2 \lfwhnf \capp~M_1~M_2}{}
$
}
\prf{$\Gamma ; \Phi \vdash \capp~M_1~M_2 : \tm$ \hfill by assumption}
\prf{$\Gamma ; \Psi \vdash \lfs\sigma\Phi (\capp~M_1~M_2) : \tm$ \hfill by LF subst. lemma}
\prf{$\Gamma ; \Psi \vdash \capp \lfs\sigma\Phi(M_1)~\lfs\sigma\Phi(M_2) : \tm$ \hfill subst. prop}
\prf{$\norm (\capp \lfs\sigma\Phi(M_1)~\lfs\sigma\Phi(M_2))$ \hfill by whnf def.}
\prf{$\Gamma ; \Psi \vdash \lfs\sigma\Phi (\capp~M_1~M_2) \lfwhnf \lfs\sigma\Phi (\capp~M_1~M_2) : \tm$ \hfill by subst. prop., $\lfwhnf$ rule, and Def. \ref{def:typedwhnf}}
\\[1em]
\pcase{$\ianc
{M \lfwhnf \lambda x.M' \quad [N/x]M' \lfwhnf \capp~M_1~M_2 }
{M\;N \lfwhnf \capp~M_1~M_2}{}
$
}
\prf{$\Gamma ; \Phi \vdash \capp~M_1~M_2 : \tm$ \hfill by assumption}
\prf{$\Gamma ; \Psi \vdash  \lfs\sigma\Phi (\capp~M_1~M_2) : \tm$ \hfill by LF subst. lemma}
\prf{$\Gamma ; \Phi \vdash M~N : \tm$ \hfill by assumption}
\prf{$\Gamma ; \Psi \vdash \lfs\sigma\Phi (M~N) : \tm$ \hfill by LF subst. lemma}
\prf{$\Gamma ; \Psi \vdash \lfs\sigma\Phi M \lfwhnf \lfs\sigma\Phi (\lambda x.M') : \tm$ \hfill by IH}
\prf{$\Gamma ; \Psi \vdash \lfs\sigma\Phi([N/x]M') \lfwhnf \lfs\sigma\Phi (\capp~M_1~M_2) : \tm$ \hfill by IH}
\prf{$\Gamma ; \Psi \vdash \lfs\sigma\Phi (M~N) \lfwhnf \lfs\sigma\Phi (\capp~M_1~M_2) : \tm$ \hfill
with $\lfs\sigma\Phi([N/x]M') = \lfs{\sigma, \lfs\sigma\Phi N}{\Phi, x{:}\tm}$ }
\\[1em]
\pcase{$\ianc
      {t \whnf \cbox{\hatctx{\Phi'} \vdash M}\quad
       \lfs {\sigma'}{\Phi'} M \lfwhnf \capp~M_1~M_2 }
{\unbox t{\sigma'} \lfwhnf \capp~M_1~M_2 }{}
$
}
\prf{$\Gamma ; \Psi \vdash \lfs{\sigma}{\Phi}\lfs {\sigma'}{\Phi'} M \lfwhnf  \lfs\sigma\Phi (\capp~M_1~M_2) : \tm$ \hfill by IH}
\prf{$\Gamma ; \Psi \vdash \unbox t{\lfs{\sigma}{\Phi}\sigma'} \lfwhnf \lfs{\sigma}{\Phi}(\capp~M_1~M_2) : \tm$ \hfill by  subst. prop., $\lfwhnf$ rule, and Def. \ref{def:typedwhnf}}

The remaining cases for (2) and (5) are similar.
}
 \end{proof}





\section{Kripke-style Logical Relation}\label{sec:logrel}
We construct a Kripke-logical relation to prove weak head
normalization. 
Our semantic definitions for computations follows closely
\citet{Abel:LMCS12} to accommodate type-level computation.

\begin{figure}[h]
\hrulefill
  \centering
  \[
    \begin{array}{c}
\infer{\Gamma ; \Psi \Vdash M = N : \tm}
{
\begin{array}{lll}
\Gamma ; \Psi \vdash M \lfwhnf \unbox {t_1} {\sigma_1} : \tm   & \typeof (\Gamma  \vdash t_1) = \cbox{\Phi_1 \vdash \tm} &
\Gamma \vdash \Phi_1 \equiv \Phi_2 : \ctx \\
\Gamma ; \Psi \vdash N \lfwhnf \unbox {t_2} {\sigma_2} : \tm & \typeof (\Gamma  \vdash t_2) = \cbox{\Phi_2 \vdash \tm} &
\Gamma  \vdash t_1 \equiv t_2 : \cbox{\Phi_1 \vdash \tm} \quad \Gamma ; \Psi \Vdash \sigma_1 = \sigma_2 : \Phi_1
\end{array}
}
\\[0.75em]
\infer {\Gamma ; \Psi \Vdash M = N : \tm}
{
      \begin{array}{ll}
\Gamma ; \Psi \vdash M \lfwhnf \clam M' : \tm & \\
\Gamma ; \Psi \vdash N \lfwhnf \clam N' : \tm & \Gamma ; \Psi, x{:}\tm \Vdash M'~x = N'~x: \tm
      \end{array}
}
\\[0.75em]
\infer {\Gamma ; \Psi \Vdash M = N : \tm}
{
      \begin{array}{ll}
\Gamma ; \Psi \vdash M \lfwhnf \capp {M_1}~{M_2} : \tm & \\
\Gamma ; \Psi \vdash N \lfwhnf \capp {N_1}~{N_2} : \tm &
\Gamma ; \Psi \Vdash M_1 = N_1 : \tm \quad \Gamma ; \Psi \Vdash M_2 = N_2 : \tm
      \end{array}
}
\\[0.75em]
\infer {\Gamma ; \Psi \Vdash M = N : \tm}
{ \Gamma ; \Psi \vdash M \lfwhnf x : \tm  &
  \Gamma ; \Psi \vdash N \lfwhnf x : \tm }
  \end{array}
\]
\hrulefill
  \caption{Semantic Equality for LF Terms: \fbox{$\Gamma ; \Psi \Vdash M = N : A$}}
  \label{fig:LFsem}
\end{figure}

We start by defining semantic equality for LF terms of type $\tm$ (Fig.~\ref{fig:LFsem}), as we restricted our LF signature and these are the terms of interest. 
To define semantic equality for LF terms $M$ and $N$, we consider different cases depending on their whnf: 1) if they reduce to $\capp M_1~M_2$ and $\capp~N_1~N_2$ respectively, then $M_i$ must be semantically equal to $N_i$; 2) if they reduce to $\clam M'$ and $\clam N'$ respectively, then the bodies of $M'$ and $N'$ must be equal. To compare their bodies, we apply both $M'$ and $N'$ to an LF variable $x$ and consider $M'~x$ and $N'~x$ in the extended LF context $\Psi, x{:}\tm$. This has the effect of opening up the body and replacing the bound LF variable with a fresh one. This highlights the difference between the intensional LF function space and the extensional nature of the computation-level functions. In the former, we can concentrate on LF variables and continue to analyze the LF function body; in the latter, we consider all possible inputs, not just variables; 3) if the LF terms $M$ and $N$ may reduce to the same LF variable in $\Psi$, then they are obviously also semantically equal; 4) last, if $M$ and $N$ reduce to $\unbox {t_i} {\sigma_i}$ respectively. In this case $t_i$ is neutral and we only need to semantically compare the LF substitutions $\sigma_i$ and check whether the terms $t_i$ are definitional equal. However, what type should we choose? -- As the computation $t_i$ is neutral, we can infer a unique type $\cbox{\Phi \vdash \tm}$ which we can use.
This is defined as follows:

\[
  \begin{array}{c}
 \multicolumn{1}{l}{\mbox{Type inference for Neutral Computations $t$}: \typeof (\Gamma \vdash t) = \tau}\\[1em]
 \infer{\typeof (\Gamma \vdash t~s) = \{s/y\}\tau_2}
 {\typeof (\Gamma \vdash t ) = \tau & \tau \whnf (y{:}\tau_1) \arrow \tau_2 &
  \Gamma \vdash s : \tau_1} \qquad
 \infer{\typeof (\Gamma \vdash x) = \tau}{x{:}\tau \in \Gamma}
\\[1em]
\infer{\typeof (\Gamma \vdash \titer{\R}{}{\IH}~\Psi~t) = \{\Psi/\psi, t/y\}\tau}
{\IH = (\psi : \tmctx) \arrow (y : \cbox{\psi \vdash \tm}) \arrow \tau}
  \end{array}
\]

\LONGVERSION{
  \begin{lemma}\label{lm:typeof}
If $\Gamma \vdash t : \tau$ and $\neut t$ then $\typeof (\Gamma \vdash t) = \tau'$ and $\Gamma \vdash \tau \equiv \tau' : u$.
  \end{lemma}
  \begin{proof}
By induction on $\neut t$.
  \end{proof}
}

Semantic equality for LF substitutions is also defined by considering different weak head normal forms (see Fig. \ref{fig:LFsemctx}). As we only work with well-typed LF objects, there is only one inhabitant for an empty context. Moreover, given a LF substitution with domain $\Phi, x{:}A$, we can weak head reduce the LF substitutions $\sigma$ and $\sigma'$ and continue to recursively compare them. Finally, for LF substitutions with domain $\psi$, a context variable, there are two cases we consider: either both LF substitution reduce to a weakening $\wk\psi$ or they reduce to substitution closure.

\begin{figure}[h]
\hrulefill
  \centering
  \[
    \begin{array}{c}
\SUBSTCLO{
\infer{\Gamma ; \Psi \Vdash \sigma = \sigma' : \phi}
{
      \begin{array}{lll}
\Gamma ; \Psi \vdash \sigma \lfwhnf (\sclo \phi {\unbox{t_1}{\sigma_1}})  : \phi
   & \typeof (\Gamma \vdash t_1) = \cbox{\Phi_1 \vdash \Phi'_1}  & \Gamma \vdash \Phi'_1 \equiv (\phi, \wvec{x{:}A}) : \ctx
\\
\Gamma ; \Psi \vdash \sigma' \lfwhnf (\sclo \phi {\unbox{t_2}{\sigma_2}}) : \phi
    & \typeof (\Gamma \vdash t_2) = \cbox{\Phi_2 \vdash \Phi'_2} & \Gamma \vdash \Phi'_2 \equiv (\phi, \wvec{x{:}A}) : \ctx
\\
\Gamma \vdash t_1 \equiv t_2 : \cbox{\Phi_1 \vdash \Phi'_1}
    & \Gamma ; \Psi \Vdash \sigma_1 = \sigma_2 : \Phi_1
    & \Gamma \vdash \Phi_1 \equiv \Phi_2 : \ctx
\end{array}
}
\\[0.5em]
}
\raisebox{1ex}{
\ianc{
\Gamma ; \Psi \vdash \sigma \lfwhnf \cdot : \cdot \quad
\Gamma ; \Psi \vdash \sigma' \lfwhnf \cdot : \cdot }
{\Gamma ; \Psi \Vdash \sigma = \sigma' : \cdot}{}}
\quad
 \ianc{
\Gamma ; \psi, \wvec{x{:}A} \vdash \sigma \lfwhnf \wk{\psi} : \psi \quad
\Gamma ; \psi, \wvec{x{:}A} \vdash \sigma' \lfwhnf \wk{\psi} : \psi
}{\Gamma ; \psi, \wvec{x{:}A} \Vdash \sigma = \sigma' : \psi}{}
\\[0.5em]
 \infer{\Gamma ; \Psi \Vdash \sigma = \sigma' : \Phi, x{:}A}
  {\begin{array}{ll}
\Gamma ; \Psi \vdash \sigma \lfwhnf \sigma_1, M : \Phi, x{:}A  & \\
\Gamma ; \Psi \vdash \sigma' \lfwhnf \sigma_2, N : \Phi, x{:}A &
\Gamma ; \Psi \Vdash \sigma_1 = \sigma_2 : \Phi \qquad \Gamma ; \Psi  \Vdash M = N : \lfs{\sigma_1}{\Phi} A
\end{array}
}
  \end{array}
\]
\hrulefill
  \caption{Semantic Equality for LF Substitutions: \fbox{$\Gamma ; \Psi \Vdash \sigma = \sigma' : \Phi$}}
  \label{fig:LFsemctx}
\end{figure}

\SHORTVERSION{
We can then lift the semantic equality definition to contextual terms: \fbox{$\Gamma \semlf C = C' : T$}.

\[
\infer{\Gamma \semlf (\hatctx{\Psi} \vdash M) = (\hatctx{\Psi} \vdash N): (\Psi \vdash A)}
{\Gamma ; \Psi \Vdash M = N : A}
\]
}

\LONGVERSION{
\begin{figure}
\hrulefill
  \centering
  \[
    \begin{array}{c}
\infer{\Gamma \semlf (\hatctx{\Psi} \vdash M) = (\hatctx{\Psi} \vdash N): (\Psi \vdash A)}
{\Gamma ; \Psi \Vdash M = N : A}
\LONGVERSION{\quad
\infer{\Gamma \semlf (\hatctx{\Psi} \vdash M) = (\hatctx{\Psi} \vdash N) : (\Psi \vdash_\# A)}
{ \Gamma ; \Psi \Vdash_\# M = N : A}}
\SUBSTCLO{\\[1em]
\infer{\Gamma \semlf (\hatctx{\Psi} \vdash \sigma) = (\hatctx{\Psi} \vdash \sigma'): (\Psi \vdash \Phi)}
{\Gamma ; \Psi \Vdash \sigma = \sigma' : \Phi}
\qquad
\infer{\Gamma \semlf (\hatctx{\Psi} \vdash \sigma) = (\hatctx{\Psi} \vdash \sigma) : (\Psi \vdash_\# \Phi)}
{ \Gamma ; \Psi \Vdash_\# \sigma = \sigma' : \Phi}}
  \end{array}
\]
\hrulefill
  \caption{Semantic Typing for Contextual LF Terms}
  \label{fig:LFsem2}
\end{figure}
}

\begin{figure}[h]
\hrulefill
\vspace{1ex}
\[
  \begin{array}{c}
\infer{\Gamma \Vdash \tau : u}
  { \Gamma \der \tau \whnf \cbox{T} : u
  & \Gamma \vdash T \equiv T }
\qquad
\infer{\Gamma \Vdash \tau : u}
  {\Gamma \der \tau \whnf u' : u & u' < u}
\qquad
\infer{\Gamma \Vdash \tau : u}
{\Gamma \vdash \tau \whnf x~\vec{t} : u \qquad \neut (x~\vec{t})}
\\[1em]
\infer[(u_1,~u_2,~u) \in \Ru]{\Gamma \Vdash \tau : u}
  {
    \begin{array}{c}
     \Gamma \der \tau \whnf (y:\ann\tau_1) \arrow \tau_2 : u
     \hfill~\quad
     \forall \Gamma' \leq_\rho \Gamma.\ \Gamma' \Vdash \{\rho\}\ann\tau_1 : u_1
   \\
   \forall \Gamma' \leq_\rho \Gamma.\ \Gamma' \Vdash s = s~ : \{\rho\}\ann\tau_1
    \Longrightarrow \Gamma' \Vdash \{\rho, s/y\} \tau_2 : u
    \end{array}
}
\qquad
\infer{\Gamma \Vdash \tmctx: u}
  {\vdash \Gamma}
  \end{array}
\]
\hrulefill
\caption{Semantic Kinding for Types \fbox{$\Gamma \sem \ann\tau : u$} (inductive)}
\label{fig:semkind}
\end{figure}

To keep the definition compact, we again overload the semantic kinding and equality for types and terms. For example, we define the judgment $\Gamma \Vdash \ann \tau : u$ which falls into two parts: $\Gamma \Vdash \tau : u$ and $\Gamma \Vdash \tmctx : u$ where the latter is simply notation, as $\tmctx$ is not a computation-level type. Similarly, we define $\Gamma \Vdash t = t' : \ann\tau$ to stand for $\Gamma \vdash t = t' : \tau$, i.e. semantic equality for terms, and semantic equality for LF contexts where we write $\Gamma \Vdash t = t' : \tmctx$, although $t$ and $t'$ stand for LF contexts.

Our semantic kinding for types (Fig.~\ref{fig:semkind}) is used as a measure to define the semantic typing for computations. In particular, we define $\Gamma \Vdash \ann\tau = \ann\tau' : u$ and $\Gamma \Vdash t = t' : \ann\tau$ recursively on the semantic kinding of $\ann\tau$. i.e. $\Gamma \Vdash \ann\tau : u$.  For better readability, we simply write for example $\Gamma \Vdash t = t': \cbox{T}$ instead of
$\Gamma \Vdash t = t': \tau$ where $\tau \whnf \cbox{T}$, and $\Gamma \vdash T \equiv T$ in proofs.
\LONGVERSION{
We note that to prove reflexivity for types, we would need to strengthen our semantic kinding definition with the additional premise:
$\forall \Gamma' \leq_\rho \Gamma.\ \Gamma' \Vdash s = s' : \{\rho\}\ann\tau_1
   \Longrightarrow \Gamma' \Vdash \{\rho, s/y\} \tau_2 = \{\rho, s'/y\} \tau_2 : u_2$. This is possible, but since semantic reflexivity for types is not needed, we keep the semantic kinding definition more compact.
}

\begin{figure}[h]
\hrulefill
  \centering
\[
  \begin{array}{l@{~}c@{~}l}
\multicolumn{3}{l}{\mbox{Semantic Equality for Types: \fbox{$\Gamma \Vdash \ann\tau = \ann\tau' :u$} recursively defined by recursion on $\Gamma \Vdash \tau : u$}}\\[0.75em]    
\Gamma \Vdash \tmctx = \tmctx : u & \defiff & true
\\[1em]
\Gamma \Vdash \cbox T = \tau' : u & \defiff &
\Gamma \vdash \tau' \whnf \cbox{T'} : u~
 \mbox{and}~\Gamma \vdash T \equiv T'
\\[1em]
\Gamma \Vdash u' = \tau' : u & \defiff &
\Gamma \vdash \tau' \whnf u' : u
\\[1em]
\Gamma \Vdash (y:\ann\tau_1) \arrow \tau_2 = \tau' : u & \defiff &
\Gamma \vdash \tau' \whnf (y' :\ann\tau_1' )\arrow \tau_2' : u ~\mbox{and}~
\forall \Gamma' \leq_\rho \Gamma.~ \Gamma'\Vdash \{\rho\}\ann\tau_1 = \{\rho\}\ann\tau_1' : u_1~\mbox{and}\\
& &\forall \Gamma'\leq_\rho \Gamma.~ \Gamma' \Vdash s = s' : \{\rho\}\ann\tau_1 \Longrightarrow \Gamma' \Vdash \{\rho, s/y\} \tau_2 = \{\rho,s'/y'\}\tau_2' : u_2 ~\\& & \mbox{where}~ (u_1, u_2, u) \in \Ru
\\[1em]
\Gamma \Vdash x~\vec{t} = \tau' : u & \defiff & \Gamma \vdash \tau' \whnf x~\vec{s} : u~\mbox{and}~\Gamma \vdash x~\vec{t} \equiv x~\vec{s} : u
\\[1em]
\multicolumn{3}{l}{\mbox{Semantic Equality for Terms (Computations): \fbox{$\Gamma \Vdash t = t' : \ann\tau$} by recursion on $\Gamma \Vdash \ann\tau : u$}}\\[0.75em]
\Gamma \Vdash \Psi = \Psi' : \tmctx & \defiff &
\Gamma \vdash \Psi \equiv \Psi' : \tmctx
\\[0.5em]
\Gamma \Vdash t = t' : \cbox {\Psi \vdash A}  & \defiff &
\Gamma \vdash t \whnf w : \cbox{\Psi \vdash A} ~\mbox{and}~
\Gamma \vdash t' \whnf w' : \cbox{\Psi \vdash A} ~\mbox{and}~\\
& & \Gamma ; \Psi \Vdash \unbox{w}{\id} = \unbox{w'}{\id} : A
    \\[0.75em]
\Gamma \Vdash t = t' : u'  &  \multicolumn{2}{@{}l} {\mbox{already defined above by recursion on $\Gamma \Vdash t : u'$}}\\
& \multicolumn{2}{@{}l}{\mbox{we note that we have  $\Gamma \Vdash \tau : u$ where $\Gamma \vdash \tau \whnf u' :  u$ and $u' < u$;}}\\
& \multicolumn{2}{@{}p{11cm}}{hence using the sem. eq. for types
  is well-founded as we do so at a smaller universe $u'$ and our sem. eq. for types is recursively defined on $\Gamma \Vdash t : u'$}
    \\[1.75em]
\Gamma \Vdash t = t' : (y:\ann\tau_1) \arrow \tau_2 & \defiff &
\Gamma \vdash t \whnf w : (y:\ann\tau_1) \arrow \tau_2 ~\mbox{and}~
\Gamma \vdash t'\whnf w': (y:\ann\tau_1) \arrow \tau_2 ~\mbox{and}~\\
&&\forall \Gamma' \leq_\rho \Gamma.
~\Gamma' \Vdash s = s' : \{\rho\}\ann\tau_1 \Longrightarrow \Gamma' \Vdash \{\rho\}w~s = \{\rho\}w'~s' : \{\rho, s/y\}\tau_2
\\[1em]
\Gamma \Vdash t = t' : x~\vec{s} & \defiff &
 \Gamma \vdash t \whnf n :  x~ \vec{s} ~\mbox{and}~
 \Gamma \vdash t' \whnf n' :  x~\vec{s}~\mbox{and}~\neut n,n'~\mbox{and}~
\Gamma \vdash n \equiv n' : x~\vec{s}
  \end{array}
  \]
\hrulefill
  \caption{Semantic Equality}
  \label{fig:sem}
\end{figure}

\begin{lemma}[Weakening of Type Inference for Neutral Computations]\label{lem:wktypinf}
If $\typeof (\Gamma \vdash t) = \tau$ and $\Gamma' \leq_\rho \Gamma$ then $\typeof (\Gamma' \vdash \{\rho\}t) = \tau'$ s.t. $\tau' = \{\rho\}\tau$.
\end{lemma}
\begin{proof}
By induction on    $\typeof (\Gamma \vdash t) = \tau$ using Lemma \ref{lem:weakwhnf} (\ref{it:sweakcomp}).
\end{proof}

\section{Semantic Properties }\label{sec:semprop}
\subsection{Semantic Properties of LF}




 \begin{lemma}[Well-Formedness of Semantic LF Typing] \quad \label{lm:semlfwf}
   \begin{enumerate}
   \item If $\Gamma ; \Psi \Vdash M = N : A$ then
         $\Gamma ; \Psi \vdash M : A$ and $\Gamma ; \Psi \vdash N : A$
         and $\Gamma \vdash M \equiv N : A$.
   \item If $\Gamma ; \Psi \Vdash \sigma_1 = \sigma_2 : \Phi$ then
         $\Gamma ; \Psi \vdash \sigma_1 : \Phi$ and $\Gamma ; \Psi \vdash \sigma_2 : \Phi$
         and $\Gamma ; \Psi \vdash \sigma_1 \equiv \sigma_2 : \Phi$.
   \end{enumerate}
 \end{lemma}
 \begin{proof}
By induction on the semantic definition. In each case, we refer the Def.~\ref{def:typedwhnf}. To illustrate, consider the case where
$\Gamma ; \Psi \vdash M \lfwhnf \lambda x.M' : \Pi x{:}A.B$ and
$\Gamma ; \Psi \vdash N \lfwhnf \lambda x.N' : \Pi x{:}A.B$, we also know that
$\Gamma ; \Psi \vdash M \equiv \lambda x.M' : \Pi x{:}A.B$ and
$\Gamma ; \Psi \vdash N \equiv \lambda x.N' : \Pi x{:}A.B$ by Def.~\ref{def:typedwhnf}.

Further, we have that $\Gamma ; \Psi, x{:}A \Vdash M' = N' : B$. By
IH, we get that $\Gamma ; \Psi, x{:}A \vdash M' \equiv N' : B$ By
dec. equivalence rules, we have $\Gamma ; \Psi \vdash \lambda x.M'
\equiv \lambda x.N' : \Pi x{:}A.B$. Therefore, by symmetry and
transitivity of $\equiv$, we have $\Gamma ; \Psi \vdash M \equiv N :
\Pi x{:}A.B$. The typing invariants are left implicit.
\LONGVERSIONCHECKED{
\\
 We show the expanded proofs below concentrating on showing
  $\equiv$ and leaving the tracking of typing invariants implicit.
  \\[1em]
\SUBSTCLO{
  \pcase {$\infer{\Gamma ; \Psi \Vdash \sigma = \sigma' : \Phi}
    {\begin{array}{lll}
        \Gamma ; \Psi \vdash \sigma \lfwhnf (\sclo \phi {\unbox{t_1}{\sigma_1}})  : \phi
        & \typeof (\Gamma \vdash t_1) = \cbox{\Phi_1 \vdash \Phi'_1}
        & \Gamma \vdash \Phi'_1 \equiv (\phi, \wvec{x{:}A}) : \ctx \\
        \Gamma ; \Psi \vdash \sigma' \lfwhnf (\sclo \phi {\unbox{t_2}{\sigma_2}}) : \phi
        & \typeof (\Gamma \vdash t_2) = \cbox{\Phi_2 \vdash \Phi'_2}
        & \Gamma \vdash \Phi'_2 \equiv (\phi, \wvec{x{:}A}) : \ctx \\
        \Gamma \vdash t_1 \equiv t_2 : \cbox{\Phi_1 \vdash \Phi'_1}
        & \Gamma ; \Psi \Vdash \sigma_1 = \sigma_2 : \Phi_1
        & \Gamma \vdash \Phi_1 \equiv \Phi_2 : \ctx
      \end{array}}$}
  \prf{$\Gamma ; \Psi \vdash \sigma \equiv (\sclo \phi {\unbox{t_1}{\sigma_1'}}) : \phi$
    and $\Gamma ; \Psi \vdash \sigma' \equiv (\sclo \phi {\unbox{t_2}{\sigma_2'}}) : \phi$
    \hfill by Def.~\ref{def:typedwhnf}}
  \prf{$\Gamma ; \Psi \vdash \sigma_1' \equiv \sigma_2' : \Phi_1$
    \hfill by induction hypothesis}
  \prf{$\Gamma ; \Psi \vdash (\sclo \phi {\unbox{t_1}{\sigma_1'}})
    \equiv (\sclo \phi {\unbox{t_2}{\sigma_2'}}) : \phi$
    \hfill by dec. equivalence rules}
  \prf{$\Gamma ; \Psi \vdash \sigma \equiv \sigma': \phi$
    \hfill by symmetry, transitivity, conversion of $\equiv$}
  \\
}
  \pcase {$\infer{\Gamma ; \Psi \Vdash \sigma = \sigma' : \cdot}
    {\Gamma ; \Psi \vdash \sigma \lfwhnf \cdot : \cdot  &
      \Gamma ; \Psi \vdash \sigma' \lfwhnf \cdot : \cdot }$}
  \prf{$\Gamma ; \Psi \vdash \sigma_1 \equiv \cdot : \cdot$
    and $\Gamma ; \Psi \vdash \sigma_2 \equiv \cdot : \cdot$
    \hfill by Def.~\ref{def:typedwhnf}}
  \prf{$\Gamma ; \Psi \vdash \sigma_1 \equiv \sigma_2 : \cdot$ \hfill   by symmetry and transitivity of $\equiv$}
  \\
  \pcase {$ \infer{\Gamma ; \psi, \wvec{x{:}A} \Vdash \sigma = \sigma' : \psi}{
      \Gamma ; \psi, \wvec{x{:}A} \vdash \sigma \lfwhnf \wk{\psi} : \psi &
      \Gamma ; \psi, \wvec{x{:}A} \vdash \sigma' \lfwhnf \wk{\psi} : \psi
    }$}
  \prf{$\Gamma ; \psi, \wvec{x{:}A} \vdash \sigma_1 \equiv \wk{\psi} : \psi$ and
    $\Gamma ; \psi, \wvec{x{:}A} \vdash \sigma_2 \equiv \wk{\psi} : \psi$
    \hfill by Def.~\ref{def:typedwhnf}}
  \prf{$\Gamma ; \psi \wvec{x{:}A} \vdash \sigma_1 \equiv \sigma_2 : \psi$
    \hfill by symmetry and transitivity of $\equiv$}
  \\
  \pcase{$ \infer{\Gamma ; \Psi \Vdash \sigma_1 = \sigma_2 : \Phi, x{:}A}
    {\begin{array}{ll}
       \Gamma ; \Psi \vdash \sigma_1 \lfwhnf \sigma'_1, M : \Phi, x{:}A  & \\
       \Gamma ; \Psi \vdash \sigma_2 \lfwhnf \sigma'_2, N : \Phi, x{:}A &
       \Gamma ; \Psi \Vdash \sigma'_1 = \sigma'_2 : \Phi
       \qquad \Gamma ; \Psi  \Vdash M = N : \lfs{\sigma'_1}{\Phi} A
      \end{array}}$}
  \prf{$\Gamma ; \Psi \vdash \sigma_1 \equiv \sigma_1', M : \Phi, x{:}A$
    and $\Gamma ; \Psi \vdash \sigma_2 \equiv \sigma_2', N : \Phi, x{:}A$
    \hfill by Def.~\ref{def:typedwhnf}}
  \prf{$\Gamma ; \Psi \vdash \sigma_1' \equiv \sigma_2' : \Phi$
    and $\Gamma ; \Psi  \vdash M \equiv N : \lfs{\sigma_1'}{\Phi} A$
    \hfill by induction hypothesis}
  \prf{$\Gamma ; \Psi \vdash \sigma_1', M \equiv \sigma_2', N : \Phi, x{:}A$
    \hfill by dec. equivalence rules}
  \prf{$\Gamma ; \Psi \vdash \sigma_1 \equiv \sigma_2 : \Phi, x{:}A$
    \hfill by symmetry and transitivity of $\equiv$}
  \\
  \pcase{$\infer{\Gamma ; \Psi \Vdash M = N : \tm}
    {\begin{array}{lll}
        \Gamma ; \Psi \vdash M \lfwhnf \unbox {t_1} {\sigma_1} : \tm
        & \typeof (\Gamma  \vdash t_1) = \cbox{\Phi_1 \vdash \tm} &
        \Gamma \vdash \Phi_1 \equiv \Phi_2 : \ctx \\
        \Gamma ; \Psi \vdash N \lfwhnf \unbox {t_2} {\sigma_2} : \tm
        & \typeof (\Gamma  \vdash t_2) = \cbox{\Phi_2 \vdash \tm} &
        \Gamma  \vdash t_1 \equiv t_2 : \cbox{\Phi_1 \vdash \tm}
        \quad \Gamma ; \Psi \Vdash \sigma_1 = \sigma_2 : \Phi_1
      \end{array}}$}
  \prf{$\Gamma ; \Psi \vdash M \equiv \unbox {t_1} {\sigma_1} : \tm $
    and $\Gamma ; \Psi \vdash N \equiv \unbox {t_2} {\sigma_2} : \tm$
    \hfill by Def.~\ref{def:typedwhnf}}
  \prf{$\Gamma ; \Psi \vdash \sigma_1 \equiv \sigma_2 : \Phi_1$ \hfill by induction hypothesis}
  \prf{$\Gamma ; \Psi \vdash \unbox {t_1} {\sigma_1} \equiv \unbox {t_2} {\sigma_2} : \tm$
    \hfill by dec. equivalence rules}
  \prf{$\Gamma ; \Psi \vdash M \equiv N : \tm$ \hfill by symmetry and transitivity  of $\equiv$}
  \\
  \pcase{$\infer {\Gamma ; \Psi \Vdash M = N : \tm}
    {\begin{array}{ll}
        \Gamma ; \Psi \vdash M \lfwhnf \clam M' : \tm & \\
        \Gamma ; \Psi \vdash N \lfwhnf \clam N' : \tm
        & \Gamma ; \Psi \Vdash M' = N': \Pi x{:}\tm.\tm
      \end{array}}$}
  \prf{$\Gamma ; \Psi \vdash M \equiv \clam M' : \tm$
    and $\Gamma ; \Psi \vdash N \equiv \clam N' : \tm$
    \hfill by Def.~\ref{def:typedwhnf}}
  \prf{$\Gamma ; \Psi \vdash M' \equiv N': \Pi x{:}\tm.\tm$
    \hfill by induction hypothesis}
  \prf{$\Gamma ; \Psi \vdash \clam M' \equiv \clam N' : \tm$
    \hfill by dec. equivalence rules}
  \prf{$\Gamma ; \Psi \vdash M \equiv N : \tm$ \hfill by symmetry and transitivity  of $\equiv$}
  \\
  \pcase{$\infer {\Gamma ; \Psi \Vdash M = N : \tm}
    {\begin{array}{ll}
        \Gamma ; \Psi \vdash M \lfwhnf \capp {M_1}~{M_2} : \tm & \\
        \Gamma ; \Psi \vdash N \lfwhnf \capp {N_1}~{N_2} : \tm &
        \Gamma ; \Psi \Vdash M_1 = N_1 : \tm \quad \Gamma ; \Psi \Vdash M_2 = N_2 : \tm
      \end{array}}$}
  \prf{$\Gamma ; \Psi \vdash M \equiv \capp {M_1}~{M_2} : \tm$
    and $\Gamma ; \Psi \vdash N \equiv \capp {N_1}~{N_2} : \tm$
    \hfill by Def.~\ref{def:typedwhnf}}
  \prf{$\Gamma ; \Psi \vdash M_1 \equiv N_1 : \tm$
    and $\Gamma ; \Psi \vdash M_2 \equiv N_2 : \tm$
    \hfill by induction hypothesis}
  \prf{$\Gamma ; \Psi \vdash \capp {M_1}~{M_2} \equiv \capp {N_1}~{N_2}$
    \hfill by dec. equivalence rules}
  \prf{$\Gamma ; \Psi \vdash M \equiv N : \tm$ \hfill by symmetry and transitivity  of $\equiv$}
  \\
  \pcase{$\infer {\Gamma ; \Psi \Vdash M = N : \tm}
    { \Gamma ; \Psi \vdash M \lfwhnf x : \tm  &
      \Gamma ; \Psi \vdash N \lfwhnf x : \tm }$}
  \prf{$\Gamma ; \Psi \vdash M \equiv x : \tm$
    and $\Gamma ; \Psi \vdash N \equiv x : \tm$
    \hfill by Def.~\ref{def:typedwhnf}}
  \prf{$\Gamma ; \Psi \vdash M \equiv N : \tm$ \hfill by symmetry and transitivity  of $\equiv$}}
 \end{proof}

\begin{lemma}[Semantic Weakening for LF]\label{lem:semweak}\quad
\begin{enumerate}
  \item \label{it:sweaksemlf} If\/ $\Gamma ; \Psi \Vdash M = N: A$ and $\Gamma' \leq_\rho \Gamma$
     then $\Gamma' ; \{\rho\}\Psi \Vdash \{\rho\}M = \{\rho\}N: \{\rho\}A$.
  \item \label{it:sweaklfsub} If\/ $\Gamma ; \Psi \Vdash \sigma = \sigma': \Phi$ and $\Gamma' \leq_\rho \Gamma$
     then $\Gamma' ; \{\rho\}\Psi \Vdash \{\rho\}\sigma = \{\rho\}\sigma': \{\rho\}\Phi$.
  \end{enumerate}
\end{lemma}
\begin{proof}
By induction on the first derivation.
\LONGVERSIONCHECKED{
\\[1em]
\pcase{$\ianc
{\begin{array}{lll}
\Gamma ; \Psi \vdash M \lfwhnf \unbox {t_1} {\sigma_1} : \tm   & \typeof (\Gamma  \vdash t_1) = \cbox{\Phi_1 \vdash \tm} &
\Gamma \vdash \Phi_1 \equiv \Phi_2 : \ctx\\
\Gamma ; \Psi \vdash N \lfwhnf \unbox {t_2} {\sigma_2} : \tm & \typeof (\Gamma  \vdash t_2) = \cbox{\Phi_2 \vdash \tm} &
\Gamma ; \Psi \Vdash \sigma_1 = \sigma_2 : \Phi_1 \quad
\Gamma  \vdash t_1 \equiv t_2 : \cbox{\Phi_1 \vdash \tm}
\end{array}}
{\Gamma \Vdash M = N : \tm}{}$}
    \prf{$\Gamma'; \{\rho\}\Psi \vdash \{\rho\}M \lfwhnf \{\rho\}(\unbox {t_1} {\sigma_1}) : \tm$ and
         $\Gamma'; \{\rho\}\Psi \vdash \{\rho\}N \lfwhnf \{\rho\}(\unbox {t_2} {\sigma_1}) : \tm$ \hfill by Lemma \ref{lem:weakwhnf}}
    \prf{$\Gamma'; \{\rho\}\Psi \vdash \{\rho\}M \lfwhnf \unbox {\{\rho\}t_1} {\{\rho\}\sigma_1} : \tm$ and
         $\Gamma'; \{\rho\}\Psi \vdash \{\rho\}N \lfwhnf \unbox {\{\rho\}t_2} {\{\rho\}\sigma_2} : \tm$ \hfill by subst. def.}
    \prf{$\Gamma'; \{\rho\}\Psi \Vdash \{\rho\}\sigma_1 = \{\rho\}\sigma_2: \{\rho\}\Phi_1$ \hfill by IH}
    \prf{$\Gamma' \vdash \{\rho\}\Phi_1 \equiv \{\rho\}\Phi_2 : \ctx$ \hfill by Lemma \ref{lem:weakcomp}}
    \prf{$\Gamma' \vdash \{\rho\}t_1 \equiv \{\rho\}t_2 : \{\rho\}\cbox{\Phi_1 \vdash \tm}$ \hfill by Lemma \ref{lem:weakcomp}}
    \prf{$\Gamma' \vdash \{\rho\}t = \{\rho\}t': \cbox{\{\rho\}\Phi_1 \vdash \tm}$ \hfill by substitution def.}
    \prf{$\typeof (\Gamma'  \vdash \{\rho\}t_1) = \cbox{\{\rho\}\Phi_1 \vdash \tm}$ \hfill by Lemma \ref{lem:wktypinf}}
    \prf{$\typeof (\Gamma'  \vdash \{\rho\}t_2) = \cbox{\{\rho\}\Phi_2 \vdash \tm}$ \hfill by Lemma \ref{lem:wktypinf} }
    \prf{$\Gamma' ; \{\rho\}\Psi \Vdash \{\rho\}M = \{\rho\}N: \{\rho\}\tm$ \hfill by rule and substitution def.}
\\
}
\end{proof}


\begin{lemma}[Backwards Closure for LF terms]\label{lem:lfbclosed}\quad
  \begin{enumerate}
  \item If $\Gamma; \Psi \Vdash  Q = N: A$ and $\Gamma ; \Psi \vdash M \lfwhnf Q : A$
then $\Gamma \sem M = N: A$
   \item If $\Gamma; \Psi \Vdash N = Q : A$ and $\Gamma ; \Psi \vdash M \lfwhnf Q : A$
then $\Gamma \sem N = M: A$
  \end{enumerate}
\end{lemma}
\begin{proof}
By case analysis on $\Gamma; \Psi \Vdash  Q = N: A$ and the fact that
$Q$ is in $\norm$. 
\LONGVERSIONCHECKED{
$\quad$\\[1em]
\pcase{$\ianc
{
\begin{array}{lll}
\Gamma ; \Psi \vdash Q \lfwhnf \unbox {t_1} {\sigma_1} : \tm   & \typeof (\Gamma  \vdash t_1) = \cbox{\Phi_1 \vdash \tm} &
\Gamma \vdash \Phi_1 \equiv \Phi_2 : \ctx \\
\Gamma ; \Psi \vdash N \lfwhnf \unbox {t_2} {\sigma_2} : \tm & \typeof (\Gamma  \vdash t_2) = \cbox{\Phi_2 \vdash \tm} &
\Gamma  \vdash t_1 \equiv t_2 : \cbox{\Phi_1 \vdash \tm} \quad \Gamma ; \Psi \Vdash \sigma_1 = \sigma_2 : \Phi_1
\end{array}
}
{\Gamma ; \Psi \Vdash Q = N : \tm}{}$}
\prf{$\Gamma ; \Psi \vdash M \lfwhnf Q : A$ \hfill by assumption}
\prf{$\norm Q$ \hfill by invariant of $\lfwhnf$}
\prf{$Q = \unbox {t_1} {\sigma_1}$ \hfill since $\norm Q$}
\prf{$\Gamma ; \Psi \Vdash M = N : \tm$ \hfill using $\M \lfwhnf \unbox{t_1}{\sigma_1}$ and sem. def.}
\\[1em]
\pcase{$\ianc
{
      \begin{array}{ll}
\Gamma ; \Psi \vdash Q \lfwhnf \clam M' : \tm & \\
\Gamma ; \Psi \vdash N \lfwhnf \clam N' : \tm & \Gamma ; \Psi \Vdash M' = N': \Pi x{:}\tm.\tm
      \end{array}
}{\Gamma ; \Psi \Vdash Q = N : \tm}{}
$}
\prf{$\Gamma ; \Psi \vdash M \lfwhnf Q : A$ \hfill by assumption}
\prf{$\norm Q$ \hfill by invariant of $\lfwhnf$}
\prf{$Q = \clam M' $ \hfill by $\Gamma ; \Psi \vdash Q \lfwhnf \clam M' : \tm $ using $\norm Q$}
\prf{$\Gamma ; \Psi \Vdash M = N : \tm$ \hfill using $\Gamma ; \Psi \vdash M \lfwhnf \clam M' : \tm$}
}
\end{proof}

 \begin{lemma}[Semantic LF Equality is preserved under LF Substitution] \label{lem:semlfeqsub}\quad
   \begin{enumerate}
   \item If $\Gamma ; \Psi \sem \sigma = \sigma' : \Phi$
and $\Gamma ; \Phi \sem M = N : A$
then $\Gamma ; \Psi \sem \lfs\sigma\Phi M = \lfs{\sigma'}\Phi N : \lfs\sigma\Phi A$.
   \item If $\Gamma ; \Psi \sem \sigma = \sigma' : \Phi$
and $\Gamma ; \Phi \sem \sigma_1 = \sigma_2 : \Phi'$
then $\Gamma ; \Psi \sem \lfs\sigma\Phi \sigma_1 = \lfs{\sigma'}\Phi \sigma_2 : \Phi'$.
   \end{enumerate}
 \end{lemma}
 \begin{proof}
Proof by mutual induction on $\Gamma ; \Phi \sem M = N : A$ and $\Gamma ; \Phi \sem \sigma = \sigma' : \Phi$ using the fact that weak head reduction is preserved under substitution (Lemma \ref{lm:lfwhnfsub}).
\LONGVERSIONCHECKED{
\\[0.5em]
 (1) \fbox{If $\Gamma ; \Psi \sem \sigma = \sigma' : \Phi$
and $\Gamma ; \Phi \sem M = N : A$
then $\Gamma ; \Psi \sem \lfs\sigma\Phi M = \lfs{\sigma'}\Phi N :\lfs\sigma\Phi A$.}
\\[0.5em]
\pcase{
$\ianc
{
\begin{array}{lll}
\Gamma ; \Phi \vdash M \lfwhnf \unbox {t_1} {\sigma_1} : \tm   & \typeof (\Gamma  \vdash t_1) = \cbox{\Phi_1 \vdash \tm} &
\Gamma \vdash \Phi_1 = \Phi_2 : \ctx\\
\Gamma ; \Phi \vdash N \lfwhnf \unbox {t_2} {\sigma_2} : \tm & \typeof (\Gamma  \vdash t_2) = \cbox{\Phi_2 \vdash \tm} &
\Gamma ; \Phi \Vdash \sigma_1 = \sigma_2 : \Phi_1 \quad
\Gamma \vdash t_1 \equiv t_2 : \cbox{\Phi_1 \vdash \tm}
\end{array}
}
{\Gamma ; \Phi \Vdash M = N : \tm}{}$
}
\prf{$\Gamma ; \Psi \vdash \sigma : \Phi$ and $\Gamma ; \Psi \vdash \sigma' : \Phi$ \hfill by well-formedness of semantic equ. (Lemma \ref{lm:semlfwf})}
\prf{$\Gamma ; \Psi \vdash \lfs\sigma\Phi M \lfwhnf \lfs\sigma\Phi (\unbox{t_1}{\sigma_1}) : \tm$ \hfill by Lemma \ref{lm:lfwhnfsub} }
\prf{$\Gamma ; \Psi \vdash \lfs{\sigma'}{\Phi} N \lfwhnf \lfs{\sigma'}{\Phi}(\unbox{t_2}{\sigma_2}) : \tm$ \hfill by Lemma \ref{lm:lfwhnfsub}}
\prf{$\Gamma ; \Phi \Vdash \lfs\sigma\Phi\sigma_1 = \lfs{\sigma'}\Phi \sigma_2 : \Phi_1$ \hfill by IH}
\prf{$\Gamma ; \Phi \Vdash \lfs{\sigma}\Phi M = \lfs{\sigma'}\Phi N : \tm$ \hfill by well-typed }
\\
\pcase{
$\ianc
{ \Gamma ; \Phi \vdash M \lfwhnf x : \tm  \quad
  \Gamma ; \Phi \vdash N \lfwhnf x : \tm }
 {\Gamma ; \Phi \Vdash M = N : \tm}{}$
}
\prf{$\Gamma ; \Psi \Vdash \sigma(x) = \sigma'(x) : \tm$ \hfill by $\Gamma ; \Phi \Vdash \sigma = \sigma' : \Psi$}
\prf{$\Gamma ; \Psi \vdash \sigma(x) \lfwhnf M' : \tm$ and $\Gamma ; \Phi \vdash \sigma'(x) \lfwhnf N' : \tm$
 \hfill by $\Gamma ; \Psi \Vdash \sigma(x) = \sigma'(x) : \tm$
}
\prf{$\Gamma ; \Psi \Vdash M' = N' : \tm$ \hfill since both $\norm M'$ and $\norm N'$}
\prf{$\Gamma ; \Psi \vdash \lfs\sigma\Phi M \lfwhnf M' : \tm$    \hfill by Lemma \ref{lm:lfwhnfsub}
}
\prf{$\Gamma ; \Psi \vdash \lfs\sigma\Phi N \lfwhnf N' : \tm$    \hfill by Lemma \ref{lm:lfwhnfsub}}
\prf{$\Gamma ; \Phi \Vdash \lfs\sigma\Phi M = \lfs{\sigma'}\Phi N : \tm$ \hfill Backwards Closure (Lemma \ref{lem:lfbclosed})}
\\
Other cases are similar.
}
\LONGVERSIONCHECKED{
\\[1em]
(2) \fbox{If $\Gamma ; \Psi \sem \sigma = \sigma' : \Phi$
and $\Gamma ; \Phi \sem \sigma_1 = \sigma_2 : \Phi'$
then $\Gamma ; \Psi \sem \lfs\sigma\Phi \sigma_1 = \lfs{\sigma'}\Phi \sigma_2 : \Phi'$.}\\[1em]
Proof by induction on $\Gamma ; \Phi \sem \sigma = \sigma' : \Phi$ using the fact that weak head reduction is preserved under substitution (Lemma \ref{lm:lfwhnfsub}).
\\[1em]
\SUBSTCLO{
 \pcase{$\ianc {
       \begin{array}{lll}
 \Gamma ; \Phi \vdash \sigma_1 \lfwhnf (\sclo \phi {\unbox{t_1}{\sigma'_1}})  : \phi
    & \typeof (\Gamma \vdash t_1) = \cbox{\Phi_1 \vdash \Phi'_1}  & \Gamma \vdash \Phi'_1 \equiv (\phi, \wvec{x{:}A}) : \ctx
 \\
 \Gamma ; \Phi \vdash \sigma_2 \lfwhnf (\sclo \phi {\unbox{t_2}{\sigma'_2}}) : \phi
     & \typeof (\Gamma \vdash t_2) = \cbox{\Phi_2 \vdash \Phi'_2} & \Gamma \vdash \Phi'_2 \equiv (\phi, \wvec{x{:}A}) : \ctx
 \\
 \Gamma \vdash t_1 \equiv t_2 : \cbox{\Phi_1 \vdash \Phi'_1}
     & \Gamma ; \Phi \Vdash \sigma'_1 = \sigma'_2 : \Phi_1
     & \Gamma \vdash \Phi_1 \equiv \Phi_2 : \ctx
 \end{array}
 }
 {\Gamma ; \Phi \Vdash \sigma_1 = \sigma_2 : \phi}{}$}
 \prf{$\Gamma ; \Psi \vdash \lfs{\sigma}{\Phi}(\sigma_1) \lfwhnf \lfs{\sigma}{\Phi} (\sclo \phi {\unbox{t_1}{\sigma'_1}})  : \phi$  \hfill by Lemma \ref{lm:lfwhnfsub}}
 \prf{$\Gamma ; \Psi \vdash \lfs{\sigma'}{\Phi}(\sigma_2) \lfwhnf \lfs{\sigma'}{\Phi} (\sclo \phi {\unbox{t_2}{\sigma'_2}})  : \phi$   \hfill by Lemma \ref{lm:lfwhnfsub}}
 \prf{$\lfs{\sigma}{\Phi} (\sclo \phi {\unbox{t_1}{\sigma'_1}}) = \sclo \phi {\unbox{t_1}{\lfs{\sigma}{\phi}{\sigma'_1}}}$ \hfill by LF subst. def.}
 \prf{$\lfs{\sigma'}{\Phi} (\sclo \phi {\unbox{t_2}{\sigma'_2}}) = \sclo \phi {\unbox{t_2}{\lfs{\sigma'}{\phi}{\sigma'_2}}}$ \hfill by LF subst. def.}
 \prf{$\Gamma ; \Psi \Vdash \lfs{\sigma}{\Phi}\sigma'_1 = \lfs{\sigma'}{\Phi}\sigma'_2: \Phi_1$ \hfill by IH}
 \prf{$\Gamma ; \Psi \Vdash \lfs{\sigma}{\Phi}\sigma_1 = \lfs{\sigma'}{\Phi}\sigma_2:\phi$ \hfill by sem. equ. def.}
 \\[1em]
}
 \pcase{$
 \ianc {
 \Gamma ; \Phi \vdash \sigma_1 \lfwhnf \cdot : \cdot  \quad
 \Gamma ; \Phi \vdash \sigma_2 \lfwhnf \cdot : \cdot }
 {\Gamma ; \Phi \Vdash \sigma_1 = \sigma_2 : \cdot}{}
 $}
\prf{$\Gamma ; \Psi \vdash \lfs{\sigma}{\Phi} \sigma_1 \lfwhnf \lfs{\sigma}{\Phi} \cdot : \cdot$ \hfill by Lemma \ref{lm:lfwhnfsub}}
\prf{$\Gamma ; \Psi \vdash \lfs{\sigma}{\Phi} \sigma_2 \lfwhnf \lfs{\sigma}{\Phi} \cdot : \cdot$ \hfill by Lemma \ref{lm:lfwhnfsub}}
\prf{$\lfs{\sigma}\Phi \cdot = \cdot$ \hfill by LF subst. def.}
\prf{$\Gamma \Psi \Vdash \lfs{\sigma}{\Phi}\sigma_1 = \lfs{\sigma}{\Phi} \sigma_2 : \cdot$ \hfill by sem. eq. def.}
 \\[1em]
 \pcase{$\ianc
 {
 \Gamma ; \phi, \wvec{x{:}A} \vdash \sigma_1 \lfwhnf \wk{\phi} : \phi \qquad
 \Gamma ; \phi, \wvec{x{:}A} \vdash \sigma_2 \lfwhnf \wk{\phi} : \phi
 }
 {\Gamma ; \phi, \wvec{x{:}A} \Vdash \sigma_1 = \sigma_2 : \phi}{}
 $}
\prf{$\Gamma ; \Psi \vdash [\sigma / \phi, \vec x] \sigma_1 \lfwhnf [\sigma / \phi, \vec x] \wk{\phi} : \phi$ \hfill by Lemma \ref{lm:lfwhnfsub}}
\prf{$[\sigma / \phi, \vec x] \wk{\phi} = \trunc_{\phi} (\sigma / \phi, \vec x) = \sigma'_1$ where $\Gamma ; \Psi \vdash \sigma'_1 : \phi$ \hfill by LF subst. def.}
\prf{$\Gamma ; \Psi \vdash [\sigma' / \phi, \vec x] \sigma_2 \lfwhnf [\sigma' / \phi, \vec x] \wk{\phi} : \phi$ \hfill by Lemma \ref{lm:lfwhnfsub}}
\prf{$[\sigma' / \phi, \vec x] \wk{\phi} = \trunc_{\phi} (\sigma' / \phi, \vec x) = \sigma'_2$ where $\Gamma ; \Psi \vdash \sigma'_2 : \phi$ \hfill by LF subst. def.}
\prf{$\Gamma ; \Psi \Vdash \sigma'_1 = \sigma'_2 : \phi$ \hfill since $\Gamma ; \Psi \Vdash \sigma = \sigma' : \phi, \vec x$ }
 \\[1em]
 \pcase{$
 \ianc {\begin{array}{ll}
 \Gamma ; \Phi \vdash \sigma_1 \lfwhnf \sigma'_1, M : \Phi', x{:}A  & \\
 \Gamma ; \Phi \vdash \sigma_2 \lfwhnf \sigma'_2, N : \Phi', x{:}A &
 \Gamma ; \Phi \Vdash \sigma'_1 = \sigma'_2 : \Phi' \qquad \Gamma ; \Psi  \Vdash M = N : \lfs{\sigma'_1}{\Phi'} A
 \end{array}
 }
 {\Gamma ; \Phi \Vdash \sigma_1 = \sigma_2 : \Phi', x{:}A}{}
 $
 }
\prf{$\Gamma ; \Psi \vdash \lfs{\sigma}{\Phi}\sigma_1 \lfwhnf \lfs{\sigma}\Phi {(\sigma'_1,M)} :  \Phi', x{:}A$ \hfill by Lemma \ref{lm:lfwhnfsub}}
\prf{$\Gamma ; \Psi \vdash \lfs{\sigma'}{\Phi}\sigma_2 \lfwhnf \lfs{\sigma'}\Phi {(\sigma'_2,N)} :  \Phi', x{:}A$ \hfill by Lemma \ref{lm:lfwhnfsub}}
\prf{$\lfs{\sigma'}\Phi {(\sigma'_2,N)} = \lfs{\sigma'}\Phi \sigma'_2,~\lfs{\sigma'}\Phi N$ \hfill by LF subst. def.}
\prf{$\lfs{\sigma}\Phi {(\sigma'_1,M)} = \lfs{\sigma}\Phi \sigma'_1,~\lfs{\sigma}\Phi M$ \hfill by LF subst. def.}
\prf{$\Gamma ; \Psi \Vdash \lfs{\sigma}\Phi \sigma'_1 = \lfs{\sigma'}\Phi \sigma'_2 :  \Phi'$ \hfill by IH}
\prf{$\Gamma ; \Psi \Vdash \lfs{\sigma}\Phi M = \lfs{\sigma'}\Phi N : \lfs{\sigma}{\Phi}(\lfs{\sigma'_1}{\Phi'}A)$ \hfill by IH}
\prf{$\Gamma ; \Psi \lfs\sigma\Phi \sigma_1 = \lfs{\sigma'}\Phi \sigma_2 : \Phi', x{:}A$ \hfill by sem. equ. def.}
\\[1em]
}
 \end{proof}

 \begin{lemma}[Semantic Weakening Substitution Exist]\label{lm:semlfwk}\quad \\
If $\Gamma ; \Psi, \wvec{x{:}A} \vdash \wk{\hatctx\Psi} : \Psi$
then $\Gamma ; \Psi, \wvec{x{:}A} \Vdash \wk{\hatctx\Psi} = \wk{\hatctx\Psi} : \Psi$.
 \end{lemma}
 \begin{proof}
By induction on the LF context $\Psi$.
\LONGVERSIONCHECKED{
\pcase{$\Psi = \cdot$.}
\prf{$\Gamma ; \cdot, \wvec{x{:}A} \vdash \wk\cdot : \cdot$ \hfill by  assumption}
\prf{$\Gamma ; \cdot, \wvec{x{:}A} \vdash \wk\cdot \whnf \cdot : \cdot$ \hfill by $\whnf$ rule and typing}
\prf{$\Gamma ; \cdot, \wvec{x{:}A} \Vdash \wk\cdot = \wk\cdot : \cdot$ \hfill by semantic. def.}

\pcase{$\Psi = \Psi', y{:}B$}
\prf{$\Gamma; \Psi', y{:}B, \wvec{x{:}A} \vdash\wk{\hatctx\Psi',y} : \Psi', y{:}B$ \hfill by assumption}
\prf{$\Gamma ; \Psi', y{:}B, \wvec{x{:}A} \vdash \wk{\hatctx\Psi'} : \Psi'$ \hfill by typing }
\prf{$\Gamma ; \Psi', y{:}B, \wvec{x{:}A} \Vdash \wk{\hatctx\Psi'} = \wk{\hatctx\Psi'} : \Psi'$ \hfill by IH}
\prf{$\Gamma ; \Psi', y{:}B, \wvec{x{:}A} \Vdash y = y : B$ \hfill by semantic eq. for LF terms, the fact that $\norm x$, and $B = \lfs{\wk{\hatctx\Psi'}}{\Psi'} B$}
\prf{$\Gamma ; \Psi', y{:}B, \wvec{x{:}A} \vdash \wk{\hatctx\Psi',y} \whnf \wk{\hatctx\Psi'}, y : \Psi', y{:}B$ \hfill by $\whnf$ and typing rules}
\prf{$\Gamma ; \Psi', y{:}B, \wvec{x{:}A} \vdash \wk{\hatctx\Psi',y} = \wk{\hatctx\Psi',y} : \Psi', y{:}B$ \hfill by sem. eq. for LF substitutions}

\pcase{$\Psi = \psi$}
\prf{$\Gamma ; \psi, \wvec{x{:}A} \vdash \wk\psi : \psi$ \hfill by assumption}
\prf{$\Gamma ; \psi, \wvec{x{:}A} \vdash \wk\psi : \psi$ \hfill by $\whnf$ and typing and the fact that $\norm \wk\psi$}
\prf{$\Gamma ; \psi, \wvec{x{:}A} \vdash \wk\psi = \wk\psi: \psi$  \hfill by sem. eq. for LF subst.}
}
 \end{proof}

\begin{lemma}[Semantic LF Context Conversion]\label{lm:semlfctxconv}
\quad
\begin{enumerate}
\item If $\Gamma ; \Psi, x{:}A_1 \Vdash M = N : B$ and $\Gamma ; \Psi \vdash A_1 \equiv A_2 : \lftype$
then $\Gamma ; \Psi, x{:}A_2 \Vdash M = N : B$

\item If $\Gamma ; \Psi, x{:}A_1 \Vdash \sigma = \sigma' : \Phi$ and $\Gamma ; \Psi \vdash A_1 \equiv A_2 : \lftype$
then $\Gamma ; \Psi, x{:}A_2 \Vdash \sigma = \sigma' : \Phi$.

\end{enumerate}
\end{lemma}
\begin{proof} The idea is to use  $\Gamma ; \Psi \vdash A_1 \equiv A_2 : \lftype$  and build LF weakening substitutions
$\Gamma ; \Psi,  x{:}A_2, y{:}A_1 \Vdash \wk{\hatctx\Psi}, y = \wk{\hatctx\Psi}, y : \Psi, x{:}A_1$
and
$\Gamma ; \Psi,  x{:}A_2 \Vdash \wk{\hatctx\Psi}, x, x = \wk{\hatctx\Psi}, x, x : \Psi, x{:}A_2, y{:}A_1$.
Using semantic LF subst. (Lemma \ref{lem:semlfeqsub}), we can then
move $\Gamma ; \Psi, x{:}A_1 \Vdash M = N : B$ to the new LF context $\Psi, x{:}A_2$.
~
\LONGVERSIONCHECKED{\\[0.5em]
 (1): \fbox{ If $\Gamma ; \Psi, x{:}A_1 \Vdash M = N : B$ and $\Gamma ; \Psi \vdash A_1 \equiv A_2 : \lftype$
then $\Gamma ; \Psi, x{:}A_2 \Vdash M = N : B$}
\\[1em]
\prf{$\Gamma \vdash \Psi : \ctx$ \hfill by Well-Formedness of Sem. LF
  Equ. (Lemma \ref{lm:semlfwf}) \\
\mbox{$\quad$}\hfill and Well-formedness of LF context (Lemma \ref{lm:lfctxwf}) }
 \prf{$\Gamma \vdash \Psi, x{:}A_2 : \ctx$ \hfill by context well-formedness rules}
 \prf{$\Gamma ; \Psi \vdash A_2 \equiv A_1 : \lftype$ \hfill by symmetry}
 \prf{$\Gamma ; \Psi, x{:}A_2 \vdash A_2 \equiv A_1 : \lftype$ \hfill by LF weakening  }
 \prf{$\Gamma ; \Psi, x{:}A_2 \vdash x : A_2$ \hfill by typing rule using $\Gamma \vdash \Psi, x{:}A_2 : \ctx$}
 \prf{$\Gamma ; \Psi, x{:}A_2 \vdash x : A_1$ \hfill conversion using $\Gamma ; \Psi, x{:}A_2 \vdash A_2 \equiv A_1 : \lftype$}
 \prf{$\Gamma ; \Psi, x{:}A_2 \vdash x : \lfs{\wk{\hatctx\Psi}}{\Psi} A_1$ \hfill as $A_1 = \lfs{\wk{\hatctx\Psi}}{\Psi} A_1$ }
 \prf{$\Gamma ; \Psi, x{:}A_2 \vdash \wk{\hatctx{\Psi}}, x, x : \Psi, x{:}A_2, y{:}A_1$ \hfill by typing rules for LF substitution}
 \prf{$\Gamma ; \Psi, x{:}A_2, y{:}A_1 \vdash \wk{\hatctx\Psi}, y : \Psi, x{:}A_1$ \hfill by typing rule for LF substitution}
\prf{$\Gamma ; \Psi, x{:}A_2, y{:}A_1 \vdash \wk{\hatctx\Psi} : \Psi$ \hfill by typing}
\prf{$\Gamma ; \Psi, x{:}A_2, y{:}A_1 \Vdash \wk{\hatctx\Psi} = \wk{\hatctx\Psi} : \Psi$ \hfill by Lemma \ref{lm:semlfwk}}
\prf{$\Gamma ; \Psi,  x{:}A_2, y{:}A_1 \Vdash \wk{\hatctx\Psi}, y = \wk{\hatctx\Psi}, y : \Psi, x{:}A_1$
   \hfill by sem. equ. for LF subst. \\\mbox{\hspace{1cm}}\hfill using the fact that $\norm y$}
\prf{$\Gamma ; \Psi, x{:}A_2 \Vdash \wk{\hatctx\Psi} = \wk{\hatctx\Psi} : \Psi$ \hfill by Lemma \ref{lm:semlfwk}}
\prf{$\Gamma ; \Psi,  x{:}A_2 \Vdash \wk{\hatctx\Psi}, x, x = \wk{\hatctx\Psi}, x, x : \Psi, x{:}A_2, y{:}A_1$
   \hfill by sem. equ. for LF subst. \\ \mbox{\hspace{1cm}}\hfill using the fact that $\norm x$}
 \prf{$\Gamma ; \Psi, x{:}A_2 \Vdash \lfs{\wk{\hatctx{\Psi}}, x, x}  {\Psi, x, y}M' = \lfs{\wk{\hatctx{\Psi}}, x, x}{\Psi, x, y}N' :
   \lfs{\wk{\hatctx{\Psi}}, x, x}{\Psi, x, y} B$\\
\mbox{$\quad$}\hfill where $M' = \lfs{\wk{\hatctx\Psi}, y}{\Psi, x}M$
                    and $N' = \lfs{\wk{\hatctx\Psi}, y}{\Psi, x}N$  by semantic LF subst. (Lemma \ref{lem:semlfeqsub} twice)}
 \prf{$\lfs{\wk{\hatctx{\Psi}}, x, x}{\Psi, x, y}(\wk{\hatctx\Psi}, y) = \wk{\hatctx\Psi}, x$ \hfill by subst. def.}
 \prf{$\Gamma ; \Psi, x{:}A_2 \vdash \lfs{\wk{\hatctx\Psi}, x}{\Psi, x} M = \lfs{\wk{\hatctx\Psi}, x}{\Psi, x} N : \lfs{\wk{\hatctx\Psi}, x}{\Psi, x} B$ \hfill by previous lines}
 \prf{$\Gamma ; \Psi, x{:}A_2 \Vdash M = N : B$ \hfill using the fact that $\lfs{\wk{\hatctx\Psi}, x}{\Psi, x} M = M$, etc.}
\\[1em]
We prove (2): \fbox{If $\Gamma ; \Psi, x{:}A_1 \Vdash \sigma = \sigma' : \Phi$ and $\Gamma ; \Psi \vdash A_1 \equiv A_2 : \lftype$
  then $\Gamma ; \Psi, x{:}A_2 \Vdash \sigma = \sigma' : \Phi$.}
\\[1em]
\prf{$\Gamma ; \Psi,  x{:}A_2, y{:}A_1 \Vdash \wk{\hatctx\Psi}, y = \wk{\hatctx\Psi}, y : \Psi, x{:}A_1$
   \hfill constructed as for case (1)}
\prf{$\Gamma ; \Psi,  x{:}A_2 \Vdash \wk{\hatctx\Psi}, x, x = \wk{\hatctx\Psi}, x, x : \Psi, x{:}A_2, y{:}A_1$
   \hfill constructed as for case (1)}
 \prf{$\Gamma ; \Psi, x{:}A_2 \Vdash \lfs{\wk{\hatctx{\Psi}}, x, x}  {\Psi, x, y}\sigma_1 = \lfs{\wk{\hatctx{\Psi}}, x, x}{\Psi, x, y}\sigma_2 : \Phi$ \\
\mbox{$\quad$}\hfill where $\sigma_1 = \lfs{\wk{\hatctx\Psi}, y}{\Psi, x} \sigma$
                    and $\sigma_2 = \lfs{\wk{\hatctx\Psi}, y}{\Psi, x}\sigma'$  by semantic LF subst. (Lemma \ref{lem:semlfeqsub} twice)}
 \prf{$\lfs{\wk{\hatctx{\Psi}}, x, x}{\Psi, x, y}(\wk{\hatctx\Psi}, y) = \wk{\hatctx\Psi}, x$ \hfill by subst. def.}
 \prf{$\Gamma ; \Psi, x{:}A_2 \vdash \lfs{\wk{\hatctx\Psi}, x}{\Psi, x} \sigma = \lfs{\wk{\hatctx\Psi}, x}{\Psi, x} \sigma' : \Phi$ \hfill by previous lines}
 \prf{$\Gamma ; \Psi, x{:}A_2 \Vdash \sigma = \sigma' : \Phi$ \hfill
   using the fact that $\lfs{\wk{\hatctx\Psi}, x}{\Psi, x} \sigma =
   \sigma$, etc.}
}
\end{proof}

Our semantic definitions are reflexive, symmetric, and transitive. Further they are stable under type conversions. Establishing these properties is tricky and intricate. We first establish these properties for LF and subsequently for computations. All proofs can be found in the long version.

\begin{lemma}[Symmetry, Transitivity, and Conversion of Semantic Equality for LF]\label{lem:semsymlf}
\quad\\
\LONGVERSION{A.~}For LF Terms:
\begin{enumerate}
\item \label{it:reflclf} (Reflexivity:) $\Gamma ; \Psi \Vdash M = M: A$.
\item \label{it:symclf} (Symmetry:)
      If $\Gamma ; \Psi \Vdash M = N : A$ then $\Gamma ; \Psi \Vdash N = M : A$.
\item \label{it:transclf} (Transitivity:)
    If $\Gamma ; \Psi \Vdash M_1 = M_2 : A$ and $\Gamma ; \Psi \Vdash M_2 = M_3 : A$
    then $\Gamma ; \Psi \Vdash M_1 = M_3 : A$.
\item \label{it:convclf} (Conversion:) If\/ $\Gamma ; \Psi \vdash A \equiv A' : \lftype$
     and $\Gamma ; \Psi \Vdash M = N : A$
     then $\Gamma ; \Psi \Vdash M = N : A'$.
\end{enumerate}

\LONGVERSION{B.~}For LF Substitutions:
\begin{enumerate}
\item \label{it:reflsub} (Reflexivity:) $\Gamma ; \Psi \Vdash \sigma = \sigma: \Phi$.
\item \label{it:symsub} (Symmetry:)
      If $\Gamma ; \Psi \Vdash \sigma = \sigma' : \Phi$ then $\Gamma ; \Psi \Vdash \sigma' = \sigma : \Phi$.
\item \label{it:transsub} (Transitivity:)
    If $\Gamma ; \Psi \Vdash \sigma_1 = \sigma_2 : \Phi$ and $\Gamma ; \Psi \Vdash \sigma_2 = \sigma_3 : \Phi$
    then $\Gamma ; \Psi \Vdash \sigma_1 = \sigma_3 : \Phi$.
\item \label{it:convsub} (Conversion:)
    If\/ $\Gamma \vdash \Phi \equiv \Phi' : \ctx$ and $\Gamma ; \Psi \Vdash \sigma = \sigma' : \Phi$
    then $\Gamma \Vdash \sigma = \sigma' : \Phi'$.
\end{enumerate}
\end{lemma}
\begin{proof}
Reflexivity follows directly from symmetry and transitivity. For LF terms and substitutions, we prove symmetry and conversion by induction on the derivation $\Gamma ; \Psi \Vdash M = N : A$ and $\Gamma; \Psi \Vdash \sigma = \sigma' : \Phi$ respectively. For transitivity, we use lexicographic induction.

We reason by induction on semantic equivalence relation where we consider any $\sigma'$ smaller than $\sigma$ if $\sigma \lfwhnf \sigma'$; the proofs is mostly straightforward exploiting symmetry of decl. equivalence ($\equiv$), determinacy of weak head reductions, and crucially relies on  well-formedness of semantic equality (Lemma \ref{lm:semlfwf}) and  functionality of LF typing (Lemma \ref{lm:func-lftyping}) for the case where $\sigma_i \lfwhnf \sigma'_i, M_i$.
\LONGVERSIONCHECKED{\\[0.5em]
\fbox{
Transitivity: If $\Gamma ; \Psi \Vdash \sigma_1 = \sigma_2 : \Phi$ and $\Gamma ; \Psi \Vdash \sigma_2 = \sigma_3 : \Phi$
    then $\Gamma ; \Psi \Vdash \sigma_1 = \sigma_3 : \Phi$.   }
\\[1em]
By lexicographic induction on the first wo derivations;
\\
\pcase{$\ianc{\Gamma ; \Psi \vdash \sigma_1 \lfwhnf \cdot : \cdot  \quad
              \Gamma ; \Psi \vdash \sigma_2 \lfwhnf \cdot : \cdot }
             {\Gamma ; \Psi \Vdash \sigma_1 = \sigma_2 : \cdot}{}$
}
\\
\prf{$\Gamma ; \Psi \Vdash \sigma_2 = \sigma_3 : \cdot$ \hfill by assumption}
\prf{$\Gamma ; \Psi \vdash \sigma_2 \lfwhnf \cdot : \cdot $ \hfill by inversion and determinacy (Lemma \ref{lem:detwhnf})}
\prf{$\Gamma ; \Psi \vdash \sigma_3 \lfwhnf \cdot : \cdot$ \hfill by inversion}
\prf{$\Gamma ; \Psi \Vdash \sigma_1 = \sigma_3 : \cdot$
\hfill using $\Gamma ; \Psi \vdash \sigma_1 \lfwhnf \cdot : \cdot$
        and $\Gamma ; \Psi \vdash \sigma_3 \lfwhnf \cdot : \cdot$}

\pcase{$\ianc{\Gamma ; \Psi, \wvec{x{:}A} \vdash \sigma_1 \lfwhnf \wk{\hatctx\Psi} : \Psi \quad
              \Gamma ; \Psi, \wvec{x{:}A} \vdash \sigma_2 \lfwhnf \wk{\hatctx\Psi} : \Psi}
             {\Gamma ; \Psi \Vdash \sigma_1 = \sigma_2 : \Psi}{}$
}
\prf{$\Gamma ; \Psi \vdash \sigma_2 \lfwhnf \cdot : \cdot $ \hfill by assumption}
\prf{$ \Gamma ; \Psi, \wvec{x{:}A} \vdash \sigma_2 \lfwhnf \wk{\hatctx\Psi} : \Psi$ \hfill  by inversion and determinacy (Lemma \ref{lem:detwhnf})}
\prf{$ \Gamma ; \Psi, \wvec{x{:}A} \vdash \sigma_3 \lfwhnf \wk{\hatctx\Psi} : \Psi$ \hfill by inversion}
\prf{$\Gamma \Vdash \sigma_1 = \sigma_3 : \cdot$
\hfill using $\Gamma ; \Psi, \wvec{x{:}A} \vdash \sigma_1 \lfwhnf \wk{\hatctx\Psi} : \Psi$
         and $\Gamma ; \Psi, \wvec{x{:}A} \vdash \sigma_3 \lfwhnf \wk{\hatctx\Psi} : \Psi$}

\pcase{$\ianc
  {\begin{array}{ll}
\Gamma ; \Psi \vdash \sigma_1 \lfwhnf \sigma'_1, M : \Phi, x{:}A  & \\
\Gamma ; \Psi \vdash \sigma_2 \lfwhnf \sigma'_2, N : \Phi, x{:}A &
\Gamma ; \Psi \Vdash \sigma'_1 = \sigma'_2 : \Phi \qquad \Gamma ; \Psi  \Vdash M = N : \lfs{\sigma'_1}{\Phi} A
\end{array}
}
{\Gamma ; \Psi \Vdash \sigma_1 = \sigma_2 : \Phi, x{:}A}{}$}
\prf{$\Gamma ; \Psi \Vdash \sigma_2 = \sigma_3 : \Phi, x{:}A$ \hfill by assumption}
\prf{$\Gamma ; \Psi \vdash \sigma_2 \lfwhnf \sigma'_2, N : \Phi, x{:}A$ \hfill by inversion and determinacy (Lemma \ref{lem:detwhnf})}
\prf{$\Gamma ; \Psi \vdash \sigma_3 \lfwhnf \sigma'_3, Q : \Phi, x{:}A$ \hfill by inversion}
\prf{$\Gamma ; \Psi  \Vdash N = Q : \lfs{\sigma_2}{\Phi} A$ \hfill inversion }
\prf{$\Gamma ; \Psi \Vdash \sigma'_2 = \sigma'_3 : \Phi, x{:}A$ \hfill by inversion}
\prf{$\Gamma ; \Psi \Vdash \sigma'_1 = \sigma'_3 : \Phi, x{:}A$ \hfill by IH}
\prf{$\Gamma ; \Psi \vdash \sigma'_2, N: \Phi, x{:}A$ \hfill by def. of well-typed whnf}
\prf{$\Gamma \vdash \Phi, x{:}A : \ctx$ \hfill by well-formedness of LF typing}
\prf{$\Gamma ; \Phi \vdash A : \lftype$ \hfill by well-formedness of LF contexts}
\prf{$\Gamma ; \Psi \vdash \sigma_1 \equiv \sigma_2 : \Phi$ \hfill by well-formedness of semantic equality (Lemma \ref{lm:semlfwf}) }
\prf{$\Gamma ; \Psi \vdash \lfs{\sigma_1}{\Phi} A \equiv \lfs{\sigma_2}{\Phi} A : \lftype$ \hfill by functionality of LF typing (Lemma \ref{lm:func-lftyping})}
\prf{$\Gamma ; \Psi \Vdash N = Q : \lfs{\sigma_1}{\Phi}A$ \hfill by IH (Conversion \ref{it:convclf})}
\prf{$\Gamma ; \Psi \Vdash M \equiv Q \lfs{\sigma_1}{\Phi} A$ \hfill by IH}
\prf{$\Gamma ; \Psi \Vdash \sigma_1 = \sigma_3 : : \Phi, x{:}A$ \hfill by sem. def. }

\SUBSTCLO{
\pcase{$\ianc {
      \begin{array}{lll}
      \Gamma ; \Psi \vdash \sigma_1 \lfwhnf (\sclo {\hatctx{\Phi}} {\unbox{t_1}{\sigma'_1})}  : \Phi
   & \typeof (\Gamma \vdash t_1) = \cbox{\Phi_1 \vdash \Phi'_1}  & \Gamma \vdash \Phi'_1 \equiv (\Phi, \wvec{x{:}A}) : \ctx
\\
\Gamma ; \Psi \vdash \sigma_2 \lfwhnf (\sclo {\hatctx\Phi} {\unbox{t_2}{\sigma'_2}}) : \Phi
    & \typeof (\Gamma \vdash t_2) = \cbox{\Phi_2 \vdash \Phi'_2} & \Gamma \vdash \Phi'_2 \equiv (\Phi, \wvec{x{:}A}) : \ctx
\\
\Gamma \vdash t_1 \equiv t_2 : \cbox{\Phi_1 \vdash \Phi'_1}
    & \Gamma ; \Psi \Vdash \sigma'_1 = \sigma'_2 : \Phi_1
    & \Gamma \vdash \Phi_1 \equiv \Phi_2 : \ctx
\end{array}
}
{\Gamma ; \Psi \Vdash \sigma_1 = \sigma_2 : \Phi}
{}$
}
\prf{$\Gamma ; \Psi \vdash \sigma_2 = \sigma_3 : \Phi$ \hfill by assumption}
\prf{$\Gamma ; \Psi \vdash \sigma_2 \lfwhnf (\sclo {\hatctx{\Phi}} {\unbox{t_2}{\sigma'_2})}  : \Phi$  \hfill by inversion and determinacy (Lemma \ref{lem:detwhnf})}
\prf{$\Gamma ; \Psi \vdash \sigma_3 \lfwhnf (\sclo {\hatctx{\Phi}} {\unbox{t_3}{\sigma'_3})}  : \Phi$  \hfill by inversion}
\prf{$\typeof (\Gamma \vdash t_2) = \cbox{\Phi_2 \vdash \Phi'_2}$ \hfill by inversion and uniqueness of $\typeof$}
\prf{$\typeof (\Gamma \vdash t_3) = \cbox{\Phi_3 \vdash \Phi'_3}$ \hfill by inversion}
\prf{$\Gamma ; \Psi \Vdash \sigma'_2 = \sigma'_3 : \Phi_2  $ \hfill by inversion}
\prf{$\Gamma \vdash t_2 \equiv t_3 : \cbox{\Phi_2 \vdash\Phi'_2} $ \hfill by inversion}
\prf{$\Gamma \vdash \Phi_2 \equiv \Phi_3 : \ctx $ \hfill by inversion}
\prf{$\Gamma \vdash \Phi'_2 \equiv (\Phi, \wvec{x{:}A'}) : \ctx $ \hfill by inversion}
\prf{$\Gamma \vdash \Phi'_3 \equiv (\Phi, \wvec{x{:}A'}) : \ctx $ \hfill by inversion}
\prf{$\Gamma \vdash (\Phi, \wvec{x{:}A})  \equiv (\Phi, \wvec{x{:}A'}) : \ctx$ \hfill by transitivity and symmetry}
\prf{$\Gamma \vdash \Phi'_3 \equiv (\Phi, \wvec{x{:}A}) :\ctx$ \hfill by transitivity and symmetry}
\prf{$\Gamma \vdash \Phi_1 \equiv \Phi_3 : \ctx$ \hfill by transitivity}
\prf{$\Gamma ; \Psi \Vdash \sigma'_2 = \sigma'_3 : \Phi_1$ \hfill by IH (\ref{it:convsub}) using $\Gamma \vdash \Phi_1 \equiv \Phi_2 : ctx$}
\prf{$\Gamma ; \Psi \Vdash \sigma'_1 = \sigma'_3 : \Phi_1$ \hfill by IH (\ref{it:transsub})}
\prf{$\Gamma \vdash t_1 \equiv t_3 : \cbox{\Phi_2 \vdash\Phi'_2} $ \hfill by transitivity}
\prf{$\Gamma ; \Psi \Vdash \sigma_1 = \sigma_3 : \Phi$ \hfill by def. of sem. typing}
}
\fbox{
Symmetry: If $\Gamma ; \Psi \Vdash \sigma_1 = \sigma_2 : \Phi$ then $\Gamma ; \Psi \Vdash \sigma_2 = \sigma_1 : \Phi$}
\\[1em]
By induction on semantic equivalence relation where we consider any $\sigma'$ smaller than $\sigma$ if $\sigma \lfwhnf \sigma'$; the proof is mostly straightforward exploiting symmetry of decl. equivalence ($\equiv$), but also relies again on  well-formedness of semantic equality (Lemma \ref{lm:semlfwf}) and  functionality of LF typing (Lemma \ref{lm:func-lftyping}) for the case where $\sigma_i \lfwhnf \sigma'_i, M_i$. We show the interesting case.
\\[1em]
\pcase{$\ianc
  {\begin{array}{ll}
\Gamma ; \Psi \vdash \sigma_1 \lfwhnf \sigma'_1, M : \Phi, x{:}A  & \\
\Gamma ; \Psi \vdash \sigma_2 \lfwhnf \sigma'_2, N : \Phi, x{:}A &
\Gamma ; \Psi \Vdash \sigma'_1 = \sigma'_2 : \Phi \qquad \Gamma ; \Psi  \Vdash M = N : \lfs{\sigma'_1}{\Phi} A
\end{array}
}
{\Gamma ; \Psi \Vdash \sigma_1 = \sigma_2 : \Phi, x{:}A}{}$}
\\[1em]
\prf{$\Gamma ; \Psi \Vdash \sigma'_2 = \sigma'_1 : \Phi$ \hfill by IH}
\prf{$\Gamma \vdash \Phi, x{:}A : \ctx$ \hfill by well-formedness of typing }
\prf{$\Gamma ; \Phi \vdash A : \lftype$ \hfill by well-formedness of LF contexts}
\prf{$\Gamma ; \Psi \vdash \sigma'_1 \equiv \sigma'_2 : \Phi$ \hfill by well-formedness of semantic equality (Lemma \ref{lm:semlfwf}) }
\prf{$\Gamma ; \Psi \vdash \lfs{\sigma_1}{\Phi} A \equiv \lfs{\sigma_2}{\Phi} A : \lftype$ \hfill by functionality of LF typing (Lemma \ref{lm:func-lftyping})}
\prf{$\Gamma ; \Psi \Vdash M = N : \lfs{\sigma_2}{\Phi}A$ \hfill by IH (Conversion \ref{it:convclf})}
\prf{$\Gamma ; \Psi \Vdash N = M : \lfs{\sigma_2}{\Phi}A$ \hfill by IH }
\prf{$\Gamma ; \Psi \Vdash \sigma_2 = \sigma_1 : \Phi, x{:}A$ \hfill by def. of semantic equivalence}

\vspace{1cm}
We now consider some cases for establishing symmetry and transitivity for semantic equality of LF terms.
\\[1em]
\pcase{$\ianc{
\begin{array}{lll}
\Gamma ; \Psi \vdash M_1 \lfwhnf \unbox {t_1} {\sigma_1} : \tm & \typeof(\Gamma \vdash t_1) = \cbox{\Phi_1 \vdash \tm}
& \Gamma \vdash \Phi_1 \equiv \Phi_2 : \ctx \\
\Gamma ; \Psi \vdash M_2 \lfwhnf \unbox {t_2} {\sigma_2} : \tm &
\typeof(\Gamma \vdash t_2) = \cbox{\Phi_2 \vdash \tm} &
\Gamma \vdash t_1 \equiv t_2 : \cbox{\Phi_1 \vdash \tm} \quad
\Gamma \Vdash \sigma'_1 = \sigma'_2 : \Phi_1
\end{array} }
            {\Gamma ; \Psi \Vdash M_1 = M_2 : \tm}{}$}
Symmetry for LF Terms.\\
\prf{$\Gamma ; \Psi \Vdash \sigma'_2 = \sigma'_1 : \Phi_1$ \hfill by IH}
\prf{$\Gamma ; \Psi \Vdash \sigma'_2 = \sigma'_1 : \Phi_2$ \hfill by IH (Conversion (\ref{it:convsub}))}
\prf{$\Gamma \vdash \cbox{\Phi_1 \vdash \tm} \equiv \cbox{\Phi_2 \vdash \tm} : u$ \hfill since $\Gamma \vdash \Phi_1 \equiv \Phi_2 : ctx$}
\prf{$\Gamma \vdash t_1 \equiv t_2 : \cbox{\Phi_2 \vdash \tm}$ \hfill by conversion  using $\Gamma \vdash \cbox{\Phi_1 \vdash \tm} \equiv \cbox{\Phi_2 \vdash \tm} : u$}
\prf{$\Gamma \vdash t_2 \equiv t_1 : \cbox{\Phi_2 \vdash \tm}$ \hfill by transitivity of $\equiv$}
\prf{$\Gamma ; \Psi \Vdash N = M : A$ \hfill by sem. def.}
\\
Transitivity for LF Terms.\\
\prf{$\Gamma \Vdash M_2 = M_3 : \tm$ \hfill by assumption}
\prf{$\Gamma \vdash M_2 \lfwhnf \unbox{t_2}{\sigma_2} : \tm$ \hfill by inversion and  determinacy (Lemma \ref{lem:detwhnf})}
\prf{$\Gamma \vdash M_3 \lfwhnf \unbox{t_3}{\sigma_3}$  \hfill by inversion}
\prf{$\typeof (\Gamma \vdash t_2) = \cbox{\Phi_2 \vdash \tm}$ \hfill by inversion and uniqueness of $\typeof$}
\prf{$\typeof (\Gamma \vdash t_3) = \cbox{\Phi_3 \vdash \tm}$ \hfill by inversion}
\prf{$\Gamma \vdash \Phi_2 \equiv \Phi_3 : \ctx$ \hfill by inversion}
\prf{$\Gamma \vdash \Phi_1 \equiv \Phi_3 : \ctx$ \hfill by transitivity ($\equiv$)}
\prf{$\Gamma \Vdash \sigma_2 = \sigma_3 : \Phi_2$ \hfill by inversion}
\prf{$\Gamma \Vdash \sigma_2 = \sigma_3 : \Phi_1$ \hfill by IH (Conversion \ref{it:convsub} using $\Gamma \Vdash \Phi_2 = \Phi_1 : \ctx$)}
\prf{$\Gamma \Vdash \sigma_1 = \sigma_3 : \Phi_1$ \hfill by IH}
\prf{$\Gamma \vdash t_2 \equiv t_3 : \cbox{\Phi_2 \vdash \tm}$ \hfill by inversion on $\Gamma ; \Psi \Vdash M_2 = M_3 : \tm$}
\prf{$\Gamma \vdash (\Phi_1 \vdash \tm) \equiv (\Phi_2 \vdash \tm) : u$ \hfill since $\Gamma \vdash \Phi_1 \equiv \Phi_2 : ctx$}
\prf{$\Gamma \vdash t_2 \equiv t_3 : \cbox{\Phi_1 \vdash \tm}$ \hfill by type conversion using $\Gamma \vdash (\Phi_1 \vdash \tm) \equiv (\Phi_2 \vdash \tm) : u$)}
\prf{$\Gamma \Vdash t_1 \equiv t_3 : \cbox{\Phi_1 \vdash \tm}$ \hfill by Transitivity of $\equiv$}
\prf{$\Gamma \Vdash M_1 = M_3 : \tm$}
}
\LONGVERSIONCHECKED{
\\[1em]
We concentrate here on proving the conversion properties:
\\[1em]
\fbox{(Conversion:) If\/ $\Gamma ; \Psi \vdash A \equiv A' : \lftype$
     and $\Gamma ; \Psi \Vdash M = N : A$
     then $\Gamma ; \Psi \Vdash M = N : A'$.}
\\[1em]
\pcase{
$\ianc
{ \Gamma ; \Psi \vdash M \lfwhnf \lambda x.M' : \Pi x{:}A.B \quad
  \Gamma ; \Psi \vdash N \lfwhnf \lambda x.N' : \Pi x{:}A.B \quad
 \Gamma ; \Psi, x{:}A \Vdash M' = N': B
}
{\Gamma ; \Psi \Vdash M = N: \Pi x{:}A. B}{}$
}
\prf{$\Gamma ; \Psi \vdash \Pi x{:}A.B \equiv \Pi x{:}A'.B' : \lftype$ \hfill by assumption}
\prf{$\Gamma ; \Psi, x{:}A \vdash B \equiv B' : \lftype$ and $\Gamma ; \Psi \vdash A \equiv A' : \lftype$ \hfill by injectivity of $\Pi$-types (Lemma \ref{lm:lfpi-inj})}
\prf{$\Gamma; \Psi, x{:}A \Vdash M' = N' : B'$ \hfill by IH}
\prf{$\Gamma ; \Psi, x{:}A' \Vdash M' = N' : B'$ \hfill by Semantic LF context conversion (Lemma \ref{lm:semlfctxconv})}
\prf{$\Gamma ; \Psi \Vdash M = N : \Pi x{:}A'.B'$ \hfill by sem. def.}
\\[1em]
Other cases are trivial since they are at type $\tm$.
}
\end{proof}

\subsection{Semantic Properties of Computations}

 \begin{lemma}[Well-Formedness of Semantic Typing] \quad \label{lm:semwf}
    If $\Gamma \Vdash t = t' : \ann\tau$ then
         $\Gamma \vdash t : \ann\tau$ and $\Gamma \vdash t' : \ann\tau$ and $\Gamma \vdash t \equiv t' : \ann\tau$.
\end{lemma}
\begin{proof}
By induction on the induction on $\Gamma \Vdash \ann\tau : u$. . In each case, we refer the Def.~\ref{def:typedwhnf}.
\end{proof}

\begin{lemma}[Semantic Weakening for Computations]\label{lem:compsemweak}\quad
\begin{enumerate}
  \item \label{it:weaksemtau} If $\Gamma \Vdash \ann\tau :u$ and $\Gamma' \leq_\rho \Gamma$ then $\Gamma' \Vdash \{\rho\}\ann\tau : u$.
  \item If\/ $\Gamma \Vdash \ann\tau = \ann\tau' : u$ and $\Gamma' \leq_\rho \Gamma$
    then $\Gamma' \Vdash \{\rho\}\ann\tau = \{\rho\}\ann\tau' : u$.
  \item  If\/ $\Gamma \Vdash t = t' : \ann\tau$ and $\Gamma' \leq_\rho \Gamma$
    then $\Gamma' \Vdash \{\rho\}t = \{\rho\} t : \{\rho\}\ann\tau$.
  \end{enumerate}
\end{lemma}
\begin{proof}
By induction on $\Gamma \Vdash \ann\tau : u$.
\LONGVERSIONCHECKED{
\\[0.25em]
We note that the theorem is trivial for $\ann\tau = \tmctx$. Hence we concentrate on proving it where $\ann\tau = \tau$ (i.e. it is a proper type).
\\[0.25em]
For better and easier readability we simply write for example
 $\tau = (y:\ann\tau_1) \arrow \tau_2$  instead of
\\
$\Gamma \Vdash \tau : u$ where
\begin{enumerate}
\item $\Gamma \vdash \tau \whnf (x:\ann\tau_1) \arrow \tau_2$
\item $(\forall \Gamma' \leq_\rho \Gamma. \Gamma' \Vdash \{\rho\}\ann\tau_1 : u_1)$
\item $\forall \Gamma'\leq_\rho \Gamma.~ \Gamma' \Vdash s = s' :\{\rho\}\ann\tau_1
    \Longrightarrow \Gamma' \Vdash \{\rho, s/x\} \tau_2 = \{\rho,s'/x\}\tau_2 : u_2$ and
\item $(u_1, u_2, u_3) \in \mathcal{R}$.
\end{enumerate}

\noindent
Weakening of semantic typing $\Gamma \Vdash \tau : u$: \\
By case analysis on $\Gamma \Vdash \tau : u$.
\\[0.5em]
    \pcase{$\tau  = (x: \ann\tau_1) \arrow \tau_2 $}
    %
    %
    \prf{$\Gamma' \vdash \{\rho\}\tau \whnf \{\rho\}((y: \ann\tau_1) \arrow \tau_2) : u_3$ \hfill by Lemma \ref{lem:weakwhnf} using $\Gamma \Vdash \tau : u$}
    \prf{$\Gamma' \vdash \{\rho\}\tau \whnf (x:\{\rho\}\ann\tau_1) \arrow \{\rho,x/x\}\tau_2: u_3$ \hfill by substitution def.}
\\
    \prf{Suppose that $\Gamma_1 \leq_{\rho_1} \Gamma'$}
    \prf{$\Gamma' \leq_\rho \Gamma$ \hfill by assumption}
    \prf{$\Gamma_1 \leq_{\{\rho_1\}\rho} \Gamma$ \hfill by composition of substitution}
    \prf{$\Gamma_1 \Vdash \{\{\rho_1\}\rho\}\ann\tau_1 : u_1$  \hfill by $\Gamma \Vdash \tau : u$}
    \prf{$\Gamma_1 \Vdash \{\rho_1\}\{\rho\}\ann\tau_1  : u_1$ \hfill by composition of substitution}
\\
    \prf{Suppose that $\Gamma_1 \leq_{\rho_1} \Gamma'$ and $\Gamma_1 \Vdash s = s' : \{\rho_1\}(\{\rho\}\ann\tau_1)$}
    \prf{$\Gamma_1 \leq_{\{\rho_1\}\rho} \Gamma$ \hfill by composition of substitution}
    \prf{$\Gamma_1 \Vdash s = s' : \{\{\rho_1\}\rho\}\ann\tau_1$\hfill by composition of substitution}
    \prf{$\Gamma_1 \Vdash \{\{\rho_1\}\rho, s/x\} \tau_2 = \{\{\rho_1\}\rho, s'/x\}\tau_2: u_2$ \hfill by $\Gamma \Vdash \tau : u$}
    \prf{$\Gamma_1 \Vdash \{\rho_1, s/x\}(\{\rho,x/x\} \tau_2)= \{\rho_1, s'/x\}(\{\rho,x/x\} \tau_2) : u_2$ \hfill by composition of substitution}
\\
    \prf{$\Gamma' \Vdash \{\rho\}\tau : u_3$ \hfill by abstraction, since $\Gamma_1, \rho_1$ where arbitrary}
\\
    \pcase{$\tau = \cbox{T}$}
    \prf{$\Gamma' \der\{\rho\}\tau \whnf \{\rho\}\cbox{T}:u$ \hfill by Lemma \ref{lem:weakwhnf}  using $\Gamma \Vdash \tau : u$}
    \prf{$\Gamma' \der\{\rho\}\tau \whnf \cbox{\{\rho\}T}:u$ \hfill by substitution def.}
    \prf{$\Gamma \semlf T = T : u$ \hfill by assumption}
    \prf{$\Gamma \vdash T \equiv T : u$\hfill by inversion}
    \prf{$\Gamma' \vdash \{\rho\}T \equiv \{\rho\}T : u$\hfill by substitution lemma}
    \prf{$\Gamma' \Vdash \{\rho\}\tau : u$ \hfill by def.}
\\
    %
     \pcase{$\Gamma \Vdash t = t' : \cbox {\Psi \vdash A}$}
     \prf{$\Gamma \vdash t \whnf w   : \cbox {\Psi \vdash A}$ \hfill by assumption}
     \prf{$\Gamma' \vdash \{\rho\}t \whnf \{\rho\}w   : \{\rho\}\cbox {\Psi \vdash A}$ \hfill by \ref{lem:weakwhnf}}
     \prf{$\Gamma \vdash t' \whnf w'   : \cbox {\Psi \vdash A}$ \hfill by assumption}
     \prf{$\Gamma' \vdash \{\rho\}t' \whnf \{\rho\}w' :  \{\rho\}\cbox{\Psi \vdash A}$ \hfill by \ref{lem:weakwhnf}}
     \prf{$\Gamma ; \Psi \Vdash \unbox w \id = \unbox{w'}{\id} : A$ \hfill by assumption}
     \prf{$\Gamma' ; \{\rho\}\Psi \Vdash \{\rho\}\unbox{w}{\id} = \{\rho\}\unbox{w'}{\id} : \{\rho\} A$ \hfill by IH}
     \prf{$\Gamma' \Vdash \{\rho\}t = \{\rho\}t' : \{\rho\}\cbox{\Psi \vdash A}$ \hfill by sem. def. using subst. properties}
\\
    \pcase{$\Gamma \Vdash t = t' : (y: \ann\tau_1) \arrow \tau_2$}
    \prf{$\Gamma \vdash t \whnf w : (y: \ann\tau_1) \arrow \tau_2$ \hfill by assumption}
    \prf{$\Gamma' \vdash \{\rho\}t \whnf \{\rho\}w :\{\rho\}((y: \ann\tau_1) \arrow \tau_2)$ \hfill by Lemma \ref{lem:weakwhnf}}
    \prf{$\Gamma \vdash t'\whnf w': (y: \ann\tau_1) \arrow \tau_2$ \hfill by assumption}
    \prf{$\Gamma' \vdash \{\rho\}t' \whnf \{\rho\}w' :\{\rho\}((y: \ann\tau_1) \arrow \tau_2)$ \hfill by Lemma \ref{lem:weakwhnf}}
    \prf{Suppose that $\Gamma_1 \leq_{\rho_1} \Gamma'$ and $\Gamma_1 \vdash t_1 : \{\rho_1\}(\{\rho\}\ann\tau_1)$ and $\Gamma_1 \vdash t_2 : \{\rho_1\}(\{\rho\}\ann\tau_1)$}
    \prf{$\Gamma_1\leq_{\{\rho_1\}\rho}\Gamma'$ \hfill by definition}
    \prf{$\Gamma_1 \vdash t_1 : \{\{\rho_1\}\{\rho\}\ann\tau_1$\hfill by composition of substitution}
    \prf{$\Gamma_1 \vdash t_2 :\{ \{\rho_1\}\rho\}\ann\tau_1$ \hfill by composition of substitution}
    \prf{$\Gamma_1 \Vdash t_1 = t_2 : \{\{\rho_1\}\rho\}\ann\tau_1 \Longrightarrow \Gamma_1 \Vdash \{\{\rho_1\}\rho\}w~t_1 = \{\{\rho_1\}\rho\}w'~t_2 : \{\{\rho_1\}\rho, t_1/x\}\tau_2$ \hfill by assumption }
    \prf{$\Gamma_1 \Vdash t_1 = t_2 : \{\rho_1\}(\{\rho\}\ann\tau_1) \Longrightarrow \Gamma_1 \Vdash(\{\rho_1\}(\{\rho\}w))~t_1 = (\{\rho_1\}\{\rho\}w')~t_2 : \{\rho_1, t_1/x\}(\{\rho\}\tau_2)$ \hfill by composition of substitution}
    \prf{$\Gamma' \Vdash \{\rho\}t = \{\rho\}t' : \{\rho\}((y: \ann\tau_1) \arrow \tau_2)$ \hfill by def.}
}
\end{proof}


Our semantic equality definition is symmetric and transitive. Our semantic equality is also reflexive -- however, note we prove a weaker reflexivity statement which says that if $t_1$ is semantically equivalent to another term $t_2$ then it is also equivalent to itself. This suffices for our proofs. We also note that our semantic equality takes into account extensionality for terms at function types and contextual types; this is in fact baked into our semantic equality definition.

\begin{lemma}[Symmetry, Transitivity, and Conversion of Semantic Equality]\label{lem:semsym}
\quad\\
Let $\Gamma \Vdash \ann\tau : u$ and $\Gamma \Vdash \ann\tau' : u$ and $\Gamma \sem \ann\tau = \ann\tau' : u$ and $\Gamma \Vdash t_1 = t_2 : \ann\tau$.  Then:
\begin{enumerate}
\item\label{it:refl} (Reflexivity for Terms:) $\Gamma \sem t_1 = t_1 : \ann\tau$.
\item\label{it:sym} (Symmetry for Terms:) $\Gamma \Vdash t_2 = t_1 : \ann\tau$.
\item\label{it:trans} (Transitivity for Terms:) If\/ $\Gamma \Vdash t_2 = t_3 : \ann\tau$ then
  $\Gamma \Vdash t_1 = t_3 : \ann\tau$.
\item \label{it:symtyp} (Symmetry for Types:) $\Gamma \sem \ann\tau' = \ann\tau : u$
\item \label{it:transtyp} (Transitivity for Types:) If\/ $\Gamma \Vdash \ann\tau' = \ann\tau'' : u$ and $\Gamma \Vdash \ann\tau'' : u$ then
$\Gamma \Vdash \ann\tau = \ann\tau'' : u$.
\item\label{it:conv} (Conversion:) $\Gamma \Vdash t_1 = t_2 : \ann\tau'$.
\end{enumerate}
\end{lemma}
\begin{proof}
  Reflexivity follows directly from symmetry and transitivity.
  We prove symmetry and transitivity for terms using a lexicographic induction on $u$ and
  $\Gamma \Vdash \tau : u$; we appeal to the induction hypothesis and
  use the corresponding properties on types if
  the universe is smaller; if the universe stays the same,
  then we may appeal to the property for terms if $\Gamma \Vdash \tau
  : u$ is smaller;

  To prove conversion and symmetry for types, we may also appeal to the induction hypothesis if $\Gamma \Vdash \tau' : u$ is smaller.
\LONGVERSIONCHECKED{
  \\[1em]
\pcase{\fbox{\fbox{$\ann\tau = \tmctx$}}}
Symmetry for Terms (Prop. \ref{it:sym}): To Show: \fbox{$\Gamma \Vdash t_2 = t_1 : \tmctx$} \\
\prf{$t_2$ and $t_1$ stand for LF context $\Psi_2$ and $\Psi_1$ respectively}
\prf{$\Gamma \vdash \Psi_1 \equiv \Psi_2 : \tmctx$ \hfill by def. $\Gamma \Vdash t_1 = t_2 : \tau$}
\prf{$\Gamma \vdash \Psi_2 \equiv \Psi_1 : \tmctx$ \hfill by symmetry of $\equiv$}
\prf{$\Gamma \Vdash t_2 = t_2 : \tmctx$ \hfill by sem. equ. def.}
\\[0.5em]
Transitivity for Terms (Prop. \ref{it:trans}): To Show:
\fbox{If\/ $\Gamma \Vdash t_2 = t_3 : \tmctx $ then  $\Gamma \Vdash t_1 = t_3 : \tmctx$.}\\
\prf{$t_1$, $t_2$, and $t_3$ stand for LF context $\Psi_1$, $\Psi_2$, and $\Psi_3$ respectively}
\prf{$\Gamma \vdash \Psi_1 \equiv \Psi_2 : \tmctx$ \hfill by $\Gamma \Vdash t_1 = t_2 : \tmctx$}
\prf{$\Gamma \vdash \Psi_2 \equiv \Psi_3 : \tmctx$ \hfill by $\Gamma \Vdash t_2 = t_3 : \tmctx$}
\prf{$\Gamma \vdash \Psi_1 \equiv \Psi_3 : \tmctx$ \hfill by transitivity of $\equiv$}
\prf{$\Gamma \Vdash t_1 = t_3 : \tmctx$ \hfill by semantic equ. def.}
%
\\[0.5em]
Other cases are trivial.
\\[1.5em]
%
%
\pcase{\fbox{\fbox{$\tau = \cbox T$,~i.e. $\Gamma \vdash \tau \whnf \cbox T : u$
  and $\Gamma \Vdash_\LF T = T $ where $T = \Psi \vdash A$}}}
Symmetry for Terms (Prop. \ref{it:sym}): To Show: \fbox{$\Gamma \Vdash t_2 = t_1 : \cbox T$.}\\
\prf{$\Gamma \vdash t_1 \whnf w_1 : \tau$ \hfill by definition of  $\Gamma \Vdash t_1 = t_2 : \tau$}
\prf{$\Gamma \vdash t_2 \whnf w_2 : \tau$ \hfill by definition of  $\Gamma \Vdash t_1 = t_2 : \tau$}
\prf{\emph{Sub-Case}: $\Gamma \vdash t_1 \whnf \cbox{C} : \cbox{T}$
  and $\Gamma \vdash t_2 \whnf \cbox{C'} : \cbox T$
  and $\Gamma \Vdash_{\LF} C = C' : T$}
\prf{Consider $C = (\hatctx{\Psi} \vdash M)$ and $C' = (\hatctx{\Psi} \vdash N)$ and $T = \Psi \vdash A$ (proof is the same for case $C = (\hatctx{\Psi} \vdash \sigma)$)}
\prf{$\Gamma ; \Psi \Vdash M = N : A$ \hfill by $\Gamma \Vdash_\LF C = C' : T$}
\prf{$\Gamma ; \Psi \Vdash N = M : A$ \hfill by Lemma \ref{lem:semsymlf} (\ref{it:symclf})}
\prf{$\Gamma \Vdash_\LF C' = C : T$ \hfill by sem. def.}
\prf{$\Gamma \Vdash t_2 = t_1 : \cbox T$ \hfill by semantic equ. def.}
\\[-0.75em]
\prf{\emph{Sub-Case}: $\neut w_1, w_2$ and $\Gamma \vdash w_1 \equiv w_2 : \cbox{T}$}
\prf{$\Gamma \vdash w_2 \equiv w_1 : \cbox{T}$ \hfill by symmetry of $\equiv$}
\prf{$\Gamma \Vdash t_2 = t_1 : \cbox T$ \hfill by semantic equ. def.}
\\
Transitivity for Terms (Prop.\ref{it:trans}):
To Show \fbox{If\/ $\Gamma \Vdash t_2 = t_3 : \cbox T $ then  $\Gamma \Vdash  t_1 = t_3 : \cbox T$.}
 \\
\prf{$\Gamma \vdash t_1 \whnf w_1 : \tau$ \hfill by definition of  $\Gamma \Vdash t_1 = t_2 : \tau$}
\prf{$\Gamma \vdash t_2 \whnf w_2 : \tau$ \hfill by definition of  $\Gamma \Vdash t_1 = t_2 : \tau$}
\prf{$\Gamma \vdash t_2 \whnf w_2' : \tau$ \hfill by definition of  $\Gamma \Vdash t_2 = t_3 : \tau$}
\prf{$\Gamma \vdash t_3 \whnf w_3 : \tau$ \hfill by definition of  $\Gamma \Vdash t_2 = t_3 : \tau$}
\prf{$w_2 = w_2'$ \hfill by Lemma~\ref{lem:detwhnf} (\ref{it:comp-detwhnf})}
\prf{$\Gamma ; \Psi \Vdash \unbox{w_1}{\id} = \unbox{w_2}{\id} : A$ \hfill by def.  $\Gamma \Vdash t_1 = t_2 : \tau$}
\prf{$\Gamma ; \Psi \Vdash \unbox{w_2}{\id} = \unbox{w_3}{\id} : A$ \hfill by def.  $\Gamma \Vdash t_2 = t_3 : \tau$}
\prf{$\Gamma ; \Psi \Vdash \unbox{w_1}{\id} = \unbox{w_3}{\id} : A$ \hfill by Lemma~\ref{lem:semsymlf} (\ref{it:transclf})}
\prf{$\Gamma \Vdash t_1 = t_3 : \cbox{\Psi \vdash A}$ \hfill by sem. equ. def.}
\\
Symmetry for Types (Prop.\ref{it:symtyp}): To Show: \fbox{$\Gamma \sem \tau' =  \cbox T : u$ where $T = \Psi \vdash A$}\\
\prf{$\Gamma \vdash \tau' \whnf \cbox{T'} : u$ and $\Gamma \vdash T \equiv T'$
  \hfill by $\Gamma \Vdash \tau = \tau' : u$}
\prf{$\Gamma \vdash T' \equiv T$ \hfill by symmetry for LF equ.}
\prf{$\Gamma \Vdash \tau' = \tau : u$ \hfill by semantic equ. def.}
\\
Transitivity for Types (Prop.\ref{it:transtyp}): To Show:
\fbox{ If\/ $\Gamma \Vdash \tau' = \tau'' : u$ and $\Gamma \Vdash \tau'' : u$ then
$\Gamma \Vdash \cbox T = \tau'' : u$.}\\
\prf{$\Gamma \vdash \tau' \whnf \cbox{T'} : u$ and $\Gamma \vdash T \equiv T'$
  \hfill by $\Gamma \Vdash \tau = \tau' : u$}
\prf{$\Gamma \vdash \tau'' \whnf \cbox{T''} : u$ and $\Gamma \vdash T' \equiv T''$
  \hfill by $\Gamma \Vdash \tau' = \tau'' : u$}
\prf{$\Gamma \vdash T \equiv T''$ \hfill by transitivity for LF equ.}
\prf{$\Gamma \Vdash \tau = \tau'' : u$ \hfill by semantic equ. def.}
\\
Conversion for Terms (Prop.\ref{it:conv}): To Show:
\fbox{$\Gamma \Vdash t_1 = t_2 : \tau'$.} \\ 
\prf{$\Gamma \vdash t_1 \whnf w_1 : \cbox T$ \hfill by definition of  $\Gamma \Vdash t_1 = t_2 : \tau$}
\prf{$\Gamma \vdash t_2 \whnf w_2 : \cbox T$ \hfill by definition of  $\Gamma \Vdash t_1 = t_2 : \tau$}
\prf{$\Gamma \vdash \tau' \whnf \cbox{T'} : u$ and $\Gamma \vdash T \equiv T'$ \hfill by def. of $\Gamma \Vdash \tau = \tau' : u$}
\prf{$\Gamma \vdash \cbox{T} \equiv \cbox{T'} : u$ \hfill by decl. equ. def.}
\prf{$\Gamma \vdash \tau \equiv \cbox{T} : u$ \hfill by $\Gamma \vdash
  \tau \whnf \cbox{T}$ (since $\whnf$ rules are a subset of $\equiv$)}
\prf{$\Gamma \vdash \tau' \equiv \cbox{T'} : u$ \hfill by $\Gamma \vdash \tau' \whnf \cbox{T'}$ (since $\whnf$ rules are a subset of $\equiv$)}
\prf{$\Gamma \vdash \tau \equiv \tau' : u$ \hfill by transitivity and
  symmetry of decl. equality ($\equiv$)}
\prf{$\Gamma \vdash t_i : \tau'$ and $\Gamma \vdash w_i : \tau'$ for $i = 1,2$
  \hfill by typing rules using $\Gamma \vdash t_i : \cbox T$}
\prf{$\Gamma \vdash t_1 \whnf w_1 : \tau' $ and $\Gamma \vdash t_2  \whnf w_2 : \tau' $
  \hfill by Def.~\ref{def:typedwhnf}}
\prf{$\Gamma \Vdash t_1 = t_2 : \tau'$ \hfill by semantic equ. def.}
\\[0.5em]
\pcase{\fbox{\fbox{$\tau = (y : \ann\tau_1) \arrow \tau_2$ i.e.\ $\Gamma \vdash \tau \whnf (y : \ann\tau_1) \arrow \tau_2 : u$}}}
Symmetry for Terms (Prop. \ref{it:sym}):  To Show:
\fbox{$\Gamma \Vdash t_2 = t_1 : (y : \ann\tau_1) \arrow \tau_2$ }\\
\prf{$\Gamma \vdash t_1 \whnf w_1 : \tau$ \hfill by definition of  $\Gamma \Vdash t_1 = t_2 : \tau$}
\prf{$\Gamma \vdash t_2 \whnf w_2 : \tau$ \hfill by definition of  $\Gamma \Vdash t_1 = t_2 : \tau$}
\prf{Assume $\Gamma' \leq_\rho \Gamma$ and $\Gamma' \Vdash s_2 = s_1 : \{\rho\}\ann\tau_1$}
\prf{\smallerderiv{$\Gamma' \Vdash \{\rho\}\ann\tau_1 : u_1$} \hfill by $\Gamma \Vdash \tau : u$}
\prf{$\Gamma' \sem s_1 = s_2 : \{\rho\}\ann\tau_1$
  \hfill by induction hypothesis (Prop. \ref{it:sym}), symmetry}
\prf{$\Gamma' \sem \{\rho\}w_1~s_1 = \{\rho\}w_2~s_2 : \{\rho,s_1/y\}\tau_2$
  \hfill by definition of  $\Gamma \Vdash t_1 = t_2 : \tau$}
\prf{\smallerderiv{$\Gamma' \sem \{\rho,s_1/y\}\tau_2 : u_2$} \hfill by assumption $\Gamma \sem \tau : u$ }
\prf{$\Gamma' \sem \{\rho\}w_2~s_2 = \{\rho\}w_1~s_1 : \{\rho,s_1/y\}\tau_2$
  \hfill by induction hypothesis (Prop. \ref{it:sym}), symmetry}
\prf{$\Gamma \vdash \tau' \whnf (y : \ann\tau_1') \arrow \tau_2' : u$
  \hfill by definition of $\Gamma \Vdash \tau = \tau' : u$}
\prf{$\Gamma' \Vdash \{\rho,s_1/y\}\tau_2 = \{\rho,s_2/y\}\tau_2' : u_2$
  \hfill by definition of $\Gamma \Vdash \tau = \tau' : u$}
\prf{$\Gamma' \Vdash s_2 = s_2 : \{\rho\}\ann\tau_1$
  \hfill by induction hypothesis (Prop. \ref{it:refl}), reflexivity}
\prf{$\Gamma' \Vdash \{\rho,s_2/y\}\tau_2 = \{\rho,s_2/y\}\tau_2' : u_2$
  \hfill by definition of $\Gamma \Vdash \tau = \tau' : u$}
\prf{$\Gamma' \Vdash \{\rho,s_2/y\}\tau_2' = \{\rho,s_2/y\}\tau_2 : u_2$
  \hfill by induction hypothesis (Prop. \ref{it:symtyp}), symmetry}
\prf{$\Gamma' \Vdash \{\rho,s_1/y\}\tau_2 = \{\rho,s_2/y\}\tau_2 : u_2$
  \hfill by induction hypothesis (Prop. \ref{it:transtyp}), transitivity}
\prf{$\Gamma' \sem \{\rho\}w_2~s_2 = \{\rho\}w_1~s_1 : \{\rho,s_2/y\}\tau_2$
  \hfill by induction hypothesis (Prop. \ref{it:conv}), conversion}
\prf{$\Gamma \Vdash t_2 = t_1 : \tau$ \hfill since $\Gamma',\rho,s_2,s_1$ were arbitrary}
\\
Transitivity for Terms (Prop. \ref{it:trans}):
\fbox{ If\/ $\Gamma \Vdash t_2 = t_3 : (y : \ann\tau_1) \arrow \tau_2$ then
  $\Gamma \Vdash t_1 = t_3 : (y : \ann\tau_1) \arrow \tau_2$.}\\
\prf{$\Gamma \vdash t_1 \whnf w_1 : \tau$ \hfill by definition of  $\Gamma \Vdash t_1 = t_2 : \tau$}
\prf{$\Gamma \vdash t_2 \whnf w_2 : \tau$ \hfill by definition of  $\Gamma \Vdash t_1 = t_2 : \tau$}
\prf{$\Gamma \vdash t_2 \whnf w_2' : \tau$ \hfill by definition of  $\Gamma \Vdash t_2 = t_3 : \tau$}
\prf{$\Gamma \vdash t_3 \whnf w_3 : \tau$ \hfill by definition of  $\Gamma \Vdash t_2 = t_3 : \tau$}
\prf{$w_2 = w_2'$ \hfill by determinacy of weak head evaluation (Lemma \ref{lem:detwhnf})}
\prf{Assume $\Gamma' \leq_\rho \Gamma$ and $\Gamma' \Vdash s_1 = s_3 : \{\rho\}\ann\tau_1$}
\prf{\smallerderiv{$\Gamma' \Vdash \{\rho\}\ann\tau_1 : u_1$} \hfill by $\Gamma \Vdash \tau : u$}
\prf{$\Gamma' \sem s_1 = s_1 : \{\rho\}\ann\tau_1$
  \hfill by induction hypothesis (Prop. \ref{it:refl}), reflexivity}
%
\prf{\smallerderiv{$\Gamma' \sem \{\rho,s_1/y\}\tau_2 : u_2$} \hfill  by assumption $\Gamma \sem \tau : u$ }
\prf{$\Gamma' \sem \{\rho\}w_1~s_1 = \{\rho\}w_2~s_1 : \{\rho,s_1/y\}\tau_2$
  \hfill by assumption $\Gamma \sem t_1 = t_2 : \tau$}
\prf{$\Gamma' \sem \{\rho\}w_2~s_1 = \{\rho\}w_3~s_3 : \{\rho,s_1/y\}\tau_2$
  \hfill by assumption $\Gamma \sem t_2 = t_3 : \tau$}
\prf{$\Gamma' \sem \{\rho\}w_1~s_1 = \{\rho\}w_3~s_3 : \{\rho,s_1/y\}\tau_2$
  \hfill by induction hypothesis (Prop. \ref{it:trans}), transitivity}
\prf{$\Gamma \Vdash t_1 = t_3 : \tau$  \hfill since $\Gamma',\rho,s_1,s_3$ were arbitrary}
\\
Symmetry for Types (Prop. \ref{it:symtyp}):\fbox{$\Gamma \sem \tau' = \tau : u$}\\
\prf{$\Gamma \sem (y : \ann\tau_1) \arrow \tau_2  = \tau' : u$ \hfill by assumption}
\prf{$\Gamma \vdash \tau' \whnf (y : \ann\tau_1') \arrow \tau_2'$
  \hfill by definition of $\Gamma \sem \tau = \tau' : u$}
\prf{Assume $\Gamma' \leq_\rho \Gamma$ and $\Gamma' \Vdash s' = s : \{\rho\}\ann\tau_1'$.}
\prf{\smallerderiv{$\Gamma' \Vdash \{\rho\}\ann\tau_1 : u_1$} \hfill by $\Gamma \Vdash \tau : u$}
\prf{$\Gamma' \Vdash \{\rho\}\ann\tau_1 = \{\rho\}\ann\tau_1' : u_1$ \hfill by $\Gamma \Vdash \tau = \tau' : u$}
\prf{$\Gamma' \Vdash \{\rho\}\ann\tau_1' = \{\rho\}\ann\tau_1 : u_1$ \hfill by induction hypothesis (Prop. \ref{it:symtyp})}
\prf{\highlight{$\Gamma' \vdash \{\rho\}\ann\tau_1' : u_1$} \hfill by $\Gamma \Vdash \tau' : u$}
\prf{$\Gamma' \Vdash s' = s : \{\rho\}\ann\tau_1$ \hfill by induction hypothesis (Prop. \ref{it:conv}), conversion}
\prf{$\Gamma' \Vdash s = s' : \{\rho\}\ann\tau_1$ \hfill by induction hypothesis (Prop. \ref{it:sym}), symmetry for terms}
\prf{$\Gamma' \sem \{\rho, s/y\}\tau_2 = \{\rho, s'/y\}\tau_2' : u_2$  \hfill by $\Gamma \sem \tau = \tau' : u$}
\prf{$\Gamma' \sem s = s : \{\rho\}\ann\tau_1$
  \hfill by induction hypothesis (Prop. \ref{it:refl}), reflexivity}
\prf{\smallerderiv{$\Gamma' \Vdash \{\rho, s/y\}\tau_2 : u_2$} \hfill by $\Gamma \Vdash \tau : u$}
\prf{$\Gamma' \sem \{\rho, s'/y\}\tau_2' =  \{\rho, s/y\}\tau_2 : u$ \hfill by induction hypothesis (Prop. \ref{it:symtyp}), symmetry for types}
%
%
\prf{$\Gamma \sem \tau' = \tau : u$ \hfill since $\Gamma', \rho, s, s'$ were arbitrary}
\\
Transitivity for Types (Prop. \ref{it:transtyp}):
\fbox{ If\/ $\Gamma \Vdash \tau' = \tau'' : u$ and $\Gamma \Vdash \tau'' : u$ then
$\Gamma \Vdash \tau = \tau'' : u$.}\\
\prf{$\Gamma \sem (y : \ann\tau_1') \arrow \tau_1 = \tau' : u$ \hfill by assumption}
\prf{$\Gamma \vdash \tau' \whnf (y:\ann\tau_2') \arrow \tau_2 : u$ \hfill by definition of $\Gamma \sem \tau = \tau':u$
and determinacy of reduction}
\prf{$\Gamma \vdash \tau'' \whnf (y:\ann\tau_3') \arrow \tau_3 : u$ \hfill by definition of $\Gamma \sem \tau' = \tau'' : u$ }
\prf{Assume $\Gamma' \leq_\rho \Gamma$ and $\Gamma' \sem s_1 = s_3 : \{\rho\}\ann\tau_1'$}
\prf{\smallerderiv{$\Gamma' \Vdash \{\rho\}\ann\tau_1' : u_1$} \hfill by $\Gamma \Vdash \tau : u$}
\prf{$\Gamma' \sem s_1 = s_1 : \{\rho\}\ann\tau_1'$
  \hfill  by induction hypothesis (Prop. \ref{it:refl}), reflexivity}
\prf{$\Gamma' \sem \{\rho,s_1/y\}\tau_1 = \{\rho, s_1/y\}\tau_2 : u_2$
   \hfill by definition of $\Gamma \sem \tau = \tau' :u$}
\prf{$\Gamma' \sem \{\rho\}\ann\tau_1' = \{\rho\}\ann\tau_2' : u_1$
   \hfill by definition of $\Gamma \sem \tau = \tau' :u$}
\prf{$\Gamma' \sem s_1 = s_1 : \{\rho\}\ann\tau_2'$ \hfill
      \hfill by induction hypothesis (Prop. \ref{it:conv}), type conversion}
\prf{$\Gamma' \sem \{\rho,s_1/y\}\tau_2 = \{\rho, s_1/y\}\tau_3 : u_2$
   \hfill by definition of $\Gamma \sem \tau' = \tau'' : u$}
\prf{\smallerderiv{$\Gamma' \Vdash \{\rho, s_1/y\}\tau_1 : u_2$} \hfill by $\Gamma \Vdash \tau : u$ }
\prf{$\Gamma' \sem \{\rho,s_1/y\}\tau_1 = \{\rho, s_1/y\}\tau_3 : u_2$
   \hfill by induction hypothesis (Prop. \ref{it:transtyp}), transitivity for types}
\prf{$\Gamma' \sem \{\rho\}\ann\tau_2' = \{\rho\}\ann\tau_3' : u_1$
   \hfill by definition of $\Gamma \sem \tau' = \tau'' :u$}
\prf{$\Gamma' \sem \{\rho\}\ann\tau_1' = \{\rho\}\ann\tau_3' : u_1$
   \hfill by induction hypothesis (Prop. \ref{it:transtyp}), transitivity for types}
\prf{$\Gamma' \sem s_1 = s_3 : \{\rho\}\tau_3'$
    \hfill by  induction hypothesis (\ref{it:conv}), type conversion}
\prf{$\Gamma' \sem \{\rho, s_1/y\}\tau_3 = \{\rho, s_3/y\}\tau_3 : u_2$
    \hfill by $\Gamma \sem \tau'' : u$}
\prf{$\Gamma \sem \{\rho, s_1/y\}\tau_1 = \{\rho, s_3/y\}\tau_3 : u_2$
   \hfill by induction hypothesis (Prop. \ref{it:transtyp}), transitivity for types}
\prf{$\Gamma \sem \tau = \tau'' : u$ \hfill since $\Gamma', \rho, s_1, s_3$ were arbitrary}
\\
Conversion (Prop. \ref{it:conv}). \fbox{$\Gamma \Vdash t_1 = t_2 : \tau'$.}
\\
\prf{$\Gamma \der \tau \equiv \tau' : u$  \hfill by Well-formedness Lemma~\ref{lm:semwf}}
\prf{$\Gamma \vdash t_1 \whnf w_1 : \tau$ \hfill by definition of  $\Gamma \Vdash t_1 = t_2 : \tau$}
\prf{$\Gamma \vdash t_1 \whnf w_1 : \tau'$ \hfill by the conversion rule}
\prf{$\Gamma \vdash t_2 \whnf w_2 : \tau'$ \hfill ditto}
\prf{$\Gamma \vdash \tau' \whnf (y : \ann\tau_1') \arrow \tau_2' : u$ \hfill by definition of  $\Gamma \Vdash \tau = \tau' : u$}
\prf{Assume $\Gamma' \leq_\rho \Gamma$ and $\Gamma' \Vdash s_1 = s_2 : \{\rho\}\ann\tau_1'$}
\prf{$\Gamma' \Vdash \{\rho\}\ann\tau_1 =\{\rho\}\ann\tau_1' : u_1$
   \hfill by definition of $\Gamma \sem \tau = \tau' : u_1$}
\prf{\smallerderiv{$\Gamma' \Vdash \{\rho\}\ann\tau_1 : u_1$} \hfill by $\Gamma \Vdash \tau : u$}
 \prf{$\Gamma' \Vdash \{\rho\}\ann\tau_1' =\{\rho\}\ann\tau_1 : u_1$
    \hfill by induction hypothesis (Prop. \ref{it:symtyp}), symmetry}
\prf{\highlight{$\Gamma' \Vdash s_1 = s_2 : \{\rho\}\ann\tau_1$
  \hfill by induction hypothesis (Prop. \ref{it:conv}) on $\Gamma' \sem \{\rho\}\ann\tau_1' : u_1$, conversion}}
\prf{$\Gamma' \sem \{\rho\}w_1~s_1 = \{\rho\}w_2~s_2 : \{\rho,s_1/y\}\tau_2$
  \hfill by assumption $\Gamma \sem t_1 = t_2 : \tau$}
\prf{$\Gamma' \sem s_1 = s_1 : \{\rho\}\ann\tau_1$
  \hfill by induction hypothesis (\ref{it:refl}), reflexivity}
\prf{\smallerderiv{$\Gamma' \sem \{\rho,s_1/y\}\tau_2 : u_2$} \hfill  by definition of $\Gamma \sem \tau : u$}
\prf{$\Gamma' \sem \{\rho,s_1/y\}\tau_2 = \{\rho,s_1/y\}\tau_2'$
  \hfill by definition of $\Gamma \sem \tau = \tau' : u$}
\prf{$\Gamma' \sem \{\rho\}w_1~s_1 = \{\rho\}w_2~s_2 : \{\rho,s_1/y\}\tau_2'$
  \hfill by induction hypothesis (Prop. \ref{it:conv}), conversion}
\prf{$\Gamma \Vdash t_1 = t_2 : \tau'$  \hfill since $\Gamma',\rho,s_1,s_2$ were arbitrary}
%
\\
\pcase{\fbox{\fbox{$\tau = u'$,~i.e. $\Gamma \Vdash \tau : u$ where
$\Gamma \vdash \tau \whnf u' : u$ and $u' < u$}}}
Symmetry for Terms (Prop. \ref{it:sym}): To Show:
\fbox{$\Gamma \Vdash t_2 = t_1 : u'$}\\
\prf{$\Gamma \Vdash t_2 = t_1 : u'$ \hfill by IH using Symmetry for Types (Prop. \ref{it:symtyp}) (since $u' < u$)}
%
\\[1em]
Transitivity for Terms (Prop. \ref{it:trans}): To Show:
\fbox{If\/ $\Gamma \Vdash t_2 = t_3 : u'$ then  $\Gamma \Vdash t_1 = t_3 : u'$.}
\\[1em]
\prf{$\Gamma \vdash t_1 = t_3:u'$ \hfill by IH using Transitivity for Types (Prop. \ref{it:transtyp}) (since $u' < u$)}
\\[1em]
Symmetry for Types (Prop. \ref{it:symtyp}): To Show:
\fbox{$\Gamma \sem \tau' = u' : u$}
\\[1em]
\prf{$\Gamma \vdash \tau' \whnf u' : u$ \hfill by $\Gamma \Vdash u' = \tau' : u$}
\prf{$\Gamma \vdash \tau' = \tau : u$ \hfill since $\Gamma \vdash \tau \whnf u' : u$ and $u' < u$  (by assumption), and $\Gamma \vdash \tau' \whnf u' : u$ }
\\[1em]
Transitivity for Types (Prop. \ref{it:transtyp}): To Show:
\fbox{ If\/ $\Gamma \Vdash \tau' = \tau'' : u$ and $\Gamma \Vdash \tau'' : u$ then
$\Gamma \Vdash u' = \tau'' : u$.}
\\[1em]
\prf{$\Gamma \vdash \tau' \whnf u' : u$ \hfill $\Gamma \Vdash u' =  \tau' : u$ }
\prf{$\Gamma \vdash \tau'' \whnf u' : u$ \hfill by $\Gamma \Vdash \tau' = \tau'' : u$ }
\prf{$\Gamma \Vdash u' = \tau'' : u$ \hfill using sem. equ. def. and the assumption $u' < u$}
\\[1em]
Conversion (Prop. \ref{it:conv}): To Show: \fbox{$\Gamma \Vdash t_1 = t_2 : \tau'$.}
\\[1em]
\prf{$\Gamma \Vdash t_1 = t_2 : \tau$ and
     $\Gamma \Vdash \tau : u$ where $\Gamma \vdash \tau \whnf u' : u$ and $u' < u$\hfill by assumption}
\prf{$\Gamma \Vdash u' = \tau' : u$ and $u' < u$ \hfill by assumption}
\prf{$\Gamma \vdash \tau' \whnf u' : u$ \hfill by $\Gamma \Vdash u' = \tau' : u$ }
\prf{$\Gamma \Vdash t_1 = t_2 : \tau'$ \hfill since $\Gamma \Vdash \tau' : u$}
\\
\pcase{\fbox{\fbox{$\tau = x~\vec{t}$ and $\Gamma \vdash \tau \whnf x~\vec{t} : u$ and $\neut (x~\vec{t})$}}}
\\
Symmetry for Terms (Prop. \ref{it:sym}): To Show: \fbox{$\Gamma \Vdash t_2 = t_1 : x~\vec{s}$.}
\\
\prf{$\Gamma \Vdash t_1 = t_2 : x~\vec{s}$ \hfill by assumption}
\prf{$\Gamma \vdash t_1 \whnf n_1 : x ~\vec{s}$,~~$\Gamma \vdash t_2 \whnf n_2 : x~\vec{s}$,~~
$\Gamma \vdash n_1 \equiv n_2 : x~\vec{s}$,~~$\neut n_1, n_2$ \hfill by $\Gamma \Vdash t_1 = t_2 : x~\vec{s}$}
\prf{$\Gamma \vdash n_2 \equiv n_1 : x~\vec{s}$ \hfill by symmetry of $\equiv$}
\prf{$\Gamma \Vdash t_2 = t_1 : x~\vec{s}$ \hfill by sem. equ. definition}
\\[0.5em]
Transitivity for Terms (Prop. \ref{it:trans}):
To Show \fbox{If\/ $\Gamma \Vdash t_2 = t_3 : x~\vec{s}$ then  $\Gamma \Vdash  t_1 = t_3 : x~\vec{s}$.}
\\
\prf{$\Gamma \Vdash t_1 = t_2 :  x~\vec{s}$ \hfill by assumption}
\prf{$\Gamma \vdash t_1 \whnf n_1 : x ~\vec{s}$,~~$\Gamma \vdash t_2 \whnf n_2 : x~\vec{s}$,~~$\Gamma \vdash n_1 \equiv n_2 : x~\vec{s}$,~~$\neut n_1, n_2$ \hfill by $\Gamma \Vdash t_1 = t_2 : x~\vec{s}$}
\prf{$\Gamma \vdash t_3 \whnf n_3 : x ~\vec{s}$,
     $\Gamma \vdash n_2 \equiv n_3 : x~\vec{s}$,~~$\neut n_3$ \hfill by $\Gamma \Vdash t_2 = t_3 : x~\vec{s}$}
\prf{$\Gamma \vdash n_1 \equiv n_3 : x~\vec{s}$ \hfill by transitivity of $\equiv$}
\prf{$\Gamma \Vdash t_1 = t_3 : x~\vec{s}$ \hfill by sem. equ. definition}
\\[0.5em]
Symmetry for Types (Prop. \ref{it:symtyp}): To Show: \fbox{$\Gamma \sem \tau' = x~\vec{s} : u$}
\\
 \prf{$\Gamma \Vdash x~\vec{s} = \tau' : u$ where $\Gamma \vdash \tau \whnf x~\vec{s} : u$ \hfill by assumption}
 \prf{$\Gamma \vdash \tau' \whnf x~\vec{t} : u$ and $\Gamma \vdash x~\vec{s} \equiv x~\vec{t} : u$}
 \prf{$\Gamma \vdash x~\vec{t} \equiv x~\vec{s} : u$ \hfill by symmetry of $\equiv$}
 \prf{$\Gamma \Vdash \tau' = \tau : u$ \hfill by sem. equ. definition}
\\[0.5em]
Transitivity for Types (Prop. \ref{it:transtyp}): To Show: \fbox{If\/ $\Gamma \Vdash \tau' = \tau'' : u$ and $\Gamma \Vdash \tau'' : u$ then
 $\Gamma \Vdash x~\vec{s} = \tau'' : u$.}
\\
\prf{$\Gamma \Vdash x~\vec{s} = \tau' : u$ \hfill by assumption}
\prf{$\Gamma \vdash \tau' \whnf x~\vec{s'} : u$ and $\Gamma \vdash x~\vec{s} \equiv x~\vec{s'} : u$ \hfill by $\Gamma \Vdash x~\vec{s} = \tau' : u$ }
\prf{$\Gamma \Vdash \tau' = \tau'' : u$ \hfill by assumption}
\prf{$\Gamma \vdash \tau'' \whnf x~\vec{s''}:u$ and $\Gamma \vdash x~\vec{s'} \equiv x~\vec{s''}:u$ \hfill by $\Gamma \Vdash \tau' = \tau'' : u$}
\prf{$\Gamma \vdash x~\vec{s} = x~\vec{s''} : u$ \hfill by transitivity of $\equiv$}
\prf{$\Gamma \vdash \tau = \tau'' : u$ \hfill by sem. equ. def.}

Conversion (Prop. \ref{it:conv}): \fbox{$\Gamma \Vdash t_1 = t_2 : \tau'$.}
\\
\prf{$\Gamma \Vdash t_1 = t_2 : x~\vec{s} $ \hfill by assumption}
\prf{$\Gamma \vdash t_1 \whnf n_1 : x~\vec{s}$ and $\Gamma \vdash t_2 \whnf n_2 : x~\vec{s}$ and $\Gamma \vdash n_1 \equiv n_2 : x~\vec{s}$ \hfill by $\Gamma \Vdash t_1 = t_2 : x~\vec{s} $}
\prf{$\Gamma \Vdash x~\vec{s} = \tau' : u$ \hfill by assumption}
\prf{$\Gamma \vdash \tau' \whnf x~\vec{s'} : u$ and $\Gamma \vdash x~\vec{s} \equiv x~\vec{s'} : u$ \hfill by $\Gamma \Vdash x~\vec{s} = \tau' : u$ }
\prf{$\Gamma \vdash n_1 \equiv n_2 : x~\vec{s'} $ \hfill using type conversion}
\prf{$\Gamma \vdash t_1 \whnf n_1 : x~\vec{s'}$ and $\Gamma \vdash t_2 \whnf n_2 : x~\vec{s'}$ \hfill using type conversion}
\prf{$\Gamma \Vdash t_1 = t_2 : \tau'$ \hfill by sem. equ. def.}
}
\end{proof}

Finally we establish various elementary properties about our semantic definition that play a key role in the fundamental lemma which we prove later. First, we show that all neutral terms are in the semantic definition.

 \begin{lemma}[Neutral Soundness]\label{lem:NeutSound}$\quad$$\quad$\\
If\/ $\Gamma \Vdash \ann\tau : u$ and $\Gamma \vdash t : \ann\tau$ and $\Gamma \vdash t' : \ann\tau$ and $\Gamma \vdash t \equiv t' : \ann\tau$ and $\neut t, t'$ then $\Gamma \Vdash t = t' : \ann\tau$.

 \end{lemma}
 \begin{proof}
By induction on $\Gamma \Vdash \tau : u$.
%
\LONGVERSIONCHECKED{
\\
\noindent
\pcase{$\tau = u$}
\prf{$\neut t$ and $\neut t'$ \hfill by assumption}
\prf{$t = x~\vec{s}$ and $t' = x~\vec{s'}$ \hfill since $\neut t$ and $\neut t$, $\Gamma \vdash t : u$, $\Gamma \vdash t':u$ and $\Gamma \vdash t \equiv t' : u$}
\prf{$\Gamma \vdash x~\vec{s} \equiv x~\vec{s'} : u$ \hfill by assumption $\Gamma \vdash t \equiv t' : u$}
\prf{$\Gamma \vdash t' \whnf  x~\vec{s'}$ and $\Gamma \vdash t \whnf  x~\vec{s}$ \hfill since $\neut t$ and $\neut t'$}
\prf{$\Gamma \vdash x~\vec{s} = t' : u$ \hfill by sem. def.}
%
%
\\
\pcase{$\tau =\cbox{T}$ where $\Psi \vdash \tm$}
\prf{$\Gamma \vdash t : \cbox{T}$ and $\Gamma \vdash t' : \cbox{T}$\hfill by assumption}
\prf{$\neut t$ and $\neut t'$\hfill by assumption}
\prf{$\norm t$ and $\norm t'$\hfill by def. of $\norm / \neut$}
\prf{$t \whnf t$ and $t' \whnf t'$ \hfill by def. of $\whnf$}
\prf{$\Gamma \vdash t \equiv t' : \cbox{T}$ \hfill by assumption}
\prf{$\Gamma \vdash t \whnf t : \cbox{T}$  and $\Gamma \vdash t' \whnf t' : \cbox{T}$ \hfill by Def.~\ref{def:typedwhnf}}
\prf{$\Gamma ; \Psi \vdash \unbox t \id \lfwhnf \unbox t \id : A$ \hfill since $\neut t$}
\prf{$\Gamma ; \Psi \vdash \unbox {t'} \id \lfwhnf \unbox {t'} \id : A$ \hfill since $\neut t'$}
\prf{$\Gamma ; \Psi \vdash \typeof (\Gamma \vdash t) = \cbox{\Phi \vdash \tm}$ and $\Gamma \vdash \Psi \equiv \Phi : \ctx$ \hfill by Lemma \ref{lm:typeof} }
\prf{$\Gamma ; \Psi \vdash \typeof (\Gamma \vdash t') = \cbox{\Phi' \vdash \tm}$ and $\Gamma \vdash \Psi \equiv \Phi' : \ctx$ \hfill by Lemma \ref{lm:typeof} }
\prf{$\Gamma \vdash \Phi \equiv \Phi'  : \ctx$ \hfill by symmetry and transitivity of $\equiv$}
\prf{$\Gamma ; \Psi \Vdash \id = \id : \Psi$ \hfill by Lemma \ref{lm:semlfwk}}
\prf{$\Gamma ; \Psi \Vdash \unbox t \id =  \unbox {t'}\id : A$ \hfill by sem. def.}
\prf{$\Gamma \Vdash t = t' : \cbox{T}$ \hfill by semantic def.}
\\
\pcase{$\tau = (y:\ann\tau_1) \arrow \tau_2$}
\prf{$\Gamma \vdash t : (y:\ann\tau_1) \arrow \tau_2$ and $\Gamma \vdash t' : (y:\ann\tau_1) \arrow \tau_2$ \hfill by assumption}
\prf{$\neut t$ and $\neut t'$ \hfill by assumption}
\prf{$\norm t$ and $\norm t'$ \hfill by def. of $\norm / \neut$}
\prf{$t \whnf t$ and $t' \whnf t'$ \hfill by def. of $\whnf$}
\prf{$\Gamma \vdash t \equiv t' : (y:\ann\tau_1)\arrow \tau_2$ \hfill by assumption}
\prf{$\Gamma \vdash t \whnf t : (y:\ann\tau_1) \arrow \tau_2$ \hfill by Def.~\ref{def:typedwhnf}}
\prf{$\Gamma \vdash t' \whnf t' : (y:\ann\tau_1) \arrow \tau_2$ \hfill by Def.~\ref{def:typedwhnf}}
\prf{Assume $\forall \Gamma' \leq_\rho \Gamma.~\Gamma' \Vdash s = s' :  \{\rho\}\ann\tau_1$}
\prf{\mbox{$\quad$}$\Gamma' \vdash \{\rho\}t \equiv \{\rho\}t': \{\rho\}((y:\ann\tau_1) \arrow \tau_2)$ \hfill by Weakening Lemma \ref{lem:weakcomp}}
\prf{\mbox{$\quad$}$\Gamma' \vdash \{\rho\}t \equiv \{\rho\}t': (y:\{\rho\}\ann\tau_1) \arrow \{\rho, y/y\}\tau_2$ \hfill by subst. def.}
\prf{\mbox{$\quad$}$\Gamma' \vdash s \equiv s' : \{\rho\}\ann\tau_1$ \hfill by  Well-formedness Lemma~\ref{lm:semwf} }
\prf{\mbox{$\quad$}$\Gamma' \vdash \{\rho\}t~s \equiv \{\rho\}t'~s' : \{\rho, s/y\}\tau_2$\hfill by rule}
\prf{\mbox{$\quad$}$\neut \{\rho\}t$ and $\neut \{\rho\}t'$ \hfill by Lemma \ref{lem:weaknorm}}
\prf{\mbox{$\quad$}$\neut \{\rho\}t~s$ and $\neut \{\rho\}t'~s'$ \hfill by def. of $\norm / \neut$}
\prf{\mbox{$\quad$}\smallerderiv{$\Gamma' \Vdash \{\rho,s/y\}\tau_2 : u_2$} \hfill by $\Gamma \Vdash \tau : u$}
\prf{\mbox{$\quad$}$\Gamma' \Vdash \{\rho\}t~s = \{\rho\}t~s' : \{\rho, s/y\}\tau_2$\hfill by IH}
\prf{$\Gamma \Vdash t = t' : (y:\ann\tau_1) \arrow \tau_2$ \hfill by semantic def.}

\pcase{$\tau = x~\vec{s}$}
\prf{$\Gamma \vdash \tau \whnf x~\vec{s} : u$ and $\neut(x~\vec{s})$ \hfill by $\Gamma \Vdash \tau : u$ }
\prf{$\Gamma \vdash t \equiv t' : x~\vec{s}$ and $\neut t, t'$ \hfill by assumption}
\prf{$\Gamma \vdash t \whnf t : x~\vec{s}$ \hfill since $\neut t$}
\prf{$\Gamma \vdash t' \whnf t' : x~\vec{s}$ \hfill since $\neut t'$}
\prf{$\Gamma \Vdash t = t' : x~\vec{s}$ \hfill by sem. equ. def.}

}
 \end{proof}

We also show that semantic definition are backwards closed.

\begin{lemma}[Backwards Closure for Computations]\label{lem:bclosed}\quad
If $\Gamma \sem t_1 = t_2: \ann\tau$ and $\Gamma \vdash t_1 \whnf w : \ann\tau$ and $\Gamma \vdash t_1' \whnf w : \ann\tau$
then $\Gamma \sem t_1' = t_2: \ann\tau$.
\end{lemma}
\begin{proof}
We analyse the definition of $\Gamma \sem t_1 = t_2: \ann\tau$ considering different cases of $\Gamma \sem \ann\tau : u$.
\end{proof}

\begin{lemma}[Typed Whnf is Backwards Closed]\label{lem:typeclosed}
    If\/ $\Gamma \vdash t \whnf w : (y:\ann\tau_1) \arrow \tau_2$
    and  $\Gamma \vdash s : \ann\tau_1$
    and  $\Gamma \vdash w~s \whnf v : \{s/y\}\tau_2$
    then $\Gamma \vdash t~s \whnf v : \{s / y \}\tau_2$.
\end{lemma}
\begin{proof}
Proof by unfolding definition and typing rules and considering different cases for $w$.
\LONGVERSIONCHECKED
{
 \quad\\
\prf{$\Gamma \vdash t : (y:\ann\tau_1) \arrow \tau_2$ \hfill by def. of $\Gamma \vdash t \whnf w : (y:\ann\tau_1) \arrow \tau_2$}
\prf{$\Gamma \vdash s : \ann\tau_1$ \hfill by assumption}
\prf{$\Gamma \vdash t~s : \{s/y\}\tau_2$ \hfill by typing rule}
\prf{$\Gamma \vdash s \equiv s : \ann\tau_1$ \hfill by reflexivity of $\equiv$}
\prf{$\Gamma \vdash t \equiv w : (y:\ann\tau_1) \arrow \tau_2$ \hfill by def. of $\Gamma \vdash t \whnf w : (y:\ann\tau_1) \arrow \tau_2$}
\prf{$\Gamma \vdash t~s \equiv w~s : \{s/y\}\tau_2$ \hfill by congruence rules of $\equiv$}
\prf{$\Gamma \vdash w~s \equiv v : \{s/y\}\tau_2$ \hfill by def. of $\Gamma \vdash w~s \whnf v : \{s/y\}\tau_2$ }
\prf{$\Gamma \vdash t~s \equiv v : \{s/y\}\tau_2$ \hfill by symmetry and transitivity of $\equiv$}
\prf{$t \whnf w$ \hfill by  def. of $\Gamma \vdash t \whnf w : (y:\ann\tau_1) \arrow \tau_2$}
\prf{$w~s \whnf v$ \hfill by def. of $\Gamma \vdash w~s \whnf v : \{s/y\}\tau_2$ }
\prf{$\Gamma \vdash v : \{s/y\}\tau_2$ \hfill by def. of $\Gamma \vdash w~ s \whnf v : \{s/y\}\tau_2$ }
\prf{$\norm w$ \hfill by definition of $\whnf$}
\\
\prf{\emph{Sub-case}: $t \whnf \tmfn x {t'}$ and $w = \tmfn x {t'}$}
\prf{$(\tmfn x t')~s \whnf v$ where $\{s/x\}t' \whnf v$ \hfill by $\Gamma \vdash w~s\whnf v: \{s/y\}\tau_2$}
\prf{$t~s \whnf v$ \hfill since $t \whnf \tmfn x t'$}
\prf{$\Gamma \vdash t~s \whnf v : \{s / y \}\tau_2$ \hfill by def. }
\\
\prf{\emph{Sub-case}: $t \whnf w$ where $\neut w$ 
                       }
\prf{$w~s\whnf w~s$ \hfill  since $\neut (w~s)$}
\prf{$t~s \whnf w~s$ \hfill by rule}
\prf{$\Gamma \vdash t~s \whnf v : \{s / y \}\tau_2$ \hfill by def. }
}
\end{proof}

\begin{lemma}[Semantic Application]\label{lem:SemTypeApp}
If $\Gamma \Vdash t = t': (y:\ann\tau_1) \arrow \tau_2$ and $\Gamma \Vdash s = s': \ann\tau_1$ then
$\Gamma \Vdash t~s = t'~s': \{s/y\}\tau_2$
 \end{lemma}
 \begin{proof}
Proof using Well-formedness Lemma \ref{lm:semwf}, Backwards closed properties (Lemma \ref{lem:typeclosed} and \ref{lem:bclosed}), and Symmetry of semantic equality (Lemma Prop. \ref{it:sym}).
\LONGVERSIONCHECKED{
$\;$\\ [1em]
\prf{$\Gamma \vdash t \whnf w   : (y:\ann\tau_1) \arrow \tau_2$ ~\mbox{and}~\\
$\Gamma \vdash t' \whnf w' : (y:\ann\tau_1) \arrow \tau_2$ ~\mbox{and}~
\\
$\forall \Gamma' \leq_\rho \Gamma.~\Gamma' \Vdash s = s' : \{\rho\}\ann\tau_1
\Longrightarrow \Gamma' \Vdash (\{\rho\}w)~s = \{\rho\}w'~s' : \{\rho, s/y\}\tau_2$
\hfill by sem. def.}
\prf{$\Gamma \Vdash w~s = w'~s': \{s/y\}\tau_2$ \hfill choosing $\Gamma$ for $\Gamma'$, $\rho$ to be the identity substitution }
\prf{$\Gamma \vdash w~s \whnf v : \{s/y\}\tau_2$ \hfill by def. of $\Gamma \Vdash w~s = w'~s': \{s/y\}\tau_2$}
\prf{$\Gamma \vdash s : \tau_1$ and $\Gamma \vdash s' \vdash \tau_1$ \hfill by Well-formedness Lemma \ref{lm:semwf}}
\prf{$\Gamma \vdash t~s \whnf v : \{s/y\}\tau_2$ \hfill by Whnf Backwards closed (Lemma \ref{lem:typeclosed})}
\prf{$\Gamma \vdash w'~s' \whnf v: \{s/y\}\tau_2$ \hfill by def. of $\Gamma \Vdash w~s = w'~s': \{s/y\}\tau_2$}
\prf{$\Gamma \vdash t'~s' \whnf v : \{s/y\}\tau_2$ \hfill by Whnf Backwards closed (Lemma \ref{lem:typeclosed})}
\prf{$\Gamma \Vdash t~s = t'~s': \{s/y\}\tau_2$ \hfill by Semantic Backwards
  Closure for Computations (Lemma \ref{lem:bclosed}) \\
\mbox{\quad} \hfill and Symmetry (Lemma Prop. \ref{it:sym})}
}
 \end{proof}


\section{Validity in the Model}\label{sec:validity}
For normalization, we need to establish that well-typed terms are
logically related. However, as we traverse well-typed terms, they do
not remain closed. As is customary, we now extend our logical relation
to substitutions defining semantic substitutions \fbox{$\Gamma \Vdash \theta = \theta' : \Gamma$}:

\[
  \begin{array}{l}
\infer{\Gamma' \Vdash \cdot = \cdot : \cdot}{\vdash \Gamma' }
\qquad
\infer{\Gamma' \Vdash \theta, t/x = \theta', t'/x : \Gamma, x{:}\ann\tau}
{ \Gamma' \Vdash \theta = \theta' : \Gamma &
 \Gamma' \Vdash \{\theta\}\ann\tau = \{\theta'\}\ann\tau : u &
 \Gamma' \Vdash t = t' : \{\theta\}\ann\tau &
 \Gamma' \Vdash \{\theta\}\ann\tau : u}
  \end{array}
\]

\begin{lemma}[Semantic Weakening of Computation-level Substitutions]\label{lem:weakcsub}
If $\Gamma' \Vdash \theta = \theta' : \Gamma$ and $\Gamma'' \leq_\rho \Gamma'$
then $\Gamma'' \Vdash \{\rho\}\theta = \{\rho\}\theta' : \Gamma$.
\end{lemma}
\begin{proof}
By induction on $\Gamma' \Vdash \theta = \theta' : \Gamma$ and using semantic weakening lemma \ref{lem:semweak}.
\end{proof}

 \begin{lemma}[Semantic Substitution Preserves Equivalence]\label{lem:semsubst}
  Let $\Gamma' \sem \theta = \theta' : \Gamma$;
    \begin{enumerate}
    \item If\/$\Gamma ; \Psi \vdash M \equiv M : A$ then
      $\Gamma;  \{\theta\}\Psi \vdash \{\theta\}M \equiv \{\theta'\}M : \{\theta\}A$.
    \item If\/$\Gamma ; \Psi \vdash \sigma \equiv \sigma : \Phi$ then
      $\Gamma;  \{\theta\}\Psi \vdash \{\theta\}\sigma \equiv \{\theta'\}\sigma : \{\theta\}\Phi$.
    \item If\/ $\Gamma \vdash t \equiv t : \ann\tau$
        then $\Gamma' \vdash \{\theta\}t \equiv \{\theta'\}t : \{\theta\}\ann\tau$.
    \end{enumerate}
  \end{lemma}
  \begin{proof}
 By induction on $M$, $\sigma$ $\tau$, and $t$. The proof is mostly straightforward; in the case where $t = x$ we know by assumption that $t_x/x \in \theta$ and $t'_x/x \in \theta'$ where $\Gamma' \sem t_x = t'_x : \{\theta\}\tau_x$. From Well-formedness of semantic typing (Lemma \ref{lm:semwf}), we know that $\Gamma' \vdash t_x \equiv t'_x : \{\theta\}\tau_x$.
  \end{proof}

Last, we define validity of computation-level contexts and and
computation-level types and terms referring to their semantic
definitions (Fig.~\ref{fig:valid}). This allows us to define compactly
the fundamental lemma which states that well typed terms correspond to
valid terms in our model. Validity here is defined in terms of the
semantic definition (Fig.~\ref{fig:sem}). In particular, we say that
two terms $M$ and $N$ are equal in our model, if for all
computation-level instantiations $\theta$ and $\theta'$ which are
considered semantically equal, we have that $\{\theta\}M$ and
$\{\theta'\}N$ are semantically equal.

\begin{figure}[htb]
  \centering
\[
  \begin{array}{l}
\mbox{Validity of Context}: \fbox{$\models \Gamma$}\quad
\\[1em]
\infer{\models \cdot}{}
\quad
\infer{\models \Gamma, x:\ann\tau}{\models \Gamma & \Gamma \models \ann\tau : u}
\\[1em]
\mbox{Validity of LF Objects}: \fbox{$\Gamma ; \Psi \models M = N : A$} \\[1em]
\infer{\Gamma ; \Psi \models M = N : A}
{\models \Gamma &
 \forall \Gamma',~\theta,~\theta'.~\Gamma' \Vdash \theta = \theta' :
                  \Gamma \Longrightarrow \Gamma' ; \{\theta\}\Psi \Vdash \{\theta\}M  = \{\theta'\}N : \{\theta\}A}
\\[1em]
\mbox{Validity of LF Substitutions}: \fbox{$\Gamma ; \Psi \models \sigma = \sigma' : \Phi$} \\[1em]
\infer{\Gamma ; \Psi \models \sigma = \sigma' : \Phi}
{\models \Gamma &
 \forall \Gamma',~\theta,~\theta'.~\Gamma' \Vdash \theta = \theta' :
                  \Gamma \Longrightarrow \Gamma' ; \{\theta\}\Psi \Vdash \{\theta\}\sigma_1  = \{\theta'\}\sigma' : \{\theta\}\Phi}
\\[1em]
\mbox{Validity of Types}: \fbox{$\Gamma \models \ann\tau = \ann\tau':u$} \quad\mbox{and}\quad \fbox{$\Gamma \models \ann\tau : u$}
\\[1em]
\infer{\Gamma \models \ann\tau = \ann\tau' : u}
      {\models \Gamma \quad
       \forall \Gamma',~\theta,~\theta'.~\Gamma' \Vdash \theta = \theta' : \Gamma
       \Longrightarrow \Gamma' \Vdash \{\theta\}\ann\tau = \{\theta'\}\ann\tau' : u}
\qquad
\infer{\Gamma \models \ann\tau : u}{\Gamma \models \ann\tau = \ann\tau : u}
\\[1em]
\mbox{Validity of Terms}: \fbox{$\Gamma \models t = t' : \ann\tau$} \quad\mbox{and}\quad \fbox{$\Gamma \models t : \ann\tau$}\\[1em]
\infer{\Gamma \models t = t' : \ann\tau}
{\models \Gamma & \Gamma \models \ann\tau : u &
 \forall \Gamma',~\theta,~\theta'.~\Gamma' \Vdash \theta = \theta' :
                  \Gamma \Longrightarrow \Gamma' \Vdash \{\theta\}t  =
                  \{\theta'\}t' : \{\theta\}\ann\tau}
\quad
\infer{\Gamma \models t : \ann\tau}
{\Gamma \models t = t : \ann\tau}
  \end{array}
\]

  \caption{Validity Definition}\label{fig:valid}
\end{figure}

\begin{lemma}[Well-formedness of Semantic Substitutions]\label{lem:wfsemsub}
If $\Gamma' \Vdash \theta = \theta' : \Gamma$ then $\Gamma' \vdash \theta : \Gamma$ and $\Gamma' \vdash \theta' : \Gamma$ and $\Gamma' \vdash \theta \equiv \theta' : \Gamma$.
\end{lemma}
\begin{proof}
By induction on   $\Gamma' \Vdash \theta = \theta' : \Gamma$.
\end{proof}

\begin{lemma}[Symmetry and Transitivity of Semantic Substitutions]\quad
  \label{lm:symsemsub}
  Assume $\models \Gamma$.
  \begin{enumerate}
  \item If $\Gamma' \Vdash \theta_1 = \theta_2 : \Gamma$ then
    $\Gamma'\Vdash \theta_2 = \theta_1: \Gamma$.
  \item If $\Gamma' \Vdash \theta_1 = \theta_2 : \Gamma$
    and $\Gamma' \Vdash \theta_2 = \theta_3 : \Gamma$
    then $\Gamma' \Vdash \theta_1 = \theta_3 : \Gamma$.
  \end{enumerate}
\end{lemma}
\begin{proof}
We prove symmetry by induction on the derivation and transitivity by induction on both derivations using Symmetry, Transitivity, and Conversion for semantic equality (Lemma \ref{lem:semsym});
; reflexivity follows from symmetry and transitivity.
\LONGVERSIONCHECKED{\\[1em]Symmetry:  By induction on derivation.
\\[1em]
\pcase{\ianc{}{\Gamma' \Vdash \cdot = \cdot : \cdot}{}}
\prf{$\Gamma' \Vdash \cdot = \cdot : \cdot$ \hfill by def.}
\\
\pcase{$\D = \ibnc{\Gamma' \Vdash \theta = \theta' : \Gamma}
            {\Gamma' \Vdash t = t' : \{\theta\}\ann\tau}
            {\Gamma' \Vdash \theta, t/x = \theta', t'/x : \Gamma, x{:}\ann\tau}{}$}
\prf{$\Gamma' \Vdash \theta' = \theta : \Gamma$ \hfill by IH}
\prf{$\Gamma' \Vdash t' = t : \{\theta\}\ann\tau$\hfill by Lemma \ref{lem:semsym} (Symmetry)}
\prf{$\Gamma' \Vdash \theta', t'/x = \theta, t/x : \Gamma, x{:}\ann\tau$ \hfill by rule}
\\
Transitivity: By induction on both derivations.\\[0.5em]
\pcase{$\theta_1= \cdot$, $\theta_2 = \cdot$ and $\theta_3 = \cdot$}
\prf{$\Gamma' \Vdash \cdot = \cdot : \cdot$ \hfill by def.}
\\
\pcase{$\theta_1= \theta_1', t_1/x$, $\theta_2 = \theta_2', t_2/x$ and $\theta_3 =\theta_3', t_3/x$
   and $\Gamma = \Gamma_0, x{:}\ann\tau$}
\prf{$\Gamma' \Vdash \theta_1' = \theta_2' : \Gamma_0$ \hfill by inversion}
\prf{$\Gamma' \Vdash \theta_2' = \theta_3' : \Gamma_0$ \hfill by inversion}
\prf{$\Gamma' \Vdash \theta_1' = \theta_3' : \Gamma_0$ \hfill by IH}
\prf{$\Gamma' \Vdash t_1 = t_2 : \{\theta'_1\}\ann\tau$ \hfill by inversion}
\prf{$\Gamma' \Vdash t_2 = t_3: \{\theta'_2\}\ann\tau$ \hfill by inversion}
\prf{$\Gamma' \Vdash \{\theta'_1\}\ann\tau = \{\theta'_2\}\ann\tau : u$\hfill by inversion}
\prf{$\Gamma' \Vdash \{\theta'_2\}\ann\tau = \{\theta'_3\}\ann\tau : u$\hfill by inversion}
\prf{$\Gamma' \Vdash \{\theta'_2\}\ann\tau = \{\theta'_1\}\ann\tau : u$\hfill by Lemma \ref{lem:semsym} (Symmetry)}
\prf{$\Gamma' \Vdash t_2 = t_3 : \{\theta'_1\}\ann\tau$ \hfill by Lemma \ref{lem:semsym} (Conversion)}
\prf{$\Gamma' \Vdash t_1 = t_3: \{\theta'_1\}\ann\tau$ \hfill by Lemma \ref{lem:semsym} (Transitivity)}
\prf{$\Gamma' \Vdash \{\theta'_1\}\ann\tau = \{\theta'_3\}\ann\tau$ \hfill by Lemma \ref{lem:semsym} (Transitivity)}
\prf{$\Gamma' \Vdash \theta_1 = \theta_3 : \Gamma$ \hfill by rule}
}
\end{proof}

\begin{lemma}[Context Satisfiability]\label{lem:ctxsat}
If $\models \Gamma$ then $\vdash \Gamma$ and $\Gamma \Vdash \id(\Gamma) = \id(\Gamma) : \Gamma$ where
\[
\begin{array}{lcl}
  \id(\cdot) & = & \cdot\\
  \id(\Gamma, x{:}\tau) & = & \id(\Gamma), x/x
\end{array}
\]
\end{lemma}
\begin{proof}
By induction on $\Gamma$ using  Neutral Soundness (Lemma \ref{lem:NeutSound}) and Semantic Weakening (Lemma \ref{lem:weakcsub}).
\LONGVERSIONCHECKED{By induction on $\Gamma$.\\[0.5em]
\pcase{$\Gamma = \cdot$}
\prf{$\vdash \cdot $ \hfill by rules}
\prf{$\id(\cdot) = \cdot$ \hfill by def. of $\id$}
\prf{$\Gamma' \Vdash \cdot = \cdot : \cdot$ \hfill by semantic def.}

\pcase{$\Gamma = \Gamma_0, x{:}\ann\tau$}
\prf{$\models \Gamma_0, x{:}\ann\tau$ \hfill by assumption}
\prf{$\models \Gamma_0$ and $\Gamma_0 \models \ann\tau : u$ \hfill by def. validity}
\prf{$\vdash \Gamma_0$ and $\Gamma_0 \Vdash \id(\Gamma_0) = \id(\Gamma_0) : \Gamma_0$ \hfill by IH}
\prf{$\forall \Gamma',~\theta,~\theta'.~\Gamma' \Vdash \theta = \theta' : \Gamma
\Longrightarrow \Gamma' \Vdash \{\theta\}\ann\tau  =  \{\theta\}\ann\tau : \{\theta\}u$  \hfill by def. validity}
\prf{$\Gamma_0 \Vdash \{\id(\Gamma_0)\}\ann\tau = \{\id(\Gamma_0)\}\ann\tau : \{\id(\Gamma_0)\}u$ \hfill by previous lines}
\prf{$\Gamma_0 \Vdash \ann\tau = \ann\tau : u$ \hfill by subst. def.}
\prf{$\Gamma_0 \vdash \ann\tau : u$ \hfill by semantic typing}
\prf{$\vdash \Gamma_0, x{:}\ann\tau$ \hfill by context def.}
\prf{$\neut x$ \hfill by def.}
\prf{$\Gamma_0, x{:}\ann\tau \vdash x : \ann\tau$ \hfill by typing rule}
\prf{$\Gamma_0, x{:}\ann\tau \vdash x \equiv x : \ann\tau$ \hfill by typing rule}
\prf{$\Gamma_0, x{:}\ann\tau \Vdash x = x : \ann\tau$ \hfill by Neutral Soundness Lemma \ref{lem:NeutSound}}
\prf{$\Gamma_0, x{:}\ann\tau \Vdash \id({\Gamma_0}) = \id({\Gamma_0}) : \Gamma_0$ \hfill by Sem. Weakening Lemma \ref{lem:weakcsub} }
\prf{$\Gamma_0, x{:}\ann\tau \Vdash \id({\Gamma_0}), x/x = \id({\Gamma_0}), x/x : \Gamma_0, x{:}\ann\tau$ \hfill by semantic def.}
\prf{$\Gamma_0, x{:}\ann\tau \Vdash \id({\Gamma_0, x{:}\tau}) = \id({\Gamma_0, x{:}\tau}) : \Gamma_0, x{:}\ann\tau$  \hfill by def. of $\id$}
}
\end{proof}

\begin{lemma}[Symmetry and Transitivity of Validity]\quad
  \begin{enumerate}
  \item If $\Gamma \models t = t' : \ann\tau$ then $\Gamma \models t' = t : \ann\tau$.
  \item If $\Gamma \models t_1 = t_2 : \ann\tau$ and $\Gamma \models t_2 =  t_3 : \ann\tau$
    then $\Gamma \models t_1 = t_3 : \ann\tau$.
  \end{enumerate}
\end{lemma}
\begin{proof}
Using Lemma \ref{lem:ctxsat} (Context Satisfiability), Lemma \ref{lm:symsemsub} (Symmetry and Transitivity for Substitutions),  Lemma \ref{lem:semsym} (Symmetry and Transitivity for Terms),  and Lemma \ref{lem:semsym} (Conversion).
\LONGVERSIONCHECKED{
\\[1em]
    \pcase{$\ibnc{\models \Gamma}
                {\Gamma \models \ann\tau : u \qquad
                 \forall \Gamma',~\theta,~\theta'.~\Gamma' \Vdash \theta = \theta' : \Gamma
         \Longrightarrow \Gamma' \Vdash \{\theta\}t  = \{\theta'\}t' : \{\theta\}\ann\tau}{\Gamma \models t = t' : \ann\tau}{}$}
\\
\prf{Assume $\Gamma' \Vdash \theta' = \theta : \Gamma$}
     \prf{$\Gamma' \Vdash \theta = \theta' : \Gamma$ \hfill by Lemma \ref{lm:symsemsub} (Symmetry)}
    \prf{$\Gamma' \Vdash \{\theta\}t  = \{\theta'\}t' : \{\theta\}\ann\tau$ \hfill by assumption $\Gamma \Vdash t = t' : \ann\tau$}
    \prf{$\Gamma' \Vdash \{\theta\}t'  = \{\theta'\}t : \{\theta\}\ann\tau$ \hfill by Lemma \ref{lem:semsym} (Symmetry)}
    \prf{$\models \Gamma$ \hfill by assumption}
    \prf{$\Gamma \Vdash \id(\Gamma) = \id(\Gamma) : \Gamma$ \hfill by Lemma \ref{lem:ctxsat}}
    \prf{$\Gamma \models \ann\tau : u$ \hfill by assumption}
    \prf{$\Gamma \models \ann\tau = \ann\tau : u$ \hfill by def.}
    \prf{$\Gamma' \Vdash \{\theta'\}\ann\tau = \{\theta\}\ann\tau : u$\hfill by $\Gamma \models \ann\tau = \ann\tau : u$}
    \prf{$\Gamma' \Vdash \{\theta'\}t'  = \{\theta\}t : \{\theta'\}\ann\tau$ \hfill by Lemma \ref{lem:semsym} (Conversion)}
    \prf{$\Gamma \models t' = t : \ann\tau$ \hfill by rule}
\\
\pcase{$\ibnc{\models \Gamma}{
             \Gamma \models \ann\tau : u \quad
             \forall \Gamma',~\theta,~\theta'.~\Gamma' \Vdash \theta = \theta' : \Gamma
\Longrightarrow \Gamma' \Vdash \{\theta\}t_1  =  \{\theta'\}t_2 : \{\theta\}\ann\tau}
           {\Gamma \models t_1 = t_2 : \ann\tau}{}$ and}
\pcase{$\ibnc{\models \Gamma}{
             \Gamma \models \ann\tau : u \quad
             \forall \Gamma',~\theta,~\theta'.~\Gamma' \Vdash \theta = \theta' : \Gamma
\Longrightarrow \Gamma' \Vdash \{\theta\}t_2  =   \{\theta'\}t_3 : \{\theta\}\ann\tau}{\Gamma \models t_2 = t_3 : \ann\tau}{}$}
\prf{Assume $\Gamma' \Vdash \theta = \theta' : \Gamma$}
\prf{$\Gamma' \Vdash \theta = \theta : \Gamma$ \hfill by symmetry and transitivity of substitution (Lemma \ref{lm:symsemsub})}
\prf{$\Gamma' \Vdash \{\theta\}t_1  = \{\theta\}t_2 : \{\theta\}\ann\tau$ \hfill by $\Gamma \models t_1 = t_2 : \ann\tau$}
\prf{$\Gamma' \Vdash \{\theta\}t_2  = \{\theta'\}t_3 : \{\theta\}\ann\tau$ \hfill by $\Gamma \models t_2 = t_3 : \ann\tau$}
\prf{$\Gamma' \Vdash \{\theta\}t_1  = \{\theta'\}t_3 : \{\theta\}\ann\tau$ \hfill by Lemma \ref{lem:semsym} (Transitivity)}
\prf{$\Gamma \models t_1 = t_3 : \ann\tau$ \hfill by rule}
}
\end{proof}

\begin{lemma}[Function Type Injectivity is valid.]
If $\Gamma \models (y:\ann\tau_1) \arrow \tau_2 = (y:\ann\tau_1') \arrow \tau'_2 : u_3$ then
$\Gamma \models \ann\tau_1 = \ann\tau_1' : u_1$ and
$\Gamma, y{:}\ann\tau_1 \models \tau_2 = \tau'_2 : u_2$ and $(u_1,~u_2,~u_3) \in \Ru$
\end{lemma}
\begin{proof}
Proof by unfolding the semantic definitions.
\LONGVERSIONCHECKED{
\quad\\[0.5em]
\prf{$\Gamma \models (y:\ann\tau_1) \arrow \tau_2 = (y:\ann\tau_1') \arrow \tau'_2 : u_3$ \hfill by assumption}
\prf{$\models \Gamma$  \hfill}
\prf{$ \forall \Gamma',~\theta,~\theta'.~\Gamma' \Vdash \theta = \theta' : \Gamma
     \Longrightarrow \Gamma' \Vdash \{\theta\}(y:\ann\tau_1) \arrow \tau_2  =
                                     \{\theta'\}(y:\ann\tau_1') \arrow \tau'_2  : \{\theta\}u_3$ \hfill by def. of validity}
\\[0.15em]
\prf{To prove: $\Gamma \models \ann\tau_1 = \ann\tau_1' : u_1$ }
\prf{Assume an arbitrary $\Gamma' \Vdash \theta = \theta' : \Gamma$.}
\prf{$\Gamma' \Vdash \{\theta\}(y:\ann\tau_1) \arrow \tau_2  =
                    \{\theta'\}(y:\ann\tau_1') \arrow \tau'_2  : u_3$ \hfill by previous lines and $\{\theta\}u_3 = u_3$}
\prf{$\Gamma' \Vdash (y: \{\theta\}\ann\tau_1) \arrow \{\theta,~y/y\}\tau_2  =
                     (y: \{\theta'\}\ann\tau_1') \arrow \{\theta', y/y\}\tau'_2  : u_3$ \hfill by subst. def.}
\prf{$(y: \{\theta'\}\ann\tau_1') \arrow \{\theta',~y/y\}\tau'_2  \whnf
  (y: \{\theta'\}\ann\tau_1') \arrow \{\theta',~y/y\}\tau'_2$ \hfill by
  unfolding semantic def. and $\whnf$\\
\mbox{\quad}\hfill since $\norm ((y: \{\theta'\}\ann\tau_1') \arrow \{\theta',~y/y\}\tau'_2 )$}
\\[0.15em]
\prf{$\forall \Gamma_0 \leq_\rho \Gamma'.~
     \Gamma_0 \Vdash  \{\rho\}\{\theta\}\ann\tau_1 = \{\rho\}\{\theta\}\ann\tau_1' : u_1$ \hfill by sem. def.}
\prf{$\Gamma' \Vdash \{\theta\}\ann\tau_1 = \{\theta\}\ann\tau_1':u_1$
   \hfill choosing $\Gamma_0 = \Gamma'$ and $\rho = \id(\Gamma')$ and the fact that $\{\id(\Gamma')\}\theta = \theta$}
\prf{$\forall \Gamma', \theta, \theta'.~\Gamma' \Vdash \theta = \theta' : \Gamma \Longrightarrow \Gamma' \Vdash \{\theta\}\ann\tau_1 = \{\theta'\}\ann\tau_1' : u_1 $ \hfill previous lines }
\prf{$\Gamma \models \ann\tau_1 = \ann\tau_1' : u_1$ \hfill def. of validity ($\models$), since $\Gamma', \theta, \theta'$ arbitrary}
\\[0.15em]
\prf{To prove: $\Gamma, y{:}\ann\tau_1 \models \tau_2 = \tau'_2 : u_2$ }
\prf{Assume an arbitrary $\Gamma' \Vdash \theta_2 = \theta'_2 : \Gamma, y{:}\ann\tau_1$.}
\prf{$\theta_2 = \theta, s/y$ and $\theta'_2 = \theta', s'/y$ \\
 $\Gamma' \Vdash \theta = \theta' : \Gamma$ and
 $\Gamma' \Vdash s = s' : \{\theta\}\ann\tau_1$ \hfill by inversion on  $\Gamma' \Vdash \theta_2 = \theta'_2 : \Gamma, y{:}\ann\tau_1$}
\prf{$\forall \Gamma_0 \leq_\rho \Gamma'.
~\Gamma_0 \Vdash s = s' : \{\rho\}\{\theta\}\ann\tau_1  \Longrightarrow
 \Gamma' \Vdash \{\rho, s/y\}\{\theta, y/y\}\tau_2 = \{\rho,s'/y\}\{\theta',y/y\}\tau'_2$ \hfill by sem. def. }
\prf{$\Gamma' \Vdash s = s' : \{\theta\}\ann\tau_1 \Longrightarrow
     \Gamma' \Vdash \{s/y\}\{\theta,y/y\}\tau_2 = \{s'/y\}\{\theta',y/y\}\tau'_2$
  \hfill by choosing $ \rho = \id(\Gamma')$}
\prf{$\Gamma' \Vdash s = s' : \{\theta\}\ann\tau_1 \Longrightarrow
     \Gamma' \Vdash \{\theta,s/y\}\tau_2 = \{\theta',s'/y\}\tau'_2$ \hfill by subst. def.}
\prf{$\Gamma' \Vdash \{\theta, s/y\}\tau_2 = \{\theta', s'/y\}\tau'_2$ \hfill by previous line}
\prf{$\Gamma' \Vdash \{\theta_2\}\tau_2 = \{\theta'_2\}\tau'_2 $ \hfill by previous lines}
\prf{$\Gamma , y{:}\ann\tau_1 \models \tau_2 = \tau'_2 : u_2$ \hfill by def. of validity, since $\Gamma', \theta, \theta'$ arbitrary}
}
\end{proof}

\begin{theorem}[Completeness of Validity]\label{lem:SemComplete}
If $\Gamma \models t = t' : \ann\tau$ then $\Gamma \vdash t : \ann\tau$ and
$\Gamma \vdash t' : \ann\tau$ and
$\Gamma \vdash t \equiv t' : \ann\tau$ and
$\Gamma \vdash \ann\tau : u$.
\end{theorem}
\begin{proof}
Unfolding of validity definition relying on context satisfiability
(Lemma \ref{lem:ctxsat}) and Well-Formedness of Semantic Typing (Lemma \ref{lm:semwf}).
\LONGVERSIONCHECKED{\\[0.5em]\quad
\prf{$\Gamma \models t = t' : \ann\tau$ \hfill by assumption}
\prf{$\models \Gamma$ \hfill by validity definition    }
\prf{$\vdash \Gamma$ and $\Gamma \Vdash \id(\Gamma) = \id(\Gamma) : \Gamma$ \hfill by Lemma \ref{lem:ctxsat}}
\prf{$\forall \Gamma',~\theta,~\theta'.~\Gamma' \Vdash \theta = \theta' :
                  \Gamma \Longrightarrow \Gamma' \Vdash \{\theta\}t  =
                  \{\theta\}t' : \{\theta\}\ann\tau$ \hfill by validity definition    }
\prf{$\Gamma \Vdash \{\id(\Gamma)\}t = \{\id(\Gamma)\}t' : \{\id(\Gamma)\}\ann\tau$ \hfill by previous lines }
\prf{$\Gamma \Vdash t = t' : \ann\tau$ \hfill by subst. def. }
\prf{$\Gamma \vdash t' : \ann\tau$,~ $\Gamma \vdash t: \ann\tau$,~
$\Gamma \vdash t \equiv t' : \ann\tau$ and $\Gamma \vdash \ann\tau : u$ \hfill
    by Well-Formedness of Seman. Typ. (Lemma \ref{lm:semwf})}
}
\end{proof}


The fundamental lemma (Lemma \ref{lm:fundtheo}) states that well-typed
terms are valid. The proof proceeds by mutual induction on the typing derivation for LF-objects and computations. To structure the proof of the fundamental lemma that well-typed computations are valid, we consider the validity of type conversion, computation-level functions, applications, and recursion
individually.

\begin{lemma}[Validity of Type Conversion]\label{lem:validtypconv}
If $\Gamma \models \ann\tau = \ann\tau' : u$ and $\Gamma \models t : \ann\tau$ then $\Gamma \models t : \ann\tau'$.
\end{lemma}
\begin{proof}
By definition relying on semantic type conversion lemma (Lemma \ref{lem:semsym} (\ref{it:conv})).
 \LONGVERSIONCHECKED{
\quad\\[0.5em]
\prf{$\Gamma \models t : \ann\tau$ \hfill by assumption }
\prf{$\Gamma \models t = t : \ann\tau$ \hfill by validity def.}
\prf{Assume $\Gamma' \Vdash \theta = \theta' : \Gamma$}
\prf{$\Gamma' \Vdash \{\theta\}t = \{\theta'\}t:\{\theta\}\ann\tau$ \hfill by validity def. $\Gamma \models t = t : \ann\tau$}
\prf{$\Gamma \Vdash \{\theta\}\ann\tau = \{\theta'\}\ann\tau' : u$ \hfill by validity def. $\Gamma \models \ann\tau = \ann\tau' : u$}
\prf{$\Gamma \Vdash \{\theta\}t = \{\theta'\} : \{\theta\}\ann\tau'$ \hfill by Lemma \ref{lem:semsym} (Conversion)}
\prf{$\Gamma \models t = t : \ann\tau'$ \hfill since $\Gamma', \theta,\theta'$ arbitrary}
}
\end{proof}

\begin{lemma}[Validity of Functions]\label{lem:ValidFun}
If  $\Gamma, y{:}\ann\tau_1 \models t : \tau_2$ then $\Gamma \models \tmfn y t : (y:\ann\tau_1) \arrow \tau_2$.
\end{lemma}
\begin{proof}
We unfold the validity definitions, relying on 
completeness of validity (Lemma \ref{lem:SemComplete}),
semantic weakening of computation-level substitutions (Lemma \ref{lem:weakcsub}),
Well-formedness Lemma \ref{lm:semwf},  Backwards Closure Lemma \ref{lem:bclosed},  Symmetry property of semantic equality (Lemma \ref{lem:semsym}).
\LONGVERSIONCHECKED{
\quad\\[0.5em]
\prf{$\Gamma, y{:}\ann\tau_1 \models t : \tau_2$ \hfill by assumption}
\prf{$\Gamma, y:\ann\tau_1 \models t = t : \tau_2$ \hfill by def. validity}
\prf{$\forall \Gamma',~\theta,~\theta'.
~\Gamma' \Vdash \theta = \theta' :\Gamma, y{:}\ann\tau_1
\Longrightarrow
\Gamma' \Vdash \{\theta\}t  = \{\theta'\}t : \{\theta\}\tau_2$ \hfill by def. of validity}
\prf{
$\Gamma, y:\ann\tau_1 \models \tau_2 : u$ \hfill by def. of validity\\
$\Gamma, y:\ann\tau_1 \models \tau_2 = \tau_2 : u$ \hfill by inversion on validity\\
$\models \Gamma, y:\ann\tau_1$ \hfill by inversion on validity\\
$\models \Gamma$ \hfill by inversion on validity\\
$\Gamma \models \ann\tau_1 : u$ \hfill by inversion on validity \\
$\Gamma \models \ann\tau_1 = \ann\tau_1 : u$ \hfill by inversion on validity
\\[-0.5em]
}
TO SHOW:
  \begin{enumerate}
  \item $\models \Gamma$
  \item $\Gamma \models (y:\ann\tau_1) \arrow \tau_2 : u$, i.e. \\
$\forall \Gamma', ~\theta,~\theta'. ~\Gamma' \Vdash \theta = \theta' :\Gamma \Longrightarrow \Gamma' \Vdash \{\theta\}((y:\ann\tau_1) \arrow \tau_2) = \{\theta'\}((y:\ann\tau_1) \arrow \tau_2) : u$
  \item $\forall \Gamma', ~\theta,~\theta'. ~
    \Gamma' \Vdash \theta = \theta' :\Gamma
    \Longrightarrow
    \Gamma' \Vdash \{\theta\}(\tmfn y t) = \{\theta'\}(\tmfn y t) : \{\theta\}((y:\ann\tau_1) \arrow \tau_2)$
  \end{enumerate}
\prf{~\\[0.5em](1) SHOW: $\models \Gamma$}
\prf{$\models \Gamma, y{:}\ann\tau_1$ \hfill by assumption $\Gamma, y:\ann\tau_1 \models t = t : \tau_2$}
\prf{$\models \Gamma$ \hfill by inversion on $\models \Gamma, y{:}\ann\tau_1$\\[-0.5em]}
\prf{(2) SHOW: $\forall \Gamma' \Vdash \theta = \theta' :\Gamma \Longrightarrow \Gamma' \Vdash \{\theta\}((y:\ann\tau_1) \arrow \tau_2) = \{\theta'\}((y:\ann\tau_1) \arrow \tau_2) : u$}
\prf{Assume $\Gamma' \Vdash \theta = \theta' :\Gamma$ \\[-0.75em]}
\prf{(2.a) SHOW: $\forall \Gamma'' \leq_\rho \Gamma'. \Gamma'' \Vdash s = s' : \{\rho\}\{\theta\}\ann\tau_1 \Longrightarrow
 \Gamma'' \Vdash \{\{\rho\}\theta, s/y\}\tau_2 = \{\{\rho\}\theta', s'/y\}\tau_2 : u$}
\prf{\mbox{$\quad$}Assume $\Gamma'' \leq_\rho \Gamma'. \Gamma'' \Vdash s = s' : \{\rho\}\{\theta\}\ann\tau_1$}
\prf{\mbox{$\quad$}$\Gamma'' \Vdash \{\rho\}\theta = \{\rho\}\theta' : \Gamma$ \hfill by weakening sem. subst. (Lemma \ref{lem:weakcsub})}
\prf{\mbox{$\quad$}$\Gamma'' \Vdash \{\rho\}\theta, s/y = \{\rho\}\theta', s'/y : \Gamma, y{:}\ann\tau_1$ \hfill by sem. def.}
\prf{\mbox{$\quad$}$\Gamma'' \Vdash \{\{\rho\}\theta, s/y\}\tau_2 = \{\{\rho\}\theta', s'/y\}\tau_2 : u$ \hfill by def. of $\Gamma, y:\ann\tau_1 \models \tau_2 = \tau_2 : u$ \\[-0.5em]}
%
\prf{(2.b) SHOW: $\forall \Gamma'' \leq_\rho \Gamma'. \Gamma'' \Vdash \{\rho\}\{\theta\}\ann\tau_1 = \{\rho\}\{\theta'\}\ann\tau_1 : u$\\}
\prf{\mbox{$\quad$}$\Gamma'' \Vdash \{\rho\}\theta = \{\rho\}\theta' : \Gamma$ \hfill by weakening sem. subst. (Lemma \ref{lem:weakcsub})}
\prf{\mbox{$\quad$}$\Gamma'' \Vdash \{\rho\}\{\theta\}\ann\tau_1 = \{\rho\}\{\theta'\}\ann\tau_1 : u$ \hfill by def. of $\Gamma \models \ann\tau_1 = \ann\tau_1 : u$\\[-0.5em] }
\prf{(2.c) SHOW: $\Gamma' \vdash \{\theta'\}((y:\ann\tau_1) \arrow \tau_2) \whnf \{\theta'\}((y:\ann\tau_1) \arrow \tau_2) :u$}
\prf{\mbox{$\quad$}$\Gamma \vdash (y:\ann\tau_1) \arrow \tau_2 : u$ \hfill by Completeness of Validity (Lemma \ref{lem:SemComplete})\\
\mbox{\hspace{1cm}}\hfill (using validity of functions which we show under (3))}
\prf{\mbox{$\quad$}$\Gamma' \vdash \theta : \Gamma$ and $\Gamma' \vdash \theta' : \Gamma$ \hfill by Well-formedness of semantic subst. (Lemma \ref{lem:wfsemsub})}
\prf{\mbox{$\quad$}$\Gamma' \vdash \{\theta'\}((y:\ann\tau_1) \arrow \tau_2) : u$ \hfill by computation-level subst. lemma (Lemma \ref{lm:compsub} )}
\prf{\mbox{$\quad$}$\norm ((y:\{\theta'\}\ann\tau_1) \arrow \{\theta', y/y\}\tau_2)$ \hfill by def. and subst. prop.\\}
\prf{(3) SHOW: $\forall \Gamma', ~\theta,~\theta'. ~
    \Gamma' \Vdash \theta = \theta' :\Gamma
    \Longrightarrow
    \Gamma' \Vdash \{\theta\}(\tmfn y t) = \{\theta'\}(\tmfn y t) : \{\theta\}((y:\ann\tau_1) \arrow \tau_2)$}
\prf{Assume $\Gamma' \Vdash \theta = \theta' :\Gamma$; Showing: $\Gamma' \Vdash \{\theta\}(\tmfn y t) = \{\theta'\}(\tmfn y t) : \{\theta\}((y:\ann\tau_1) \arrow \tau_2)$ \\[-0.75em]}
\prf{(3.a) SHOW: $\Gamma' \vdash \{\theta\}(\tmfn y t) \whnf w : \{\theta\}((y:\ann\tau_1) \arrow \tau_2)$ \\\mbox{$\quad\qquad$}and
                 $\Gamma' \vdash \{\theta'\}(\tmfn y t) \whnf w' : \{\theta\}((y:\ann\tau_1) \arrow \tau_2)$ }
\prf{\mbox{$\quad$}$\tmfn y \{\theta,y/y\} t \whnf w$ \hfill where $w = \tmfn y \{\theta,y/y\}t$ since $\norm (\tmfn y \{\theta,y/y\}t)$}
\prf{\mbox{$\quad$}$\tmfn y \{\theta',y/y\} t \whnf w'$ \hfill where $w' = \tmfn y \{\theta',y/y\}t$ since $\norm (\tmfn y \{\theta',y/y\}t)$}
\prf{\mbox{$\quad$}$\Gamma, y{:}\ann\tau_1 \vdash t : \tau_2$ \hfill by Well-formedness of semantic typing (Lemma \ref{lm:semwf})}
\prf{\mbox{$\quad$}$\Gamma' \vdash \theta : \Gamma$ and $\Gamma' \vdash \theta' : \Gamma$ \hfill by Well-formedness of semantic subst. (Lemma \ref{lem:wfsemsub})}
\prf{\mbox{$\quad$}$\Gamma', y{:}\{\theta\}\ann\tau_1 \vdash \theta, y/y : \Gamma, y{:}\ann\tau_1$
     and $\Gamma', y{:}\{\theta'\}\ann\tau_1 \vdash \theta', y/y : \Gamma, y{:}\ann\tau_1'$  \hfill by comp. subst.}
\prf{\mbox{$\quad$}$\Gamma', y{:}\{\theta\}\tau_1 \vdash \{\theta, y/y\} t : \{\theta, y/y\}\tau_2$ and \\
     \mbox{$\quad$}$\Gamma', y{:}\{\theta'\}\tau_1 \vdash \{\theta', y/y\} t : \{\theta', y/y\}\tau_2$ \hfill by computation-level subst. lemma (Lemma \ref{lm:compsub})}
\prf{\mbox{$\quad$}$\Gamma' \vdash \{\theta\}(\tmfn y t) : \{\theta\}((y:\ann\tau_1) \arrow \tau_2)$ and\\
     \mbox{$\quad$}$\Gamma' \vdash \{\theta'\}(\tmfn y t) : \{\theta'\}((y:\ann\tau_1) \arrow \tau_2)$ \hfill by typing rule}
\prf{\mbox{$\quad$}$\Gamma' \vdash \{\theta\}((y:\ann\tau_1) \arrow \tau_2) = \{\theta'\}((y:\ann\tau_1) \arrow \tau_2) : u$ \hfill by (2)}
\prf{\mbox{$\quad$}$\Gamma' \vdash \{\theta'\}(\tmfn y t) : \{\theta\}((y:\ann\tau_1) \arrow \tau_2)$ \hfill Conversion (Lemma \ref{lem:semsym}(\ref{it:conv}))\\[0.5em]}
\prf{(3.b) SHOW: $\forall \Gamma'' \leq_\rho \Gamma'. \Gamma'' \Vdash s = s' : \{\rho\}\{\theta\}\ann\tau_1 \Longrightarrow
  \Gamma'' \Vdash \{\rho\}w~s = \{\rho\}w'~s' : \{\rho, s/y\}\{\theta, y/y\}\tau_2$}
\prf{\mbox{$\quad$}Assume $ \Gamma'' \leq_\rho \Gamma'. \Gamma'' \Vdash s = s' : \{\rho\}\{\theta\}\ann\tau_1$}
\prf{\mbox{$\quad$}$\Gamma'' \Vdash \{\rho\}\theta = \{\rho\}\theta' : \Gamma$ \hfill by weakening sem. subst. (Lemma \ref{lem:weakcsub})}
\prf{\mbox{$\quad$}$\Gamma'' \Vdash \{\rho\}\theta, s/y = \{\rho\}\theta', s'/y : \Gamma, y{:}\ann\tau_1$ \hfill by sem. def.}
\prf{\mbox{$\quad$}$\Gamma'' \Vdash \{\{\rho\}\theta, s/y\}t = \{\{\rho\}\theta', s'/y\}t :  \{\{\rho\}\theta, s/y\}\tau_2$ \hfill using $\Gamma, y:\ann\tau_1 \models t = t : \tau_2$}
\prf{\mbox{$\quad$}$\Gamma'' \Vdash  \{\{\rho\}\theta, s/y\}t \whnf v :  \{\{\rho\}\theta, s/y\}\tau_2$ \\
\mbox{$\quad$} $\Gamma''\Vdash \{\{\rho\}\theta', s'/y\}t \whnf v':  \{\{\rho\}\theta, s/y\}\tau_2$
\hfill sem. definition $\Vdash$}
\prf{\mbox{$\quad$}$\Gamma'' \Vdash \{\rho\}\{\theta\}\tmfn y t \whnf \{\rho\}\{\theta\}\tmfn y t : \{\rho\}\{\theta\}((y:\ann\tau_1) \arrow \tau_2)$ \hfill  since $\norm (\{\rho\}\{\theta\}\tmfn y t)$ \\
\mbox{$\quad$}\hfill and $\Gamma'' \vdash  \tmfn y \{\{\rho\}\theta, y/y\}t : \{\rho\}\{\theta\}((y:\ann\tau_1) \arrow \tau_2)$}
\prf{\mbox{$\quad$}$\Gamma'' \Vdash (\tmfn y  \{\{\rho\}\theta,y/y\}t)~s \whnf v : \{\{\rho\}\theta, s/y\}\tau_2$
\hfill by rules (typing and $\whnf$)}
\prf{\mbox{$\quad$}$\Gamma'' \Vdash (\tmfn y \{\{\rho\}\theta,y/y\}t) ~s =
       \{\{\rho\}\theta', s'/y\}t : \{\{\rho\}\theta, s/y\}\tau_2$ \hfill by Backwards Closed (Lemma \ref{lem:bclosed})
}
\prf{\mbox{$\quad$} $\Gamma'' \Vdash \{\{\rho\}\theta, s/y\}\tau_2 = \{\{\rho\}\theta', s'/y\}\tau_2 : u$
\hfill by $\Gamma, y:\ann\tau_1 \models \tau_2 = \tau_2 : u$ }
\prf{\mbox{$\quad$} $\Gamma'' \Vdash  \{\{\rho\}\theta', s'/y\}t = (\tmfn y \{\{\rho\}\theta,y/y\}t) ~s : \{\{\rho\}\theta, s/y\}\tau_2$
\hfill Symmetry (Lemma \ref{lem:semsym}(\ref{it:sym}))}
\prf{\mbox{$\quad$}$\Gamma'' \Vdash \{\rho\}\{\theta'\}\tmfn y t \whnf \{\rho\}\{\theta'\}\tmfn y t : \{\rho\}\{\theta\}((y:\ann\tau_1) \arrow \tau_2)$ \hfill  since $\norm (\{\rho\}\{\theta'\}\tmfn y t)$ \\
\mbox{$\quad$}\hfill and $\Gamma'' \vdash  \tmfn y \{\{\rho\}\theta', y/y\}t : \{\rho\}\{\theta\}((y:\ann\tau_1) \arrow \tau_2)$}
\prf{\mbox{$\quad$}$\Gamma'' \Vdash (\tmfn y  \{\{\rho\}\theta',y/y\}t)~s' \whnf v' : \{\{\rho\}\theta, s/y\}\tau_2$
\hfill by rules (typing and $\whnf$)}
\prf{\mbox{$\quad$}$\Gamma'' \Vdash (\tmfn y  \{\{\rho\}\theta',y/y\}t)~s' = (\tmfn y  \{\{\rho\}\theta,y/y\}t)~s : \{\rho\{\theta\}, s/y\}\tau_2$ \hfill  \\
\mbox{$\quad$}\hfill by Backwards Closed (Lemma \ref{lem:bclosed})}
\prf{\mbox{$\quad$}$\Gamma'' \Vdash (\tmfn y  \{\{\rho\}\theta,y/y\}t)~s = (\tmfn y  \{\{\rho\}\theta',y/y\}t)~s': \{\rho\{\theta\}, s/y\}\tau_2$   \\
\mbox{$\quad$}\hfill by  Symmetry (Lemma \ref{lem:semsym}(\ref{it:sym}))}
\prf{\mbox{$\quad$}$\Gamma'' \Vdash (\{\rho\}\tmfn y  \{\theta,y/y\}t)~s = (\{\rho\}\tmfn y  \{\theta',y/y\}t)~s': \{\rho, s/y\}\{\theta, y/y\}\tau_2$ \hfill by comp. subst.}
\prf{\mbox{$\quad$}$\Gamma'' \Vdash \{\rho\}w~s = \{\rho\}w'~s' : \{\rho, s/y\}\{\theta, y/y\}\tau_2$ \hfill since $w = \tmfn y \{\theta,y/y\}t$  \\\mbox{$\quad$}\hfill and $w' = \tmfn y \{\theta',y/y\}t$}
}
\end{proof}

\begin{lemma}[Validity of Recursion] \label{lem:ValidRec}
\[
  \begin{array}{ll@{~}c@{~}l}
\multicolumn{4}{p{14cm}}{Let
$\Gamma \vdash \tmrec {\IH} {b_v} {b_{\mathsf{app}}} {b_{\clam}} \rappto \Psi~t: \{\Psi/\psi, t/m\}\tau$
and $\IH = (\psi:\tmctx) \arrow (m:\cbox{\unboxc{\psi} \vdash \tm}) \arrow \tau$.}
\\[0.25em]
\mbox{If} & \Gamma \models \IH : u ~\mbox{and}~
 \Gamma \models t : \cbox{\Psi \vdash \tm}~\mbox{and}\\
 & \Gamma, \psi:{\tmctx}, p:\cbox{\unboxc{\psi} \vdash_\# \tm} & \models &  t_v : \{p / y\}\tau ~\mbox{and}\\
& \Gamma,  \psi:{\tmctx}, m:\cbox{\unboxc{\psi} \vdash \tm}, n:\cbox{\unboxc{\psi} \vdash \tm},\\
& \qquad        f_m: \{m/y\}\tau, f_n: \{n/y\}\tau & \models & t_{\mathsf{app}} : \{\cbox{\unboxc{\psi} \vdash\capp~\unbox{m}{\id}~\unbox{n}{\id}}/y\}\tau ~\mbox{and}\\
& \Gamma, \psi:{\tmctx},  m:\cbox{\unboxc{\psi}, x:\tm \vdash \tm}, & & \\
&  \qquad        f_m:\{\cbox{\unboxc{\psi}, x:\tm}/\psi, m /y \} \tau & \models &  t_{\clam} : \{\cbox{\unboxc{\psi} \vdash \clam~\lambda x.\unbox{m}{\id}~} / y\}\tau
\\[0.25em]
\multicolumn{4}{p{13cm}}{then $\Gamma \models \tmrec {\IH} {b_v} {b_{\mathsf{app}}} {b_{\clam}} \rappto \Psi~t:   \{\Psi/\psi, t/m\}\tau$.}
  \end{array}
\]
\end{lemma}
\begin{proof}
We assume $\Gamma' \Vdash \theta = \theta' : \Gamma$, and show
$\Gamma' \Vdash \{\theta\}(\trec{\R}{\tm}\IH \rappto \Psi~t )
    = \{\theta'\}(\trec{\R}{\tm}\IH \rappto \Psi~t): \{\theta\}\{\Psi/\psi,~t/m\}\tau$ by considering different cases for
$\Gamma' \Vdash \{\theta\}t = \{\theta'\}t : \cbox{\{\theta\}\Psi \vdash \tm}$. Let $\Psi' = \{\theta\}\Psi$. In the case where $\Gamma' \vdash \{\theta\}t \whnf \cbox{\hatctx\Psi' \vdash M} : \cbox{\Psi' \vdash \tm}$ and $\Gamma' \vdash \{\theta'\}t \whnf \cbox {\Psi' \vdash N} : \cbox{\{\theta\}\Psi \vdash \tm}$ and $\Gamma' ; \Psi' \Vdash M = N : \tm$, we proceed by an inner induction on $\Gamma' ; \Psi' \Vdash M = N : \tm$ exploiting on Back. Closed Lemma (\ref{lem:bclosed}).
\LONGVERSIONCHECKED{
$\quad$\\[1em]
We now give the full proof.
\\[1em]
Let $\trec{\R}{\tm}\IH = \tmrec {\IH} {b_v} {b_{\mathsf{app}}} {b_{\clam}}$
\\[1em]
\prf{Assume $\Gamma' \Vdash \theta = \theta' : \Gamma$; \\
     TO SHOW: $\quad\Gamma' \Vdash \{\theta\}(\trec{\R}{\tm}\IH \rappto \Psi~t )
    = \{\theta'\}(\trec{\R}{\tm}\IH \rappto \Psi~t): \{\theta\}\{\Psi/\psi,~t/m\}\tau$
}
\\
\prf{$\Gamma' \Vdash \{\theta\}t = \{\theta'\}t : \cbox{\{\theta\}\Psi \vdash \tm}$ \hfill by validity of $\Gamma \models t : \cbox{\Psi \vdash \tm}$\\[-0.5em]}
 \prf{Let $\Psi' = \{\theta\}\Psi$. We now proceed to prove:
 \\[0.25em]
\pcase{$\Gamma' \vdash \{\theta\}t \whnf w : \cbox{\Psi' \vdash \tm}$  \\
     \mbox{\hspace{0.2cm}}and $\Gamma' \vdash \{\theta'\}t \whnf w' : \cbox{\Psi' \vdash \tm}$ \\
     \mbox{\hspace{0.2cm}}and $\Gamma' ; \Psi' \Vdash \unbox w \id = \unbox {w'}\id : \tm$
}
\prf{We write $M$ for $\unbox w \id$ and $N$ for $\unbox {w'}\id$ below.}
 \begin{center}
 \begin{tabular}{p{13cm}}
 If $\Gamma' ; \Psi' \Vdash M = N : \tm$ \\
     then $\Gamma' \Vdash \{\theta\}(\trec{\R}{\tm}{\IH}~\rappto \cbox{\Psi})~\cbox{\hatctx{\Psi} \vdash M}
              = \{\theta'\}(\trec{\R}{\tm}{\IH}~\rappto \cbox{\Psi})~ \cbox{\hatctx{\Psi} \vdash N} :
      \{\theta,~\Psi'/\psi, \cbox{\hatctx{\Psi} \vdash M}/m\}\tau$   \\[0.25em] $\quad$
 \end{tabular}
 \end{center}
 by induction on $M$, i.e. we may appeal to the induction hypothesis if the term $M$ has made progress and has stepped using $\lfwhnf$ and hence is ``smaller'' \\[0.5em]}
 \prf{\emph{Sub-case}:
   $\Gamma'; \Psi' \vdash M \lfwhnf \capp M_1~M_2 : \tm$ and $\Gamma' ; \Psi' \vdash N \lfwhnf \capp N_1~N_2 : \tm$ \\
 \mbox{\qquad\qquad}$\Gamma'; \Psi' \Vdash M_1 = N_1 : \tm$ and $\Gamma' ; \Psi' \Vdash M_2 = N_2 : \tm$\\[-0.5em]}
 \prf{$\Gamma' \Vdash \{\theta\}(\trec{\R}{\tm}{\IH}~\rappto \cbox{\Psi})~\cbox{\hatctx{\Psi} \vdash M_1}
              = \{\theta'\}(\trec{\R}{\tm}{\IH}~\rappto \cbox{\Psi})~\cbox{\hatctx{\Psi} \vdash N_1} :
      \{\theta,~\Psi'/\psi, \cbox{\hatctx{\Psi} \vdash M_1/m}\}\tau$
      \hfill by IH(i)}
 \prf{$\Gamma' \Vdash \{\theta\}(\trec{\R}{\tm}{\IH}~\rappto \cbox{\Psi})~\cbox{\hatctx{\Psi'} \vdash M_2}
                    = \{\theta'\}(\trec{\R}{\tm}{\IH}~\rappto \cbox{\Psi'})~\cbox{\hatctx{\Psi'} \vdash N_2} :
      \{\theta,~\Psi'/\psi, \cbox{\hatctx{\Psi} \vdash M_2}/m\}\tau $
      \hfill by IH(i)}
 \prf{$\Gamma' \vdash \{\theta\}(\trec{\R}{\tm}{\IH}~\rappto \Psi)~\cbox{\hatctx \Psi \vdash M_i} :  \{\theta,~\Psi'/\psi, \cbox{\hatctx\Psi \vdash M_i}/m\}\tau$ \hfill by Well-form. of Sem. Def.}
 \prf{$\Gamma' \vdash \{\theta\}(\trec{\R}{\tm}{\IH}~\rappto \Psi)~\cbox{\hatctx \Psi \vdash N_i} :  \{\theta',~\Psi'/\psi, \cbox{\hatctx\Psi \vdash N_i}/m\}\tau$ \hfill by Well-form. of Sem. Def. and Type Conv.}
 \prf{let $\theta_\capp = \theta, \Psi'/\psi, \cbox{\hatctx{\Psi} \vdash M_1} / m,~\cbox{\hatctx{\Psi} \vdash M_2} / n\\
 \mbox{\qquad\quad$\;$\quad} \{\theta\}(\trec{\R}{\tm}{\IH}~\rappto \cbox{\Psi})~\cbox{\hatctx{\Psi} \vdash M_1} / f_m,
  \{\theta\}(\trec{\R}{\tm}{\IH}~\rappto \Psi)~\cbox{\hatctx{\Psi} \vdash M_2} / f_n
          $
 }
 \prf{let $\theta'_\capp = \theta', \Psi'/\psi, \cbox{\hatctx{\Psi} \vdash  N_1} / m,~\cbox{\hatctx{\Psi} \vdash  N_2} / n\\
 \mbox{\qquad\quad$\;$\quad} \{\theta'\}(\trec{\R}{\tm}{\IH} \rappto {\Psi})~\cbox{\hatctx{\Psi} \vdash N_1} / f_m,
  \{\theta'\}(\trec{\R}{\tm}{\IH} \rappto {\Psi})~\cbox{\hatctx{\Psi} \vdash N_2} / f_n
          $
 }
 \prf{let $\Gamma_\capp =  \Gamma,  \psi:\cbox{\tmctx}, m:\cbox{\unboxc{\psi} \vdash \tm}, n:\cbox{\unboxc{\psi} \vdash \tm},  f_m: \{m/y\}\tau, f_n: \{n/y\}\tau$}
 \prf{$\Gamma' \Vdash \theta_\capp = \theta'_\capp : \Gamma_\capp$}
 \prf{$\Gamma' \Vdash \{\theta_\capp\}t_\capp = \{\theta'_\capp\}t_\capp :
       \{\theta_\capp\} \{\cbox{\psi \vdash\capp~\unbox{m}{\id}~\unbox{n}{\id}}/y\}\tau$
    \hfill by sem. def. of $t_\capp$}
 \prf{$\Gamma' \Vdash  \{\theta_\capp\}t_\capp = \{\theta'_\capp\}t_\capp :
      \{\theta, \Psi'/\psi,~\cbox{\hatctx \Psi \vdash\capp~M_1~M_2}/m\}\tau$
    \hfill by subst. def.}
  \prf{$\Gamma' \Vdash C = C : (\{\theta\}\Psi \vdash \tm)$ \hfill by reflexivity (Lemma \ref{lem:semsymlf})
 where $C = \hatctx\Psi \vdash \capp M_1~M_2$ }
 \prf{$\Gamma' \Vdash \{\theta\}t = \cbox{\hatctx{\Psi} \vdash \capp M_1~M_2} : \Psi' \vdash \tm$
    \hfill since $\Gamma' \vdash \{\theta\}t \whnf \cbox{C} : (\{\theta\}\Psi \vdash \tm)$ }
 \prf{$\Gamma' \Vdash \{\theta, \Psi'/\psi,~\{\theta\}t/m\}\tau =  \{ \theta,~\Psi'/\psi, \cbox{\hatctx\Psi \vdash \capp M_1~M_2}/m\}\tau : u$
        \hfill by sem. def. of $\Gamma \Vdash \IH : u$}
 \prf{$\Gamma' \Vdash  \{\theta_\capp\}t_\capp = \{\theta'_\capp\}t_\capp :
         \{\theta, \Psi'/\psi,~\{\theta\}t/m\}\tau$ \hfill by type conversion}
 \prf{$\Gamma' \vdash \{\theta_\capp\}t_\capp \whnf v :  \{\theta, \Psi'/\psi,~\{\theta\}t/m\}\tau$ \hfill by previous sem. def. }
 \prf{$\Gamma' ; \Psi' \vdash  \{\theta\}(\trec{\R}{\tm}{\IH}~@~\Psi~t) \whnf v :   \{\theta, \Psi'/\psi,~\{\theta\}t/m\}\tau$ \\
 \mbox{\qquad} \hfill since $\Gamma' \vdash \{\theta\}t \whnf \cbox{\hatctx \Psi \vdash M}: \tm$ and $\Gamma' ; \Psi' \vdash M \lfwhnf \capp M_1~M_2 : \tm$}
 \prf{$\Gamma' \Vdash \{\theta\}(\trec{\R}{\tm}{\IH}~@~\Psi~t) =  \{\theta'_\capp\}t_\capp :
         \{\theta, \Psi'/\psi,~\{\theta\}t/m\}\tau$ \hfill by Back. Closed (Lemma \ref{lem:bclosed})}
 \prf{$\Gamma' \vdash \{\theta'_\capp\}t_\capp \whnf v' :  \{\theta, \Psi'/\psi,~\{\theta\}t/m\}\tau$ \hfill by previous sem. def. }
 \prf{$\Gamma' ; \Psi' \vdash  \{\theta'\}(\trec{\R}{\tm}{\IH}~@~\Psi~t) \whnf v' :   \{\theta, \Psi'/\psi,~\{\theta\}t/m\}\tau$ \hfill (using type conversion)\\
 \mbox{\qquad} \hfill since $\Gamma' \vdash \{\theta'\}t \whnf \cbox{\hatctx \Psi \vdash N} : \tm$ and $\Gamma' ; \Psi' \vdash N \lfwhnf \capp N_1~N_2 : \tm$}
 \prf{$\Gamma' \Vdash \{\theta\}(\trec{\R}{\tm}{\IH}~@~\Psi~t) = \{\theta'\}(\trec{\R}{\tm}{\IH}~@~\Psi~t) :
         \{\theta, \Psi'/\psi,~\{\theta\}t/m\}\tau$ \hfill by Back. Closed (Lemma \ref{lem:bclosed})}
 \prf{$\Gamma \models \trec{\R}{\tm}{\IH}\rappto \Psi~t : \IH$ \hfill since $\Gamma', \theta, \theta'$ were arbitrary\\[0.5em]}
 \prf{\emph{Sub-Case.}  Other Sub-Cases where $\Gamma' ; \Psi' \vdash \{\theta\}M \lfwhnf \clam \lambda x.M : \tm$ and
  $\Gamma' ; \Psi' \vdash \{\theta\}M \lfwhnf x : \tm$ where $x:\tm \in \Psi'$ are similar.\\[0.5em]}
 \prf{\emph{Sub-Case.} $\Gamma' ; \Psi' \vdash \{\theta\}M \lfwhnf r_1  : \tm$ where $r_1 = \unbox{t_1}{\sigma_1}$ and $\neut t_1$\\
 \mbox{\hspace{0.8cm}}and $\Gamma' ; \Psi' \vdash \{\theta'\}M \lfwhnf r_2 : \tm$ where $r_2 = \unbox{t_2}{\sigma_2}$ and $\neut t_2$ }
\prf{$\Gamma' ; \Psi' \Vdash \sigma_1 = \sigma_2 : \Phi$ \hfill since $\Gamma' ; \Psi' \Vdash_{LF} \{\theta\}M = \{\theta'\}M : \tm$}
\prf{$\Gamma' ; \Psi' \vdash \sigma_1 \equiv \sigma_2 : \Phi$ \hfill by Well-formedness Lemma}
\prf{$\Gamma' \vdash t_1 \equiv t_2 : \cbox{\Phi \vdash \tm}$ \hfill since $\Gamma' ; \Psi' \Vdash_{LF} \{\theta\}M = \{\theta'\}M : \tm$}
\prf{$\Gamma' ; \Psi' \vdash r_1 \equiv r_2 : \tm$}
\prf{$\Gamma' \vdash \{\theta\}\Psi \equiv \{\theta'\}\Psi : \tmctx$ \hfill Sem. Subst. Preserve Equiv. (Lemma \ref{lem:semsubst})\\
\mbox{\hspace{1cm}}\hfill since $\Gamma' \vdash \theta \equiv \theta' : \Gamma$}
\prf{$\Gamma' \vdash \{\theta\}(\trec{\R}{\tm}{\IH}\rappto \Psi)~\cbox{\hatctx\Psi \vdash r_1}
\equiv \{\theta'\}(\trec{\R}{\tm}{\IH}\rappto \Psi)~\cbox{\hatctx\Psi \vdash r_2} : \{\theta, \Psi'/\psi,~\{\theta\}t/m\}\tau$ \hfill by $\equiv$}
\prf{$\Gamma' \vdash \{\theta\}(\trec{\R}{\tm}{\IH}\rappto \Psi)~\{\theta\}t \whnf \{\theta\}(\trec{\R}{\tm}{\IH}\rappto \Psi)~\cbox{\hatctx\Psi \vdash r_1}  : \{\theta, \Psi'/\psi,~\{\theta\}t/m\}\tau$  \hfill by $\whnf$}
\prf{$\Gamma' \vdash \{\theta'\}(\trec{\R}{\tm}{\IH}\rappto \Psi)~\{\theta'\}t \whnf \{\theta'\}(\trec{\R}{\tm}{\IH}\rappto \Psi)~\cbox{\hatctx\Psi \vdash r_2} :  \{\theta, \Psi'/\psi,~\{\theta\}t/m\}\tau$ \hfill by $\whnf$ and type conversion}
\prf {$\Gamma' \Vdash \{\theta\}(\trec{\R}{\tm}{\IH}~@~\Psi~t) = \{\theta'\}(\trec{\R}{\tm}{\IH}~@~\Psi~t) :
         \{\theta, \Psi'/\psi,~\{\theta\}t/m\}\tau$ \hfill Back. Closed (Lemma \ref{lem:bclosed})\\
\mbox{\hspace{1cm}}\hfill and Symmetry}
 %
 \pcase{$\neut w$ and $\neut w'$ and $\Gamma' \vdash w \equiv w' : \cbox{\Psi' \vdash \tm}$}
\prf{$\neut \{\theta\}(\trec{\R}{\tm}{\IH}~@~\Psi)~w$ and $\neut \{\theta'\}(\trec{\R}{\tm}{\IH}~@~\Psi)~w'$ \hfill by def. of $\neut$}
\prf{$\Gamma' \vdash  \{\theta\}\trec{\R}{\tm}{\IH} \equiv  \{\theta'\}\trec{\R}{\tm}{\IH} $ (short for all the branches remain equivalent) \hfill Sem. Subst. Preserve Equiv. (Lemma \ref{lem:semsubst})}
\prf{$\Gamma' \vdash \{\theta\}\Psi \equiv \{\theta'\}\Psi : \tmctx$ \hfill Sem. Subst. Preserve Equiv. (Lemma \ref{lem:semsubst})}
\prf{$\Gamma' \Vdash \{\theta\}(\trec{\R}{\tm}{\IH}~@~\Psi)~w \equiv \{\theta'\}(\trec{\R}{\tm}{\IH}~@~\Psi)~w' :
         \{\theta, \Psi'/\psi,~\{\theta\}t/m\}\tau$ \hfill by $\equiv$ }
%
 \prf{$\Gamma' \vdash  \{\theta\}(\trec{\R}{\tm}{\IH} \rappto \Psi)~t \whnf  \{\theta\}(\trec{\R}{\tm}{\IH}\rappto \Psi)~w
           :  \{\theta, \Psi/\psi,~\{\theta\} t/m\}\tau $
     \hfill by type conversion and $\whnf$ rule}
 \prf{$\Gamma' \vdash  \{\theta\}(\trec{\R}{\tm}{\IH} \rappto \Psi)~t \whnf  \{\theta\}(\trec{\R}{\tm}{\IH}~\rappto \Psi)~w'
           :  \{\theta, \Psi/\psi,~\{\theta\} t/m\}\tau $
     \hfill by type conversion and $\whnf$ rule}
  \prf{$\Gamma' \Vdash \{\theta\}(\trec{\R}{\tm}{\IH}~\rappto\Psi~t) = \{\theta'\}(\trec{\R}{\tm}{\IH}~\rappto \Psi ~t):
             \{\theta, \Psi/\psi,~\{\theta\}t/m\}\tau$
    \hfill by Back. Closed Lemma (twice) (\ref{lem:bclosed}) \\ \mbox{\qquad}\hfill and Symmetry }
 \prf{$\Gamma \models \trec{\R}{\tm}{\IH}\rappto \Psi~t : \IH$ \hfill since $\Gamma', \theta, \theta'$ were arbitrary\\[0.5em]}
}
\end{proof}

\begin{lemma}[Validity of Application]\label{lem:ValidTypeApp}
  If $\Gamma \models t : (y:\ann\tau_1) \arrow \tau_2$ and $\Gamma \models s : \ann\tau_1$
then $\Gamma \models t~s : \{s/y\}\tau_2$.
\end{lemma}
\begin{proof}
By definition relying on semantic type application lemma (Lemma \ref{lem:SemTypeApp}).
\end{proof}

 \begin{theorem}[Fundamental Theorem]\quad
   \label{lm:fundtheo}
  \begin{enumerate}
  \item If $\vdash \Gamma$ then $\models \Gamma$.
  \item If $\Gamma ; \Psi \vdash M : A$ then $\Gamma ; \Psi \models M = M : A$.
  \item If $\Gamma ; \Psi \vdash \sigma  : \Phi$ then $\Gamma ; \Psi \models \sigma = \sigma : \Phi$.
  \item If $\Gamma ; \Psi \vdash M \equiv N : A$ then $\Gamma ; \Psi \models M = N : A$.
  \item If $\Gamma ; \Psi \vdash \sigma \equiv \sigma' : \Phi$ then $\Gamma ; \Psi \models \sigma = \sigma' :  \Phi$.
  \item If $\Gamma \vdash t : \tau$ then $\Gamma \models t : \tau$.
  \item If $\Gamma \vdash t \equiv t' : \tau$ then $\Gamma \models t = t' : \tau$.
  \end{enumerate}
\end{theorem}
\begin{proof}
By induction on the first derivation using the previous lemma on validity of application (\ref{lem:ValidTypeApp}), Backwards Closed (\ref{lem:bclosed}), Well-formedness Lemma~\ref{lm:semwf}, 
Lemma  \ref{lem:validtypconv}, Lemma \ref{lm:ctxwf}, Lemma \ref{lem:ValidFun}, Sem. weakening lemma \ref{lem:semweak}, Validity of Recursion Lemma \ref{lem:ValidRec}.
\LONGVERSIONCHECKED{
\\[0.5em]
\fbox{ If $\Gamma \vdash t : \tau$ then $\Gamma \models t : \tau$.}
\\[0.5em]
\paragraph{Case} $\D = \ianc{\Gamma \vdash t : \cbox{\Phi \vdash A} \quad \Gamma ; \Psi \vdash \sigma : \Phi}
                            {\Gamma ; \Psi \vdash \unbox t \sigma : [\sigma]A}{}$\\[0.5em]
\prf{$\Gamma \models t : \cbox{\Phi \vdash A}$ \hfill by IH}
\prf{$\Gamma ; \Psi \models \sigma = \sigma: \Phi$ \hfill by IH}\\[-0.75em]
\prf{Assume $\Gamma' \Vdash \theta = \theta' : \Gamma$}
\prf{$\Gamma' \Vdash \{\theta\}t = \{\theta'\}t : \{\theta\}\cbox{\Phi \vdash A}$ \hfill by def. $\Gamma \models t : \cbox{\Phi \vdash A}$}
\\[-0.75em]
\prf{\emph{Sub-case}: $\Gamma' \vdash \{\theta\}t \whnf \cbox C : \{\theta\}\cbox{\Phi \vdash A}$ and $\Gamma' \vdash \{\theta'\}t \whnf \cbox {C'} : \{\theta\}\cbox{\Phi \vdash A}$ \\
\mbox{\qquad} and $\Gamma' \Vdash C = C' : \{\theta\}(\Phi \vdash A)$}
\prf{$C = \hatctx{\Phi} \vdash M$ and $C' = \hatctx{\Phi} \vdash N$
   \hfill by inversion on $\Gamma' \Vdash C = C' : \{\theta\}(\Phi \vdash A)$}
\prf{$\Gamma' ; \{\theta\}\Psi \Vdash \{\theta\} \sigma = \{\theta'\}\sigma: \{\theta\}\Phi$
   \hfill by $\Gamma ; \Psi \models \sigma : \Phi$  }
\prf{$\Gamma' ; \{\theta\}\Phi \Vdash M = N: \{\theta\}A$
   \hfill by def. of $\Gamma' \Vdash C = C' : \{\theta\}(\Phi \vdash A)$}
\prf{$\Gamma' ; \{\theta\}\Psi \Vdash [\{\theta\}\sigma] M = [\{\theta'\}\sigma]N: \{\theta\}[\sigma]A$
   \hfill by Lemma \ref{lem:semlfeqsub} }
\prf{$\Gamma' ; \{\theta\}\Psi \vdash [\{\theta\}\sigma] M \lfwhnf W : \{\theta\}[\sigma]A$ \hfill by previous line}
\prf{$\Gamma' \vdash \{\theta\}t \whnf \cbox C : \{\theta\}\cbox{\Phi \vdash A}$ \hfill by restating the condition of the case we are in}
\prf{$\Gamma' ; \{\theta\}\Psi \vdash \{\theta\}(\unbox t \sigma) \lfwhnf W : \{\theta\}[\sigma]A$ \hfill by whnf rules}
\prf{$\Gamma' ; \{\theta\}\Psi \vdash [\{\theta'\}\sigma]N \lfwhnf W' : \{\theta\}[\sigma]A$ \hfill by previous line}
\prf{$\Gamma' \vdash \{\theta'\}t \whnf \cbox {C'} : \{\theta\}\cbox{\Phi \vdash A}$ \hfill by restating the condition of the case we are in}
\prf{$\Gamma' ; \{\theta\}\Psi \vdash \{\theta'\}(\unbox t \sigma) \lfwhnf W'  : \{\theta\}[\sigma]A$ \hfill by whnf rules}
\prf{$\Gamma' ; ; \{\theta\}\Psi  \Vdash \{\theta\}(\unbox t \sigma) = \{\theta'\}(\unbox t \sigma) : \{\theta\}[\sigma]A$
   \hfill by Backwards Closed Lemma \ref{lem:bclosed} (twice) \\
\mbox{$\quad$}\hfill and symmetry.}
\prf{$\Gamma ; \Psi \models \unbox t \sigma : [\sigma]A$
  \hfill by abstraction, since $\Gamma', \theta, \theta'$ were arbitrary}
\\
\prf{\emph{Sub-case}: $\Gamma' \vdash \{\theta\}t \whnf w : \{\theta\}\cbox{\Phi \vdash A}$ and
                      $\Gamma' \vdash \{\theta'\}t \whnf w': \{\theta\}\cbox{\Phi \vdash A}$ \\
\mbox{\qquad} and $\neut w$ and $\neut w'$ and $\Gamma \vdash w \equiv w' : \{\theta\}\cbox{\Phi \vdash A}$}
\\[-0.75em]
\prf{$\Gamma' \vdash \theta(\unbox t \sigma) \whnf \unbox w {\{\theta\}\sigma}$ \hfill by whnf rules}
\prf{$\Gamma' \vdash \theta'(\unbox t \sigma) \whnf \unbox {w'} {\{\theta\}\sigma}$ \hfill by whnf rules}
\prf{$\Gamma' \{\theta\}\Psi \Vdash \{\theta\}\sigma = \{\theta'\}\sigma : \{\theta\}\Phi$ \hfill by def. $\Gamma ; \Psi \models \sigma = \sigma : \Phi$}
\prf{$\Gamma' ; \{\theta\}\Psi \vdash \{\theta\}\sigma \equiv \{\theta'\}\sigma : \{\theta\}\Phi$ \hfill by Well-formedness Lemma~\ref{lm:semwf}}
\prf{$\Gamma' ; \{\theta\}\Psi \vdash \{\theta\}(\unbox t \sigma) \equiv \{\theta'\}(\unbox t \sigma) : \{\theta\}[\sigma]A$ \hfill by $\equiv$ rules}
\prf{$\Gamma' \Vdash \{\theta\}\Phi = \{\theta\}\Phi : \ctx$ \hfill Reflexivity}
\prf{$\typeof (\Gamma' \vdash w) = \cbox{\{\theta\}\Phi \vdash \tm}$ \hfill since $\Gamma' \vdash w : \{\theta\}\cbox{\Phi \vdash A}$ }
\prf{$\typeof (\Gamma' \vdash w') = \cbox{\{\theta\}\Phi \vdash \tm}$ \hfill since $\Gamma' \vdash w' : \{\theta\}\cbox{\Phi \vdash A}$ }
\prf{$\Gamma' ; \{\theta\}\Psi \Vdash \{\theta\}(\unbox t \sigma) = \{\theta'\}(\unbox t \sigma) : \{\theta\}[\sigma]A$
    \hfill by semantic def. }
\prf{$\Gamma ; \Psi \models \unbox t \sigma : [\sigma]A$
   \hfill by abstraction, since $\Gamma', \theta, \theta'$ were
   arbitrary}
}
\LONGVERSIONCHECKED{\\
\pcase{$\D = \ianc {y:\ann\tau \in \Gamma}{\Gamma \vdash y : \ann\tau}{}$}
\prf{$\Ca: ~~\vdash \Gamma$ and $\Ca < \D$ \hfill by Lemma \ref{lm:ctxwf}}
\prf{Assume $\Gamma', \theta, \theta'.~\Gamma' \Vdash \theta = \theta' : \Gamma$}
\prf{$\Gamma' \Vdash t = s : \{\theta_i\}\ann\tau$ \hfill by sem. def. of $\Gamma' \Vdash \theta = \theta' : \Gamma$}
\prf{$\Gamma' \Vdash \{\theta\}y = \{\theta'\}y : \{\theta\}\ann\tau$ \hfill subst. def. where $\{\theta\}y = t$ and $\{\theta'\}y = s$}
\prf{$\Gamma \models y = y : \ann\tau$ \hfill by def. of validity}
\\
\pcase{ $\D = \ibnc
{\Gamma \vdash t : \tau'}{\Gamma \vdash \tau' \equiv \tau : u}
{\Gamma \vdash t : \tau}{}$}
\prf{$\Gamma \models t : \tau'$ \hfill by IH }
\prf{$\Gamma \models \tau = \tau' : u$ \hfill by IH }
\prf{$\Gamma \models t : \tau$ \hfill by Lemma \ref{lem:validtypconv}}
\\
\pcase{$\D = \ianc{\Gamma \vdash C : T}
                           {\Gamma \vdash \cbox C : \cbox T}{}$
}
\prf{$\Ca: ~~\vdash \Gamma$ and $\Ca < \D$ \hfill by Lemma \ref{lm:ctxwf}}
\prf{$\models \Gamma$ \hfill  by IH}
\prf{Let $C = \hatctx{\Psi} \vdash M$ and $T = \Psi \vdash A$.}
\prf{$\Gamma ; \Psi \vdash M : A$ \hfill by inversion on $\Gamma \vdash C : T$}
\prf{$\Gamma ; \Psi \models M : A$ \hfill by IH}
\prf{Assume $\Gamma',~\theta,~\theta'.~ \Gamma' \Vdash \theta = \theta'$}
\prf{$\Gamma ; \{\theta\}\Psi \Vdash \{\theta\}M = \{\theta'\}M : \{\theta\}A$ \hfill using $\Gamma ; \Psi \models M : A$}
\prf{$\Gamma \Vdash \{\theta\}(\hatctx{\Psi} \vdash M) = \{\theta'\}(\hatctx{\Psi} \vdash M) : \{\theta\}T$ \hfill by sem. def.}
\prf{$\Gamma \Vdash \{\theta\}\cbox C = \{\theta'\}\cbox C: \{\theta\}\cbox T$ \hfill by previous line}
\prf{$\Gamma \models \cbox C = \cbox C: \cbox T$ \hfill by abstraction, since $\Gamma', ~\theta,~\theta'$ were arbitrary}
\prf{$\Gamma \models \cbox C : \cbox T$ \hfill by def. of validity}
\\
\pcase{$\D = \ibnc
{\Gamma \vdash t : (y:\ann\tau_1) \arrow \tau_2}{ \Gamma \vdash s : \ann\tau_1}
{\Gamma \vdash t~s : \{s/y\}\tau_2}{}
$}
\prf{$\Gamma \models s : \ann\tau_1$ \hfill by IH }
\prf{$\Gamma \models t : (y:\ann\tau_1) \arrow \tau_2$ \hfill by IH}
\prf{$\Gamma \models t~s : \{s/y\}\tau_2$ \hfill by Lemma \ref{lem:ValidTypeApp}}
\pcase{$\D = \ianc{\Gamma, y:\ann\tau_1 \vdash t : \tau_2 }
                            {\Gamma \vdash \tmfn y t : (y:\ann\tau_1) \arrow \tau_2}{}$}
\\[0.2em]
\prf{$\Ca: ~~\vdash \Gamma, y{:}\ann\tau_1$ and $\Ca < \D$ \hfill by Lemma \ref{lm:ctxwf}}
\prf{$\models \Gamma, y{:}\ann\tau_1$ \hfill  by IH }
\prf{$\models \Gamma$ \hfill by def. of validity }
\prf{$\Gamma, y:\ann\tau_1 \models t : \tau_2$ \hfill by IH}
\prf{$\Gamma \models  (\tmfn y t) : (y:\ann\tau_1) \arrow \tau_2$ \hfill Lemma \ref{lem:ValidFun}}

\pcase{$\D = \ianc{\vdash \Gamma}{\Gamma \vdash u_1: u_2}{(u_1, u_2) \in \Ax}$}
\prf{$\Gamma \vdash u_1 \whnf u_1 : u_2$ \hfill by rules and typing assumption}
\prf{$u_1 \leq u_2$ \hfill by $(u_1, u_2) \in \Ax$}
\\
\pcase{$\D = \ibnc {\Gamma \vdash \ann\tau_1 : u_1}
       {\Gamma, y{:}\ann\tau_1 \vdash \tau_2 : u_2}
       {\Gamma \vdash (y:\ann\tau_1) \arrow \tau_2 : u_3}{(u_1,~u_2,~u_3) \in \Ru}$}
\prf{$\Gamma \models \ann\tau_1 : u_1$ \hfill by IH}
\prf{$\forall \Gamma'.~\Gamma' \Vdash \theta = \theta' : \Gamma \Longrightarrow \Gamma' \Vdash \{\theta\}\ann\tau_1 = \{\theta'\}\ann\tau_1 : u_1$ \hfill by sem. def.}
\prf{$\Gamma, y{:}\ann\tau_1 \models \tau_2 : u_2$ \hfill by IH}
\prf{$\forall ~\Gamma' \Vdash \theta = \theta' : \Gamma, y{:}\ann\tau_1 \Longrightarrow \Gamma' \Vdash  \{\theta\}\tau_2 = \{\theta'\}\tau_2 : u_1$ \hfill by sem. def.}
\prf{Assume $\Gamma' \Vdash \theta = \theta' : \Gamma$}
\prf{$\Gamma' \Vdash \{\theta\}\ann\tau_1 = \{\theta'\}\ann\tau_1 : u_1$ \hfill by previous lines}
\prf{Assume $\Gamma'' \leq_\rho \Gamma'$}
\prf{$\Gamma'' \Vdash \{\rho\}\{\theta\}\ann\tau_1 = \{\rho\}\{\theta'\}\ann\tau_1 : u_1$ \hfill sem. weakening for computations lemma \ref{lem:compsemweak}}
\prf{$\forall \Gamma''.~\Gamma'' \leq_\rho \Gamma'. \Gamma'' \Vdash \{\rho\}\{\theta\}\ann\tau_1 = \{\rho\}\{\theta'\}\ann\tau_1 : u_1$\hfill by previous lines}
\prf{Assume $\Gamma'' \leq_\rho \Gamma'$ and $\Gamma'' \Vdash t = t' : \{\{\rho\}\theta\}\ann\tau_1$}
\prf{$\Gamma'' \Vdash \{\rho\}\theta = \{\rho\}\theta' : \Gamma$ \hfill by sem. weakening lemma \ref{lem:weakcsub}}
\prf{$\Gamma '' \Vdash \{\rho\}\theta, t/y = \rho\{\theta'\}, t'/y : \Gamma, y{:}\ann\tau_1$ \hfill by sem. def.}
\prf{$\Gamma'' \Vdash \{\{\rho\}\theta, t/y\}\tau_2 = \{\rho\{\theta'\}, t'/y\}\tau_2 : u_2$ \hfill by $\Gamma, y{:}\ann\tau_1 \models \tau_2 : u_2$ \\
\mbox{$\quad$}\hfill choosing $\theta =  \{\rho\}\theta, t/y$ and $\theta' = \rho\{\theta'\}, t'/y$ }
\prf{$\Gamma'' \Vdash \{\rho, t/y\}\{\theta, y/y\}\tau_2 = \{\rho, t'/y\}\{\theta', y/y\}\tau_2 : u_2$\hfill by subst. def.}
\prf{$\forall \Gamma''.~\Gamma'' \leq_\rho \Gamma'. ~\Gamma'' \Vdash t = t' : \{\{\rho\}\theta\}\ann\tau_1 \Longrightarrow \Gamma'' \Vdash \{\rho, t/y\}\{\theta, y/y\}\tau_2 = \{\rho, t'/y\}\{\theta', y/y\}\tau_2 : u_2$ \hfill \\
\mbox{$\quad$}\hfill by previous lines}
\prf{$\Gamma' \Vdash \{\theta\}((y:\ann\tau_1) \arrow \tau_2 : u_3) = \{\theta'\}((y:\ann\tau_1) \arrow \tau_2 : u_2)$ \hfill by sem. def.}
\prf{$\Gamma \models ((y:\ann\tau_1) \arrow \tau_2 : u_3) = ((y:\ann\tau_1) \arrow \tau_2 : u_3)$ \hfill by sem. def.}
\prf{$\Gamma \models (y:\ann\tau_1) \arrow \tau_2 : u_3$ \hfill by def. of validity }
\\
\pcase{$\D = \ianc{\Gamma \vdash T}{\Gamma \vdash \cbox{T} : u}{}$}
\prf{$\models \Gamma$ \hfill }
\prf{Assume $\Gamma' \Vdash \theta = \theta' : \Gamma $}
\prf{$\Gamma \models_\LF T = T$ \hfill by IH}
\prf{$\Gamma' \Vdash_\LF \{\theta\} T = \{\theta'\}T$ \hfill since $\Gamma \models_\LF T = T$ }
\prf{$\Gamma' \vdash \cbox{\{\theta\}T}\whnf \cbox{\{\theta\}T} : u$ \hfill since $\norm \cbox{\{\theta\}T}$}
\prf{$\Gamma' \Vdash \{\theta\}\cbox{T} = \{\theta'\}\cbox{T} : u$ \hfill by sem. def.}
\prf{$\Gamma \models \cbox{T}  = \cbox{T} : \univ k$ \hfill since $\Gamma', \theta, \theta'$ are arbitrary}
\prf{$\Gamma \models \cbox{T} : \univ k$ \hfill by def. of validity}
}
\LONGVERSIONCHECKED{
\\[1em]
\pcase{$\D =
\ianc
{\Gamma \vdash t : \cbox{\Psi \vdash \tm} \quad
 \Gamma \vdash \IH : u \quad
 \Gamma \vdash b_v : \IH \quad
 \Gamma \vdash b_{\mathsf{app}} : \IH \quad
 \Gamma \vdash b_{\mathsf{lam}} : \IH}
{\Gamma \vdash \tmrec {\IH} {b_v} {b_{\mathsf{app}}} {b_{\clam}} \rappto \Psi~t : \{{\Psi}/\psi,~t/y\}\tau}{}
$\\[0.25em]
\mbox{where}~$\IH = (\psi : {\tmctx}) \arrow (y:\cbox{\psi \vdash \tm}) \arrow \tau$}
\\[-1em]
\prf{$\Gamma \models \IH : u $ \hfill by IH}
\prf{$\Gamma \models t : \cbox{\Psi \vdash \tm}$ \hfill by IH}
\prf{$ \Gamma, \psi:\tmctx, p:\cbox{~\psi \vdash_\# \tm}  \vdash  t_v : \{p / y\}\tau$ \hfill by typing inversion}
\prf{$\Gamma,  \psi:\tmctx, m:\cbox{~\psi \vdash \tm}, n:\cbox{~\psi \vdash \tm}
         f_m: \{m/y\}\tau, f_n: \{n/y\}\tau$}
\prf{\mbox{$\qquad$}\hfill$\qquad \vdash  t_{\mathsf{app}} : \{\cbox{~\psi \vdash\capp~\unbox{m}{\id}~\unbox{n}{\id}}/y\}\tau$$\qquad$ by typing inversion}
\prf{$\Gamma, \phi:{\tmctx},  m:\cbox{~\unboxc{\phi}, x:\tm \vdash \tm},
          f_m:\{\cbox{~\unboxc{\phi}, x:\tm}/\psi, m /y \} \tau$}
\prf{\mbox{$\qquad$}$\quad$\hfill $\vdash t_{\clam} : \{\phi/\psi, \cbox{~\unboxc{\phi} \vdash \clam~\lambda x.\unbox{m}{\id}~} / y\}\tau$ $\qquad$ by typing inversion}
\prf{$\Ca: ~~\vdash \Gamma, \psi:\tmctx, p:\cbox{~\unboxc{\psi} \vdash_\# \tm}$ and $\Ca < \D$ \hfill by Lemma \ref{lm:ctxwf}}
\prf{$\models \Gamma$ \hfill by def. of validity }
\prf{$ \Gamma, \psi:\tmctx, p:\cbox{~\unboxc{\psi} \vdash_\# \tm}  \models  t_v : \{p / y\}\tau$ \hfill by IH}
\prf{$\Gamma,  \psi:\tmctx, m:\cbox{~\unboxc{\psi} \vdash \tm}, n:\cbox{~\unboxc{\psi} \vdash \tm}
         f_m: \{m/y\}\tau, f_n: \{n/y\}\tau$}
\prf{\mbox{$\qquad$}\hfill$\qquad \models  t_{\mathsf{app}} : \{\cbox{~\unboxc{\psi} \vdash\capp~\unbox{m}{\id}~\unbox{n}{\id}}/y\}\tau\qquad$ by IH}
\prf{$\Gamma, \phi:\tmctx,  m:\cbox{~\unboxc{\phi}, x:\tm \vdash \tm},
          f_m:\{\cbox{~\unboxc{\phi}, x:\tm}/\psi, m /y \} \tau$}
\prf{\mbox{$\quad$}\hfill $\models t_{\clam} : \{\phi/\psi, \cbox{~\unboxc{\phi} \vdash \clam~\lambda x.\unbox{m}{\id}~} / y\}\tau\qquad$  by IH}
\prf{$\Gamma \models \tmrec {\IH} {b_v} {b_{\mathsf{app}}} {b_{\clam}} \rappto \Psi~t :   \IH$
 \hfill by Validity of Recursion Lemma \ref{lem:ValidRec}}
\\
\fbox{ If $\Gamma \vdash t \equiv t' : \tau$ then $\Gamma \models t =  t' : \tau$.}
\\[1em]
\pcase{
$\D = \ianc {\Gamma \vdash \tmfn y t : (y{:}\ann\tau_1) \arrow \tau_2
     \qquad \Gamma \vdash s : \ann\tau_1}
           {\Gamma  \vdash (\tmfn y t)~s \equiv \{s/y\}t : \{s/y\}\tau_2}{}$
                }
\prf{To Show: $\Gamma \models (\tmfn y t)~s = \{s/y\}t : \{s/y\}\tau_2
  $}
\prf{$\Ca: ~~\vdash \Gamma$ and $\Ca < \D$ \hfill by Lemma \ref{lm:ctxwf}}
\prf{$\models \Gamma$ \hfill by IH}
\prf{Assume $\Gamma' \Vdash \theta = \theta' : \Gamma$}
\prf{$\Gamma \models \tmfn y t : (y : \ann\tau_1) \arrow \tau_2$ \hfill by IH}
\prf{$\Gamma \models s : \ann\tau_1$ \hfill by IH}
\prf{$\Gamma' \Vdash \{\theta\}s = \{\theta'\}s : \{\theta\}\ann\tau_1$ \hfill
   by $\Gamma \models s : \ann\tau_1$}
\prf{$\Gamma \Vdash \{\theta\}(\tmfn y t) = \{\theta'\}(\tmfn y t) : \{\theta\}((y : \ann\tau_1) \arrow \tau_2)$ \hfill by previous lines (def. of $\models$)}
\prf{$\Gamma \vdash \{\theta\}(\tmfn y t) \whnf \{\theta\}(\tmfn y t) : \{\theta\}((y : \ann\tau_1) \arrow \tau_2)$ \hfill by sem. equ. def.}
\prf{$\Gamma \vdash \{\theta'\}(\tmfn y t) \whnf \{\theta'\}(\tmfn y t) : \{\theta\}((y : \ann\tau_1) \arrow \tau_2)$ \hfill by sem. equ. def.}
\prf{$\Gamma \Vdash \{\theta\}(\tmfn y t)~\{\theta\}s = \{\theta'\}(\tmfn y t)~\{\theta'\}s' : \{\theta, \{\theta\}s/y\}\tau_2$ \hfill by sem. equ. def.}
\prf{$\Gamma \vdash \{\theta\}((\tmfn y t)~s) \whnf w :  \{\theta, s/y\}\tau_2$ \hfill by sem. equ. def}
\prf{$\Gamma \vdash \{\theta, \{\theta\}s/y\}t \whnf w :  \{\theta, s/y\}\tau_2$ \hfill by inversion on $\whnf$}
\prf{$\Gamma \vdash \{\theta'\}((\tmfn y t)~s) \whnf w' :  \{\theta, s/y\}\tau_2$ \hfill by sem. equ. def}
\prf{$\Gamma \vdash \{\theta', \{\theta'\}s/y\}t \whnf w' :  \{\theta, s/y\}\tau_2$ \hfill by inversion on $\whnf$}
\prf{$\Gamma \Vdash (\tmfn y t)~s = \{s/y\}t : \{s/y\}\tau_2$ \hfill by Backwards Closure (Lemma \ref{lem:bclosed})}
\\[0.5em]
\pcase{
$\D =  \ianc{\Gamma \vdash t : \cbox{\Psi \vdash A}}
            {\Gamma \vdash \cbox{\hatctx \Psi \vdash \unbox{t}{\wk{\hatctx\Psi}}} \equiv t : \cbox{\Psi \vdash A}}{}$
}
\prf{To Show: $\Gamma \models \cbox{\hatctx \Psi \vdash \unbox{t}{\wk{\hatctx\Psi}}} = t : \cbox{\Psi \vdash A}$}
\prf{Assume $\Gamma' \Vdash \theta = \theta' : \Gamma$}
\prf{$\Gamma \models t : \cbox{\Psi \vdash A}$ \hfill by IH}
\prf{$\Gamma' \Vdash \{\theta\} t = \{\theta'\}t : \{\theta\} \cbox{\Psi \vdash A}$ \hfill by $\Gamma \models t : \cbox{\Psi \vdash A}$}
\prf{$\Gamma' \vdash \{\theta\} t \whnf w : \{\theta\} \cbox{\Psi \vdash A}$ \hfill by sem. def.}
\prf{$\Gamma' \vdash \{\theta'\} t \whnf w' : \{\theta\} \cbox{\Psi \vdash A}$ \hfill by sem. def.}
\prf{$\Gamma' ; \{\theta\}\Psi \Vdash \unbox w \id = \unbox{w'}\id : \{\theta\}A$ \hfill by sem. def.}
\prf{We note that $w$ is either $n$ where $\neut n$ or $\cbox{\hatctx\Psi \vdash M}$ }
\prf{\emph{Sub-case} $w$ is neutral, i.e. $\neut w$}
\prf{$\Gamma' ; \{\theta\}\Psi \vdash \unbox {\{\theta\}t} {\wk{\hatctx\Psi}} \lfwhnf \unbox w {\id} : \{\theta\}A$
  \hfill since $\Gamma' \vdash \{\theta\} t \whnf w : \{\theta\} \cbox{\Psi \vdash A}$ }
\prf{$\Gamma' ; \{\theta\}\Psi \vdash \unbox{\{\theta\}\cbox{\hatctx\Psi \vdash   \unbox{t}{\wk{\hatctx\Psi}} }}\id \lfwhnf \unbox w {\id} : \{\theta\}A$  \hfill by $\lfwhnf$ }
\prf{$\Gamma ; \{\theta\}\Psi \Vdash \unbox{\{\theta\}\cbox{\hatctx\Psi \vdash   \unbox{t}{\wk{\hatctx\Psi}} }}\id = \unbox{w'}\id : \{\theta\}A$
\hfill \\\mbox{\hspace{2cm}}\hfill using $\Gamma' ; \{\theta\}\Psi \Vdash \unbox w \id = \unbox{w'}\id : \{\theta\}A$ and Backwards Closure (Lemma \ref{lem:bclosed})}
\prf{$\Gamma ; \{\theta\}\Psi \vdash \{\theta\}\cbox{\hatctx \Psi \vdash \unbox{t}{\wk{\hatctx\Psi}}} \whnf \{\theta\}\cbox{\hatctx \Psi \vdash \unbox{t}{\wk{\hatctx\Psi}}} : \{\theta\}\cbox{\Psi \vdash A}$\hfill since $\norm \{\theta\}\cbox{\hatctx \Psi \vdash \unbox{t}{\wk{\hatctx\Psi}}}$}
\prf{$\Gamma' \Vdash \{\theta\}\cbox{\hatctx \Psi \vdash \unbox{t}{\wk{\hatctx\Psi}}} = \{\theta'\} t : \{\theta\}\cbox{\Psi \vdash A}$ \hfill by sem. def.}
\\
\prf{\emph{Sub-case} $w = \cbox{\hatctx \Psi \vdash M}$}
\prf{$\Gamma' ; \{\theta\}\Psi \vdash \unbox{\cbox{\hatctx\Psi \vdash M}}{\id} \lfwhnf N : \{\theta\}A$ \hfill
 by $\Gamma' ; \{\theta\}\Psi \Vdash \unbox w \id = \unbox{w'}\id : \{\theta\}A$ }
\prf{$\Gamma'; \{\theta\}\Psi \vdash  M \lfwhnf N : \{\theta\}A$ \hfill by $\lfwhnf$}
\prf{$\Gamma' ; \{\theta\}\Psi \vdash \unbox{\{\theta\}\cbox{\hatctx\Psi \vdash   \unbox{t}{\wk{\hatctx\Psi}} }}\id \lfwhnf N : \{\theta\}A $}
\prf{$\Gamma' ; \{\theta\}\Psi \Vdash N = \unbox{w'}\id : \{\theta\}A$ \hfill by sem. equ. } 
\prf{$\Gamma' ; \{\theta\}\Psi \Vdash \unbox{\{\theta\}\cbox{\hatctx\Psi \vdash   \unbox{t}{\wk{\hatctx\Psi}} }}\id = \unbox{w'}\id: \{\theta\}A$ \hfill by Backwards Closure (Lemma  \ref{lem:bclosed})}
\prf{$\Gamma' \Vdash \{\theta\}\cbox{\hatctx \Psi \vdash \unbox{t}{\wk{\hatctx\Psi}}} = \{\theta'\} t : \{\theta\}\cbox{\Psi \vdash A}$ \hfill by sem. def.}
}
\end{proof}

From the fundamental lemma follows our main theorem normalization and subject reduction.

\begin{theorem}[Normalization and Subject Reduction]\label{lm:norm}
If $\Gamma \vdash t : \tau$ then $t \whnf w$ and $\Gamma \vdash t \equiv w : \tau$
\end{theorem}
\begin{proof}
By the Fundamental theorem (Lemma \ref{lm:fundtheo}), we have $\Gamma \Vdash t = t : \tau$ (choosing the identity substitution for $\theta$ and $\theta'$).
This includes a definition $t \whnf w$. Since $w$ is in weak head normal form (i.e. $\norm w$), we have $w \whnf w$. Therefore, we can easily show that also $\Gamma \Vdash t = w : \tau$.  By well-formedness (Lemma \ref{lm:semwf}), we also have that $\Gamma \vdash t \equiv w : \tau$ and more specifically, $\Gamma \vdash w : \tau$.
\end{proof}

Using the fundamental lemma, we can also show that every term has a unique type. This requires first showing some standard inversion lemmas and then showing function type injectivity.
\LONGVERSION{
\begin{lemma}[Inversion]
  \begin{enumerate}
  \item If $\Gamma \vdash x : \ann\tau$ then $x:\ann\tau' \in \Gamma$ for some $\ann\tau'$ and $\Gamma \vdash \ann\tau \equiv \ann\tau' : u$.
\item If $\Gamma \vdash \tmfn y t : \tau$ then $\Gamma, y:\ann\tau_1 \vdash t : \tau_2$ for some $\ann\tau_1$, $\tau_2$
  and If $\Gamma \vdash \tau \equiv (y : \ann\tau_1) \arrow \tau_2 : u$.
\item If $\Gamma \vdash t~s : \tau$ then
     $\Gamma \vdash t : (y : \ann\tau_1) \arrow \tau_2$ and
     $\Gamma \vdash s : \ann\tau_1$ for some $\ann\tau_1$ and $\tau_2$ and
     $\Gamma \vdash \tau \equiv \{s/y\}\tau_2 : u$.
\item If $\Gamma \vdash \cbox C : \tau$ then
     $\Gamma \vdash \cbox C : \cbox T$ for some contextual type $T$ and
     $\Gamma \vdash \tau \equiv \cbox T : u$.
\item If $\Gamma \vdash \titer {(b_v \mid b_\capp \mid b_\clam)}{}\IH~\Psi~t : \tau$ where
      $\IH = (\psi : \tmctx) \arrow (y : \cbox{\Psi \vdash \tm}) \arrow \tau$
     then
   $\Gamma \vdash t : \cbox{\Psi \vdash \tm}$ and
   $\Gamma \vdash \IH : u$ and
   $\Gamma \vdash b_v : \IH$ and $\Gamma \vdash b_\capp : \IH$ and $\Gamma \vdash b_\clam : \IH$ and
   $\Gamma \vdash \tau \equiv \{\Psi/\psi, t/y\}\tau : u$.
\item If $\Gamma \vdash u_1 : \tau$ then there some $u_2$ s.t. $(u_1, u_2) \in \Ax$ and $\Gamma \vdash \tau \equiv u_2 : u$.
\item If $\Gamma \vdash (y : \ann\tau_1) \arrow \tau_2 : \tau$ then there is some $u_1$, $u_2$, and $u_3$ s.t.
  $(u_1, u_2, u_3) \in \Ru$ and $\Gamma \vdash \ann\tau_1 : u_1$, $\Gamma, y:\ann\tau_1 \vdash \tau_2 : u_2$ and
$\Gamma \vdash u_3 \equiv \tau : u$.
  \end{enumerate}
\end{lemma}
\begin{proof}
By induction on the typing derivation.
\end{proof}
}

\begin{lemma}[Injectivity of Function Type]\label{lm:funinj}
If $\Gamma \vdash (y : \ann\tau_1) \arrow \tau_2 \equiv (y : \ann\tau'_1) \arrow \tau'_2 :u$
then
$\Gamma \vdash \ann\tau_1 \equiv \ann\tau'_1 : u_1$ and
$\Gamma, y:\ann\tau_1 \vdash \tau_2 \equiv \tau_2 : u_2$ and $(u_1, u_2, u) \in \Ru$.
\end{lemma}
\begin{proof}
By the fundamental theorem (Lemma \ref{lm:fundtheo})
$\Gamma \Vdash (y : \ann\tau_1) \arrow \tau_2 \equiv (y : \ann\tau'_1) \arrow \tau'_2 :u$
(choosing the identity substitution for $\theta$ and $\theta'$).
By the sem. equality def., we have $\Gamma \Vdash \ann\tau_1 = \ann\tau'_1 : u_1$
and $\Gamma, y:\ann\tau_1 \Vdash \tau_2 = \tau'_2 : u_2$ and
$(u_1, u_2, u) \in \Ru$.
By well-formedness of semantic typing (Lemma \ref{lm:semwf}), we have
$\Gamma \vdash \ann\tau_1 \equiv \ann\tau'_1 : u_1$ and
and $\Gamma, y:\ann\tau_1 \vdash \tau_2 \equiv \tau'_2 : u_2$
\end{proof}

\begin{theorem}[Type Uniqueness]$\quad$
  \begin{enumerate}
  \item If $\Gamma ; \Psi \vdash M : A$ and $\Gamma ; \Psi \vdash M : B$ then $\Gamma \vdash A \equiv B : \lftype$.
  \item If $\Gamma \vdash t : \ann\tau$ and $\Gamma \vdash t : \ann\tau'$ then $\Gamma \vdash \ann\tau \equiv \ann\tau' : u$.
  \end{enumerate}
\end{theorem}
\begin{proof}
By mutual induction on the typing derivation exploiting typing inversion lemmas.
\end{proof}

Last but not least, the fundamental lemma allows us to show that not every type is inhabited and thus \cocon can be used as a logic. To establish this stronger notion of consistency, we first prove that we can discriminate type constructors.

\begin{lemma}[Type Constructor Discrimination]\label{lm:typdisc}
Neutral types, sorts, and function types are can be discriminated.
\LONGVERSION{\begin{enumerate}
\item If $\Gamma \vdash u_1 \equiv u_2 : u_3$  then $u_1 = u_2$ (they are the same).
\item $\Gamma \vdash x~\vec t \neq u : u'$.
\item $\Gamma \vdash x~\vec t  \neq (y : \ann\tau_1) \arrow \tau_2 $.
\item $\Gamma \vdash x~\vec t \neq \cbox T$.
\item $\Gamma \vdash u \neq (y : \ann\tau_1) \arrow \tau_2$.
\end{enumerate} }
\end{lemma}
\begin{proof}
Proof by contradiction. To show for example that $\Gamma \vdash x~\vec t  \neq (y : \ann\tau_1) \arrow \tau_2 $, we
assume $\Gamma \vdash x~\vec t \equiv (y : \ann\tau_1) \arrow \tau_2 : u$. By the fundamental lemma (Lemma \ref{lm:fundtheo}), we have
$\Gamma \Vdash  x~\vec t \equiv (y : \ann\tau_1) \arrow \tau_2 : u$ (choosing the identity substitution for $\theta$ and $\theta'$); but this is impossible given the semantic equality definition (Fig.~\ref{fig:sem}).
\end{proof}

\begin{theorem}[Consistency]\label{lm:consistency}
$x:u_0 \not\vdash t : x$.
\end{theorem}
\begin{proof}
Assume $x:u_0 \vdash t : x$. By subject reduction (Lemma \ref{lm:norm}), there is some $w$ s.t. $t \whnf w$ and $\Gamma \vdash t \equiv w : x$ and in particular, we must have $\Gamma \vdash w : x$. As $x$ is neutral, it cannot be equal $u$, $(y : \ann\tau_1) \arrow \tau_2$, or $\cbox T$ (Lemma \ref{lm:typdisc}). Thus $w$ can also not be a sort, function, or contextual object. Hence, $w$ can only be neutral, i.e. given the assumption $x:u_0$, the term $w$ must be $x$. This implies that $\Gamma \vdash x : x$ and implies $\Gamma \vdash x \equiv u_0 : u$ by inversion lemma for typing. But this is impossible by Lemma \ref{lm:typdisc}.
\end{proof}





\section{Related Work}
\paragraph{HOAS within dependent type theory} We propose a new
type theoretic foundation where LF is integrated within a Martin
L{\"o}f type theory. This is in some sense a radical step. A more
lightweight approach is to integrate at least some of the benefits of
LF within an existing type theory. This is for example accomplished by
weak HOAS approaches \cite{Despeyroux:TLCA95,Chlipala:ICFP08} where we
get $\alpha$-renaming for free but still have to deal with
capture-avoiding substitutions. The Hybrid library \cite{Felty12} in
Coq goes further supporting both $\alpha$-renaming and substitution
by encoding a specification logic within Coq. 
However it is unclear whether these approaches scale to dependently
typed encodings and can be integrated smoothly into practice. 

\paragraph{Metaprogramming}Dependent type theory is  flexible enough to serve as its own
metaprogramming language.
For this reason many dependently typed
systems \cite{Walt:IFL12,Ebner:ICFP17,Christiansen:PHD,Christiansen:IFL14} try to support
meta-programming in practice using quote operator to turn an
expression into its syntactical representation and unquote operator to
escape the quotation and refer to another computation whose value will
be plugged in at its place. However, a clean theoretical foundation is
missing.

\citet{Davies:ACM01} observed the similarity between modal types in S4
and quasi-quotation to support simply typed metaprogramming. However
their work concentrated on closed simply-typed code. In
\cocon, we can describe open pieces of code, i.e. code that
depends on a context of assumptions, and work within a Martin L{\"o}f
type theory. Hence, \cocon has the potential to provide a basis for dependently
typed metaprogramming. 



\section{Conclusion}
\cocon is a first step towards integrating LF methodology into
Martin-L{\"o}f style dependent type theories and and bridges the
longstanding gap between these two worlds. We have established
type uniqueness, normalization, and consistency. The next immediate step is
to derive an equivalence algorithm based on weak head reduction and
show its completeness. We expect that this will follow a similar
Kripke-style logical relation as the one we described. This would
allow us to justify that type checking \cocon programs is decidable.

It should be possible to implement \cocon as an extension to \beluga
-- from a syntactic point of view, it would be a small change. It also
 seems possible to extend existing implementation of Agda, however this
 might be more work, as in this case one needs to implement the LF
 infrastructure.



\begin{acks}                            
  This material is based upon work supported by the Humboldt
  Foundation
   and NSERC (Natural Science and Engineering Research Council, Canada)  No.~\grantnum{GS100000001}{mmmmmmm}.  Any opinions, findings, and
  conclusions or recommendations expressed in this material are those
  of the author and do not necessarily reflect the views of the
  funding agencies.
\end{acks}



\appendix
\section{Appendix}

\section*{Examples}

\paragraph{Cost-Semantics}
As a first example, consider the definition of a cost semantics for our small term language described earlier. As we aim to reason by structural induction on the evaluation judgement, we define the evaluation judgment $\evaln m v k$ which says that the term $m$ evaluates in at most steps $k$ to the value $v$ as an inductive type:

\begin{lstlisting}
inductive Eval : (m:[ |- tm])(n:[ |- tm])(k: nat) type =
| E_Lam : (m : [x:tm |- tm])(k:nat)
          Eval [ |- lam \x.u<m>u] [ |- lam \x.u<m>u] k
| E_App : (m : [ |- tm]) (n: [ |- tm])(k:nat)(l:nat)(j:nat)
          (m': [x:tm |- tm])(v:[ |- tm])(w:[ |- tm])
          Eval m [ |- lam \x.u<m'>u] k -> Eval n v l -> Eval [ |- u< m'>u with $\unbox{\texttt{v}}{}$/x ] w j
       -> Eval [ |- app u<m>u u<n>u] w (k + l + j + 1)
| E_Let : (m : [ |- tm])(n: [x:tm |- tm])(k:nat)(l:nat)(v : [ |- tm]) (w : [ |- tm])
          Eval m v k
       -> Eval [ |- u<n>u with u<v>u/x ] w l
       -> Eval [ |- letv u<m>u \x.u<n>u ] w (k + l + 1)
\end{lstlisting}

We define here an inductive type that relates closed \lstinline!tm!-objects \lstinline!m! and \lstinline!n! with the cost \lstinline!k!. We rely on the inductive type \lstinline!nat! in addition to functions such as addition on natural numbers. In a dependent type theory such as Coq or Agda, we would not be able to define an inductive type over (open) LF objects and exploit HOAS. In this example, we clearly state that we evaluate closed terms. To evaluate \lstinline![ |- app u<m>u u<n>u]! where \lstinline!m! and \lstinline!n! denote a closed \lstinline!tm!-objects, we evaluate \lstinline!m! and \lstinline!n! respectively. Note that when we refer to variables inside a box (or quoted) expression, we need to first unbox (or unquote) them. We write the unboxing here as \lstinline!u<m>u!.

In general, we are unboxing open terms, i.e. terms that may contain free variables. Hence, we are unboxing a term together with a substitution. For example, to evaluate the body \lstinline!m'! where we replace the LF variable \lstinline!x! with the closed term \lstinline!v!, we unbox \lstinline!m'! with the substitution \lstinline!u<v>u/x!. This is written as \lstinline! [ |- u<m'>u with u<v>u/x]!. In general, we may omit mentioning the substitution that is associated with every unbox-operation, if the substitution is simply the identity.

In \beluga, an inductive definition about open LF objects is possible, but we would not be able to compute \lstinline!k + l + 1!, as \beluga is an indexed language -- not a full dependently typed language. Therefore, we cannot refer to addition function in the index.

\paragraph{Compilation} As a simple example of compilation, we consider here a function \lstinline!trans! which eliminates let-expressions. As we also must traverse the body of lambda-abstractions and let-expression, this function takes a term in the context $\gamma$ as input and returns a term in the same context as output. As for example in \beluga, contexts are first-class in our language and we specify their shape using a context schema \lstinline!tm_ctx! which states that the context only contains \lstinline!tm! declarations.

\begin{lstlisting}
rec trans: (\gamma : [tm_ctx]) [\gamma |- tm] -> [\gamma |- tm] =
fun  \gamma   (p : [\gamma |-$_\#$ tm])        = p
   | \gamma  [\gamma |- app u<m>u u<n>u]       = [\gamma |- app u<trans \gamma m>u u<trans \gamma n>u]
   | \gamma  [\gamma |- lam \x.u<m>u]        = [\gamma |- lam \x. u<trans (\gamma,x:tm) m>u ]
   | \gamma  [\gamma |- letv u<m>u \y.u<n>u]   = [\gamma |- app (lam \x. u<trans (\gamma,x:tm) n>u u<trans \gamma m>u ]

\end{lstlisting}

We write the translation by pattern matching on the HOAS tree of type \lstinline![\gamma |- tm]!. Four different cases arise. First, we might encounter a variable from $\gamma$. As in \beluga, we use a pattern variable \lstinline!p! of type \lstinline![\gamma $\vdash_\#$ tm]! which can only be instantiated with variables from $\gamma$. Second, we translate \lstinline![\gamma |- app u<m>u u<n>u]! by simply recursively translating \lstinline!m! and \lstinline!n! and rebuilding our term. Third, we translate \lstinline![\gamma |- lam \x.u<m>u]! by translating \lstinline!m!. Note that \lstinline!m! has type \lstinline![\gamma, x:tm |- tm]! and hence \lstinline!trans m! returns a term in the context \lstinline!\gamma, x:tm!. Last, we translate \lstinline![\gamma |- letv u<m>u \y.u<n>u]! by translating each part and replacing the let-expression with the application of a lambda-abstraction.

The function \lstinline!trans! is close to what we can already write in \beluga with one major differences: we are able to inline recursive calls such as \lstinline!trans n! within a HOAS tree by supporting the boxing (quote) and the unboxing (unquote) of contextual objects; this is in contrast to \beluga where we are forced to write programs in a let-box style. Further, we only distinguish between LF variables that are bound by a $\lambda$-abstraction or in a LF context and computation-level variables. If computation-level variables have a contextual type, then we can use them to construct an LF object (or LF context) by unboxing them. In \beluga, we essentially distinguish between three different kinds of variables: LF variables, computation-level variables, and meta-variables (or contextual variables) that are of contextual type. Our treatment here unifies the latter two classes into one.

We now prove that the operational meaning of a term is preserved and it still evaluates in at most \lstinline!k! steps. In other words, our optimization did not add any additional costs. This is done by recursively analyzing and pattern matching on the derivation
\lstinline!Eval m v k!.  We write the type of each of the recursive calls as comments to illuminate what is happening in the background.


\begin{lstlisting}
rec ctrns : [ |- tm] -> [ |- tm] = fun m => trans [ ] m ;

rec val_preserve : (m:[ |- tm])(v:[ |- tm])(k : nat) Eval m v k -> Eval (ctrns m) (ctrns v) k =
fun [ |- lam \x.u<m>u] [ |- lam \x.u<m>u] k (E_Lam m k) = E_Lam (trans [x:tm] m) k
  | [ |- app u<m>u u<n>u ] w (k+l+j+1) (E_App m n k l j m' v w e1 e2 e3) =
    E_App (ctrns m) (ctrns n) k l j (trans [x:tm] m') (ctrns v) (ctrns w)
          (val_preserve m [ |- lam \x.u<m'>u] k e1)
                       % Eval (ctrns m) [ |- lam \x.u<trans [x:tm] m'>u] k
          (val_preserve n v l e2)                      % Eval (ctrns n) (ctrns v) l
          (subst [ |- tm] (fun e => Eval e (ctrns w) j)
                 (ctrns [ |- u<m'>u with u<v>u/x]) ([ |- u<trans [x:tm] m'>u with u<ctrns v>u/x])
                 (lemma_trans [] m' v) (val_preserve [ |- u<m'>u with u<v>u/x] w j e3))
                       % Eval (ctrns [ |- u<m'>u with u<v>u/x]) (ctrns w) j
  | [ |- letv u<m>u \y.u<n>u] w (k+l+1) (E_Let m n k l v w e1 e2) =
    E_App [ |- lam \x. u<trans [x:tm] n>u] (ctrns m) 0 k l (trans [x:tm] n) (ctrns v) (ctrns w)
          (E_Lam (trans [x:tm] n) 0)
                                  % Eval (ctrns [ |- lam \y.u<n>u]) (ctrns [ |- lam \y.u<n>u]) 0)
          (val_preserve m v k e1)                    % Eval (ctrns m) (ctrns v) k)
          (subst [ |- tm] (fun e => Eval e (ctrns w) l)
                 (ctrns [ |- u<n>u with u<v>u/x]) ([ |- u<trans [x:tm] n>u with u<ctrns v>u/x])
                 (lemma_trans [] n v) (val_preserve [ |- u<n>u with u<v>u/x] w l e2))
                                 % Eval (ctrns [ |- u<n>u with u<v>u/x]) (ctrns w) l)
\end{lstlisting}

The proof above relies on a lemma that states that it does not matter whether we
translate first a term \lstinline!m! of type \lstinline![\gamma, x:tm |- tm]! and
then replace \lstinline!x! with the translation of the term \lstinline!v! or we
translate directly the term \lstinline!m! where we already substituted for \lstinline!x!
the term \lstinline!v!. It is applied by using substitutivity of identity type whose type
is:

\noindent
\lstinline!subst : (A : type) (P : A -> type) (a b : A) a = b -> P a -> P b!

\begin{lstlisting}
rec lemma_trans : (\gamma : ctx) (m : [\gamma, x:tm |- tm]) (v : [\gamma |- tm])
         trans \gamma [\gamma |- u<m>u with $\wk{\gamma}$, u<v>u/x]
         = [\gamma |- u<trans (\gamma, x:tm) m>u with $\wk{\gamma}$, u<trans \gamma v>u/x] =
fun \gamma [\gamma, x |- x] = refl  % since trans \gamma v = [\gamma |- u<trans \gamma v >u]
  | \gamma [\gamma, x |- u<p>u with $\wk{\gamma}$] where p : [\gamma |-$_\#$ \tm] =
       refl  %trans \gamma p = [\gamma |- u<p>u with $\wk\gamma$]
  | \gamma [\gamma, x |- lam \y.u<m>u] =
    cong [\gamma,y:tm |- tm] [\gamma,y:tm |- tm]
      (fun e => [\gamma |- lam (\y. u<e>u)])
      (trans (\gamma,y:tm) [\gamma,y |- u<m>u with $\wk\gamma$,u<v>u/x,y/y])
      [\gamma,y |- u<trans (\gamma,x:tm,y:tm) m>u with $\wk\gamma$, u<trans \gamma v>u/x, y/y]
      (lemma_trans (\gamma, y:tm) [\gamma, y, x |- u<m>u with $\wk\gamma$,x/x, y/y])

  | \gamma [\gamma, x |- app u<m>u u<n>u] =
    cong2 [\gamma |- tm] [\gamma |- tm] [\gamma |- tm]
      (fun m n => [\gamma |- app u<m>u u<n>u])
      (trans \gamma [\gamma |- u<m with $\wk\gamma$,u<v>u/x>u])
      ([\gamma |- u<trans (\gamma,x:tm) m with $\wk\gamma$,u<trans (\gamma) v>u,x>u])
      (trans \gamma [\gamma |- u<n with $\wk\gamma$,u<v>u/x>u])
      ([\gamma |- u<trans (\gamma,x:tm) n with $\wk\gamma$,u<trans (\gamma) v>u,x>u])
      (lemma_trans \gamma m v) (lemma_trans (\gamma) n v)

  | \gamma [\gamma, x |- letv u<m>u \y.u<n>u] =
    cong2 [\gamma |- tm] [\gamma |- tm] [\gamma |- tm]
      (fun m n => [\gamma |- app (lam \y.u<n>u) u<m>u])
      (trans \gamma [\gamma |- u<m>u with $\wk\gamma$,u<v>u/x])
      [\gamma |- u<trans (\gamma,x:tm) m>u with $\wk\gamma$,u<trans \gamma v>u/x]
      (trans (\gamma,y:tm)[\gamma,y |- u<n>u with $\wk\gamma$,u<v>u/x, y/y])
      [\gamma,y |- u<trans (\gamma,x:tm,y:tm) n>u with $\wk\gamma$,u<trans \gamma v>u/x,y/y]
      (lemma_trans \gamma m v)
      (lemma_trans (\gamma,y:tm) [\gamma, y, x |- u<n>u with $\wk\gamma$,x/x,y/y])
\end{lstlisting}

In the first case, we have by definitional equality that indeed \lstinline!trans \gamma v = [\gamma |- u<trans \gamma v >u]!.
In the second case, we in addition exploit that composing
\lstinline!$\wk{\gamma}$, u<trans \gamma v>u/x! with $\wk\gamma$ simply is $\wk\gamma$ effectively dropping
\lstinline!u<trans \gamma v>u/x!.

The recursive cases are handled by means of congruence. We rebuild the term on each side by joining together the equalities
obtained for the subterms. The types for the functions on congruences are respectively:
\begin{lstlisting}
cong1 : (A B : type) (f : A -> B) (a1 a2 : A)-> a1 = a2 -> (f a1) = (f a2)
cong2 : (A B C : type) (f : A -> B -> C) (a1 a2 : A) (b1 b2 : B) -> a1 = a2 -> b1 = b2 -> (f a1 b1)=(f a2 b2)
\end{lstlisting}
In the \lstinline!lam! case and in the \lstinline!letv! case, we exploit that we can build a substitution that exchanges variables such that we can keep the
\lstinline!x!, the variable that we want to replace, to the right most position. This is a standard technique employed in HOAS systems.

Note that we would not be able to implement such functions in Beluga for several reasons:

\begin{itemize}
\item We directly refer to the function \lstinline!ctrns! in the type of \lstinline!val_preserve! to indicate that if \lstinline!m! evaluates to a value \lstinline!v! then the translation of \lstinline!m! (i.e. \lstinline!ctrns m!) also evaluates to the translation of the value of \lstinline!v! (i.e. \lstinline!ctrns v!). In Beluga, we would need to reify the function \lstinline!ctrns! as an inductive type and then pass it as an additional argument to \lstinline!val_preserve!.

\item The lemma \lstinline!lemma_trans! also directly references the function type of \lstinline!trans! in stating that the equality property.

\item We reason by equality, in particular we use the functions \lstinline!subst!, \lstinline!cong1!, and \lstinline!cong2! which all are polymorphic. Polymorphism is presently not supported in Beluga, or any system we are aware of that supports HOAS.

\end{itemize}




\section*{Recursors over dependently typed LF objects}
LF definitions are not inductive -- however, we can generate recursors
for each LF type following the procedure described in
\cite{Pientka:TLCA15}.

To illustrate concretely, how recursors look for a dependently typed
LF signature, we consider here another example: the representation of well-typed
terms. In this case, our LF signature contains the following type
families and constants.

\[
  \begin{array}{p{3cm}@{~}l@{~}r@{~}l}
LF Signature & \Sigma & \bnfas & \tp:\lftype, \mathsf{nat}:\tp,~
                                 \mathsf{arr}:\Pityp y \tp {\Pityp y \tp \tp} \\
& & & \tm:\Pityp a \tp \lftype,~
     ~ \mathsf{z}:\tm~\mathsf{nat},~
       \mathsf{suc}:\Pityp y {\tm~\mathsf{nat}} {\tm~\mathsf{nat}},~\\
& & & \clam: \Pityp a \tp {\Pityp b \tp {\Pityp y {(\Pityp x {\tm~a} {\tm~b})} {\tm~(\mathsf{arr}~a~b)}}},~\\
&&&    \mathsf{app}: \Pityp a \tp {\Pityp b \tp {\Pityp x {\tm~(\mathsf{arr}~a~b)} {\Pityp y {\tm~a} {\tm~b}}}}
  \end{array}
\]

For easier readability, we simply write how it looks when we declare these constants in \beluga or Twelf. Note that we write simply \lstinline!->! if $B$ does not depend on $x$ in $\Pi x{:}A.B$. We also omit abstracting over the implicit arguments -- this is common practice in logical frameworks, as the type for \lstinline!A! and \lstinline!B! can be inferred.

\begin{lstlisting}
tp : type
nat: tp.
arr: tp -> tp -> tp.

tm: tp -> type.
z  : tm nat.
suc: tm nat -> tm nat.
lam: (tm A -> tm B) -> tm (arr A B).
app: tm (arr A B) -> tm A -> tm B.
\end{lstlisting}

To build the recursor for the type family $\mathsf{tm}~a$ we proceed as follows:

\begin{itemize}
\item We generalize the recursor to $\titer{\R}{}{\IH} \rappto\Psi~s~t$. Here the intention is that $s$ has type $\cbox{ \vdash \mathsf{\tp}}$ and $t$ has type $\cbox{\Psi \vdash \tm~\unbox{s}{\wk{\hatctx\Psi}}}$. Hence \lstinline!t! depends not only on the LF context $\Psi$ but also on the type \lstinline!s!. As \lstinline!s! denotes a closed type, we weaken it to be used within the LF context $\Psi$. \newline
 In general, we have a vector $\vec s$ to describe all the implicit arguments $t$ depends on. Note that even in an LF signature that is simply typed, i.e. we have for example defined $\tm : \lftype$, the type of $t$ already depends on $\Psi$, since it has contextual type $\cbox{\Psi \vdash \tm}$. Moreover, we already are tracking this dependency on the LF context. Hence, the generalization to more dependent arguments is quite natural.

\item The recursor for iterating over contextual terms of type $\cbox{\psi \vdash \mathsf{tm}~a}$ will have 5 branches: 4 branches covering each constructor and one branch for the variable case.

\[
\trec {(b_v \mid b_{\mathsf{z}} \mid b_{\mathsf{suc}} \mid b_{\mathsf{app}}  \mid b_{\mathsf{lam}})}{}{\IH}
\]

where

\[
  \begin{array}{l@{~}r@{~}l}
b_v         & \bnfas & \psi, a, p \mto t_v \\
b_{\mathsf{z}} & \bnfas & \psi \mto t_{\mathsf{z}} \\
b_{\mathsf{suc}} & \bnfas & \psi, n, f_n \mto t_{\mathsf{suc}} \\
b_{\mathsf{lam}} & \bnfas & \psi, a, b, m, f_m \mto t_{\mathsf{lam}} \\
b_{\mathsf{app}} & \bnfas & \psi, a, b, n, m, f_n, f_m \mto t_{\mathsf{app}}
  \end{array}
\]

\item We give the typing rules for the recursor over  terms of type $\cbox{\psi \vdash \mathsf{tm}~a}$. Each branch gives rise to a specific typing rule and we label the $\vdash_{l}$ with the label $l$ where $l = \{v,~\mathsf{z},~\mathsf{suc},~\mathsf{lam},~\mathsf{app}
\}$ for clarity.

\[
\begin{array}{c}
\multicolumn{1}{l}{\mbox{Recursor over LF Terms}~\IH = (\psi : \tmctx)
  \arrow (a:\cbox{~\vdash \tp}) \arrow (y:\cbox{\psi \vdash \tm~a}) \arrow \tau }\\[0.5em]
\multicolumn{1}{l}{\mbox{where}~l = \{
v,~\mathsf{z},~\mathsf{suc},~\mathsf{lam},~\mathsf{app}
\} }\\[0.5em]
\infer
{\Gamma \vdash
\trec {(b_v \mid b_{\mathsf{z}} \mid b_{\mathsf{suc}} \mid b_{\mathsf{lam}}
   \mid b_{\mathsf{app}})}{}{\IH} \rappto \Psi~s~t : \{{\Psi}/\psi,~s/a,~t/y\}\tau}
{\Gamma \vdash s : \cbox{~\vdash \tp} \qquad
 \Gamma \vdash t :  \cbox{\Psi \vdash \tm~\unbox{s}{\wk{\Psi}}} & \Gamma \vdash \IH : u &
 \Gamma \vdash_l b_{l} : \IH }
\\[1em]
\multicolumn{1}{p{13.5cm}}{\mbox{Branches where}
}
\\[1em]
\infer{\Gamma \vdash_v ({\psi,a, p \mto t_v}) : \IH }
{ \Gamma, \psi:\tmctx, a:\cbox{~ \vdash \tp},~p:\cbox{~\unboxc{\psi} \vdash_\# \tm~\unbox{a}{\wk\psi}}  \vdash t_v : \{a / a,~p / y\}\tau}
\\[1em]
\infer{\Gamma \vdash_{\mathsf{z}} (\psi \mto t_{\mathsf{z}}) : \IH }
{\Gamma, \psi:\tmctx \vdash t_{\mathsf{z}} : \{\cbox{~\vdash \mathsf{nat}} / a,~\cbox{\psi\vdash \mathsf{z}}/y\}\tau}
\\[1em]
\infer{\Gamma \vdash_{\mathsf{suc}} (\psi, m \mto t_{\mathsf{suc}}) : \IH}{
 \begin{array}{l@{}lcl}
  \Gamma & \psi:\tmctx, m:\cbox{\psi \vdash \tm~\mathsf{nat}} & & \\
          & f_m : \{\cbox{~\vdash\mathsf{nat}} / a,~m/y\}\tau               & \vdash & t_{\mathsf{suc}} :
 \{\cbox{~\vdash \mathsf{nat}}/a,~\cbox{\psi \vdash \mathsf{suc}~\unbox{m}{\id}}/y\}\tau
  \end{array}
 }
\\[1em]
\infer{\Gamma \vdash_{\mathsf{app}} (\psi, a, b, m, n, f_n, f_m \mto t_{\mathsf{app}}) : \IH}
{
  \begin{array}{l@{}lcl}
\Gamma, &\psi:\tmctx, a:\cbox{~\vdash \tp},~b:\cbox{~\vdash \tp},& &    \\
& m:\cbox{~\unboxc{\psi} \vdash \tm~(\mathsf{arr}~\unbox{a}{\wk\psi}~\unbox{b}{\wk\psi})}, n:\cbox{\psi \vdash \tm~\unbox{a}{\id}}& & \\
&         f_m: \{\mathsf{arr}~\unbox{a}{\id}~\unbox{b}{\id}/a,m/y\}\tau, f_n: \{a/a,~n/y\}\tau
& \vdash &  t_{\mathsf{app}} : \{b/a,~\cbox{~\unboxc{\psi} \vdash \capp~\unbox{m}{\id}~\unbox{n}{\id}}/y\}\tau
  \end{array}
}
\\[1em]
\infer{\Gamma \vdash_{\clam} (\psi, a, b, m, f_m \mto t_{\clam}) : \IH}
{ \begin{array}{l@{}lcl}
\Gamma, & \phi:\tmctx,  a:\cbox{~ \vdash \tp}, b:\cbox{~\vdash \tp}& & \\
        & m:\cbox{\phi, x:\tm~\unbox{a}{\wk\psi} \vdash \tm~\unbox{b}{\wk\psi}}, & & \\
        &  f_m:\{(\phi, x:\tm~a)/\psi, b/a,~m /y \} \tau & \vdash & t_{\clam} : \{\phi/\psi,~\cbox{~\vdash \mathsf{arr}~\unbox{a}{\id}~\unbox{b}{\id}}/a,~\cbox{~\unboxc{\phi} \vdash \clam~\lambda x.\unbox{m}{\id}~} / y\}\tau
 \end{array}
 }
\end{array}
\]
\end{itemize}

\end{document}